\documentclass[english, 11pt, letter]{article}
\usepackage{amsmath,amssymb,amsthm}
\usepackage{mathtools}
\usepackage{natbib}
\usepackage[nodisplayskipstretch]{setspace}
\usepackage{bm}
\usepackage{threeparttable}
\usepackage{graphicx}
\usepackage{enumerate}
\usepackage{commath}			
\usepackage{booktabs}
\usepackage[titletoc,title]{appendix}
\usepackage{multirow}
\usepackage{arydshln}
\usepackage{tabularx}
\usepackage{dsfont}
\usepackage{rotating} 
\usepackage[font={small}]{subcaption}
\usepackage[left=1in, right=1in, top=1in, bottom=1in]{geometry}

\allowdisplaybreaks

\usepackage{dcolumn}
\newcolumntype{d}[1]{D{.}{.}{#1}}

\usepackage{url}
\usepackage[
	breaklinks,
	colorlinks=true,
	pagebackref=false,
	pdfauthor={},%
	pdftitle={},%
	citecolor=blue,
	linkcolor=blue,
	urlcolor=blue
	]{hyperref}      

\newtheorem{lem}{Lemma}
\newtheorem{assumption}{Assumption}

\newtheorem{prop}{Proposition}
\newtheorem{cor}{Corollary}

\newtheorem{thm}{Theorem}

\newtheoremstyle{remark2}{1ex}{1ex}%
      {}
      {}
      {\bf}
      {.}
      {5pt}
      {\thmname{#1}\thmnumber{ #2}\thmnote{ \slshape{(#3)}}} 
\theoremstyle{remark2}
\newtheorem{rem}{Remark}

\newtheoremstyle{remark3}{1ex}{1ex}%
      {\it}
      {}
      {\bf}
      {.}
      {5pt}
      {\thmname{#1}\thmnumber{.#2}\thmnote{ \slshape{(#3)}}} 
\theoremstyle{remark3}

\newcommand{\1}{\mathds{1}}

\usepackage{babel}
\input{ee.sty}

\@ifundefined{bibfont}{  } 
    {  }

\makeatletter
\renewenvironment{proof}[1][\bfseries\proofname]{\par
   \pushQED{\qed}%
   \normalfont \topsep6\p@\@plus6\p@\relax
   \trivlist
   \item[\hskip\labelsep
     #1\@addpunct{:}]\ignorespaces
}{%
   \popQED\endtrivlist\@endpefalse
}
\makeatother

\newcommand{\Comments}{1}
\newcommand{\mynote}[2]{\ifnum\Comments=1\textcolor{#1}{#2}\fi}
\newcommand{\mytodo}[2]{\ifnum\Comments=1%
  \todo[linecolor=#1!80!black,backgroundcolor=#1,bordercolor=#1!80!black]{#2}\fi}

\ifnum\Comments=1               
\fi

\ifnum\Comments=1               
\fi

\renewcommand\appendixpagename{Appendix}
\newcommand{\CoVaR}{\operatorname{CoVaR}}
\newcommand{\VIX}{\operatorname{VIX}}

\newcommand{\divp}{\operatorname{dp}}

\newcommand{\D}{\,\mathrm{d}}
\renewcommand{\E}{\mathbb{E}}
\renewcommand{\P}{\mathbb{P}}

\begin{document}

\baselineskip18pt
\renewcommand\floatpagefraction{.9}
\renewcommand\topfraction{.9}
\renewcommand\bottomfraction{.9}
\renewcommand\textfraction{.1}
\setcounter{totalnumber}{50}
\setcounter{topnumber}{50}
\setcounter{bottomnumber}{50}
\abovedisplayskip1.5ex plus1ex minus1ex
\belowdisplayskip1.5ex plus1ex minus1ex
\abovedisplayshortskip1.5ex plus1ex minus1ex
\belowdisplayshortskip1.5ex plus1ex minus1ex

\title{Persistence-Robust Break Detection in Predictive CoVaR Regressions\thanks{The author would like to thank Matei Demetrescu for his insightful comments that significantly improved the quality of the paper. Support of the Deutsche Forschungsgemeinschaft (DFG, German Research Foundation) through grants 460479886 and 531866675 is gratefully acknowledged.}
}

\author{
	Yannick Hoga\thanks{Faculty of Economics and Business Administration, University of Duisburg-Essen, Universit\"atsstra\ss e 12, D--45117 Essen, Germany, \href{mailto:yannick.hoga@vwl.uni-due.de}{yannick.hoga@vwl.uni-due.de}.}
}

\date{\today}
\maketitle

\begin{abstract}
	\noindent 
	Forecasting systemic risk (as measured by \citeauthor{AB16}'s \citeyearpar{AB16} CoVaR) is important in economics and finance.
	However, predictive relationships may be unstable over time.
	Therefore, this paper develops structural break tests in predictive CoVaR regressions.
	These tests can detect changes in the forecasting power of covariates, and are based on the principle of self-normalization. 
	We show that our tests are valid irrespective of whether the predictors are stationary or near-stationary, rendering the inference procedures suitable for a wide range of practical applications.
	Also, we derive the power of our tests under local alternatives, which shows that power is higher when predictors are near-stationary instead of stationary.
	As a final theoretical contribution, we propose unsupervised change point tests that are consistent even in the presence of arbitrarily many breaks in the predictive relationship.
	Simulations illustrate the good finite-sample properties of our methods.
	An empirical application to systemic risk forecasting models for the US banking system shows the usefulness of our tests by uncovering changes in the predictive content of the VIX.\\
	
\end{abstract}


\section{Motivation}


Recent financial crises and their widespread impacts have heightened awareness of the \textit{systemic} nature of risk \citep{AB16,Aea17,VZ19}. 
While markets often experience higher levels of volatility during such turbulent periods---indicating greater overall risk---this does not automatically translate to increased systemic risk. 
For example, banks may exhibit more volatile returns, signaling higher individual risk, yet without a corresponding rise in their co-movements (i.e., no increase in systemic risk). 
However, it is this rise in commonality, that has become a key concern for regulators, as it suggests the potential for widespread financial distress and significant economic consequences \citep{GKP16}.

Therefore, researchers have investigated the predictive content of numerous variables for future levels of systemic risk.
The perhaps most important systemic risk measure in this context is the conditional Value-at-Risk (CoVaR) of \citet{AB16}.
Many researchers have related this measure to lagged covariates via a so-called \textit{CoVaR regression}.
Among the studied predictors are bank-level variables, such as (the log of) total assets \citep{BDP20,BRS20}, and financial variables, such as equity volatility \citep{AB16,BDP20,BRS20} or the VIX \citep{Hea16}. 
But also macroeconomic variables have been investigated, such as inflation \citep{BRS20}, the TED spread as a measure of liquidity risk \citep{AB16,BDP20}, and 10-year government bond rates \citep{BRS20}.

To this date, there are no methods that allow to test for the stability of these predictive relations.
Such change point tests for predictive CoVaR regressions are, however, important for at least two reasons. 
First, unstable models should be used with caution (if at all) for forecasting purposes \citep{PPP13}. 
Second, knowledge of structural breaks may improve modeling attempts and generate further research into the underlying causes of the instabilities.
For instance, \citet{ST21} carry out this program for predictive mean regressions.
Specifically, they show that suitably accounting for breaks improves the forecasting performance of several predictors of stock returns.
It is the main purpose of this paper to fill this gap in the literature by proposing suitable change point tests for predictive CoVaR regressions.

Indeed, the theoretical literature on CoVaR regressions seems underdeveloped.
Currently, there exist only theoretical results for standard inference on (constant) parameters in CoVaR regressions with stationary regressors \citep{Lea24,DH26}.
In particular, these papers do not consider the problem of structural break testing.
We stress that, while the theoretical literature on CoVaR regressions is limited, they are frequently applied in practice \citep{AB16,BRS20,BDP20}.

One key feature of our change point tests is their persistence-robustness in the sense that they allow predictors to be stationary or near-stationary.
This is useful because there are many applications of CoVaR regressions, where covariates exhibit different degrees of serial dependence.
For instance, among the predictors mentioned above, the short-term TED spread (i.e., the difference between the 3-month LIBOR rate and the 3-month secondary market Treasury bill rate) and equity volatility, are known to be reasonably persistent. 
Likewise, inflation as a predictor of systemic risk is strongly serially dependent.
Moreover, 10-year government bond rates, that account for the connection between sovereigns and banks, are a classic example of a highly persistent time series.
In light of this variety of predictors displaying different degrees of serial dependence, change point tests possessing some persistence-robustness---such as ours---seem to be called for.

We mention that structural break tests in regression contexts typically require knowledge of the persistence of the predictors.
For instance, in their change point tests for predictive mean regressions, \citet{Zea23a} show that asymptotic distributions differ between stationary and non-stationary covariates, such that the correct choice of critical values depends on whether the predictors are stationary or non-stationary.
Conveniently, this is not the case for our test.

Although our main emphasis in this paper is on developing change point tests for CoVaR regressions, as a corollary we also obtain change point tests for simple predictive quantile regressions (QRs).
The reason for this is that CoVaR regressions can only be estimated jointly with quantile regressions (\citealp{FH24}).
Of course, testing for changes in predictive QRs is equally important as doing so in CoVaR regressions, due to the significant interest that the former have garnered in the past (\citealp{Lee16,FL19,cai2022new,FLS23,Lea24a,MSK24,MWT25}).
The works most closely related to our persistence-robust change point tests for predictive QRs are those of \citet{Qu08} and \citet{OQ11}, who develop structural break tests for quantile regressions with exclusively stationary covariates.
Our test differs from theirs in the following respects.
First, our structural change test is robust to whether predictors are stationary or near-stationary.
What is more, under the alternative, we show that local power is higher when the predictors are highly persistent.
Second, our tests are based on the principle of self-normalization (SN), pioneered in a change point context by \citet{SZ10}.
Third, we also derive ``unsupervised'' tests in the spirit of \citet{ZL18}, which do not require \textit{a priori} knowledge of the number of change points.

All our persistence-robust change point tests work by comparing estimates of regression coefficients based on suitably chosen subsamples.
We show that these subsample estimates converge (in a functional sense) to Brownian motion both when predictors are stationary and near-stationary.
Because these Brownian motions only differ in their covariance matrices, our self-normalized structural break tests following \citet{SZ10} are valid for stationary and near-stationary predictors in the sense that no \textit{a priori} knowledge of their persistence is required.
The technical development relies on convexity results of \citet{Kat09} and recent results for stationary and near-stationary variables in \citet{MP20}.
Note that we only draw on \citet{MP20} to utilize certain properties of stationary and near-stationary predictors.
We use these results for a distinctly different end, viz.~\textit{(quantile, CoVaR)} regressions---much unlike \citet{MP20}, who consider (cointegrating) \textit{mean} regressions.
Moreover, we go beyond \citet{MP20} by developing \textit{functional} central limit theory for the regression estimators.

Apart from delivering persistence-robustness, SN also offers further advantages.
First, it is straightforward to implement, as it only requires recursively estimated coefficients as inputs.
Second, our functional results underlying the SN-based change point tests also facilitate ``off-the-shelf'' inference for the full-sample estimates when no break was detected; see Section~\ref{Exposure CoVaR} for an application.
Third, compared with approaches that directly estimate asymptotic variances, SN-based tests often offer more adequate size control in finite samples \citep{Sha10,Sha15}.
Fourth, ``standard'' change point tests that involve consistent estimation of asymptotic (long-run) variances often suffer from the problem of non-monotonic power, where power may even \textit{decrease} in the distance from the null \citep{Vog99}.
By a clever construction of the normalizer, SN sidesteps this problem and, hence, does not exhibit this non-monotonicity \citep{SZ10}.

Simulations show that our tests possess excellent finite-sample properties. 
For both stationary and near-stationary predictors, size is close to the nominal level for sample sizes typically encountered in practice.
This not only holds for the quantile regression but also for the CoVaR regression, where the effective sample size is severely reduced by the very definition of CoVaR as a systemic risk measure.
We also show that our tests have high power in detecting deviations from the null of a stable predictive relationship.


Our second main contribution is to apply our test to a systemic risk prediction model for the US financial sector.
As a predictor, we use the volatility index (VIX), which measures the short-term future volatility in the S\&P~500 expected by market participants \citep{Wha09}.
The VIX and other measures of market volatility are widely used as predictors of systemic risk \citep{AB16,Hea16,Ned20}.
We measure systemic risk for all global systemically important US banks with respect to the S\&P~500~Financials.
While throughout our sample a high VIX signals higher future systemic risk, our test shows that the precise predictive power of the VIX is liable to change over time.

The remainder of this paper is structured as follows.
We introduce basic notation and definitions in Section~\ref{Preliminaries}.
Then, Section~\ref{Main Results} presents our change point tests for (quantile, CoVaR) regressions.
Section~\ref{sec:Simulation} summarizes the results of extensive Monte Carlo simulations, and Section~\ref{The VIX as a Predictor for Systemic Risk} contains our systemic risk application.
The final Section~\ref{sec:Conclusion} concludes.
To promote flow in the main paper, we relegate the introduction of some assumptions and the asymptotic variance-covariance matrices of the parameter estimators to Appendices~\ref{Assumptions on the Quantile Regression}--\ref{Asymptotic Variance-Covariance Matrices}. 
The appendix further contains detailed proofs and an additional application in Sections~\ref{sec:thm1}--\ref{Quantile Predictability of the Equity Premium}.

\section{Preliminaries}\label{Preliminaries}

\subsection{Defining Quantiles and the CoVaR}\label{Defining Quantiles and the CoVaR}

Let $Z_t$ denote the variable of interest and let $Y_t$ stand for the ``distress'' variable.
The systemic importance of $Y_t$ for $Z_t$ will then be measured by the CoVaR.
In the empirical application in Section~\ref{The VIX as a Predictor for Systemic Risk}, $Z_t$ will denote the log-losses of a large US bank and $Y_t$ those of the S\&P~500 Financials.
The $(k+1)\times 1$-vector of predictors is $\mX_{t-1}=(1,\vx_{t-1}^\prime)^\prime$, where $\vx_{t-1}$ contains the $\mathbb{R}^{k}$-valued, non-constant predictors. 
(We use bold font to denote vector- or matrix-valued quantities and normal font to denote scalar-valued objects.)
As mentioned in the Motivation, several possibly persistent variables have been used in CoVaR regressions as predictors.

To formally introduce quantiles and the CoVaR, let $\mathcal{F}_{t}=\sigma(Z_t, Y_t, \mX_t,Z_{t-1}, Y_{t-1}, \mX_{t-1},\ldots)$ denote the time-$t$ information set, and let $\alpha,\beta\in(0,1)$ be two probability levels. 
The quantity to be modeled in the auxiliary QR is the conditional $\alpha$-quantile $Q_{\alpha}(Y_t\mid\mathcal{F}_{t-1}):=Q_{\alpha}(F_{Y_t\mid\mathcal{F}_{t-1}}):=F_{Y_t\mid\mathcal{F}_{t-1}}^{\leftarrow}(\alpha)$ of the cumulative distribution function $F_{Y_t\mid\mathcal{F}_{t-1}}(\cdot):=\P\{Y_t\leq\cdot\mid\mathcal{F}_{t-1}\}$.

The CoVaR is the $\beta$-quantile of the distribution of $Z_t\mid \big\{Y_t\geq Q_{\alpha}(Y_t\mid\mathcal{F}_{t-1}),\mathcal{F}_{t-1}\big\}$, i.e.,
\[	
	\CoVaR_{\beta\mid\alpha}\big((Z_t,Y_t)^\prime\mid\mathcal{F}_{t-1}\big)=Q_{\beta}\big(F_{Z_t\mid Y_t\geq Q_{\alpha}(Y_t\mid\mathcal{F}_{t-1}),\, \mathcal{F}_{t-1}}\big).
\]
Thus, the CoVaR is only distinct from a standard quantile by the conditioning event $\{Y_t\geq Q_{\alpha}(Y_t\mid\mathcal{F}_{t-1})\}$, which we interpret as a ``distress event'' for large $\alpha$.
Therefore, as $\alpha\uparrow1$, the CoVaR measures the risk in $Z_t$ (via its $\beta$-quantile) conditional on $Y_t$ being in distress, such that the \textit{systemic} risk interpretation of the CoVaR becomes apparent.
For $\alpha=0$, the distress event occurs with probability one, and the CoVaR degenerates to a quantile, i.e., $\CoVaR_{\beta\mid0}\big((Z_t,Y_t)^\prime\mid\mathcal{F}_{t-1}\big)=Q_{\beta}(Z_t\mid\mathcal{F}_{t-1})$.

\subsection{Predictive (Quantile, CoVaR) Regressions}
\label{sec:Predictive QR}

We model quantiles and the CoVaR via the linear (quantile, CoVaR) regression (henceforth simply called a \textit{CoVaR regression})
\begin{align}
	Y_t&=\mX_{t-1}^\prime\valpha_0 + \epsilon_t,\qquad Q_{\alpha}(\epsilon_t\mid\mathcal{F}_{t-1})=0,\label{eq:quantile model}\\
	Z_t&=\mX_{t-1}^\prime\vbeta_0 + \delta_t,\qquad \CoVaR_{\beta\mid\alpha}\big((\delta_t,\epsilon_t)^\prime\mid\mathcal{F}_{t-1}\big)=0.\label{eq:CoVaR model}
\end{align}
The model in \eqref{eq:quantile model} is a standard linear predictive quantile regression.
The assumption on the QR error $\epsilon_t$ in \eqref{eq:quantile model} ensures that $Q_{\alpha}(Y_t\mid\mathcal{F}_{t-1})=\mX_{t-1}^\prime\valpha_0$, whence the coefficient $\valpha_0$ measures the strength of the predictive content of $\mX_{t-1}$ for the (conditional) $\alpha$-quantile of $Y_t$.
Analogously, the assumption on the (quantile, CoVaR) errors $(\epsilon_t,\delta_t)^\prime$ in \eqref{eq:CoVaR model} implies that $\CoVaR_{\beta\mid\alpha}\big((Z_t,Y_t)^\prime\mid\mathcal{F}_{t-1}\big)=\mX_{t-1}^\prime\vbeta_0$, such that $\vbeta_0$ quantifies the forecasting power of $\mX_{t-1}$ for future levels of systemic risk.
In Appendix~\ref{DGP} of the simulations, we present a concrete data-generating process (DGP) for $(Y_t,Z_t)^\prime$ that gives rise to linear (quantile, CoVaR) regressions as in \eqref{eq:quantile model}--\eqref{eq:CoVaR model}.

To estimate the predictive CoVaR regression in \eqref{eq:quantile model}--\eqref{eq:CoVaR model}, we adapt the estimator of \citet{DH26} to our context.
Given a sample $\big\{(Z_t,Y_t,\mX_{t-1}^\prime)^\prime\big\}_{t=1,\ldots,n}$, we estimate the parameter vector $\valpha_0$ via
\begin{equation*}
	\widehat{\valpha}_n = \argmin_{\valpha\in\mathbb{R}^{k+1}} \sum_{t=1}^{n}\rho_{\alpha}(Y_t - \mX_{t-1}^\prime\valpha),
\end{equation*}
where $\rho_{\alpha}(u)=u(\alpha - \1_{\{u\leq0\}})$ is the standard pinball loss known from quantile regressions \citep{KB78}.
To estimate $\vbeta_0$, we use
\begin{equation*}
	\widehat{\vbeta}_n = \argmin_{\vbeta\in\mathbb{R}^{k+1}} \sum_{t=1}^{n}\1_{\{Y_t>\mX_{t-1}^\prime\widehat{\valpha}_n\}}\rho_{\beta}(Z_t - \mX_{t-1}^\prime\vbeta).
\end{equation*}
This estimator is similar in spirit to the QR estimator $\widehat{\valpha}_n$, except that in minimizing the loss, only those observations $(Z_t,Y_t,\mX_{t-1}^\prime)^\prime$ are used for which $Y_t$ is larger than its predicted conditional quantile, i.e., $Y_t>\mX_{t-1}^\prime\widehat{\valpha}_n$.
We mention that this restriction to observations with a quantile exceedance is unavoidable by definition of the CoVaR as a quantile of the distribution $F_{Z_t\mid Y_t\geq Q_{\alpha}(Y_t\mid\mathcal{F}_{t-1}),\mathcal{F}_{t-1}}$ that conditions on $Y_t\geq Q_{\alpha}(Y_t\mid\mathcal{F}_{t-1})$.
For a more technical argument, we refer to Theorem~4.2~(ii) of \citet{FH24}.
Underlying the estimators $\widehat{\valpha}_n$ and $\widehat{\vbeta}_n$ is the assumption that the predictive relationships in \eqref{eq:quantile model}--\eqref{eq:CoVaR model} are unchanged throughout the sample.

\subsection{CoVaR Regressions With Instabilities}
\label{sec:CoVaR Regressions With Instabilities}

Since the predictive content of $\mX_{t-1}$ may change over time, the coefficients $\valpha_0$ and $\vbeta_0$ in \eqref{eq:quantile model}--\eqref{eq:CoVaR model} may vary with $t$, such that
\begin{align}
		Y_t&=\mX_{t-1}^\prime\valpha_{0,t} + \epsilon_t,\qquad Q_{\alpha}(\epsilon_t\mid\mathcal{F}_{t-1}),\label{eq:(QRalt)}\\
		Z_t&=\mX_{t-1}^\prime\vbeta_{0,t} + \delta_t,\qquad \CoVaR_{\beta\mid\alpha}\big((\delta_t,\epsilon_t)^\prime\mid\mathcal{F}_{t-1}\big)=0.\label{eq:(CoVaRalt)}
\end{align}
Our goal is to construct a test of the no-break hypothesis for model \eqref{eq:(QRalt)}--\eqref{eq:(CoVaRalt)}, i.e.,
\[
	\mathcal{H}_0^{\CoVaR}\colon \begin{pmatrix}\valpha_{0,1}\\ \vbeta_{0,1}\end{pmatrix}=\ldots=\begin{pmatrix}\valpha_{0,n}\\ \vbeta_{0,n}\end{pmatrix}\equiv\begin{pmatrix}\valpha_{0}\\ \vbeta_{0}\end{pmatrix}.
\]
To do so, we follow \citet{SZ10} and rely on estimates based on certain subsamples $\big\{(Z_t, Y_t,\mX_{t-1}^\prime)^\prime\big\}_{t=\lfloor nr\rfloor + 1,\ldots,\lfloor ns\rfloor}$, where $0\leq r<s\leq1$ and $\lfloor \cdot\rfloor$ denotes the floor function. 
Specifically, the subsample estimates for the CoVaR regression are
\begin{align*}
	\widehat{\valpha}_n(r,s) &= \argmin_{\valpha\in\mathbb{R}^{k+1}} \sum_{t=\lfloor nr\rfloor+1}^{\lfloor ns\rfloor}\rho_{\alpha}(Y_t - \mX_{t-1}^\prime\valpha),\\
	\widehat{\vbeta}_n(r, s) &= \argmin_{\vbeta\in\mathbb{R}^{k+1}} \sum_{t=\lfloor nr\rfloor+1}^{\lfloor ns\rfloor}\1_{\{Y_t>\mX_{t-1}^\prime\widehat{\valpha}_n(r,s)\}}\rho_{\beta}(Z_t - \mX_{t-1}^\prime\vbeta).
\end{align*}

Define the stacked estimator $\widehat{\vgamma}_n(r,s):=\big(\widehat{\valpha}_n^\prime(r,s),\widehat{\vbeta}_n^\prime(r,s)\big)^\prime$.
Then, our self-normalized test statistic for $\mathcal{H}_0^{\CoVaR}$, inspired by \citet{SZ10}, is given by
\[
\mathcal{U}_{n,\vgamma}:=\sup_{s\in[\iota,1-\iota]} s^2(1-s)^2\big[\widehat{\vgamma}_n(0,s) - \widehat{\vgamma}_n(s,1)\big]^\prime \bm{\mathcal{N}}_{n,\vgamma}^{-1}(s)\big[\widehat{\vgamma}_n(0,s) - \widehat{\vgamma}_n(s,1)\big]
\]
with cut-off $\iota\in(0,1/2)$ and normalizer
\begin{multline*}
\bm{\mathcal{N}}_{n,\vgamma}(s) =\int_{\iota}^{s}r^2\big[\widehat{\vgamma}_n(0,r) - \widehat{\vgamma}_n(0,s)\big]\big[\widehat{\vgamma}_n(0,r) - \widehat{\vgamma}_n(0,s)\big]^\prime\D r \\
	+ \int_{s}^{1-\iota}(1-r)^2\big[\widehat{\vgamma}_n(r,1) - \widehat{\vgamma}_n(s,1)\big]\big[\widehat{\vgamma}_n(r,1) - \widehat{\vgamma}_n(s,1)\big]^\prime\D r.
\end{multline*}

The cut-off $\iota$ ensures that estimates are based on at least a $\iota$-fraction of the available sample, that is, on $\lfloor n\iota\rfloor$ many observations.
For technical reasons, the choice $\iota=0$ is not allowed, which is similar as in the fixed-regressor QRs considered in \citet[Theorem~2]{ZS13}.

The test statistic $\mathcal{U}_{n,\vgamma}$ has two different components. 
First, the outer terms involving $\big[\widehat{\vgamma}_n(0,s) - \widehat{\vgamma}_n(s,1)\big]$ indicate a structural instability in \eqref{eq:(QRalt)}--\eqref{eq:(CoVaRalt)} if they are far from zero and, thus, give the test its power.
When $s$ is close to 0 or 1, this difference is weighted down by the factor $s^2(1-s)^2$, because then either $\widehat{\vgamma}_n(0,s)$ or $\widehat{\vgamma}_n(s,1)$ is based on very few observations, such that larger differences between the two estimates are not uncommon.
The second component of $\mathcal{U}_{n,\vgamma}$ is the normalizer $\bm{\mathcal{N}}_{n,\vgamma}(s)$, which serves to give a nuisance parameter-free limiting distribution (see the proof of Corollary~\ref{cor:SBT CoVaR}).

To derive the asymptotic limit of our test statistic $\mathcal{U}_{n,\vgamma}$, we require some assumptions on the predictors $\vx_{t-1}$.
We introduce these in the following Section~\ref{Assumptions on the Predictors} (see Assumptions~\ref{ass:N}--\ref{ass:LP}).
Additional required assumptions on the predictive QR model in \eqref{eq:(QRalt)} (see Assumptions~\ref{ass:innov}--\ref{ass:K} in Appendix~\ref{Assumptions on the Quantile Regression}), and the predictive CoVaR regression in \eqref{eq:(CoVaRalt)} (see Assumptions~\ref{ass:innov CoVaR}--\ref{ass:K ast} in Appendix~\ref{Assumptions on the CoVaR Regression}) are relegated to the appendix to promote flow.
Importantly, Assumptions~\ref{ass:innov}--\ref{ass:K ast} in Appendices~\ref{Assumptions on the Quantile Regression}--\ref{Assumptions on the CoVaR Regression} are sufficiently general to allow for conditional heteroskedasticity of the regression errors $(\epsilon_t,\delta_t)^\prime$.

\section{Main Results}\label{Main Results}

\subsection{Assumptions on the Predictors}\label{Assumptions on the Predictors}

In line with most of the predictive regression literature, we assume the stochastic predictors $\vx_{t}$ in $\mX_{t}=(1,\vx_{t}^\prime)^\prime$ to be generated by the additive components model
\[
	\vx_{t}=\vmu_{x} + \vxi_t,
\]
where $\vmu_{x}$ is an $\mathbb{R}^{k}$-valued constant and the (zero-mean) stochastic component $\vxi_t$ obeys the following autoregression:

\begin{assumption}[Predictors]\label{ass:N}
The stochastic $\vxi_t$ are generated by
\begin{equation}\label{eq:xi}
	\vxi_t  = \mR_n\vxi_{t-1} + \vu_{t},\qquad t\in\mathbb{N},
\end{equation}
where $\vu_{t}$ is a linear process defined in Assumption~\ref{ass:LP} below and the autoregressive matrix $\mR_n$ satisfies
\begin{equation}\label{eq:Cn}
	\mC_n:=n^{\kappa}(\mR_n-\mI_k)\underset{(n\to\infty)}{\longrightarrow}\mC,
\end{equation}
where $\mI_{k}$ denotes the $(k\times k)$-identity matrix and $\kappa\geq0$.
The variables $\vxi_t$ belong to one of the following classes:
\begin{enumerate}
	\item[(I0)]\label{it:I0} \textbf{Stationary predictors:} Equation~\eqref{eq:Cn} holds with $\kappa=0$ and $\mR:=\mR_n=\mI_k + \mC$ has spectral radius $\rho(\mR)<1$.
	
	\item[(NS)]\label{it:MI} \textbf{Near-stationary predictors:} Equation~\eqref{eq:Cn} holds with $\kappa\in(0,1)$ and $\mC$ is a negative stable matrix, i.e., all its eigenvalues have negative real part.
	
\end{enumerate}
The process $\vxi_t$ in \eqref{eq:xi} is initialized at $\vxi_0=O_{\P}(1)$ under (I0), and $\vxi_0=o_{\P}(n^{\kappa/2})$ under (NS).
\end{assumption}

A similar condition to Assumption~\ref{ass:N} is entertained by \citet[Assumption~P]{KMS23}.
Our Assumption~\ref{ass:N} covers the cases of stationary and near-stationary regressors in Assumption~N from \citet{MP20} with, in their notation, $\kappa_n=n^{\kappa}$.
Variables generated according to (NS) in Assumption~\ref{ass:N} have also been termed mildly integrated \citep{PM09,Lee16} or moderately integrated \citep{MP09}.
However, we adopt the terminology of \citet{MP20} here, as it is closest to our framework.
Allowing $\kappa=1$ in (NS) would lead to what \citet{MP20} call \textit{near-nonstationary regressors}.
These are closely related to nearly integrated variables, considered in the earlier literature \citep{CES95,JM06a}.

Of course, it would be desirable for additional robustness to also allow for near-nonstationary covariates with $\kappa=1$ in Assumption~\ref{ass:N}.
However, in that case, limit theory is already non-Gaussian for the quantile regression \citep{Lee16}, such that our self-normalized approach cannot be expected to work. 
Covering near-nonstationary regressors requires use of other methods, such as IVX-based approaches. 
Yet, investigating these is beyond the scope of the present paper, and is left for future research.

It will turn out that the estimators $\widehat{\valpha}_n(r,s)$ and $\widehat{\vbeta}_n(r,s)$ have a different convergence rate depending on whether the predictors are stationary or near-stationary in Assumption~\ref{ass:N}. 
To unify notation across these two cases, we introduce the normalizing matrix
\[
	\mD_n=\begin{cases}
		\sqrt{n}\mI_{k+1} & \text{for (I0)},\\
		\diag(\sqrt{n},n^{\frac{1+\kappa}{2}}\mI_{k}) & \text{for (NS)}.
	\end{cases}
\]

For the predictor innovations $\vu_t$, we follow the linear framework of \citet[Assumption~LP]{MP20}.
To that end, let $\Vert\cdot\Vert$ denote the spectral norm.

\begin{assumption}[Predictor innovations]\label{ass:LP}
For each $t\in\mathbb{N}$, $\vu_t$ has linear process representation
\[
	\vu_t=\sum_{j=0}^{\infty}\mF_j\vvarepsilon_{t-j},\qquad \sum_{j=0}^{\infty}\Vert\mF_j\Vert<\infty,\qquad \sum_{j=1}^{\infty}j\Vert\mF_j\Vert^2<\infty,
\]
where $\mF_0:=\mI_{k}$, $\mF(1):=\sum_{j=1}^{\infty}\mF_j$ has full rank, and $\vvarepsilon_t$ is a $\mathbb{R}^{k}$-valued martingale difference sequence with respect to $\widetilde{\mathcal{F}}_{t}=\sigma(\vvarepsilon_t,\vvarepsilon_{t-1},\ldots)$, such that $\E_{\widetilde{\mathcal{F}}_{t-1}}[\vvarepsilon_t\vvarepsilon_t^\prime]=\mSigma_{\varepsilon}>0$ and the sequence $\{\Vert\vvarepsilon_t\Vert^2\}_{t\in\mathbb{Z}}$ is uniformly integrable.
\end{assumption}

Often in the predictive regression literature, it is assumed that the regression disturbances are uncorrelated with the predictor increments $\vu_t$ \citep[see, e.g.,][Assumption~3]{Dea22}.
We do not require a similar condition in our linear process Assumption~\ref{ass:LP} for the $\vu_t$.
Also note that while the requirement $\E_{\widetilde{\mathcal{F}}_{t-1}}[\vvarepsilon_t\vvarepsilon_t^\prime]=\mSigma_{\varepsilon}>0$ rules out conditional heteroskedasticity of the $\vvarepsilon_t$ with respect to their own past, conditional on the full information set $\mathcal{F}_{t-1}$, variances of the $\vvarepsilon_t$ are allowed to vary.

The sole purpose of Assumptions~\ref{ass:N}--\ref{ass:LP} is to ensure that we can draw on certain results of \citet{MP20} for stationary and near-stationary predictors $\vx_{t-1}$.
The remaining development for our CoVaR regressions in \eqref{eq:(QRalt)}--\eqref{eq:(CoVaRalt)} instead relies on results for estimators derived from convex minimization problems \citep[see, e.g.,][]{Kat09}.
To be able to apply those, the appendix imposes some regularity conditions on \eqref{eq:(QRalt)} (Assumptions~\ref{ass:innov}--\ref{ass:K} in Appendix~\ref{Assumptions on the Quantile Regression}) and \eqref{eq:(CoVaRalt)} (Assumptions~\ref{ass:innov CoVaR}--\ref{ass:K ast} in Appendix~\ref{Assumptions on the CoVaR Regression}).

\subsection{Change Point Test for Predictive CoVaR Regressions}\label{Change Point Test for Predictive CoVaR Regressions}

Before presenting our first main result, we have to introduce some additional notation.
Define $\mathcal{D}_{\iota}=\big\{(r,s)\in[0,1]^2\colon \iota\leq r<s\leq 1-\iota,\ s-r\geq\iota\big\}$.
Let $\ell^{\infty}(\mathcal{D}_{\iota})$ denote the space of real-valued, bounded functions on $\mathcal{D}_{\iota}$ endowed with the uniform topology \citep{VW96}.
The $d$-fold product of this space is denoted by $(\ell^{\infty}(\mathcal{D}_{\iota}))^d$, which comes equipped with the product topology.

\begin{thm}\label{thm:CoVaR est}
Suppose $\mathcal{H}_0^{\CoVaR}$ holds true for the model \eqref{eq:(QRalt)}--\eqref{eq:(CoVaRalt)}.
If Assumptions~\ref{ass:N}--\ref{ass:K ast} are satisfied, then, as $n\to\infty$,
\begin{equation*}
(s-r)\begin{pmatrix}\mD_n\big[\widehat{\valpha}_n(r,s) - \valpha_0\big]\\
\mD_n\big[\widehat{\vbeta}_n(r,s) - \vbeta_0\big]
\end{pmatrix}
\overset{d}{\longrightarrow}\overline{\mSigma}^{1/2}\big[\overline{\mW}(s)-\overline{\mW}(r)\big]\qquad\text{in }(\ell^{\infty}(\mathcal{D}_{\iota}))^{2k+2},
\end{equation*}
where $\overline{\mW}(\cdot)$ is a $(2k+2)$-variate standard Brownian motion, and $\overline{\mSigma}$ is defined in Appendix~\ref{Asymptotic Variance-Covariance Matrices}.
\end{thm}

\begin{proof}
See Appendix~\ref{sec:proof sketch} for a sketch of the proof.
Appendices~\ref{sec:thm1}--\ref{sec:CoVaRest Lemmas} provide full detail.
\end{proof}

The proof of Theorem~\ref{thm:CoVaR est} draws on several sources.
The general outline of the proof is similar to that of Theorem~2 in \citet{HS25}, where they show the functional convergence of parameter estimators in linear (quantile, expected shortfall) regressions. 
Apart from considering a completely different functional, the main technical differences are as follows.
First, we go beyond \citet{HS25} by considering \textit{two-parameter} convergence in Theorem~\ref{thm:CoVaR est}.
Doing so allows us to derive an ``unsupervised'' break test in the spirit of \citet{ZL18} later on (see Section~\ref{Multiple Breaks}), which does not require the number of breaks to be pre-specified.
Second, in contrast to the stationary setup in \citet{HS25}, the limit theory in Theorem~\ref{thm:CoVaR est} unifies the cases of stationary and near-stationary regressors.
We are able to do so by leveraging several results of \citet{MP20} for (I0) and (NS) variables.

The precise form of the asymptotic variance-covariance matrix $\overline{\mSigma}$ is irrelevant for our test, because---by virtue of self-normalization---it simply cancels out in the limiting distribution of our test statistic $\mathcal{U}_{n,\vgamma}$ (see Corollary~\ref{cor:SBT CoVaR}).
The (non-functional) convergence for fixed $r=0$ and $s=1$ of Theorem~\ref{thm:CoVaR est} for $\widehat{\valpha}_n(r,s)$ is similar to that in \citet{Lee16} and \citet{FL19} (who, however, go beyond the (I0) and (NS) case of Assumption~\ref{ass:N} by also considering near-nonstationary and mildly explosive regressors).
Yet, on account of the functional convergence, Theorem~\ref{thm:CoVaR est} provides a much stronger conclusion in the (I0) and (NS) case, which may be used to derive uniformly valid inference on structural breaks in the predictive relationship, as we show next.

\begin{cor}\label{cor:SBT CoVaR}
Under the conditions of Theorem~\ref{thm:CoVaR est} it holds that, as $n\to\infty$,
\begin{equation*}
	\mathcal{U}_{n,\vgamma}	\overset{d}{\longrightarrow}\sup_{s\in[\iota,1-\iota]}\big[\overline{\mW}(s)-s\overline{\mW}(1)\big]^\prime \overline{\bm{\mathcal{\mW}}}^{-1}(s)\big[\overline{\mW}(s)-s\overline{\mW}(1)\big]=:\mathcal{U}_{2k+2},
\end{equation*}
where $\overline{\mW}(\cdot)$ is a $(2k+2)$-variate standard Brownian motion and the normalizer is
\begin{align*}
	\overline{\bm{\mathcal{\mW}}}(s) &= \int_{\iota}^{s}\big\{\overline{\mW}(r)-(r/s)\overline{\mW}(s)\big\}\big\{\overline{\mW}(r)-(r/s)\overline{\mW}(s)\big\}^\prime\D r \\
	&\hspace{1cm} + \int_{s}^{1-\iota}\Big\{\big[\overline{\mW}(1)-\overline{\mW}(r)\big] - \big((1-r)/(1-s)\big)\big[\overline{\mW}(1)-\overline{\mW}(s)\big]\Big\}\times\\
	&\hspace{2cm}\times\Big\{\big[\overline{\mW}(1)-\overline{\mW}(r)\big] - \big((1-r)/(1-s)\big)\big[\overline{\mW}(1)-\overline{\mW}(s)\big]\Big\}^\prime\D r.
\end{align*}
\end{cor}

\begin{proof}
The result follows by a straightforward application of the continuous mapping theorem to Theorem~\ref{thm:CoVaR est}; see Appendix~\ref{sec:thm2} for full detail.
\end{proof}

Denote by $\mathcal{U}_{\ell,1-\nu}$ the $(1-\nu)$-quantile of the distribution of $\mathcal{U}_{\ell}$ $(\ell\in\mathbb{N})$ for $\nu\in(0,1)$.
Then, Corollary~\ref{cor:SBT CoVaR} shows that rejecting $\mathcal{H}_{0}^{\CoVaR}$ if $\mathcal{U}_{n,\vgamma}	> \mathcal{U}_{2k+2,1-\nu}$ leads to an asymptotic level-$\nu$ test.

Table~\ref{tab:cv} shows some selected $(1-\nu)$-quantiles of the limiting distribution $\mathcal{U}_{\ell}$ from Corollary~\ref{cor:SBT CoVaR} for different $\ell$'s.
The critical values have been computed based on 500,000 draws from the limiting distribution, where the Brownian motion was approximated on a grid of 5,000 equally spaced points in $[0,1]$.

\begin{table}[t!]
		\centering
		\begin{tabular}{crrrrrr}
			\toprule
		$\ell$ & \multicolumn{6}{c}{$1-\nu$} 	\\
	\cline{2-7} \noalign{\vspace{0.5ex}}
		& 80\% & 90\% & 95\% & 97.5\%& 99\%  &   99.5\%  \\
\midrule
2		&  48.7 &  70.1 &  93.0 & 117.2 & 152.6 & 180.5 \\
3		&  78.9 & 108.6 & 138.8 & 170.9 & 215.4 & 250.7 \\
4		& 113.4 & 151.4 & 190.0 & 229.3 & 283.7 & 327.1 \\
5		& 152.9 & 199.5 & 245.6 & 293.0 & 357.5 & 407.5 \\
6		& 196.7 & 251.9 & 306.7 & 362.5 & 437.9 & 495.5 \\
7   & 244.8 & 309.4 & 373.0 & 436.9 & 522.6 & 589.3 \\
8   & 297.5 & 372.1 & 444.2 & 516.2 & 611.4 & 682.5 \\
9   & 354.8 & 438.6 & 519.7 & 599.6 & 706.7 & 790.3 \\
10  & 416.5 & 509.8 & 601.1 & 691.2 & 809.5 & 900.5 \\
			\bottomrule
		\end{tabular}
\caption{$(1-\nu)$-quantiles $\mathcal{U}_{\ell,1-\nu}$ of the limiting distribution $\mathcal{U}_{\ell}$ for different values of $\ell$ and fixed $\iota=0.1$.}
\label{tab:cv}
\end{table}

Importantly, these critical values are valid irrespective of whether predictors are (I0) or (NS), which delivers the persistence-robustness of our test.
The reason we obtain such a unified approach lies in Theorem~\ref{thm:CoVaR est}, which shows that the functional limit of the subsample estimates is Brownian motion with only the variance-covariance matrix differing between the (I0) and (NS) case.
These covariance matrices, however, cancel out in the limit due to SN.
The only quantity appearing in $\mathcal{U}_{2k+2}$ that is related to the data-generating process is the dimension $k$ of the predictors.
In our case, $k$ is a fixed integer. 

\begin{rem}[Computation of test statistic]
It is often computationally easier to approximate the test statistic $\mathcal{U}_{n,\vgamma}$ in the following way.
First, for $j\in\mathbb{N}$ define $\mT_n(j)=(j/n)(1-j/n)\big[\widehat{\vgamma}(0,j/n) - \widehat{\vgamma}(j/n, 1)\big]$, such that $\widehat{\vgamma}(0,j/n)$ ($\widehat{\vgamma}(j/n,1)$) is the regression estimate based on observations from $t=1,\ldots,j$ ($t=j+1,\ldots,n$).
Then, $\mathcal{U}_{n,\vgamma}$ can be approximated via
\begin{equation*}
	\widetilde{\mathcal{U}}_{n,\vgamma}=\sup_{j=\lfloor\iota n\rfloor+1,\ldots,\lfloor(1-\iota) n\rfloor}\mT_n^{\prime}(j)\mV_n^{-1}(j)\mT_n(j),
\end{equation*}
where 
\begin{multline*}
	\mV_n(j) = \frac{1}{n}\bigg\{\sum_{i=\lfloor\iota n\rfloor+1}^{j}(i/n)^2\big[\widehat{\vgamma}(0,i/n) - \widehat{\vgamma}(0, j/n)\big]\big[\widehat{\vgamma}(0,i/n) - \widehat{\vgamma}(0,j/n)\big]^\prime\\
	\sum_{i=j}^{\lfloor(1-\iota) n\rfloor}(1-i/n)^2\big[\widehat{\vgamma}(i/n,1) - \widehat{\vgamma}(j/n, 1)\big]\big[\widehat{\vgamma}(i/n,1) - \widehat{\vgamma}(j/n, 1)\big]^\prime\bigg\}.
\end{multline*}
Lengthy, but straightforward, calculations show that $\mathcal{U}_{n,\vgamma}-\widetilde{\mathcal{U}}_{n,\vgamma}=o_{\P}(1)$, such that a test of $\mathcal{H}_{0}^{\CoVaR}$ may also be based on $\widetilde{\mathcal{U}}_{n,\vgamma}$.
\end{rem}


\begin{rem}[Mixed-persistence predictors]
A careful reading of the proofs shows that Corollary~\ref{cor:SBT CoVaR} also holds for ``mixed''-persistence regressors, where $k_{1}$ predictors (say $\mX_{t-1,(I0)}$) are stationary and $k-k_1$ are near-stationary (say $\mX_{t-1,(NS)}$), as long as Lemma~\ref{lem:SUM} in Appendix~\ref{sec:QRest Lemmas} continues to hold for the $\mX_{t-1}=(\mX_{t-1,(I0)}^\prime, \mX_{t-1,(NS)}^\prime)^\prime$ and some limiting matrix $\mOmega$ (then with normalizing matrix $\mD_n=\diag(\sqrt{n}\mI_{k_1},\, n^{\frac{1+\kappa}{2}}\mI_{k-k_1})$).
\end{rem}

Finally, we show how Theorem~\ref{thm:CoVaR est} can be specialized to construct a change point test for predictive QRs.

\begin{cor}\label{cor:SBT}
Suppose $\mathcal{H}_0^{Q}\colon\valpha_{0,1}=\ldots=\valpha_{0,n}\equiv\valpha_0$ holds true for model \eqref{eq:(QRalt)}.
If Assumptions~\ref{ass:N}--\ref{ass:K} are satisfied, then, as $n\to\infty$,
\begin{equation*}
	\mathcal{U}_{n,\valpha}:=\sup_{s\in[\iota,1-\iota]} s^2(1-s)^2\big[\widehat{\valpha}_n(0,s) - \widehat{\valpha}_n(s,1)\big]^\prime \bm{\mathcal{N}}_{n,\valpha}^{-1}(s)\big[\widehat{\valpha}_n(0,s) - \widehat{\valpha}_n(s,1)\big]\overset{d}{\longrightarrow} \mathcal{U}_{k+1},
\end{equation*}
where the normalizer $\bm{\mathcal{N}}_{n,\valpha}(s)$ is defined as $\bm{\mathcal{N}}_{n,\vgamma}(s)$ with $\widehat{\vgamma}_n(\cdot,\cdot)$ replaced by $\widehat{\valpha}_n(\cdot,\cdot)$ at every occurrence.
\end{cor}

\begin{proof}
Analogous to the proof of Corollary~\ref{cor:SBT CoVaR}, but using Theorem~\ref{thm:std est} in Appendix~\ref{sec:thm1} in place of Theorem~\ref{thm:CoVaR est}.
Therefore, Corollary~\ref{cor:SBT} holds under a subset of the conditions of Corollary~\ref{cor:SBT CoVaR}.
\end{proof}



\subsection{Power Analysis}\label{Power Analysis}

Here, we show that our test based on $\mathcal{U}_{n,\vgamma}$ has power against a one-break local alternative.
To do so, consider the model \eqref{eq:(QRalt)}--\eqref{eq:(CoVaRalt)} under
\[
	\mathcal{H}_1^{\CoVaR}\colon \begin{pmatrix}\valpha_{0,t}\\ \vbeta_{0,t} \end{pmatrix}=\begin{pmatrix}\valpha_{0}+\mD_n^{-1}\va(t/n)\\ \vbeta_{0}+\mD_n^{-1}\vb(t/n) \end{pmatrix},\qquad t=1,\ldots,n,
\]
where $\va(\cdot)$ and $\vb(\cdot)$ are $\mathbb{R}^{k+1}$-valued, componentwise step functions on $[0,1]$.
This local alternative is inspired by \citet[Sec.~3]{KPA88}.
Extensions to certain smooth $\va(\cdot)$ and $\vb(\cdot)$ are also possible along the lines of \citet{KPA88} and \citet{Qu08}.
However, in view of later results on multiple discrete breaks (see Section~\ref{Multiple Breaks}) and to keep the exposition as simple as possible, we confine ourselves to step functions here.

\begin{thm}\label{thm:CoVaR est alt}
Suppose $\mathcal{H}_1^{\CoVaR}$ holds true for the model \eqref{eq:(QRalt)}--\eqref{eq:(CoVaRalt)}.
If Assumptions~\ref{ass:N}--\ref{ass:K ast} are satisfied, then, as $n\to\infty$,
\begin{multline*}
(s-r)\begin{pmatrix}\mD_n\big[\widehat{\valpha}_n(r,s) - \valpha_0\big]\\
\mD_n\big[\widehat{\vbeta}_n(r,s) - \vbeta_0\big]
\end{pmatrix}
\overset{d}{\longrightarrow}\overline{\mSigma}^{1/2}\big[\overline{\mW}(s)-\overline{\mW}(r)\big]
+\begin{pmatrix}\int_{r}^{s}\va(x)\D x\\ \int_{r}^{s}\vb(x)\D x\end{pmatrix}\quad \text{in }(\ell^{\infty}(\mathcal{D}_{\iota}))^{2k+2},
\end{multline*}
where $\overline{\mW}(\cdot)$ and $\overline{\mSigma}$ are as in Theorem~\ref{thm:CoVaR est}.
\end{thm}

\begin{proof}
See Appendix~\ref{CoVaR alt}.
\end{proof}

Our test statistic $\mathcal{U}_{n,\vgamma}$ is specifically designed to have power under one-break alternatives.
Formally, we show that when the single break point $s^\ast$ lies in the interval $(\iota,1-\iota)$, then our test is consistent against certain local alternatives.

\begin{cor}\label{cor:one-break}
Suppose $\mathcal{H}_1^{\CoVaR}$ holds true for the model \eqref{eq:(QRalt)}--\eqref{eq:(CoVaRalt)}, where $\vc(x):=\big(\va^\prime(x),\vb^\prime(x)\big)^\prime=\vc_1\1_{\{x\leq s^{\ast}\}} + \vc_2\1_{\{x>s^{\ast}\}}$ for some break point $s^\ast\in(\iota,1-\iota)$ and $\vc_1,\vc_2\in\mathbb{R}^{2k+2}$ with $\vc_1\neq\vc_2$.
If Assumptions~\ref{ass:N}--\ref{ass:K ast} are satisfied, then for any $\nu\in(0,1)$,
\[
	\lim_{\Vert\vc_2-\vc_1\Vert\to\infty}\lim_{n\to\infty}\P\big\{\mathcal{U}_{n,\vgamma}>\mathcal{U}_{2k+2,1-\nu}\big\}=1.
\]
\end{cor}

\begin{proof}
See Appendix~\ref{CoVaR alt}.
\end{proof}

We omit the corresponding local power result for the QR, because it follows analogously.
Importantly, Corollary~\ref{cor:one-break} shows that local power is higher for near-stationary than for stationary predictors.
This is because under stationarity the test is consistent against alternatives in a $n^{-1/2}$-neighborhood of the null, whereas for near-stationary variables consistency is obtained even in ``more local'' $n^{-(1+\kappa)/2}$-neighborhoods of the slope coefficients (see $\mathcal{H}_{1}^{\CoVaR}$).
The intuition behind this result is that the signal strength in regressions with highly persistent variables is much larger, allowing the parameters to be estimated much more precisely, which---in turn---aids in detecting potential breaks.

\subsection{Multiple Breaks}\label{Multiple Breaks}	

The $\mathcal{U}_{n,\vgamma}$ statistic is specifically designed for one-break alternatives.
Therefore, in this section, we briefly propose a test statistic $\mathcal{V}_{n,\vgamma}$ tailored for multiple possible breaks along the lines of the ``unsupervised'' approach of \citet{ZL18}.
We derive the asymptotic limit of $\mathcal{V}_{n,\vgamma}$ under the null and show consistency under local alternatives with arbitrarily many breaks.
Crucially, the number of break points does not have to be pre-specified \textit{ex ante} in this test and, in this sense, it is ``unsupervised''.

To introduce the test statistic, let $\mG=\mG(r,s)$ be some arbitrary function on $\mathcal{D}=\big\{(r,s)\in[0,1]^2\colon 0\leq r\leq s\leq 1\big\}$ and define
\begin{align*}
	\vd(\mG, s_1,s_2,s_3)  &= \frac{1}{(s_3-s_1)^{3/2}}\big[(s_3-s_2)\mG(s_1,s_2)-(s_2-s_1)\mG(s_2,s_3)\big],\\
	\mXi(\mG, s_1,s_2,s_3) &=  \frac{1}{(s_3-s_1)^{2}}\Big\{(s_2-s_1)\int_{s_1}^{s_2}\vd(\mG,s_1,s,s_2)\vd^\prime(\mG,s_1,s,s_2)\D s\\
	&\hspace{3cm}+(s_3-s_2)\int_{s_2}^{s_3}\vd(\mG,s_2,s,s_3)\vd^\prime(\mG,s_2,s,s_3)\D s\Big\}.
\end{align*}

Then, our test statistic is defined as
\begin{multline}\label{eq:Unsupervised T}
	\mathcal{V}_{n,\vgamma}=\sup_{(r_1,r_2)\in\mathcal{D}_{\iota}}\vd^\prime(\widehat{\mS}_n, 0,r_1,r_2)\mXi^{-1}(\widehat{\mS}_n, 0,r_1,r_2)\vd(\widehat{\mS}_n, 0,r_1,r_2)\\
	+\sup_{(s_1,s_2)\in\mathcal{D}_{\iota}}\vd^\prime(\widehat{\mS}_n, s_1,s_2,1)\mXi^{-1}(\widehat{\mS}_n,s_1,s_2,1)\vd(\widehat{\mS}_n, s_1,s_2,1),
\end{multline}
where $\iota\in(0,1/2)$ and $\widehat{\mS}_n(r,s):= (s-r)\widehat{\vgamma}_n(r,s)$.
To gain some intuition, note that straightforward algebra implies that
\begin{equation}\label{eq:(L.2)}
	\vd(\widehat{\mS}_n, s_1,s_2,s_3)=\frac{(s_2-s_1)(s_3-s_2)}{(s_3-s_1)^{3/2}}\big[\widehat{\vgamma}_n(s_1,s_2) - \widehat{\vgamma}_n(s_2,s_3)\big],
\end{equation}
such that the $\vd$-terms in $\mathcal{V}_{n,\vgamma}$ compare subsample estimates, giving the test its power.
In contrast, the ``meat'' matrices $\mXi^{-1}$ in \eqref{eq:Unsupervised T} simply serve as (cleverly constructed) normalizers.

\begin{cor}\label{thm:CoVaR est unsupervised}
Suppose $\mathcal{H}_0^{\CoVaR}$ holds true for the model \eqref{eq:(QRalt)}--\eqref{eq:(CoVaRalt)}.
If Assumptions~\ref{ass:N}--\ref{ass:K ast} are satisfied, then, as $n\to\infty$,
\begin{multline*}
\mathcal{V}_{n,\vgamma}
\overset{d}{\longrightarrow}\sup_{(r_1,r_2)\in\mathcal{D}_{\iota}}\vd^\prime(\Delta\overline{\mW}, 0,r_1,r_2)\mXi^{-1}(\Delta\overline{\mW}, 0,r_1,r_2)\vd(\Delta\overline{\mW}, 0,r_1,r_2)\\
	+\sup_{(s_1,s_2)\in\mathcal{D}_{\iota}}\vd^\prime(\Delta\overline{\mW}, s_1,s_2,1)\mXi^{-1}(\Delta\overline{\mW},s_1,s_2,1)\vd(\Delta\overline{\mW}, s_1,s_2,1)=:\mathcal{V}_{2k+2},
\end{multline*}
where $\Delta\overline{\mW}(r,s)=\overline{\mW}(s)-\overline{\mW}(r)$ for the $(2k+2)$-variate standard Brownian motion $\overline{\mW}(\cdot)$.
\end{cor}

\begin{proof}
See Appendix~\ref{sec:proofs of cors}.
\end{proof}

The limiting distribution $\mathcal{V}_{2k+2}$ coincides with $T^{\ast}(\mathbb{B})$ from Theorem~3.2 in \citet{ZL18}.

Corollary~\ref{thm:CoVaR est unsupervised loc alt} shows consistency of our test even for the case of $M^{\ast}$ breaks.

\begin{cor}\label{thm:CoVaR est unsupervised loc alt}
Suppose $\mathcal{H}_1^{\CoVaR}$ holds true for the model \eqref{eq:(QRalt)}--\eqref{eq:(CoVaRalt)}, where $\vc(x):=\big(\va^\prime(x), \vb^\prime(x)\big)=\vc_i$ for $x\in(s_{i-1}^{\ast},s_{i}^{\ast}]$ are the break magnitudes of the $M^{\ast}\in\mathbb{N}$ break points occurring at times $0=s_{0}^{\ast}<s_{1}^{\ast}<\ldots<s_{M^{\ast}}^{\ast}<s_{M^{\ast}+1}^{\ast}=1$.
Moreover, assume for the break magnitudes that $\vc_{i}\neq\vc_{i+1}$ for all $i=1,\ldots,M^{\ast}$ and the break times obey $\min_{0\leq i\leq M^{\ast}}|s_{i+1}^{\ast}-s_{i}^{\ast}|>\iota$.
If Assumptions~\ref{ass:N}--\ref{ass:K ast} are satisfied, then, for any $\nu\in(0,1)$,
\[
	\lim_{\min_{1\leq i\leq M^{\ast}}\Vert\vc_{i+1}-\vc_{i}\Vert\to\infty}\lim_{n\to\infty}\P\big\{\mathcal{V}_{n,\vgamma}>\mathcal{V}_{2k+2,1-\nu}\big\}=1,
\]
where $\mathcal{V}_{2k+2,1-\nu}$ denotes the $(1-\nu)$-quantile of $\mathcal{V}_{2k+2,1-\nu}$.
\end{cor}

\begin{proof}
See Appendix~\ref{sec:proofs of cors}.
\end{proof}

\section{Simulations}\label{sec:Simulation}

We provide simulation evidence on the finite-sample properties of our tests (size and power) in Appendix~\ref{sec:Main Simulation}.
We draw the following main conclusions from the simulation study. 
First, the change point tests for the predictive CoVaR regressions have reasonable size even for samples as small as $n=1000$.
Second, for extremely persistent predictors, the tests display some liberal tendencies. 
Yet, under near-stationarity, the size distortions vanish for large $n$, as predicted by our theory. 
Third, the power to identify structural breaks is higher for more persistent predictors---consistent with the higher local power for near-stationary covariates in Corollary~\ref{cor:one-break} (recall Section~\ref{Power Analysis}).

\section{Stability of the VIX as a Systemic Risk Predictor}\label{The VIX as a Predictor for Systemic Risk}

Up until the global financial crisis (GFC) of 2007--9, banking regulation focused almost exclusively on microprudential objectives; that is, it focused on limiting the risks of each financial institution in isolation.
However, such individual regulations were insufficient to prevent the crisis.
Thus, in the aftermath, attention has shifted towards macroprudential objectives, where---next to ensuring the financial viability of each bank in isolation---the goal is to improve the stability of the financial system as a whole. 
Achieving this requires a gauge of the interconnectedness of financial institutions, which is provided by systemic risk measures \citep{AB16,Aea17}.

As a consequence, researchers have aimed for a better understanding of systemic risk.
For instance, \citet{AB16} and \citet{Hea16} investigate the predictive content of volatility for future levels of systemic risk.
However, whether such predictive relationships remain stable over time has not been investigated before due to a lack of appropriate statistical methodology, which this paper provides.

Particular interest attaches to the VIX as a predictor for future systemic riskiness, because of what \citet{BS14} term the ``volatility paradox''.
This phenomenon describes systemic risk as building up in times of low volatility in equities, only to materialize during crises.
This occurs endogenously because financial institutions take excessive risk when volatility is low.
Due to this tight relationship, the VIX as a forward-looking volatility measure can be seen as a natural predictor of systemic risk.

The VIX, whose value at time $t$ we denote by $\VIX_t$, measures the market participants' risk-neutral expectation of the variance in stock market returns on the S\&P~500.
It is computed by the Chicago Board Options Exchange (CBOE) based on option prices.
Figure~\ref{fig:VIX} plots the VIX over our 10 year sample period from 2005--2014, which is chosen to include the GFC and the European sovereign debt crisis.
The most marked spike in expected volatility occurs after the collapse of Lehman Brothers in late 2008 during the GFC.
One can see clearly that the VIX took a long time to return to its pre-crisis level, suggesting a reasonably high persistence.
Further spikes of the VIX, such as during the European sovereign debt crisis of the early 2010s, are also noticeable. 

\begin{figure}[t!]
	\centering
	\includegraphics[width=\linewidth]{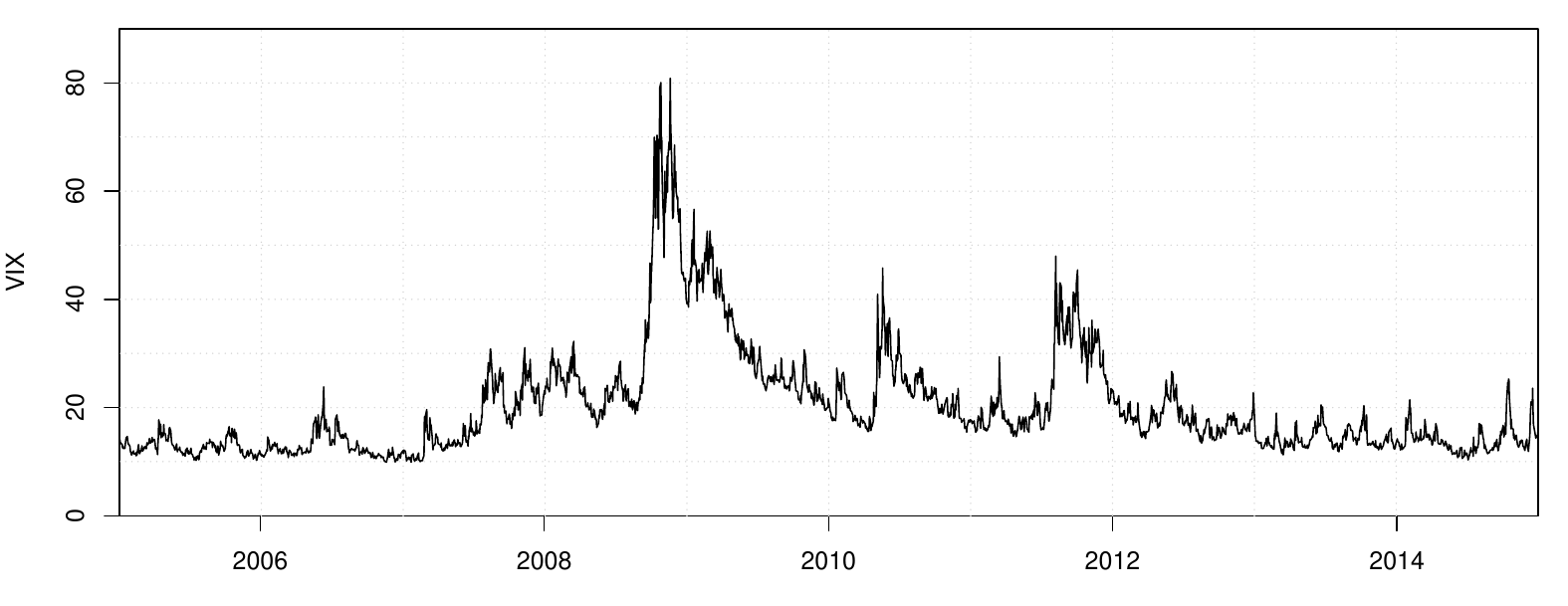}
	\caption{Level of the VIX from 2005--2014.}
	\label{fig:VIX}
\end{figure}

More formally, consistent with \citet{HZ12}, we find that the VIX is somewhere on the borderline of stationarity.
As in \citet{CWW15} and in line with Assumption~\ref{ass:N}, we report the AR(1) coefficient estimate for the VIX, which is 0.9824.
The KPSS test of \citet{Kea92} rejects the null of stationarity at any conventional significance level.
However, the ADF--GLS test of \citet{ERS96} rejects the null of a unit root in the VIX, with a $p$-value well below $1\%$.\footnote{A similar result is obtained when applying the Dickey--Fuller-type test of \citet{CT09}, which is robust to unconditional heteroskedasticity. Specifically, the wild bootstrap-based implementation of this test yields a $p$-value of 1.2\%. In computing this result, we have adopted the standard settings implemented in the \texttt{bootUR} package of \citet{SW23}.}
In light of this conflicting evidence where both stationarity and non-stationarity are rejected, the exact degree of persistence in the VIX seems difficult to determine.
Therefore, the persistence-robustness of our test becomes an empirically desirable feature when testing for changes next.

Specifically, we now investigate the stability of the VIX as a predictor of systemic risk in the US financial system. 
We do so for the ``standard'' CoVaR in Section~\ref{``Standard'' CoVaR}, where the conditioning is on the individual financial institution, and for the Exposure CoVaR in Section~\ref{Exposure CoVaR}, where we condition on the financial system.

\subsection{``Standard'' CoVaR}\label{``Standard'' CoVaR}

We now apply our test from Corollary~\ref{cor:SBT CoVaR} to check the constancy of the predictive relationship between the VIX and future levels of systemic risk in the US banking sector.
To do so, we use daily log-losses of the S\&P~500 Financials (SPF) index as our $Z_t$, and the log-losses on a particular US bank as $Y_t$.
We focus on those US banks that, at the end of our sample, were ranked as a global systemically important bank (G-SIB) by the Financial Stability Board \citep{FSB15}.
Then, $\CoVaR_{\beta\mid\alpha}\big((Z_t,Y_t)^\prime\mid\mathcal{F}_{t-1}\big)$ measures the impact that an extreme loss of the bank has on the financial sector, i.e., its riskiness to the system.
We examine the structural stability of the following predictive model for the CoVaR:
\begin{align}
	Y_t &= \alpha_{0,t}+\alpha_{1,t}\VIX_{t-1}+\epsilon_t, && Q_{\alpha}(\epsilon_t\mid\mathcal{F}_{t-1})=0,\label{eq:QR appl}\\
	Z_t &= \beta_{0,t}+\beta_{1,t}\VIX_{t-1}+\delta_t, && \CoVaR_{\beta\mid\alpha}\big((\delta_t, \epsilon_t)^\prime\mid\mathcal{F}_{t-1}\big)=0.\label{eq:CoVaR appl}
\end{align}
Therefore, we investigate whether the predictive power of volatility for systemic risk is subject to change over time.
To estimate the model, we use daily data from 2005 to 2014, all obtained from \textit{finance.yahoo.com}, resulting in a sample of size $n=2515$.

\citet[Sec.~6.2]{Hea16} use the \textit{change} in the VIX as a systemic risk predictor, citing concerns of non-stationarity of the VIX in levels.
This practice was also followed in the subsequent literature \citep[e.g.,][Sec.~6]{Ned20}.
While this step may be necessary to achieve stationarity (as required by the theory developed in \citet{Hea16}), going from levels to differences reduces the discriminatory power of subsequent tests, as is well-known from the predictive regression literature \citep[see, e.g.,][]{BD15a}.
Due to the persistence-robustness of our change point test, we can work directly with the levels of the VIX, which optimizes power; see the discussion below Corollary~\ref{cor:one-break}.

\begin{table}[t!]
	\centering
		\begin{tabular}{lcd{3.1}d{3.4}d{3.4}d{3.4}d{3.1}d{3.2}d{3.2}d{3.2}}
			\toprule
	$\alpha=\beta$& $Z_t$ 							&  \multicolumn{8}{c}{$Y_t$} \\
	\cline{3-10}\noalign{\vspace{0.5ex}}
								&											& \multicolumn{1}{c}{JPM}  &  \multicolumn{1}{c}{BAC}   &  \multicolumn{1}{c}{C}  	 & \multicolumn{1}{c}{GS} 	& \multicolumn{1}{c}{BK} 		& \multicolumn{1}{c}{MS} 		& \multicolumn{1}{c}{STT}  & \multicolumn{1}{c}{WFC} \\
		\midrule
	0.9 					& SPF									&154.5^{\ast} & 273.1^{\ast\ast} & 138.2 & 359.6^{\ast\ast\ast} &  98.4 & 249.3^{\ast\ast} & 213.9^{\ast\ast} & 213.9^{\ast\ast} \\
	0.95 					&       							& 43.3 & 197.8^{\ast\ast} & 298.0^{\ast\ast\ast} & 110.8 & 173.6^{\ast} &  26.3 &  58.1 &  58.1 \\
			\bottomrule
		\end{tabular}
	\caption{Values of test statistic $\mathcal{U}_{n,\vgamma}$ for CoVaR regression in \eqref{eq:QR appl}--\eqref{eq:CoVaR appl} with indicated $Y_t$ and $Z_t$. Ticker symbols JPM, BAC, C, GS, BK, MS, STT, and WFC correspond to JP Morgan Chase, Bank of America, Citigroup, Goldman Sachs, Bank of New York Mellon, Morgan Stanley, State Street, and Wells Fargo, respectively. Significances at the 10\%, 5\% and 1\% level are indicated by $^\ast$, $^{\ast\ast}$ and $^{\ast\ast\ast}$, respectively.}
\label{tab:SR pred}
\end{table}

For all US G-SIBs, Table~\ref{tab:SR pred} shows realizations of the test statistic $\mathcal{U}_{n,\vgamma}$ for $\alpha=\beta\in\{0.9,\ 0.95\}$ to truly capture \textit{systemic} risk.
The test statistics for all banks exceed the 10\%-critical value of 151.4 (see Table~\ref{tab:cv}) at least once, such that there is evidence for structural breaks in the predictive content of the VIX of varying degree for each institution.

\begin{figure}[t!]
	\centering
	\includegraphics[width=\linewidth]{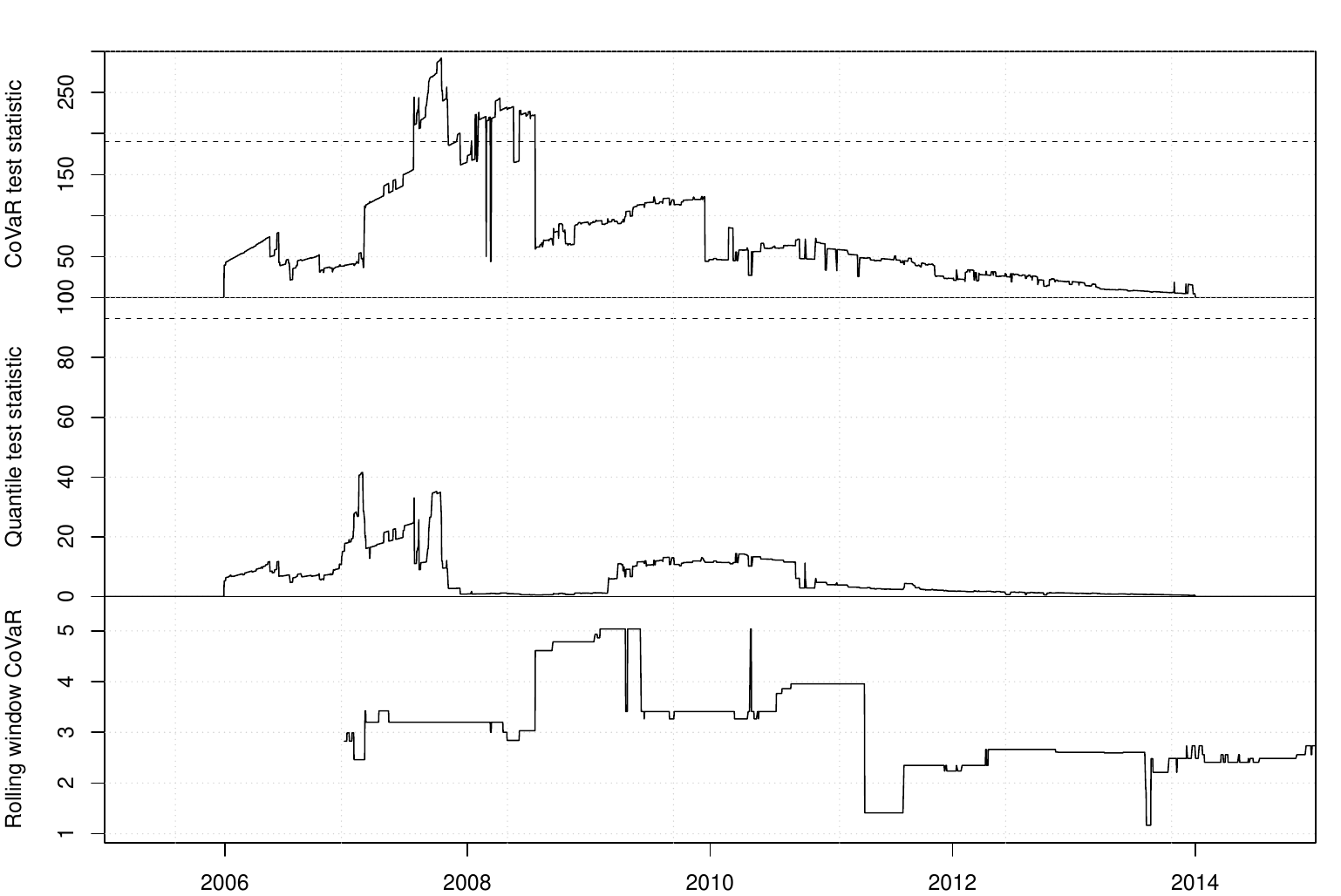}
	\caption{Top panel: Plot of the function $s\mapsto s^2(1-s)^2\big[\widehat{\vgamma}_n(0,s) - \widehat{\vgamma}_n(s,1)\big]^\prime \bm{\mathcal{N}}_{n,\vgamma}^{-1}(s)\big[\widehat{\vgamma}_n(0,s) - \widehat{\vgamma}_n(s,1)\big]$ for $\alpha=\beta=0.95$. The 5\%-critical value is indicated by the dashed horizontal line. 
	Middle panel: Plot of the function $s\mapsto s^2(1-s)^2\big[\widehat{\valpha}_n(0,s) - \widehat{\valpha}_n(s,1)\big]^\prime \bm{\mathcal{N}}_{n,\valpha}^{-1}(s)\big[\widehat{\valpha}_n(0,s) - \widehat{\valpha}_n(s,1)\big]$ for $\alpha=0.95$. The 5\%-critical value is indicated by the dashed horizontal line. 
	Bottom panel: Rolling-window estimates of the slope coefficient $\beta_{1,t}$ of the VIX in the linear predictive CoVaR regression \eqref{eq:CoVaR appl}. The rolling window estimates are based on 500 daily returns.}
	\label{fig:CPT CoVaR}
\end{figure}

To provide additional insight on a specific rejection in Table~\ref{tab:SR pred}, we focus on Citigroup.
Figure~\ref{fig:CPT CoVaR} provides an overview for the case $\alpha=\beta=0.95$.
Consistent with Table~\ref{tab:SR pred}, the 5\%-critical value of our CoVaR test is exceeded at some point in time by the function $s\mapsto s^2(1-s)^2\big[\widehat{\vgamma}_n(0,s) - \widehat{\vgamma}_n(s,1)\big]^\prime \bm{\mathcal{N}}_{n,\vgamma}^{-1}(s)\big[\widehat{\vgamma}_n(0,s) - \widehat{\vgamma}_n(s,1)\big]$, which forms the basis of our test statistic $\mathcal{U}_{n,\vgamma}$.
However, the top panel allows us to see more clearly the period responsible for the break. 
It is obvious that the evidence against a stable predictive relationship is strongest during the global financial crisis, with the maximum value of the test statistic attained on October 15th, 2007.

In principle, this break could be due to an instability in the quantile model \eqref{eq:QR appl} \textit{or} the CoVaR model \eqref{eq:CoVaR appl}.
However, when plotting the analogous function $s\mapsto s^2(1-s)^2\big[\widehat{\valpha}_n(0,s) - \widehat{\valpha}_n(s,1)\big]^\prime \bm{\mathcal{N}}_{n,\valpha}^{-1}(s)\big[\widehat{\valpha}_n(0,s) - \widehat{\valpha}_n(s,1)\big]$ for only the QR in \eqref{eq:QR appl}, we see that it solidly stays below the 5\%-critical value (see the middle panel of Figure~\ref{fig:CPT CoVaR}).
This points to an instability exclusively in the predictive relationship for future systemic risk in \eqref{eq:CoVaR appl}.

The bottom panel of Figure~\ref{fig:CPT CoVaR} sheds additional light on how the predictive content of the VIX varies over time for systemic risk.
It does so by plotting rolling-window estimates of the slope coefficient $\beta_{1,t}$ of the VIX in the CoVaR regression \eqref{eq:CoVaR appl}, where each window consists of 500 observations.
During the GFC we observe elevated levels of predictability by the VIX, which subsequently declines (but stays positive).
Therefore, as expected, during times where systemic risk materializes, the VIX as a ``fear index'' provides a good indication of high future systemic risk in the financial sector.
Yet even during times of lower predictability, the VIX and systemic risk remain tightly positively linked, consistent with the notion of \citet{AB16} that ``low volatility environments breed systemic risk''.


\subsection{Exposure CoVaR}\label{Exposure CoVaR}

Our above definition of the CoVaR corresponds to the standard definition of \citet{AB16}, where the system (here: S\&P~500 Financials) is considered conditional on some institution (here: US G-SIB) being in distress.
However, as pointed out by \citet[Sec.~II.D]{AB16}, the conditioning may be reversed to obtain the \textit{Exposure CoVaR}, which measures the risk of an institution given system failure.
The two CoVaR definitions give complementary information. 
For instance, a small, but highly connected, bank may have a large Exposure CoVaR, due to its many links with other institutions.
However, its ``standard'' CoVaR---measuring its impact on the system---would be low because of its small size, that would allow other banks to quickly pick up its positions in case of failure.
Since the Exposure CoVaR is also of interest in its own right, we repeat part of the above analysis now for the Exposure CoVaR.

Specifically, we investigate the structural stability of the Exposure CoVaR regression
\begin{align}
	Y_t &= \alpha_{0,t}+\alpha_{1,t}\VIX_{t-1}+\epsilon_t, && Q_{\alpha}(\epsilon_t\mid\mathcal{F}_{t-1})=0,\label{eq:QR appl2}\\
	Z_t &= \beta_{0,t}+\beta_{1,t}\VIX_{t-1}+\delta_t, && \CoVaR_{\beta\mid\alpha}\big((\delta_t, \epsilon_t)^\prime\mid\mathcal{F}_{t-1}\big)=0,\label{eq:CoVaR appl2}
\end{align}
which essentially corresponds to \eqref{eq:QR appl}--\eqref{eq:CoVaR appl}.
However, here we interchange the roles of $Y_t$ and $Z_t$, such that $Y_t$ now denotes the S\&P~500 Financials log-losses and $Z_t$ the log-losses on one of the eight US G-SIBs.
Doing so allows us to study the Exposure CoVaR for each G-SIB.

\begin{table}[t!]
	\centering
		\begin{tabular}{lcd{3.1}d{3.2}d{3.1}d{3.1}d{3.1}d{3.1}d{3.2}d{3.2}}
			\toprule
	$\alpha=\beta$& $Y_t$ 							&  \multicolumn{8}{c}{$Z_t$} \\
	\cline{3-10}\noalign{\vspace{0.5ex}}
								&											& \multicolumn{1}{c}{JPM}  &  \multicolumn{1}{c}{BAC}   &  \multicolumn{1}{c}{C}  	 & \multicolumn{1}{c}{GS} 	& \multicolumn{1}{c}{BK} 		& \multicolumn{1}{c}{MS} 		& \multicolumn{1}{c}{STT}  & \multicolumn{1}{c}{WFC} \\
		\midrule
	0.9 					& SPF									& 140.3 & 111.6 &  90.2 & 142.1 &  92.0 &  93.3 & 112.8 & 112.8 \\        
	0.95 					&       							& 144.4 & 163.0^{\ast} & 139.6 & 137.7 & 149.5 & 149.8 & 175.6^{\ast} & 175.6^{\ast} \\
			\bottomrule
		\end{tabular}
	\caption{Values of test statistic $\mathcal{U}_{n,\vgamma}$ for Exposure CoVaR regression in \eqref{eq:QR appl2}--\eqref{eq:CoVaR appl2} with indicated $Y_t$ and $Z_t$. Significances at the 10\%, 5\% and 1\% level are indicated by $^\ast$, $^{\ast\ast}$ and $^{\ast\ast\ast}$, respectively.}
\label{tab:SR pred2}
\end{table}

Table~\ref{tab:SR pred2}, which is the analog of Table~\ref{tab:SR pred}, shows the results of our stability tests for the Exposure CoVaR regression in \eqref{eq:QR appl2}--\eqref{eq:CoVaR appl2}. 
Here, we find little evidence of instability in the predictive relationship.
Only for three banks (for $\alpha=\beta=0.95$) do we find some indications for change at the 10\%-level.
However, with 16 tests being carried out at the 10\%-level, finding three rejections is not uncommon even when the null holds true in all cases.
Overall, full-sample estimates of the Exposure CoVaR regression in \eqref{eq:QR appl2}--\eqref{eq:CoVaR appl2} seem credible.

To obtain these full-sample estimates, we only have to run the single quantile regression for the S\&P~500 Financials in \eqref{eq:QR appl2}, which gives
\[
	\widehat{Q}_{\alpha}(Y_t\mid\mathcal{F}_{t-1}) = -0.021 + 0.0024\cdot \VIX_{t-1}.
\]
The positive sign of the estimated slope coefficient $\alpha_1$ indicates that a high expectation of future volatility (as measured by the VIX) is indicative of higher levels of risk for the US financial system.
We display the slope coefficient estimates $\widehat{\beta}_1$ of the appertaining Exposure CoVaR regression \eqref{eq:CoVaR appl2} in Table~\ref{tab:SR full sample rev}.
The VIX has a strong positive influence on the future exposure of banks to trouble in the financial system, as the positive and (mostly) significant estimates of $\beta_1$ suggest.

The $p$-values in Table~\ref{tab:SR full sample rev} (displayed in parentheses below the estimates) are calculated based on the functional convergence in Theorem~\ref{thm:CoVaR est}.
To see precisely how, write $\widehat{\vbeta}_{n}(0,s)=\big(\widehat{\beta}_{0,n}(0,s),\ldots,\widehat{\beta}_{k,n}(0,s)\big)^\prime$ and $\vbeta_0=(\beta_{0},\ldots,\beta_{k})^\prime$, and define the self-normalizer $\mathcal{S}_{n,\beta_i}=\int_{\iota}^{1}s^2\big[\widehat{\beta}_{i,n}(0,s) - \widehat{\beta}_{i,n}(0,1)\big]^2\D s$.
Then, following ideas from \citet{Sha10}, Theorem~\ref{thm:CoVaR est} and the continuous mapping theorem imply under $\mathcal{H}_0^{\CoVaR}$ and $\beta_{i}=0$ that, as $n\to\infty$,
\[
\mathcal{T}_n:=n\widehat{\beta}_{i,n}^2(0,1)/\mathcal{S}_{n,\beta_i}\overset{d}{\longrightarrow}W^2(1)/\int_{\iota}^{1}\big[W(s)-sW(1)\big]^2\D s=:\mathcal{T},
\]
where $W(\cdot)$ denotes a standard Brownian motion.
(A similar result also holds for the QR coefficients.)
Then, if $\mathcal{T}_n$ exceeds the $(1-\nu)$-quantile of $\mathcal{T}$, we conclude that the $i$-th CoVaR coefficient is significantly different from zero at level $\nu$.
Importantly, this ``off-the-shelf'' inference method afforded by our functional convergence result remains valid regardless of whether the predictors are stationary or near-stationary.
Therefore, the $p$-values of Table~\ref{tab:SR full sample rev} are persistence-robust in the same sense as our structural break tests.

\begin{table}[t!]
	\centering
		\begin{tabular}{lcrrrrrrrr}
			\toprule
			Coef.  	& $Y_t$ 						&  \multicolumn{8}{c}{$Z_t$} \\
	\cline{3-10}\noalign{\vspace{0.5ex}}
							&											& JPM  &  BAC   &  C  	 & GS 	& BK 		& MS 		& STT  & WFC \\
		\midrule
	$\beta_0$   & SPF	   							& $-39.4$ & $-62.7$ & $-55.6$ & $-23.6$ & $-36.7$ & $-38.8$ & $-75.5$ & $-52.4$  \\
							&											& (0.000) & (0.074) & (0.014) & (0.190) & (0.212) & (0.078) & (0.418) & (0.005)  \\
	$\beta_1$   &       							& $4.79$  & $7.14$  & $6.75$  & $4.78$  & $4.63$  & $6.05$  & $8.34$  & $5.72$   \\
							&											& (0.000) & (0.031) & (0.000) & (0.005) & (0.082) & (0.006) & (0.339) & (0.001)  \\
			\bottomrule
		\end{tabular}
	\caption{Full sample coefficient estimates of predictive Exposure CoVaR regression in \eqref{eq:QR appl2}--\eqref{eq:CoVaR appl2} for $\alpha=\beta=0.95$ with indicated $Y_t$ and $Z_t$. All estimates are premultiplied with 1,000 for better readability. $p$-values displayed below estimates in parentheses.}
\label{tab:SR full sample rev}
\end{table}

Overall, this section's results suggest that the forecasting power of the VIX for risk and exposure systemic risk is rather stable.
Moreover, the VIX is a statistically significant covariate, with high expected volatility predicting higher levels of risk and systemic risk in the US financial system.
We stress that this significance result for the VIX is robust to whether or not the VIX is stationary or near-stationary.

Comparing the results of Sections~\ref{``Standard'' CoVaR} and \ref{Exposure CoVaR}, we find that the VIX has somewhat instable forecasting power for predicting the riskiness of individual banks to the system, yet is a rather stable predictor of the riskiness the financial system poses for a single bank.
This suggests that the channel linking individual institutions to the system is more prone to changes than that connecting the system with a specific bank. 
In turn, this could be due to idiosyncratic effects on the level of the bank, such as---in the most extreme case---the collapse of an entire institution (e.g., Lehman Brothers) that destabilizes the financial system.
During its demise, the risk to the system is significantly increased relative to the pre-collapse time, leading naturally to time-variation in the bank's risk to the system.

\section{Conclusion}\label{sec:Conclusion}

This paper is the first to propose structural break tests for CoVaR regressions.
Importantly, our tests are valid irrespective of whether predictors are stationary or near-stationary.
In his review of self-normalization, \citet[p.~1814]{Sha15} writes on the use of SN that ``when the dependence in the time series is too strong (say, near-integrated time series), inference of certain parameter becomes difficult because the information aggregated over time does not accumulate quickly due to strong dependence.'' 
We show in this paper that when the dependence in the predictors is not too strong (at most near-stationary), then self-normalized change point tests still work in the sense of providing unified inference on possible breaks in the predictive relationship.
Crucially, no pre-tests for the serial dependence properties of the predictors are required, which would otherwise impair the validity of the subsequent break test.
As a corollary, we also obtain change point tests for QRs, where the results are novel for near-stationary predictors.

An empirical application highlights the importance of persistence-robust tests, where the VIX offers an example of a predictor somewhere on the boundary between stationarity and non-stationarity. 
Our test can be applied validly (in the sense of holding size) regardless of whether the VIX is stationary or near-stationary.
We find the predictive relationship between volatility and future systemic risk in the US financial system to be stable or instable depending on the direction of conditioning in the CoVaR.
We also provide an application of our persistence-robust QR stability test to some putative equity premium predictors (Appendix~\ref{Quantile Predictability of the Equity Premium}) that are also borderline stationary.
Our test reveals that there are significant instabilities in quantile prediction models for the US equity premium.

Future work could consider robustifying inference on change points also with respect to (I1) predictors. 
In this case, we no longer expect SN to work, but the IVX approach of \citet{MP09} and \citet{PM09} may be a feasible option. 
Indeed, IVX has been used successfully in developing unified inference in predictive QR \citep{Lee16}.
One drawback, however, is that IVX leads to slower rates of convergence for the (NS) case and, thus, possibly lower power.
Investigations such as these are left for future research.


%

%

%

%

\appendix
\renewcommand\appendixpagename{Appendix}
\appendixpage

\numberwithin{equation}{section}
\renewcommand\thesection{\Alph{section}}
\renewcommand\theequation{\thesection.\arabic{equation}}

This appendix contains Sections~\ref{Assumptions on the Quantile Regression}--\ref{Quantile Predictability of the Equity Premium}.
Specifically, Appendices~\ref{Assumptions on the Quantile Regression}--\ref{Assumptions on the CoVaR Regression} present additional regularity conditions required for the quantile and the CoVaR regression, respectively.
In Appendix~\ref{Asymptotic Variance-Covariance Matrices}, we introduce variance-covariance matrices for the (quanitle, CoVaR) regression estimators.
 A sketch of the proof of Theorem~\ref{thm:CoVaR est} is provided in Appendix~\ref{sec:proof sketch}.
	Appendix~\ref{sec:thm1} presents Theorem~\ref{thm:std est}, which establishes the functional convergence of the quantile regression estimator.
	Additional propositions for the proof of Theorem~\ref{thm:std est} are established in Appendix~\ref{sec:QRest Lemmas}.
	Theorem~\ref{thm:CoVaR est} and Corollary~\ref{cor:SBT CoVaR} are proven in Appendix~\ref{sec:thm2}.	
	Supplementary propositions for the proof of Theorem~\ref{thm:CoVaR est} are proven in Appendix~\ref{sec:CoVaRest Lemmas}. 
	Appendix~\ref{sec:thm1 alt} presents the local power result for the quantile regression in Theorem~\ref{thm:std est alt}, where additional technical detail is provided in Appendix~\ref{sec:Prop alt QR}.
	Theorem~\ref{thm:CoVaR est alt}, i.e., the local power result for the CoVaR regression, is proven in Appendix~\ref{CoVaR alt}, with supplementary propositions established in Appendix~\ref{sec:Prop alt CoVaR}.
	Appendix~\ref{sec:proofs of cors} contains the proofs of Corollaries~\ref{thm:CoVaR est unsupervised}--\ref{thm:CoVaR est unsupervised loc alt}.
	Appendix~\ref{sec:Main Simulation} presents Monte Carlo simulations for our tests.
	The final Appendix~\ref{Quantile Predictability of the Equity Premium} investigates the stability of quantile prediction models for the US equity premium.

\section*{Notation}

We use the following notational conventions throughout this appendix. The probability space that we work on is $\big(\Omega, \mathcal{F}, \P\big)$. We denote by $K>0$ a large positive constant that may change from line to line. 
If not specified otherwise, all convergences are to be understood with respect to $n\to\infty$, and the symbols $o_{\P}$ and $O_{\P}$ carry their usual meaning.
We write $\E_{t-1}[\cdot]=\E[\cdot\mid\mathcal{F}_{t-1}]$ and $\P_{t-1}\{\cdot\}=\P\{\cdot\mid\mathcal{F}_{t-1}\}$ for short. 
The undesignated norm $\Vert\cdot\Vert$ denotes the spectral norm.
We exploit without further mention that the spectral norm is submultiplicative, i.e., that $\Vert\mA\mB\Vert\leq\Vert\mA\Vert\cdot\Vert\mB\Vert$ for conformable matrices $\mA$ and $\mB$.

\section{Assumptions on the Quantile Regression}\label{Assumptions on the Quantile Regression}

We place the following assumption on the conditional distribution of the QR errors $\epsilon_t$ in \eqref{eq:(QRalt)}.

\begin{assumption}[Conditional density of QR errors]\label{ass:innov}
\begin{enumerate}[(i)]
	\item\label{it:dens} There exists some $d>0$, such that the distribution of $\epsilon_t$ conditional on $\mathcal{F}_{t-1}$ has a Lebesgue density $f_{\epsilon_t\mid\mathcal{F}_{t-1}}(\cdot)$ on $[-d,d]$.
	
	\item\label{it:dens bound} $0<\underline{f}<\inf_{x\in[-d,d]}f_{\epsilon_t\mid\mathcal{F}_{t-1}}(x)$ and $\sup_{x\in[-d,d]}f_{\epsilon_t\mid\mathcal{F}_{t-1}}(x)\leq \overline{f}<\infty$ for all $t\geq 1$.
	
	\item\label{it:Lipschitz} There exists some $L>0$, such that for all $t\geq1$ and all $-d\leq x_1,x_2\leq d$ it holds that $\big|f_{\epsilon_t\mid\mathcal{F}_{t-1}}(x_1) - f_{\epsilon_t\mid\mathcal{F}_{t-1}}(x_2)\big|\leq L |x_1-x_2|$.

\end{enumerate}
\end{assumption}

Assumption~\ref{ass:innov} imposes smoothness conditions on the conditional distribution of $\epsilon_t\mid\mathcal{F}_{t-1}$, viz.~the existence of a uniformly bounded and Lipschitz-continuous density.
The conditions are virtually identical to the assumptions on the QR error densities in \citet[Theorem~2.2]{Fit97}.

In their cointegrated system, \citet[Eqns.~(1)--(2)]{MP20} consider $\mY_t=\mA\vx_{t-1}+\vu_{0t}$ (instead of $Y_t=(1,\vx_{t-1}^\prime)^\prime\valpha_{0,t}+ \epsilon_t$ as we do in \eqref{eq:(QRalt)}).
In particular, \citet[Assumption~LP]{MP20} impose a linear process assumption jointly on $(\vu^\prime_{0t},\vu_t^\prime)^\prime$, where $\vu_t$ are the predictor innovations.
We do not have to impose a similar joint linear process assumption on our $(\epsilon_t,\vu_t^\prime)^\prime$.
Instead, the following suffices for our purposes:

\begin{assumption}[Joint behavior of QR errors and predictors]\label{ass:K}
It holds uniformly in $s\in[0,1]$ that, as $n\to\infty$,
\[
	\sum_{t=1}^{\lfloor ns\rfloor}f_{\epsilon_t\mid\mathcal{F}_{t-1}}(0)\mD_n^{-1}\mX_{t-1}\mX_{t-1}^\prime\mD_n^{-1}\overset{\P}{\longrightarrow}s\mK
\]
for some positive definite $\mK\in\mathbb{R}^{(k+1)\times(k+1)}$.
\end{assumption}

Note that the predictors $\mX_{t-1}$ are solely functions of past $\vu_t$.
Therefore, Assumption~\ref{ass:K} implicitly restricts the serial dependence in $(\epsilon_t,\vu_t^\prime)^\prime$ by imposing a uniform law of large numbers to hold for $\big\{f_{\epsilon_t\mid\mathcal{F}_{t-1}}(0)\mD_n^{-1}\mX_{t-1}\mX_{t-1}^\prime\mD_n^{-1}\big\}$.
For stationary predictors, where $\mD_n=\sqrt{n}\mI_{k+1}$, Assumption~\ref{ass:K} reads as $n^{-1}\sum_{t=1}^{\lfloor ns\rfloor}f_{\epsilon_t\mid\mathcal{F}_{t-1}}(0)\mX_{t-1}\mX_{t-1}^\prime\overset{\P}{\longrightarrow}s\mK$.
In this form, it is virtually identical to conditions entertained by \citet[Assumption~3~(b)]{Qu08}, \citet[Assumption~A6~(ii)]{SX08}, \citet[Assumption~5~(b)]{OQ11} and \citet[Assumption C.2]{MWT25}.


\section{Assumptions on the CoVaR Regression}\label{Assumptions on the CoVaR Regression}

Deriving functional central limit theory for the CoVaR regression subsample estimates requires some further assumptions.
Specifically, we impose the following analog of Assumption~\ref{ass:innov} for the CoVaR regression errors $(\epsilon_t,\delta_t)^\prime$ in \eqref{eq:(CoVaRalt)}:

\begin{assumption}[Conditional density of CoVaR errors]\label{ass:innov CoVaR}
\begin{enumerate}[(i)]
	\item\label{it:dens CoVaR} The distribution of $(\epsilon_t,\delta_t)^\prime$ conditional on $\mathcal{F}_{t-1}$ has a Lebesgue density $f_{(\epsilon_t,\delta_t)^\prime\mid\mathcal{F}_{t-1}}(\cdot,\cdot)$ on $\mathbb{R}^2$. 
	
	\item\label{it:dens bound CoVaR} $0<\underline{f}<f_{(\epsilon_t,\delta_t)^\prime\mid\mathcal{F}_{t-1}}(x,y)$ and $f_{(\epsilon_t,\delta_t)^\prime\mid\mathcal{F}_{t-1}}(x,y)\leq \overline{f}<\infty$ for all $t\geq 1$, and all $(x,y)^\prime\in\mathbb{R}^2$ such that $F_{(\epsilon_t,\delta_t)^\prime\mid\mathcal{F}_{t-1}}(x,y)\in(0,1)$.
	
	\item\label{it:Lipschitz CoVaR} There exists some $L>0$, such that for all $t\geq1$, all $-d\leq x_1,x_2\leq d$ and all $-d\leq y_1,y_2\leq d$,
	\begin{align*}
		&\bigg|\int_{0}^{\infty}f_{(\epsilon_t,\delta_t)^\prime\mid\mathcal{F}_{t-1}}(x_1,y)\D y - \int_{0}^{\infty}f_{(\epsilon_t,\delta_t)^\prime\mid\mathcal{F}_{t-1}}(x_2,y)\D y\bigg|\leq L |x_1-x_2|,\\
		&\bigg|\int_{0}^{\infty}f_{(\epsilon_t,\delta_t)^\prime\mid\mathcal{F}_{t-1}}(x,y_1)\D x - \int_{0}^{\infty}f_{(\epsilon_t,\delta_t)^\prime\mid\mathcal{F}_{t-1}}(x,y_2)\D x\bigg|\leq L |y_1-y_2|
	\end{align*}
	with $d>0$ from Assumption~\ref{ass:innov}.
\end{enumerate}
\end{assumption}

The following condition is the counterpart to Assumption~\ref{ass:K}.

\begin{assumption}[Joint behavior of CoVaR errors and predictors]\label{ass:K ast}
It holds uniformly in $s\in[0,1]$ that, as $n\to\infty$,
\begin{align*}
	\sum_{t=1}^{\lfloor ns\rfloor}\bigg(\int_{0}^{\infty}f_{(\epsilon_t,\delta_t)^\prime\mid\mathcal{F}_{t-1}}(x,0)\D x\bigg)\mD_n^{-1}\mX_{t-1}\mX_{t-1}^\prime\mD_n^{-1}\overset{\P}{\longrightarrow}s\mK_\ast,\\
	\sum_{t=1}^{\lfloor ns\rfloor}\bigg(\int_{0}^{\infty}f_{(\epsilon_t,\delta_t)^\prime\mid\mathcal{F}_{t-1}}(0,y)\D y\bigg)\mD_n^{-1}\mX_{t-1}\mX_{t-1}^\prime\mD_n^{-1}\overset{\P}{\longrightarrow}s\mK_\dagger
\end{align*}
for some positive definite matrices $\mK_{\ast}\in\mathbb{R}^{(k+1)\times(k+1)}$ and $\mK_{\dagger}\in\mathbb{R}^{(k+1)\times(k+1)}$.
\end{assumption}

\section{Asymptotic Variance-Covariance Matrices}\label{Asymptotic Variance-Covariance Matrices}


The asymptotic variance $\mSigma$ of the QR estimator $\widehat{\valpha}_n(r,s)$ (see also Theorem~\ref{thm:std est} in Appendix~\ref{thm:std est}) depends on several matrices.
To introduce these, put $\mV_{\xi\xi}=\int_{0}^{\infty}e^{r\mC} \mOmega_{\xi\xi} e^{r\mC}\D r$ with $\mOmega_{\xi\xi}=\mF(1)\mSigma_{\varepsilon}\mF^\prime(1)$, where $\mF(1)$ and $\mSigma_{\varepsilon}$ are defined in Assumption~\ref{ass:LP}.
Also, set $\mOmega_{xx}=\vmu_x\vmu_x^\prime + \sum_{j=0}^{\infty}\mR^j\big[\mGamma_u(0) + \mR\mGamma + \mGamma^\prime\mR^\prime\big](\mR^j)^\prime$ with $\mGamma=\sum_{l=1}^{\infty}\mR^{l-1}\mGamma_u^\prime(l)$ and $\mGamma_u(l)=\E[\vu_{t}\vu_{t-l}^\prime]$, and $\mR$ from Assumption~\ref{ass:N}.
Define
\begin{equation}\label{eq:Omega}
	\mOmega=\begin{cases}\mOmega_{XX}& \text{if \textit{(I0)}},\\ \mV_{XX}& \text{if \textit{(NS)}},\end{cases}\quad\text{where}\quad\mOmega_{XX}=\begin{pmatrix}1 & \vmu_{x}^\prime\\ \vmu_{x} & \mOmega_{xx}  \end{pmatrix},\qquad \mV_{XX}=\begin{pmatrix}1 & \vzero\\ \vzero & \mV_{\xi\xi}  \end{pmatrix}.
\end{equation}
Then, the asymptotic variance of the estimator $\widehat{\valpha}_n(r,s)$ is given by
\[
	\mSigma=\alpha(1-\alpha)\mK^{-1}\mOmega\mK^{-1}.
\]

The asymptotic variance of the estimator $\widehat{\vbeta}_n(r,s)$ of the CoVaR parameters depends on the following matrices; see Theorem~\ref{thm:CoVaR est}.
Define 
\[
	\mOmega_{\ast}=\alpha^{-1}\beta(1-\beta)\mOmega + \alpha^{-1}(1-\alpha)^{-1}\big[(1-\beta)\mK-\mK_{\dagger}\big]\mK^{-1}\mOmega\mK^{-1} \big[(1-\beta)\mK-\mK_{\dagger}\big]
\]
and
\[
\overline{\mK}=\begin{pmatrix}
		\mK & \vzero\\
		\vzero & \mK_{\ast}
	\end{pmatrix},\qquad
	\overline{\mOmega}=\begin{pmatrix}
	\mOmega & \mOmega\mK^{-1}\big[(1-\beta)\mK-\mK_{\dagger}\big]\\
	\big[(1-\beta)\mK-\mK_{\dagger}\big]\mK^{-1}\mOmega & \mOmega_{\ast}
	\end{pmatrix}.
\]
Then, the asymptotic variance-covariance matrix of $\big(\widehat{\valpha}_n^\prime(r,s), \widehat{\vbeta}_n^\prime(r,s)\big)^\prime$ is given by
\[
	\overline{\mSigma}=\alpha(1-\alpha)\overline{\mK}^{-1}\cdot\overline{\mOmega}\cdot\overline{\mK}^{-1}.
\]

\section{Proof Sketch for Theorem~\ref{thm:CoVaR est}}\label{sec:proof sketch}

Here, we provide a sketch of the proof of Theorem~\ref{thm:CoVaR est}, with full detail provided in Appendices~\ref{sec:thm1}--\ref{sec:CoVaRest Lemmas}.

\begin{proof}[{\textbf{Proof of Theorem~\ref{thm:CoVaR est} (Sketch):}}]
The proof proceeds in three main steps. 
First, we establish the functional convergence of the QR estimator, viz.
\begin{equation}\label{eq:first step}
(s-r)\mD_n[\widehat{\valpha}_n(r,s) - \valpha_0]\overset{d}{\longrightarrow}\mSigma^{1/2}\big[\mW(s)-\mW(r)\big]\qquad\text{in }(\ell^{\infty}(\mathcal{D}_{\iota}))^{k+1},
\end{equation}
where $\mW(\cdot)$ is a $(k+1)$-variate standard Brownian motion, and $\mSigma$ is defined in Appendix~\ref{Asymptotic Variance-Covariance Matrices}.

Second, we similarly show the following functional convergence of the CoVaR parameter estimator:
\begin{equation}\label{eq:second step}
(s-r)\mD_n\big[\widehat{\vbeta}_n(r,s) - \vbeta_0\big]\overset{d}{\longrightarrow}\mSigma_{\ast}^{1/2}\big[\mW_{\ast}(s)-\mW_{\ast}(r)\big]\qquad\text{in }(\ell^{\infty}(\mathcal{D}_{\iota}))^{k+1},
\end{equation}
where $\mW_{\ast}(\cdot)$ is a $(k+1)$-variate standard Brownian motion and $\mSigma_{\ast}=\alpha(1-\alpha)\mK_{\ast}^{-1}\mOmega_{\ast}\mK_{\ast}^{-1}$ with $\mOmega_{\ast}$ defined in Appendix~\ref{Asymptotic Variance-Covariance Matrices}.

The third and final step then proves the desired \textit{joint} functional convergence of the QR and CoVaR estimator.

We now give a sketch of the proof of the first step. 
(Appendix~\ref{sec:thm1} contains a complete proof.)
To do so, note that by \eqref{eq:quantile model} the estimator $\widehat{\valpha}_n(r,s)$ can equivalently be written as
\begin{equation*}
\widehat{\valpha}_n(r,s)= \argmin_{\valpha\in\mathbb{R}^{k+1}} \sum_{t=\lfloor nr\rfloor+1}^{\lfloor ns\rfloor}\big[\rho_{\alpha}\big(\epsilon_{t} - (\valpha-\valpha_0)^\prime\mX_{t-1}\big) - \rho_{\alpha}(\epsilon_{t})\big].
\end{equation*}
Hence, if we define 
\[
	f_n(\vw,r,s)=\sum_{t=\lfloor nr\rfloor+1}^{\lfloor ns\rfloor}\big[\rho_{\alpha}(\epsilon_{t} - \vw^\prime\mD_n^{-1}\mX_{t-1}) - \rho_{\alpha}(\epsilon_{t})\big],
\] 
then the minimizer $\vw_n(r,s)$ of $f_n(\cdot,r,s)$ satisfies that 
\[
	\vw_n(r,s)=\mD_n\big[\widehat{\valpha}_n(r,s) - \valpha_0\big].
\]
We derive the weak limit of $\vw_n(\cdot,\cdot)$ by invoking Theorem~2 of \citet{Kat09}.
This theorem shows that if $f_n(\vw,r,s)$ is asymptotically quadratic in $\vw$, one can derive an asymptotic representation of the minimizer $\vw_n(r,s)$. 
Specifically, using the relation
\begin{equation*}
	\rho_{\alpha}(u-v) - \rho_{\alpha}(u)=-v\psi_{\alpha}(u) + (u-v)\big[\1_{\{v<u<0\}} - \1_{\{0<u<v\}}\big]\quad\text{for}\  u\neq0,
\end{equation*}
we show for the quadratic form 
\begin{align*}
	g_n(\vw,r,s)&:=-\vw^\prime \mW_n(r,s) + \frac{1}{2}\vw^\prime(s-r)\mK\vw\\
	&:=-\vw^\prime\sum_{t=\lfloor nr\rfloor+1}^{\lfloor ns\rfloor}\psi_{\alpha}(\epsilon_{t})\mD_n^{-1}\mX_{t-1} + \frac{1}{2}\vw^\prime(s-r)\mK\vw
\end{align*}
that $\big|f_n(\vw,r,s) - g_n(\vw,r,s)\big|=o_{\P}(1)$ uniformly in $(r,s)$ (see Proposition~\ref{lem:LLN} in Appendix~\ref{sec:thm1}).
Theorem~2 of \citet{Kat09} then implies the following asymptotic representation of the minimizer $\vw_n(r,s)$:
\[
	\vw_n(r,s)=\frac{1}{s-r}\mK^{-1}\mW_n(r,s)+\vr_n(r,s)
\]
for uniformly negligible $\vr_n(r,s)$.
Therefore, the convergence of $\vw_n(r,s)$ can be derived from that of $\mW_n(r,s)$.
We derive functional central limit theory for $\mW_n(r,s)$ by drawing on classic results for martingale difference arrays in Proposition~\ref{lem:CLT} from Appendix~\ref{sec:thm1}.
Doing so, we obtain that, as $n\to\infty$,
\begin{align*}
	\vw_n(r,s)\overset{d}{\longrightarrow}\vw_{\infty}(r,s)&= \frac{\sqrt{\alpha(1-\alpha)}}{s-r}\mK^{-1}\mOmega^{1/2}\big[\mW(s)-\mW(r)\big],
\end{align*}
i.e., \eqref{eq:first step}.

The proof of \eqref{eq:second step} proceeds by similar arguments because $\widehat{\vbeta}_n(r,s)$ is also a quantile regression-type estimator, comparable in structure to $\widehat{\valpha}_n(r,s)$.
 
The third step that establishes the joint convergence is standard.

For details on the last two steps, we refer to the complete proof of Theorem~\ref{thm:CoVaR est} in Appendix~\ref{sec:thm2}.
\end{proof}

\section{Functional Convergence of the QR Estimator}
\label{sec:thm1}

As a first step towards proving Theorem~\ref{thm:CoVaR est}, we show the following

\begin{thm}\label{thm:std est}
Suppose $\mathcal{H}_0^{Q}\colon\valpha_{0,1}=\ldots=\valpha_{0,n}\equiv\valpha_0$ holds true for the model \eqref{eq:(QRalt)}.
If Assumptions~\ref{ass:N}--\ref{ass:K} are satisfied, then, as $n\to\infty$,
\begin{equation*}
(s-r)\mD_n[\widehat{\valpha}_n(r,s) - \valpha_0]\overset{d}{\longrightarrow}\mSigma^{1/2}\big[\mW(s)-\mW(r)\big]\qquad\text{in }(\ell^{\infty}(\mathcal{D}_{\iota}))^{k+1},
\end{equation*}
where $\mW(\cdot)$ is a $(k+1)$-variate standard Brownian motion, and $\mSigma$ is defined in Appendix~\ref{Asymptotic Variance-Covariance Matrices}.
\end{thm}

The proof of Theorem~\ref{thm:std est} is structured similarly as that of Theorem~1 in \citet{HS25}.
Of course, the same comments as below Theorem~\ref{thm:CoVaR est} apply.

Theorems~\ref{thm:CoVaR est} and \ref{thm:std est} provide functional central limit theory for estimators derived from a convex minimization problem.
To prove the convergences, we draw on \citet{Kat09}.
But before doing so, we introduce some additional notation.
Let $T\subset\mathbb{R}^{q}$ denote a compact set.
Then, $\ell^{\infty}(T)$ is the space of real-valued bounded functions on $T$ endowed with the uniform topology \citep{Bil99,VW96}.
The $d$-fold product of this space is denoted by $(\ell^{\infty}(T))^d$, equipped with the product topology.

The proof of Theorem~\ref{thm:std est} further requires the following two preliminary Propositions~\ref{lem:CLT}--\ref{lem:LLN}, which are similar to Lemmas~2--3 in \citet{HS25}.
To introduce the first, recall that $\mathcal{D}=\big\{(r,s)\in[0,1]^2\colon 0\leq r\leq s\leq 1\big\}$ and define an empty sum to be zero.
We also write $\E_{t-1}[\cdot]=\E[\cdot\mid\mathcal{F}_{t-1}]$ and $\P_{t-1}\{\cdot\}=\P\{\cdot\mid\mathcal{F}_{t-1}\}$ for short.

\begin{prop}\label{lem:CLT}
Under the assumptions of Theorem~\ref{thm:std est} it holds that, as $n\to\infty$,
\begin{equation*}
	\sum_{t=\lfloor nr\rfloor+1}^{\lfloor ns\rfloor}\psi_{\alpha}(\epsilon_t)\mD_n^{-1}\mX_{t-1}
	\overset{d}{\longrightarrow}\sqrt{\alpha(1-\alpha)}\mOmega^{1/2}\big[\mW(s)-\mW(r)\big]\qquad\text{in }(\ell^{\infty}(\mathcal{D}))^{k+1},
\end{equation*}
where $\mW(\cdot)$ is a $(k+1)$-variate standard Brownian motion, and $\mOmega$ is defined in \eqref{eq:Omega}.
\end{prop}

\begin{proof}
See Appendix~\ref{sec:QRest Lemmas}.
\end{proof}

\begin{prop}\label{lem:LLN}
Under the assumptions of Theorem~\ref{thm:std est} it holds for fixed $\vw\in\mathbb{R}^{k+1}$ that, as $n\to\infty$,
\begin{multline*}
	\sup_{0\leq r< s\leq 1}\bigg|\sum_{t=\lfloor nr\rfloor+1}^{\lfloor ns\rfloor}\big(\epsilon_{t} - \vw^\prime\mD_n^{-1}\mX_{t-1}\big)\big[\1_{\{\vw^\prime\mD_n^{-1}\mX_{t-1}<\epsilon_{t}<0\}} - \1_{\{0<\epsilon_{t}<\vw^\prime\mD_n^{-1}\mX_{t-1}\}}\big]
\\
-\frac{1}{2}(s-r)\vw^\prime\mK \vw\bigg|=o_{\P}(1).
\end{multline*}
\end{prop}

\begin{proof}
See Appendix~\ref{sec:QRest Lemmas}.
\end{proof}

\begin{proof}[{\textbf{Proof of Theorem~\ref{thm:std est}:}}]
The estimator $\widehat{\valpha}_n(r,s)$ can equivalently be written as
\begin{align*}
\widehat{\valpha}_n(r,s) &= \argmin_{\valpha\in\mathbb{R}^{k+1}} \sum_{t=\lfloor nr\rfloor+1}^{\lfloor ns\rfloor}\big[\rho_{\alpha}(Y_t-\mX_{t-1}^\prime\valpha) - \rho_{\alpha}(\epsilon_{t})\big]\\
&= \argmin_{\valpha\in\mathbb{R}^{k+1}} \sum_{t=\lfloor nr\rfloor+1}^{\lfloor ns\rfloor}\big[\rho_{\alpha}\big(Y_t-\mX_{t-1}^\prime\valpha_0-\mX_{t-1}^\prime(\valpha-\valpha_0)\big) - \rho_{\alpha}(\epsilon_{t})\big]\\
&\overset{\eqref{eq:quantile model}}{=} \argmin_{\valpha\in\mathbb{R}^{k+1}} \sum_{t=\lfloor nr\rfloor+1}^{\lfloor ns\rfloor}\big[\rho_{\alpha}\big(\epsilon_{t} - (\valpha-\valpha_0)^\prime\mX_{t-1}\big) - \rho_{\alpha}(\epsilon_{t})\big].
\end{align*}
Therefore, if we define 
\[
 f_n(\vw,r,s)=\sum_{t=\lfloor nr\rfloor+1}^{\lfloor ns\rfloor}\big[\rho_{\alpha}(\epsilon_{t} - \vw^\prime\mD_n^{-1}\mX_{t-1}) - \rho_{\alpha}(\epsilon_{t})\big],
\] 
then the minimizer $\vw_n(r,s)$ of $f_n(\cdot,r,s)$ satisfies that 
\[
	\vw_n(r,s)=\mD_n\big[\widehat{\valpha}_n(r,s) - \valpha_0\big].
\]
To derive the weak limit of $\vw_n(\cdot,\cdot)$, we invoke Theorem~2 of \citet{Kat09}.
We first rewrite $f_n(\cdot, r,s)$ by exploiting that
\begin{equation}\label{eq:(1)}
	\rho_{\alpha}(u-v) - \rho_{\alpha}(u)=-v\psi_{\alpha}(u) + (u-v)\big[\1_{\{v<u<0\}} - \1_{\{0<u<v\}}\big]\quad\text{for}\  u\neq0.
\end{equation}
Observe that from Assumption~\ref{ass:innov},
\[
	\P\big\{\exists\ t\in\mathbb{N}\colon \epsilon_t=0\big\}=\P\bigg\{\bigcup_{t\in\mathbb{N}}\{ \epsilon_t=0\}\bigg\}\leq\sum_{t\in\mathbb{N}}\P\big\{ \epsilon_t=0\big\}=0.
\]
In light of this and \eqref{eq:(1)}, we may write that almost surely (a.s.)
\begin{align*}
	f_n(\vw,r,s) &= -\vw^\prime\sum_{t=\lfloor nr\rfloor+1}^{\lfloor ns\rfloor}\psi_{\alpha}(\epsilon_{t})\mD_n^{-1}\mX_{t-1}\\
	& \hspace{2cm} + \sum_{t=\lfloor nr\rfloor+1}^{\lfloor ns\rfloor}\big(\epsilon_{t} - \vw^\prime\mD_n^{-1}\mX_{t-1}\big)\big[\1_{\{\vw^\prime\mD_n^{-1}\mX_{t-1}<\epsilon_{t}<0\}} - \1_{\{0<\epsilon_{t}<\vw^\prime\mD_n^{-1}\mX_{t-1}\}}\big].
\end{align*}

Define
\[
	g_n(\vw,r,s)=-\vw^\prime\sum_{t=\lfloor nr\rfloor+1}^{\lfloor ns\rfloor}\psi_{\alpha}(\epsilon_{t})\mD_n^{-1}\mX_{t-1} + \frac{1}{2}\vw^\prime(s-r)\mK\vw.
\]
By Proposition~\ref{lem:LLN}, we obtain that
\[
	\sup_{(r,s)\in\mathcal{D}_{\iota}}\big|f_n(\vw,r,s) - g_n(\vw,r,s)\big|=o_{\P}(1)
\]
for each $\vw\in\mathbb{R}^{k+1}$ (which is the equivalent of equation (10) in \citet{Kat09}).
Moreover, as $n\to\infty$,
\begin{equation}\label{eq:WN1}
	\mW_n(r,s):=\sum_{t=\lfloor nr\rfloor+1}^{\lfloor ns\rfloor}\psi_{\alpha}(\epsilon_{t})\mD_n^{-1}\mX_{t-1}\overset{d}{\longrightarrow}\sqrt{\alpha(1-\alpha)}\mOmega^{1/2}\big[\mW(s)-\mW(r)\big]\quad\text{in }(\ell^{\infty}(\mathcal{D}_{\iota}))^{k+1}
\end{equation}
by Proposition~\ref{lem:CLT}.
Hence, by the continuous mapping theorem (CMT),
\[
	\limsup_{n\to\infty}\P\bigg\{\sup_{(r,s)\in\mathcal{D}_{\iota}}\big\Vert\mW_n(r,s)\big\Vert>M\bigg\}=\P\bigg\{ \sup_{(r,s)\in\mathcal{D}_{\iota}}\Big\Vert\sqrt{\alpha(1-\alpha)}\mOmega^{1/2}\big[\mW(s)-\mW(r)\big]\Big\Vert>M\bigg\},
\]
which can be made arbitrarily small by choosing $M>0$ sufficiently large. (This is the analog of equation (11) in \citet{Kat09}.)

We may now invoke Theorem~2 in \citet{Kat09} to deduce that
\[
	\vw_n(r,s)=\frac{1}{s-r}\mK^{-1}\mW_n(r,s)+\vr_n(r,s),
\]
where $\sup_{(r,s)\in\mathcal{D}_{\iota}}\big\Vert\vr_n(r,s)\big\Vert=o_{\P}(1)$.
Therefore, it holds by \eqref{eq:WN1} that, as $n\to\infty$,
\begin{equation}\label{eq:thm1}
	\vw_n(r,s)\overset{d}{\longrightarrow}\frac{\sqrt{\alpha(1-\alpha)}}{s-r}\mK^{-1}\mOmega^{1/2}\big[\mW(s)-\mW(r)\big]\qquad\text{in }(\ell^{\infty}(\mathcal{D}_{\iota}))^{k+1}.
\end{equation}
Premultiplying this with $(s-r)$, the conclusion follows.
\end{proof}

\section{Proofs of Propositions~\ref{lem:CLT}--\ref{lem:LLN}}\label{sec:QRest Lemmas}

The proofs of Propositions~\ref{lem:CLT}--\ref{lem:LLN} require two preliminary lemmas.
For these, use Assumption~\ref{ass:N} to write
\begin{align}
	\vxi_t &= \mR_n^{t}\vxi_0 + \sum_{j=0}^{t-1}\mR_n^{j} \vu_{t-j}\notag\\
	 &=: \mR_n^{t}\vxi_0 + \vxi_{0t}\label{eq:(B.1m)}\\
	&=: \mR_n^{t}\vxi_0 + n^{\kappa/2}\vxi_{nt},\label{eq:(p.4.decomp)}
\end{align}
where $\vxi_{0t}$ corresponds to a predictor with zero initialization (i.e., $\vxi_0=\vzero$) and $\vxi_{nt}=n^{-\kappa/2}\sum_{j=0}^{t-1}\mR_n^{j} \vu_{t-j}$ denotes what \citet{MP20} call a \textit{normalized near-stationary process} in their Proposition~A1~(ii).
Note that $\kappa=0$ in the (I0) case, such that $\vxi_{0t}=\vxi_{nt}$ under (I0).
To emphasize the dependence of $\vxi_t$ on the starting value $\vxi_0$, we often write $\vxi_t=\vxi_t(\vxi_0)$ and, by relation \eqref{eq:(B.1m)},
\begin{equation}\label{eq:(E.222)}
	\mX_{t}=\mX_{t}(\vmu_x,\vxi_0)=\begin{pmatrix} 1\\ \vmu_{x} + \vxi_{t}(\vxi_0)\end{pmatrix}=\begin{pmatrix}1\\ \vmu_{x} + \mR_n^{t}\vxi_0 + \vxi_{0t}\end{pmatrix}.
\end{equation}

\begin{lem}\label{lem:MAX}
Under Assumptions~\ref{ass:N}--\ref{ass:LP} it holds that, as $n\to\infty$,
\[
	\max_{t=1,\ldots,n}\big\Vert\mD_n^{-1}\mX_{t-1}\big\Vert=o_{\P}(1).
\]
\end{lem}

\begin{proof}
See Appendix~\ref{sec:VaRest Lemmas help}.
\end{proof}

\begin{lem}\label{lem:SUM}
Under Assumptions~\ref{ass:N}--\ref{ass:LP} it holds for every $s\in[0,1]$ that, as $n\to\infty$,
\[
	\sum_{t=1}^{\lfloor ns\rfloor}\mD_n^{-1}\mX_{t-1}(\vmu,\vxi)\mX_{t-1}^\prime(\vmu,\vxi)\mD_n^{-1}\overset{\P}{\longrightarrow}s\mOmega,
\]
where $\mOmega$ is defined in \eqref{eq:Omega}, $\vmu=\vmu_{x}$ under (I0), and for the initialization $\vxi=O_{\P}(1)$ under (I0) and $\vxi=o_{\P}(n^{\kappa/2})$ under (NS).
\end{lem}

\begin{proof}
See Appendix~\ref{sec:VaRest Lemmas help}.
\end{proof}

Sometimes it is necessary to treat the cases of (I0) and (NS) predictors (see Assumption~\ref{ass:N}) separately in the proofs.
The following proof is an example of such an instance.

\begin{proof}[{\textbf{Proof of Proposition~\ref{lem:CLT}:}}]
\textit{\textbf{(NS) case:}} 
Our first goal is to show that, as $n\to\infty$,
\begin{equation}\label{eq:(B.3m)}
	\sum_{t=1}^{\lfloor ns\rfloor}\psi_{\alpha}(\epsilon_t)\mD_n^{-1}\mX_{t-1}\overset{d}{\longrightarrow}\sqrt{\alpha(1-\alpha)}\mOmega^{1/2}\mW(s)\qquad\text{in }(D[0,1])^{k+1}.
\end{equation}
With \eqref{eq:(p.4.decomp)} and $\vx_{t-1}=\vmu_{x} + \vxi_{t-1}$,
\begin{align}
	\psi_{\alpha}(\epsilon_t)\mD_n^{-1}\mX_{t-1} &= \psi_{\alpha}(\epsilon_t)\begin{pmatrix}n^{-1/2}\\ n^{-(1+\kappa)/2}\vx_{t-1}\end{pmatrix}\notag\\
	&= \psi_{\alpha}(\epsilon_t)\begin{pmatrix}0\\ n^{-(1+\kappa)/2}\vmu_{x} + n^{-(1+\kappa)/2}\mR_{n}^{t-1}\vxi_0\end{pmatrix} + \psi_{\alpha}(\epsilon_t)\begin{pmatrix}n^{-1/2}\\ n^{-1/2}\vxi_{n,t-1}\end{pmatrix}.\label{eq:decomp QR}
\end{align}
We consider each of the above terms separately.

Denote by $R_n^{t-1}$ some generic element of the matrix $\mR_{n}^{t-1}$.
Observe that $\big\{\psi_{\alpha}(\epsilon_t)n^{-(1+\kappa)/2}R_n^{t-1}\big\}_{t\in\mathbb{N}}$ is a martingale difference array (MDA), because
\begin{align*}
	\E_{t-1}\big[\psi_{\alpha}(\epsilon_t)n^{-(1+\kappa)/2}R_n^{t-1}\big] &= \E_{t-1}\big[(\alpha-\1_{\{\epsilon_t\leq0\}})\big]n^{-(1+\kappa)/2}R_n^{t-1}\\
	&=\big[\alpha - \P_{t-1}\{\epsilon_t\leq 0\}\big]n^{-(1+\kappa)/2}R_n^{t-1}\\
	&\overset{\eqref{eq:quantile model}}{=}0.
\end{align*}
By a standard maximal inequality for martingales \citep[e.g.,][Theorem~15.14]{Dav94},
\begin{align}
	\P\bigg\{\max_{i=1,\ldots,n}\bigg|\sum_{t=1}^{i}\psi_{\alpha}(\epsilon_t)n^{-(1+\kappa)/2}R_n^{t-1}\bigg|>Kn^{-\kappa/2}\bigg\}&\leq K^{-2}n^{\kappa}\E\bigg[\Big\{\sum_{t=1}^{n}\psi_{\alpha}(\epsilon_t)n^{-(1+\kappa)/2}R_n^{t-1}\Big\}^2\bigg]\notag\\
	&=K^{-2}n^{\kappa}\sum_{t=1}^{n}\E\Big[\psi_{\alpha}^2(\epsilon_t)n^{-(1+\kappa)}\big(R_n^{t-1}\big)^2\Big]\notag\\
	&\leq K^{-2}n^{\kappa}n^{-\kappa}\Big\{\max_{t=1,\ldots,n} \big|R_n^{t-1}\big|\Big\}^2 \E\big[\psi_{\alpha}^2(\epsilon_t)\big]\notag\\
	&=K^{-2}O(1)\alpha(1-\alpha)\notag\\
	&=O(K^{-2}).\label{eq:max ineq}
\end{align}
Here, we have used in the penultimate step that $\Vert\mR_n^{t-1}\Vert\leq\max_{j=1,\ldots,n}\Vert\mR_n^{j}\Vert=O(1)$ \citep[Lemma~2.1~(ii)]{MP20}, and \eqref{eq:(pp.4)}.
We conclude from \eqref{eq:max ineq} that 
\[
	\max_{i=1,\ldots,n}\bigg\Vert\sum_{t=1}^{i}\psi_{\alpha}(\epsilon_t)n^{-(1+\kappa)/2}\mR_n^{t-1}\bigg\Vert=O_{\P}(n^{-\kappa/2}),
\]
such that
\begin{align*}
	\sup_{s\in[0,1]}\bigg\Vert\sum_{t=1}^{\lfloor ns\rfloor}\psi_{\alpha}(\epsilon_t)n^{-(1+\kappa)/2}\mR_n^{t-1}\vxi_0\bigg\Vert& \leq \Vert\vxi_0\Vert \sup_{s\in[0,1]}\bigg\Vert\sum_{t=1}^{\lfloor ns\rfloor}\psi_{\alpha}(\epsilon_t)n^{-(1+\kappa)/2}\mR_n^{t-1}\bigg\Vert\\
	& =o_{\P}(n^{\kappa/2})O_{\P}(n^{-\kappa/2})\\
	& =o_{\P}(1).
\end{align*}
Similar, but simpler, arguments also show that
\[
	\sup_{s\in[0,1]}\bigg\Vert\sum_{t=1}^{\lfloor ns\rfloor}\psi_{\alpha}(\epsilon_t)n^{-(1+\kappa)/2}\vmu_x\bigg\Vert=o_{\P}(1).
\]


In light of \eqref{eq:decomp QR} and the previous two displays, it remains to show that, as $n\to\infty$,
\begin{equation}\label{eq:(B.2m)}
	\sum_{t=1}^{\lfloor ns\rfloor} \vzeta_{nt}:=\sum_{t=1}^{\lfloor ns\rfloor} \psi_{\alpha}(\epsilon_t)\begin{pmatrix}	
	n^{-1/2}\\
	n^{-1/2}\vxi_{n,t-1}
	\end{pmatrix}\overset{d}{\longrightarrow}\sqrt{\alpha(1-\alpha)}\mOmega^{1/2}\mW(s)\qquad\text{in }(D[0,1])^{k+1}.
\end{equation}
Note that $\{\vzeta_{nt}\}_{t=1,\ldots,n;n\in\mathbb{N}}$ is a vector MDA because
\[
	\E_{t-1}\big[\vzeta_{nt}\big]=\begin{pmatrix}n^{-1/2}\\ n^{-1/2}\vxi_{n,t-1}\end{pmatrix}\E_{t-1}\big[\psi_{\alpha}(\epsilon_t)\big]\overset{\eqref{eq:quantile model}}{=}\vzero.
\]
Therefore, we only have to verify the conditions of the functional central limit theorem (FCLT) for vector MDAs in \citet[Theorem~3.33 in Chapter VIII]{JS87}.
First, the conditional variances converge as follows:
\begin{align*}
	\sum_{t=1}^{\lfloor ns\rfloor}\E_{t-1}\big[\vzeta_{nt}\vzeta_{nt}^\prime\big] &\overset{\eqref{eq:(E.222)}}{=} \sum_{t=1}^{\lfloor ns\rfloor} \E_{t-1}\big[\psi_{\alpha}^2(\epsilon_t)\big]\mD_n^{-1}\mX_{t-1}(\vzero,\vzero)\mX_{t-1}^\prime(\vzero,\vzero)\mD_n^{-1}\\
	&\overset{\eqref{eq:(pp.4)}}{=} \alpha(1-\alpha)\sum_{t=1}^{\lfloor ns\rfloor} \mD_n^{-1}\mX_{t-1}(\vzero,\vzero)\mX_{t-1}^\prime(\vzero,\vzero)\mD_n^{-1}\\
	&\overset{\P}{\longrightarrow}s\alpha(1-\alpha)\mOmega,
\end{align*}
where the last line follows from Lemma~\ref{lem:SUM}.

The second condition to verify for the FCLT is the conditional Lindeberg condition (CLC):
\[
		\sum_{t=1}^{n} \E_{t-1}\Big[\Vert\vzeta_{nt}\Vert^2 \1_{\{\Vert\vzeta_{nt}\Vert^2>\delta^2\}}\Big]\overset{\P}{\underset{(n\to\infty)}{\longrightarrow}}0
\]
for any $\delta>0$.
To do so, it suffices to verify the (unconditional) Lindeberg condition (LC)
\[
		\sum_{t=1}^{n} \E\Big[\Vert\vzeta_{nt}\Vert^2 \1_{\{\Vert\vzeta_{nt}\Vert^2>\delta^2\}}\Big]\underset{(n\to\infty)}{\longrightarrow}0,
\]
because then, by Markov's inequality and the law of iterated expectations,
\begin{align*}
	\P\bigg\{ \sum_{t=1}^{n}\E_{t-1}\Big[\Vert\vzeta_{nt}\Vert^2 \1_{\{\Vert\vzeta_{nt}\Vert^2>\delta^2\}}\Big]>\varepsilon\bigg\} &\leq \frac{1}{\varepsilon}\E\bigg[\sum_{t=1}^{n}\E_{t-1}\big\{\Vert\vzeta_{nt}\Vert^2 \1_{\{\Vert\vzeta_{nt}\Vert^2>\delta^2\}}\big\}\bigg]\\
	&=\frac{1}{\varepsilon}\sum_{t=1}^{n}\E\Big[\Vert\vzeta_{nt}\Vert^2 \1_{\{\Vert\vzeta_{nt}\Vert^2>\delta^2\}}\Big]\underset{(n\to\infty)}{\longrightarrow}0,
\end{align*}
such that the CLC is implied by the LC.
Since $|\psi_{\alpha}(\epsilon_t)|\leq 1$, we get that
\begin{align*}
	\Vert\vzeta_{nt}\Vert^2 &\leq \Bigg\Vert\begin{pmatrix} n^{-1/2}\\ n^{-1/2}\vxi_{n,t-1}\end{pmatrix}\Bigg\Vert^2\\
	&= \Bigg\Vert\begin{pmatrix} n^{-1/2}\\ \vzero \end{pmatrix} + \begin{pmatrix} 0\\ n^{-1/2}\vxi_{n,t-1}\end{pmatrix}\Bigg\Vert^2\\
	&\leq K\max\big\{n^{-1}, n^{-1}\Vert\vxi_{n,t-1}\Vert^2\big\}.
\end{align*}
Hence, it is sufficient to verify the LC separately for $n^{-1}$ and $n^{-1}\Vert\vxi_{n,t-1}\Vert^2$.
For the deterministic quantity this is immediate because
\[
	\sum_{t=1}^{n}\E\big[n^{-1}\1_{\{n^{-1}>\delta^2\}}\big]=0
\]
for sufficiently large $n$.
The LC for $n^{-1}\Vert\vxi_{n,t-1}\Vert^2$ follows from Proposition~A1~(ii) in \citet{MP20} as follows:
\begin{equation*}
\sum_{t=1}^{n} \E\bigg[\frac{1}{n}\Vert\vxi_{n,t-1}\Vert^2 \1_{\{\frac{1}{n}\Vert\vxi_{n,t-1}\Vert^2>\delta^2\}}\bigg]\leq \max_{t=1,\ldots,n}\E\Big[\Vert\vxi_{n,t-1}\Vert^2\1_{\{\Vert\vxi_{n,t-1}\Vert^2>\delta^2n\}}\Big]\underset{(n\to\infty)}{\longrightarrow}0.	
\end{equation*}

Now, Theorem~3.33 of \citet[Chapter VIII]{JS87} implies the existence of a continuous Gaussian martingale $\vzeta(s)$ with quadratic variation
\[
	\langle\vzeta\rangle_s=s\alpha(1-\alpha)\mOmega,
\]
such that $\sum_{t=1}^{\lfloor ns\rfloor}\vzeta_{nt}\overset{d}{\longrightarrow}\vzeta(s)$ in $(D[0,1])^{k+1}$.
By Levy's characterization of Brownian motion \citep[e.g.,][Theorem~4.4 in Chapter II]{JS87}, $\vzeta(\cdot)$ is a Brownian motion with covariance matrix $\alpha(1-\alpha)\mOmega$. 
This establishes \eqref{eq:(B.2m)}, concluding the proof of \eqref{eq:(B.3m)}.

Since convergence in \eqref{eq:(B.3m)} is to a continuous limit, it also holds in $(D[0,1])^{k+1}$ equipped with the uniform (product) topology \citep{Pol84} instead of the usual Skorohod topology.
Therefore, the two right-hand side terms in
\[
	\sum_{t=\lfloor nr\rfloor+1}^{\lfloor ns\rfloor}\psi_{\alpha}(\epsilon_t)\mD_n^{-1}\mX_{t-1}=\sum_{t=1}^{\lfloor ns\rfloor}\psi_{\alpha}(\epsilon_t)\mD_n^{-1}\mX_{t-1}-\sum_{t=1}^{\lfloor nr\rfloor}\psi_{\alpha}(\epsilon_t)\mD_n^{-1}\mX_{t-1}
\]
converge in $(D[0,1])^{k+1}$ endowed with the uniform (product) topology.
The claim of the lemma in the (NS) case thus follows from standard arguments in, e.g., \citet[Sec.~4.2]{VS14}.

\textit{\textbf{(I0) case:}}
The proof strategy is similar as in the (NS) case.
Recall from \eqref{eq:(B.1m)} that $\vx_{t-1}=\vmu_{x} + \vxi_{t-1}=\vmu_x + \mR^{t-1}\vxi_0 + \vxi_{0,t-1}$, such that
\begin{align*}
	\psi_{\alpha}(\epsilon_t)\mD_n^{-1}\mX_{t-1} &= \psi_{\alpha}(\epsilon_t)\begin{pmatrix}0\\ n^{-1/2}\mR^{t-1}\vxi_0\end{pmatrix} + \psi_{\alpha}(\epsilon_t)\begin{pmatrix}n^{-1/2}\\ n^{-1/2}(\vmu_{x} + \vxi_{0,t-1})\end{pmatrix}.
\end{align*}
Since $|\psi_{\alpha}(\epsilon_t)|\leq 1$,
\[
	\sup_{s\in[0,1]}\Bigg\Vert\sum_{t=1}^{\lfloor ns\rfloor}\psi_{\alpha}(\epsilon_t)\begin{pmatrix}0\\ n^{-1/2}\mR^{t-1}\vxi_0\end{pmatrix}\Bigg\Vert \leq \Vert\vxi_0\Vert n^{-1/2}\sum_{t=1}^{n}\Vert\mR^{t-1}\Vert=O_{\P}(1)n^{-1/2}O(1)=o_{\P}(1),
\]
where $\sum_{t=1}^{n}\Vert\mR^{t-1}\Vert=O(1)$ follows from $\rho(\mR)<1$ and Rule (3) in \citet[p.~657]{Lüt05}.

In light of this, we only have to show that, as $n\to\infty$,
\[
	\sum_{t=1}^{\lfloor ns\rfloor} \vzeta_{nt}:=\sum_{t=1}^{\lfloor ns\rfloor}\psi_{\alpha}(\epsilon_t)\begin{pmatrix}n^{-1/2}\\ n^{-1/2}(\vmu_{x} + \vxi_{0,t-1})\end{pmatrix}\overset{d}{\longrightarrow}\sqrt{\alpha(1-\alpha)}\mOmega^{1/2}\mW(s)\qquad\text{in }(D[0,1])^{k+1}.
\]
(Note that we slightly overload notation here by redefining $\vzeta_{nt}$. We do so to highlight the similarities to the proof in the (NS) case.)
We again verify the conditions of the FCLT in Theorem~3.33 of \citet[Chapter~VIII]{JS87}.
First, by Lemma~\ref{lem:SUM} and \eqref{eq:(pp.4)} the conditional variances converge as follows:
\begin{align*}
	\sum_{t=1}^{\lfloor ns\rfloor}\E_{t-1}[\vzeta_{nt}\vzeta_{nt}^\prime] &= \sum_{t=1}^{\lfloor ns\rfloor} \E_{t-1}[\psi_{\alpha}^2(\epsilon_t)]\mD_n^{-1}\mX_{t-1}(\vmu_x,\vzero)\mX_{t-1}^{\prime}(\vmu_x,\vzero)\mD_n^{-1}\\
	& \overset{\P}{\longrightarrow}s\alpha(1-\alpha)\mOmega.
\end{align*}

Second, to verify the CLC we check the unconditional LC
\[
	\sum_{t=1}^{n}\E\Big[\Vert\vzeta_{nt}\Vert^2\1_{\{\Vert\vzeta_{nt}\Vert^2>\delta^2\}}\Big]\underset{(n\to\infty)}{\longrightarrow}0.
\]
Since $|\psi_{\alpha}(\epsilon_t)|\leq 1$, we obtain from standard norm inequalities that
\begin{align*}
	\Vert\vzeta_{nt}\Vert^2 &= \Bigg\Vert \begin{pmatrix}n^{-1/2}\\ \vzero\end{pmatrix} + \begin{pmatrix}0\\ n^{-1/2}\vmu_{x}\end{pmatrix} + \begin{pmatrix}0\\ n^{-1/2}\vxi_{0,t-1}\end{pmatrix}\Bigg\Vert^2\\
	&\leq K\big\{ n^{-1} + n^{-1}\Vert\vmu_{x}\Vert^2 + n^{-1}\Vert\vxi_{0,t-1}\Vert^2\big\}\\
	&\leq K\max\big\{n^{-1}, n^{-1}\Vert\vmu_{x}\Vert^2, n^{-1}\Vert\vxi_{0,t-1}\Vert^2\big\}.
\end{align*}
Due to this, we can verify the LC separately for each term in the maximum.
For the deterministic terms, the LC is immediate.
For the only stochastic quantity,
\[
	\sum_{t=1}^{n}\E\Big[\frac{1}{n}\Vert\vxi_{0,t-1}\Vert^2\1_{\{n^{-1}\Vert\vxi_{0,t-1}\Vert^2>\delta^2\}}\Big]\leq \max_{t=1,\ldots,n}\E\Big[\Vert\vxi_{0,t-1}\Vert^2\1_{\{\Vert\vxi_{0,t-1}\Vert^2>\delta^2 n \}}\Big]\underset{(n\to\infty)}{\longrightarrow}0
\]
by Proposition~A1~(ii) of \citet{MP20}.
The remainder of the proof now follows as in the (NS) case.
\end{proof}

\begin{proof}[{\textbf{Proof of Proposition~\ref{lem:LLN}:}}]
Since
\begin{align}
	&\sup_{0\leq r< s\leq 1}\bigg|\sum_{t=\lfloor nr\rfloor+1}^{\lfloor ns\rfloor}\big(\epsilon_{t} - \vw^\prime\mD_n^{-1}\mX_{t-1}\big)\big[\1_{\{\vw^\prime\mD_n^{-1}\mX_{t-1}<\epsilon_{t}<0\}} - \1_{\{0<\epsilon_{t}<\vw^\prime\mD_n^{-1}\mX_{t-1}\}}\big]
-\frac{1}{2}(s-r)\vw^\prime\mK \vw\bigg|\notag\\
& =\sup_{0\leq r< s\leq 1}\bigg|\sum_{t=1}^{\lfloor ns\rfloor}\big(\epsilon_{t} - \vw^\prime\mD_n^{-1}\mX_{t-1}\big)\big[\1_{\{\vw^\prime\mD_n^{-1}\mX_{t-1}<\epsilon_{t}<0\}} - \1_{\{0<\epsilon_{t}<\vw^\prime\mD_n^{-1}\mX_{t-1}\}}\big]
-\frac{1}{2}s\vw^\prime\mK \vw\notag\\
&\hspace{1.4cm} - \bigg\{\sum_{t=1}^{\lfloor nr\rfloor}\big(\epsilon_{t} - \vw^\prime\mD_n^{-1}\mX_{t-1}\big)\big[\1_{\{\vw^\prime\mD_n^{-1}\mX_{t-1}<\epsilon_{t}<0\}} - \1_{\{0<\epsilon_{t}<\vw^\prime\mD_n^{-1}\mX_{t-1}\}}\big]-\frac{1}{2}r\vw^\prime\mK \vw\bigg\}\bigg|\notag\\
&\leq 2\sup_{0\leq s\leq 1}\bigg|\sum_{t=1}^{\lfloor ns\rfloor}\big(\epsilon_{t} - \vw^\prime\mD_n^{-1}\mX_{t-1}\big)\big[\1_{\{\vw^\prime\mD_n^{-1}\mX_{t-1}<\epsilon_{t}<0\}} - \1_{\{0<\epsilon_{t}<\vw^\prime\mD_n^{-1}\mX_{t-1}\}}\big]
-\frac{1}{2}s\vw^\prime\mK \vw\bigg|,\label{eq:(p.9)}
\end{align}
it suffices to show that
\begin{equation}\label{eq:(SC)}
	\sup_{0\leq s\leq 1}\bigg|\sum_{t=1}^{\lfloor ns\rfloor}\big(\epsilon_{t} - \vw^\prime\mD_n^{-1}\mX_{t-1}\big)\big[\1_{\{\vw^\prime\mD_n^{-1}\mX_{t-1}<\epsilon_{t}<0\}} - \1_{\{0<\epsilon_{t}<\vw^\prime\mD_n^{-1}\mX_{t-1}\}}\big]
-\frac{1}{2}s\vw^\prime\mK \vw\bigg|=o_{\P}(1).
\end{equation}
To do so, define
\begin{align*}
	\nu_{t}(\vw) &:= (\epsilon_{t} - \vw^\prime\mD_n^{-1}\mX_{t-1})\big(\1_{\{\vw^\prime\mD_n^{-1}\mX_{t-1}< \epsilon_{t}<0\}} - \1_{\{0<\epsilon_{t}< \vw^\prime\mD_n^{-1}\mX_{t-1}\}}\big),\\
	\overline{\nu}_{t}(\vw) &:= \E_{t-1}\Big[(\epsilon_{t} - \vw^\prime\mD_n^{-1}\mX_{t-1})\big(\1_{\{\vw^\prime\mD_n^{-1}\mX_{t-1}< \epsilon_{t}<0\}} - \1_{\{0<\epsilon_{t}< \vw^\prime\mD_n^{-1}\mX_{t-1}\}}\big)\Big],\\
	V_n(\vw,s) &:= \sum_{t=1}^{\lfloor ns\rfloor}\nu_t(\vw),\\
	\overline{V}_n(\vw,s) &:= \sum_{t=1}^{\lfloor ns\rfloor}\overline{\nu}_t(\vw).
\end{align*}
We establish \eqref{eq:(SC)} by showing that, uniformly in $s\in[0,1]$,
\begin{align}
	\overline{V}_n(\vw,s) &\overset{\P}{\longrightarrow}\frac{1}{2}s\vw^\prime\mK\vw, \label{eq:(P.7.0)}\\
	V_n(\vw,s) - \overline{V}_n(\vw,s) &= o_{\P}(1).  \label{eq:(P.7.1)}
\end{align}

We first prove \eqref{eq:(P.7.0)}.
It suffices to show that the convergence holds on the set 
\[
	\Big\{\max_{t=1,\ldots,n}|\vw^\prime\mD_n^{-1}\mX_{t-1}|\leq d\Big\}
\]
with $d>0$ from Assumption~\ref{ass:innov}~\eqref{it:dens}.
This is because, as $n\to\infty$,
\[
	\P\Big\{\max_{t=1,\ldots,n}|\vw^\prime\mD_n^{-1}\mX_{t-1}|\leq d\Big\}\longrightarrow1
\]
by Lemma~\ref{lem:MAX}.
On this set, we may exploit Assumption~\ref{ass:innov}~\eqref{it:dens} to deduce that
\begin{align*}
	\overline{V}_n(\vw,s) &= \sum_{t=1}^{\lfloor ns\rfloor}\E_{t-1}\Big[(\epsilon_{t} - \vw^\prime\mD_n^{-1}\mX_{t-1})\1_{\{\vw^\prime\mD_n^{-1}\mX_{t-1}< \epsilon_{t}<0\}}\Big]\\
	&\hspace{1cm} + \sum_{t=1}^{\lfloor ns\rfloor}\E_{t-1}\Big[(\vw^\prime\mD_n^{-1}\mX_{t-1} - \epsilon_{t}) \1_{\{0<\epsilon_{t}< \vw^\prime\mD_n^{-1}\mX_{t-1}\}}\Big]\\
	&= \sum_{t=1}^{\lfloor ns\rfloor}\1_{\{\vw^\prime\mD_n^{-1}\mX_{t-1}<0\}}\int_{\vw^\prime\mD_n^{-1}\mX_{t-1}}^{0}(x-\vw^\prime\mD_n^{-1}\mX_{t-1})f_{\epsilon_t\mid\mathcal{F}_{t-1}}(x)\D x\\
	&\hspace{1cm} +\sum_{t=1}^{\lfloor ns\rfloor}\1_{\{\vw^\prime\mD_n^{-1}\mX_{t-1}>0\}}\int_{0}^{\vw^\prime\mD_n^{-1}\mX_{t-1}}(\vw^\prime\mD_n^{-1}\mX_{t-1} - x)f_{\epsilon_t\mid\mathcal{F}_{t-1}}(x)\D x\\
	&= \sum_{t=1}^{\lfloor ns\rfloor}\1_{\{\vw^\prime\mD_n^{-1}\mX_{t-1}<0\}}\int_{\vw^\prime\mD_n^{-1}\mX_{t-1}}^{0}\big[F_{\epsilon_t\mid\mathcal{F}_{t-1}}(0) - F_{\epsilon_t\mid\mathcal{F}_{t-1}}(x)\big]\D x\\
	&\hspace{1cm} +\sum_{t=1}^{\lfloor ns\rfloor}\1_{\{\vw^\prime\mD_n^{-1}\mX_{t-1}>0\}}\int_{0}^{\vw^\prime\mD_n^{-1}\mX_{t-1}}\big[F_{\epsilon_t\mid\mathcal{F}_{t-1}}(x) - F_{\epsilon_t\mid\mathcal{F}_{t-1}}(0)\big]\D x,
\end{align*}
where the final line follows from integration by parts, and $F_{\epsilon_t\mid\mathcal{F}_{t-1}}(\cdot)$ denotes the cumulative distribution function of $\epsilon_t\mid\mathcal{F}_{t-1}$.
By the mean value theorem it holds for some $x^\ast=x^\ast(x)$ between $0$ and $x$ that
\begin{align}
 \frac{F_{\epsilon_{t}\mid\mathcal{F}_{t-1}}(x)-F_{\epsilon_{t}\mid\mathcal{F}_{t-1}}(0)}{x-0} &= f_{\epsilon_{t}\mid\mathcal{F}_{t-1}}(x^\ast)\notag\\
&= f_{\epsilon_{t}\mid\mathcal{F}_{t-1}}(0) + \big[f_{\epsilon_{t}\mid\mathcal{F}_{t-1}}(x^\ast) - f_{\epsilon_{t}\mid\mathcal{F}_{t-1}}(0)\big].\label{eq:MVT}
\end{align}
Plugging this into the previous display yields that
\begin{align*}
	\overline{V}_n(\vw,s) &=\sum_{t=1}^{\lfloor ns\rfloor}\bigg\{-\1_{\{\vw^\prime\mD_n^{-1}\mX_{t-1}<0\}}\int_{\vw^\prime\mD_n^{-1}\mX_{t-1}}^{0}xf_{\epsilon_{t}\mid\mathcal{F}_{t-1}}(0)\D x\\
	&\hspace{2cm} + \1_{\{\vw^\prime\mD_n^{-1}\mX_{t-1}>0\}}\int_{0}^{\vw^\prime\mD_n^{-1}\mX_{t-1}}xf_{\epsilon_{t}\mid\mathcal{F}_{t-1}}(0)\D x\bigg\}\\
	&\hspace{1cm} + \sum_{t=1}^{\lfloor ns\rfloor}\bigg\{-\1_{\{\vw^\prime\mD_n^{-1}\mX_{t-1}<0\}}\int_{\vw^\prime\mD_n^{-1}\mX_{t-1}}^{0}x\big[f_{\epsilon_{t}\mid\mathcal{F}_{t-1}}(x^\ast) - f_{\epsilon_{t}\mid\mathcal{F}_{t-1}}(0)\big]\D x\\
	&\hspace{2cm} + \1_{\{\vw^\prime\mD_n^{-1}\mX_{t-1}>0\}}\int_{0}^{\vw^\prime\mD_n^{-1}\mX_{t-1}}x\big[f_{\epsilon_{t}\mid\mathcal{F}_{t-1}}(x^\ast) - f_{\epsilon_{t}\mid\mathcal{F}_{t-1}}(0)\big]\D x\bigg\}\\
	&=:\overline{V}_{1n}(\vw,s) + \overline{V}_{2n}(\vw,s).
\end{align*}
The convergence in \eqref{eq:(P.7.0)} is established if we can prove that, uniformly in $s\in[0,1]$,
\begin{align}
	\overline{V}_{1n}(\vw,s) & \overset{\P}{\longrightarrow}\frac{1}{2}s\vw^\prime\mK\vw,\label{eq:(P.8.0)}\\
	\overline{V}_{2n}(\vw,s) &=o_{\P}(1). \label{eq:(P.8.1)}
\end{align}

To show \eqref{eq:(P.8.1)}, use Assumption~\ref{ass:innov}~\eqref{it:Lipschitz} to deduce that
\begin{align*}
	\big|\overline{V}_{2n}(\vw,s)\big| &\leq 2\sum_{t=1}^{n}\int_{0}^{|\vw^\prime\mD_n^{-1}\mX_{t-1}|}xL|\vw^\prime\mD_n^{-1}\mX_{t-1}|\D x\\
	 &\leq 2L\sum_{t=1}^{n}|\vw^\prime\mD_n^{-1}\mX_{t-1}|\int_{0}^{|\vw^\prime\mD_n^{-1}\mX_{t-1}|}x\D x\\
	 &\leq 2L\Vert\vw\Vert \sum_{t=1}^{n} \big\Vert\mD_n^{-1}\mX_{t-1}\big\Vert \frac{1}{2}\vw^\prime\mD_n^{-1}\mX_{t-1}\mX_{t-1}^\prime\mD_n^{-1}\vw\\
	&\leq L\Vert\vw\Vert \max_{t=1,\ldots,n} \big\Vert\mD_n^{-1}\mX_{t-1}\big\Vert\vw^\prime\bigg(\sum_{t=1}^{n}\mD_n^{-1}\mX_{t-1}\mX_{t-1}^\prime\mD_n^{-1}\bigg)\vw\\
	&=o_{\P}(1)O_{\P}(1)\\
	&= o_{\P}(1)
\end{align*}
uniformly in $s\in[0,1]$, where the penultimate step uses Lemmas~\ref{lem:MAX} and \ref{lem:SUM}.

To prove \eqref{eq:(P.8.0)}, use Assumption~\ref{ass:K} to conclude that
\begin{align*}
	\overline{V}_{1n}(\vw,s) &=\sum_{t=1}^{\lfloor ns\rfloor}\bigg\{\1_{\{\vw^\prime\mD_n^{-1}\mX_{t-1}<0\}}f_{\epsilon_{t}\mid\mathcal{F}_{t-1}}(0)\frac{1}{2}(\vw^\prime\mD_n^{-1}\mX_{t-1})^2\\
	&\hspace{2cm} + \1_{\{\vw^\prime\mD_n^{-1}\mX_{t-1}>0\}}f_{\epsilon_{t}\mid\mathcal{F}_{t-1}}(0)\frac{1}{2}(\vw^\prime\mD_n^{-1}\mX_{t-1})^2\bigg\}\\
	&=\frac{1}{2}\vw^\prime\bigg(\sum_{t=1}^{\lfloor ns\rfloor}f_{\epsilon_{t}\mid\mathcal{F}_{t-1}}(0)\mD_n^{-1}\mX_{t-1}\mX^\prime_{t-1}\mD_n^{-1}\bigg)\vw\\
	&\overset{\P}{\longrightarrow}\frac{1}{2}s\vw^\prime\mK\vw
\end{align*}
uniformly in $s\in[0,1]$.
This proves \eqref{eq:(P.8.0)}, whence \eqref{eq:(P.7.0)} follows.

It remains to show \eqref{eq:(P.7.1)}.
Observe that
\[
	V_n(\vw,s) - \overline{V}_n(\vw,s) = \sum_{t=1}^{\lfloor ns\rfloor}\big[\nu_t(\vw) - \overline{\nu}_t(\vw)\big]
\]
is a sum of MDAs.
Our goal is to invoke Theorem~3.33 of \citet[Chapter VIII]{JS87} once more.
First, the conditional variances converge as follows:
\begin{align*}
\sum_{t=1}^{\lfloor ns\rfloor}&\E_{t-1}\big[\{\nu_t(\vw) - \overline{\nu}_t(\vw)\}^2\big]\\
	&= \sum_{t=1}^{\lfloor ns\rfloor}\Big\{\E_{t-1}\big[\nu_t^2(\vw)\big] - \overline{\nu}^2_t(\vw)\Big\}\\
	&\leq \sum_{t=1}^{\lfloor ns\rfloor}\E_{t-1}[\nu_t^2(\vw)]\\
	&= \sum_{t=1}^{\lfloor ns\rfloor}\E_{t-1}\Big[(\epsilon_{t} - \vw^\prime\mD_n^{-1}\mX_{t-1})^2\big(\1_{\{\vw^\prime\mD_n^{-1}\mX_{t-1}< \epsilon_{t}<0\}} - \1_{\{0<\epsilon_{t}<\vw^\prime\mD_n^{-1}\mX_{t-1}\}}\big)^2\Big]\\
	&= \sum_{t=1}^{\lfloor ns\rfloor}\E_{t-1}\Big[(\epsilon_{t} - \vw^\prime\mD_n^{-1}\mX_{t-1})^2\big(\1_{\{\vw^\prime\mD_n^{-1}\mX_{t-1}< \epsilon_{t}<0\}} + \1_{\{0<\epsilon_{t}<\vw^\prime\mD_n^{-1}\mX_{t-1}\}}\big)\Big]\\
	&\leq \sum_{t=1}^{\lfloor ns\rfloor}|\vw^\prime\mD_n^{-1}\mX_{t-1}|\times\\
	&\hspace{1cm}\times\E_{t-1}\Big[|\epsilon_{t} - \vw^\prime\mD_n^{-1}\mX_{t-1}|\big(\1_{\{\vw^\prime\mD_n^{-1}\mX_{t-1}< \epsilon_{t}<0\}} + \1_{\{0<\epsilon_{t}<\vw^\prime\mD_n^{-1}\mX_{t-1}\}}\big)\Big]\\
	&\leq K \max_{t=1,\ldots,n}\Vert\mD_n^{-1}\mX_{t-1}\Vert\times\\
	&\hspace{1cm}\times\sum_{t=1}^{\lfloor ns\rfloor}\E_{t-1}\Big[(\epsilon_{t} - \vw^\prime\mD_n^{-1}\mX_{t-1})\big(\1_{\{\vw^\prime\mD_n^{-1}\mX_{t-1}< \epsilon_{t}<0\}} - \1_{\{0<\epsilon_{t}<\vw^\prime\mD_n^{-1}\mX_{t-1}\}}\big)\Big]\\
	&=K \max_{t=1,\ldots,n}\Vert\mD_n^{-1}\mX_{t-1}\Vert \sum_{t=1}^{\lfloor ns\rfloor}\overline{\nu}_t(\vw)\\
	&=o_{\P}(1)O_{\P}(1)=o_{\P}(1),
\end{align*}
where we used that $\max_{t=1,\ldots,n}\Vert\mD_n^{-1}\mX_{t-1}\Vert=o_{\P}(1)$ (from Lemma~\ref{lem:MAX}) and $\overline{V}_n(\vw,s)=\sum_{t=1}^{\lfloor ns\rfloor}\overline{\nu}_t(\vw)=O_{\P}(1)$ (from \eqref{eq:(P.7.0)}).

Second, to verify the CLC
\[
	\sum_{t=1}^{n}\E_{t-1}\Big[\big|\nu_t(\vw)-\overline{\nu}_t(\vw)\big|^2\1_{\{|\vnu_t(\vw) - \overline{\vnu}_t(\vw)|^2>\delta^2\}}\Big]\overset{\P}{\underset{(n\to\infty)}{\longrightarrow}}0,
\]
it is once again sufficient to verify the LC
\[
	\sum_{t=1}^{n}\E\Big[\big|\nu_t(\vw)-\overline{\nu}_t(\vw)\big|^2\1_{\{|\vnu_t(\vw) - \overline{\vnu}_t(\vw)|^2>\delta^2\}}\Big]\underset{(n\to\infty)}{\longrightarrow}0.
\]
Since 
\begin{align}
	\big|\nu_t(\vw)-\overline{\nu}_t(\vw)\big|^2 &\leq K\Vert\mD_n^{-1}\mX_{t-1}\Vert^2 = K\Bigg\Vert\begin{pmatrix}n^{-1/2}\\ n^{-(1+\kappa)/2}\vx_{t-1}\end{pmatrix}\Bigg\Vert^2\notag\\
	&\leq \frac{K}{n} + Kn^{-(1+\kappa)}\Vert\vx_{t-1}\Vert^2\notag\\
	&\leq 2K\max\big\{n^{-1},n^{-(1+\kappa)}\Vert\vx_{t-1}\Vert^2\big\},\label{eq:(p.12.LC)}
\end{align}
we can check the LC separately for $n^{-1}$ and $n^{-(1+\kappa)}\Vert\vx_{t-1}\Vert^2$.
(Here, $\kappa=0$ corresponds to the (I0) case.)
This, however, follows as in the proof of Proposition~\ref{lem:CLT} via Proposition~A1~(ii) of \citet{MP20}.
Hence, the CLC is satisfied.

Overall, we may apply Theorem~3.33 of \citet[Chapter~VIII]{JS87} to conclude that
\[
	V_n(\vw,s) - \overline{V}_n(\vw,s) = \sum_{t=1}^{\lfloor ns\rfloor}\big[\nu_t(\vw) - \overline{\nu}_t(\vw)\big]\overset{d}{\longrightarrow}0\qquad\text{in }D[0,1].
\]
Applying the continuous mapping theorem (CMT) to this shows that $\sup_{s\in[0,1]}\big|V_n(\vw,s) - \overline{V}_n(\vw,s)\big|=o_{\P}(1)$, which establishes \eqref{eq:(P.7.1)}.
\end{proof}

\subsection{Proofs of Lemmas~\ref{lem:MAX}--\ref{lem:SUM}}\label{sec:VaRest Lemmas help}

\begin{proof}[{\textbf{Proof of Lemma~\ref{lem:MAX}:}}]
We only prove the result in the (NS) case, because the (I0) case follows similarly.
From \eqref{eq:(p.4.decomp)},
\begin{align*}
	\mD_n^{-1}\mX_{t-1} &= \begin{pmatrix}n^{-1/2}\\ n^{-(1+\kappa)/2}\vx_{t-1}\end{pmatrix}\\
	&= \begin{pmatrix}n^{-1/2}\\ n^{-(1+\kappa)/2}\vmu_x + n^{-(1+\kappa)/2}\mR_n^{t-1}\vxi_0 + n^{-1/2}\vxi_{n,t-1}\end{pmatrix}.
\end{align*}
In light of this, it suffices to prove that
\begin{align}
	\max_{t=1,\ldots,n} \Vert n^{-(1+\kappa)/2}\mR_n^{t-1}\vxi_0\Vert =o_{\P}(1),\label{eq:(p.6.1)}\\
	\max_{t=1,\ldots,n} \Vert n^{-1/2}\vxi_{n,t-1}\Vert =o_{\P}(1).\label{eq:(p.6.2)}
\end{align}
For \eqref{eq:(p.6.1)}, submultiplicativity implies
\begin{align*}
	\max_{t=1,\ldots,n} \big\Vert n^{-(1+\kappa)/2}\mR_n^{t-1}\vxi_0\big\Vert &\leq n^{-(1+\kappa)/2}\Big(\max_{t=1,\ldots,n}\big\Vert\mR_n^{t-1}\big\Vert\Big)\big\Vert\vxi_0\big\Vert\\
	&=n^{-(1+\kappa)/2}O(1)o_{\P}(n^{\kappa/2})\\
	&=o_{\P}(1),
\end{align*}
where $\max_{t=1,\ldots,n}\Vert\mR_n^{t-1}\Vert=O(1)$ follows from Lemma~2.1~(ii) in \citet{MP20}.
To prove \eqref{eq:(p.6.2)}, note that for arbitrary $\delta>0$,
\begin{align*}
	\P\Big\{\max_{t=1,\ldots,n}\Vert n^{-1/2}\vxi_{n,t-1}\Vert>\varepsilon\Big\}&\leq \frac{1}{\varepsilon^2}\E\Big[\max_{t=1,\ldots,n}\Vert n^{-1/2}\vxi_{n,t-1}\Vert^2\Big]\\
	&\leq \frac{1}{\varepsilon^2}\bigg\{\delta^2 + \E\Big[\max_{t=1,\ldots,n}\Vert n^{-1/2}\vxi_{n,t-1}\Vert^2 \1_{\{\Vert n^{-1/2}\vxi_{n,t-1}\Vert^2>\delta^2\}}\Big]\bigg\}\\
	&\leq \frac{\delta^2}{\varepsilon^2} + \frac{1}{\varepsilon^2}\E\bigg[\sum_{t=1}^{n}\Vert n^{-1/2}\vxi_{n,t-1}\Vert^2 \1_{\{\Vert n^{-1/2}\vxi_{n,t-1}\Vert^2>\delta^2\}}\bigg]\\
	&\leq \frac{\delta^2}{\varepsilon^2} + \frac{1}{\varepsilon^2}\max_{t=1,\ldots,n}\E\Big[\Vert\vxi_{n,t-1}\Vert^2 \1_{\{\Vert \vxi_{n,t-1}\Vert^2>\delta^2 n\}}\Big]\\
	&\underset{(n\to\infty)}{\longrightarrow}\frac{\delta^2}{\varepsilon^2},
\end{align*}
where the first step uses Markov's inequality and the final step Proposition~A1~(ii) of \citet{MP20}.
Since $\delta>0$ was arbitrary, $\max_{t=1,\ldots,n}\Vert n^{-1/2}\vxi_{n,t-1}\Vert=o_{\P}(1)$ follows.
\end{proof}

\begin{proof}[{\textbf{Proof of Lemma~\ref{lem:SUM}:}}]
To avoid notational clutter in this proof, we slightly abuse notation and write $\vxi_{t-1}=\vxi_{t-1}(\vxi)$.
However, keep in mind that $\vxi_{t-1}$ as defined in \eqref{eq:(p.4.decomp)} actually pertains to the \textit{true} starting value $\vxi_0$ (such that, in fact, $\vxi_{t-1}=\vxi_{t-1}(\vxi_0)$).

\textit{\textbf{(NS) case:}} It follows from the proof of Lemma~2.2~(ii) in \citet{MP20} that, as $n\to\infty$,
\begin{align}
	\frac{1}{n^{1+\kappa}}\sum_{t=1}^{\lfloor ns\rfloor}\vxi_{t-1}\vxi_{t-1}^\prime & \overset{\P}{\longrightarrow}s\mV_{\xi\xi},\label{eq:(p.4.4.1)}\\
	\frac{1}{n^{1/2+\kappa}}\sum_{t=1}^{n}\vxi_{t-1} & \overset{d}{\longrightarrow}N\big(\vzero, \mC^{-1}\mOmega_{\xi\xi}(\mC^{-1})^{\prime}\big),\label{eq:(p.4.4.2)}
\end{align}
where $\mOmega_{\xi\xi}$ is defined in Appendix~\ref{Asymptotic Variance-Covariance Matrices}, and $\vxi_{t-1}$ in our notation corresponds to $\vx_{t-1}$ in the notation of \citet{MP20}.
A slight adaptation of their arguments (replacing $n$ by $\lfloor ns\rfloor$ in the summation) yields that
\begin{align}
	\frac{1}{n^{1+\kappa/2}}\sum_{t=1}^{\lfloor ns\rfloor} \vxi_{t-1} &= \frac{1}{n^{(1-\kappa)/2}}\frac{1}{n^{1/2+\kappa}}\sum_{t=1}^{\lfloor ns\rfloor} \vxi_{t-1}\notag\\
	&= \frac{1}{n^{(1-\kappa)/2}}O_{\P}(1)\notag\\
	&= o_{\P}(1),\label{eq:(p.4.3.2)}
\end{align}
where the second step follows from arguments used to deduce \eqref{eq:(p.4.4.2)} in \citet{MP20}.
We also refer to the proof of Lemma~3.3 in \citet{MP20}---in particular the penultimate display on p.~50---for equivalents of \eqref{eq:(p.4.4.1)} and \eqref{eq:(p.4.3.2)}.

Using \eqref{eq:(p.4.4.1)} and \eqref{eq:(p.4.3.2)}, we obtain
\begin{align*}
	\sum_{t=1}^{\lfloor ns\rfloor}&\mD_n^{-1}\mX_{t-1}(\vmu,\vxi)\mX_{t-1}^\prime(\vmu,\vxi)\mD_n^{-1}\\
	& =\sum_{t=1}^{\lfloor ns\rfloor} \begin{pmatrix} \frac{1}{n} & \frac{1}{n^{1+\kappa/2}}\vmu^\prime + \frac{1}{n^{1+\kappa/2}}\vxi_{t-1}^\prime\\
	\frac{1}{n^{1+\kappa/2}}\vmu + \frac{1}{n^{1+\kappa/2}}\vxi_{t-1} & \frac{1}{n^{1+\kappa}}\big[\vmu\vmu^\prime + \vmu\vxi_{t-1}^\prime + \vxi_{t-1}\vmu^\prime + \vxi_{t-1}\vxi_{t-1}^\prime\big]
	\end{pmatrix}\\
	&\overset{\P}{\longrightarrow}s\begin{pmatrix}1 & \vzero\\
	\vzero & \mV_{\xi\xi}\end{pmatrix}.
\end{align*}

\textit{\textbf{(I0) case:}} Lemma~2.2~(ii) in \citet{MP20} implies that \eqref{eq:(p.4.4.2)} also holds in the (I0) case with $\kappa=0$.
Therefore, as $n\to\infty$,
\begin{equation}\label{eq:(p.14.0)}
	\frac{1}{n}\sum_{t=1}^{\lfloor ns\rfloor} (\vmu_{x} + \vxi_{t-1})=\frac{\lfloor ns\rfloor}{n}\vmu_{x} + \frac{1}{\sqrt{n}}\frac{1}{\sqrt{n}}\sum_{t=1}^{\lfloor ns\rfloor}\vxi_{t-1}=s\vmu_{x}+o(1) + \frac{1}{\sqrt{n}}O_{\P}(1)\overset{\P}{\longrightarrow}s\vmu_{x}.
\end{equation}
From the penultimate display on p.~43 of \citet{MP20},
\[
	\frac{1}{n}\sum_{t=1}^{\lfloor ns\rfloor}\vxi_{t-1}\vxi_{t-1}^\prime\overset{\P}{\longrightarrow}s\sum_{j=0}^{\infty}\mR^{j}\big[\mGamma_{u}(0) + \mR\mGamma + \mGamma^\prime\mR^\prime\big](\mR^{j})^\prime,
\]
where $\mGamma$ and $\mGamma_u(0)$ are defined in Appendix~\ref{Asymptotic Variance-Covariance Matrices}.
Thus, as $n\to\infty$,
\begin{align}
	\frac{1}{n}\sum_{t=1}^{\lfloor ns\rfloor} (\vmu_x + \vxi_{t-1})(\vmu_x^\prime + \vxi_{t-1}^\prime) &= \frac{1}{n}\sum_{t=1}^{\lfloor ns\rfloor} \vmu_x\vmu_x^\prime +  n^{-1/2}\vmu_x\frac{1}{\sqrt{n}}\sum_{t=1}^{\lfloor ns\rfloor}\vxi_{t-1}^\prime\notag\\
	&\hspace{1cm} + n^{-1/2}\bigg(\frac{1}{\sqrt{n}}\sum_{t=1}^{\lfloor ns\rfloor}\vxi_{t-1}\bigg)\vmu_{x}^\prime + \frac{1}{n}\sum_{t=1}^{\lfloor ns\rfloor}\vxi_{t-1}\vxi_{t-1}^\prime\notag\\
	&\overset{\P}{\longrightarrow}s\bigg[\vmu_{x}\vmu_{x}^\prime + \sum_{j=0}^{\infty}\mR^{j}\big[\mGamma_u(0) + \mR\mGamma + \mGamma^\prime\mR^\prime\big](\mR^{j})^\prime\bigg].\label{eq:(p.14.1)}
\end{align}

Combining \eqref{eq:(p.14.0)} and \eqref{eq:(p.14.1)} gives that, as $n\to\infty$,
\begin{align*}
	\sum_{t=1}^{\lfloor ns\rfloor}&\mD_n^{-1}\mX_{t-1}(\vmu_{x},\vxi)\mX_{t-1}^\prime(\vmu_{x},\vxi)\mD_n^{-1}\\
	& =\frac{1}{n}\sum_{t=1}^{\lfloor ns\rfloor} \begin{pmatrix} 1 & \vmu_{x}^\prime + \vxi_{t-1}^\prime\\
	\vmu_{x} + \vxi_{t-1} & \big[\vmu_{x}+\vxi_{t-1}\big]\big[\vmu_{x}^\prime+\vxi_{t-1}^\prime\big]
	\end{pmatrix}\\
	&\overset{\P}{\longrightarrow}s\begin{pmatrix}1 & \vmu_{x}^\prime\\
	\vmu_{x} & \mOmega_{xx}\end{pmatrix}=s\mOmega_{XX}.
\end{align*}
This finishes the proof.
\end{proof}

\section{Proof of Theorem~\ref{thm:CoVaR est} and Corollary~\ref{cor:SBT CoVaR}}
\label{sec:thm2}

The proof of Theorem~\ref{thm:CoVaR est} requires the following preliminary Propositions~\ref{lem:CLT2}--\ref{lem:LLN3}, that serve similar purposes as Lemmas~4--7 from \citet{HS25}.
\begin{prop}\label{lem:CLT2}
Under the assumptions of Theorem~\ref{thm:CoVaR est} it holds that, as $n\to\infty$,
\begin{equation*}
	\sum_{t=\lfloor nr\rfloor+1}^{\lfloor ns\rfloor}\1_{\{\epsilon_t>0\}}\psi_{\beta}(\delta_t)\mD_n^{-1}\mX_{t-1}
	\overset{d}{\longrightarrow}\sqrt{(1-\alpha)(1-\beta)\beta}\mOmega^{1/2}\big[\mW_{\dagger}(s)-\mW_{\dagger}(r)\big]\qquad\text{in }(\ell^{\infty}(\mathcal{D}))^{k+1},
\end{equation*}
where $\mW_{\dagger}(\cdot)$ is a $(k+1)$-variate standard Brownian motion, and $\mOmega$ is defined in \eqref{eq:Omega}. 
\end{prop}

\begin{proof}
See Appendix~\ref{sec:CoVaRest Lemmas}.
\end{proof}

\begin{prop}\label{lem:LLN2}
Under the assumptions of Theorem~\ref{thm:CoVaR est} it holds for fixed $\vw\in\mathbb{R}^{k+1}$ that, as $n\to\infty$,
\begin{multline*}
	\sup_{0\leq r<s\leq 1}\bigg|\sum_{t=\lfloor nr\rfloor+1}^{\lfloor ns\rfloor}\1_{\{\epsilon_{t}>0\}}\big(\delta_{t} - \vw^\prime\mD_n^{-1}\mX_{t-1}\big)\big[\1_{\{\vw^\prime\mD_n^{-1}\mX_{t-1}<\delta_{t}<0\}} - \1_{\{0<\delta_{t}<\vw^\prime\mD_n^{-1}\mX_{t-1}\}}\big]
\\
-\frac{1}{2}(s-r)\vw^\prime\mK_\ast \vw\bigg|=o_{\P}(1).
\end{multline*}
\end{prop}

\begin{proof}
See Appendix~\ref{sec:CoVaRest Lemmas}.
\end{proof}

\begin{prop}\label{lem:LLN CLT}
Under the assumptions of Theorem~\ref{thm:CoVaR est} it holds that, as $n\to\infty$,
\begin{multline*}
	\sum_{t=\lfloor nr\rfloor+1}^{\lfloor ns\rfloor}\big[\1_{\{\mX_{t-1}^\prime[\widehat{\valpha}_n(r,s) - \valpha_0]<\epsilon_t\leq0\}} - \1_{\{0<\epsilon_t\leq \mX_{t-1}^\prime[\widehat{\valpha}_n(r,s) - \valpha_0]\}}\big]\psi_{\beta}(\delta_{t})\mD_n^{-1}\mX_{t-1}\\
\overset{d}{\longrightarrow}\big[(1-\beta)\mK - \mK_{\dagger}\big]\mSigma^{1/2} \big[\mW(s)-\mW(r)\big]\qquad\text{in }(\ell^{\infty}(\mathcal{D}_{\iota}))^{k+1},
\end{multline*}
where $\mW(\cdot)$ and $\mSigma$ are the same as in Theorem~\ref{thm:std est}
\end{prop}

\begin{proof}
See Appendix~\ref{sec:CoVaRest Lemmas}.
\end{proof}

\begin{prop}\label{lem:LLN3}
Under the assumptions of Theorem~\ref{thm:CoVaR est} it holds for any $\vw\in\mathbb{R}^{k+1}$ that, as $n\to\infty$,
\begin{multline*}
	\sup_{(r,s)\in\mathcal{D}_{\iota}}\bigg|\sum_{t=\lfloor nr\rfloor+1}^{\lfloor ns\rfloor}\big[\1_{\{\mX_{t-1}^\prime[\widehat{\valpha}_n(r,s) - \valpha_0]<\epsilon_t\leq0\}} - \1_{\{0<\epsilon_t\leq \mX_{t-1}^\prime[\widehat{\valpha}_n(r,s) - \valpha_0]\}}\big]\big(\delta_{t} - \vw^\prime\mD_n^{-1}\mX_{t-1}\big)\times\\
	\times\big[\1_{\{\vw^\prime\mD_n^{-1}\mX_{t-1}<\delta_{t}<0\}} - \1_{\{0<\delta_{t}<\vw^\prime\mD_n^{-1}\mX_{t-1}\}}\big]\bigg|=o_{\P}(1).
\end{multline*}
\end{prop}

\begin{proof}
See Appendix~\ref{sec:CoVaRest Lemmas}.
\end{proof}

\begin{proof}[{\textbf{Proof of Theorem~\ref{thm:CoVaR est}:}}]
In the first part of this proof, we show that
\begin{equation}\label{eq:thm2 simple}
(s-r)\mD_n\big[\widehat{\vbeta}_n(r,s) - \vbeta_0\big]\overset{d}{\longrightarrow}\mSigma_{\ast}^{1/2}\big[\mW_{\ast}(s)-\mW_{\ast}(r)\big]\qquad\text{in }(\ell^{\infty}(\mathcal{D}_{\iota}))^{k+1},
\end{equation}
where $\mW_{\ast}(\cdot)$ is a $(k+1)$-variate standard Brownian motion and $\mSigma_{\ast}=\alpha(1-\alpha)\mK_{\ast}^{-1}\mOmega_{\ast}\mK_{\ast}^{-1}$ with 
\begin{equation}\label{eq:(E.2)}
	\mOmega_{\ast}=\alpha^{-1}\beta(1-\beta)\mOmega + \alpha^{-1}(1-\alpha)^{-1}\big[(1-\beta)\mK - \mK_{\dagger}\big]\mK^{-1}\mOmega\mK^{-1}\big[(1-\beta)\mK - \mK_{\dagger}\big]
\end{equation}
as defined in Appendix~\ref{Asymptotic Variance-Covariance Matrices}.
This part is similar to the proof of Theorem~\ref{thm:std est}.
To highlight the similarities, we often overload notation by redefining quantities that already appeared in the proof of Theorem~\ref{thm:std est}.
The estimator $\widehat{\vbeta}_n(r,s)$ can equivalently be written as
\begin{align*}
\widehat{\vbeta}_n(r,s) &= \argmin_{\vbeta\in\mathbb{R}^{k+1}} \sum_{t=\lfloor nr\rfloor+1}^{\lfloor ns\rfloor}\1_{\{Y_t>\mX_{t-1}^\prime\widehat{\valpha}_n(r,s)\}}\big[\rho_{\beta}(Z_t-\mX_{t-1}^\prime\vbeta) - \rho_{\beta}(\delta_{t})\big]\\
&= \argmin_{\vbeta\in\mathbb{R}^{k+1}} \sum_{t=\lfloor nr\rfloor+1}^{\lfloor ns\rfloor}\1_{\{Y_t>\mX_{t-1}^\prime\widehat{\valpha}_n(r,s)\}}\Big[\rho_{\beta}\big(Z_t-\mX_{t-1}^\prime\vbeta_0-\mX_{t-1}^\prime(\vbeta-\vbeta_0)\big) - \rho_{\beta}(\delta_{t})\Big]\\
&\overset{\eqref{eq:CoVaR model}}{=} \argmin_{\vbeta\in\mathbb{R}^{k+1}} \sum_{t=\lfloor nr\rfloor+1}^{\lfloor ns\rfloor}\1_{\{Y_t>\mX_{t-1}^\prime\widehat{\valpha}_n(r,s)\}}\Big[\rho_{\beta}\big(\delta_{t} - (\vbeta-\vbeta_0)^\prime\mX_{t-1}\big) - \rho_{\beta}(\delta_{t})\Big].
\end{align*}
Therefore, if we define 
\begin{equation}\label{eq:(p.16W)}
 f_n(\vw,r,s)=\sum_{t=\lfloor nr\rfloor+1}^{\lfloor ns\rfloor}\1_{\{Y_t>\mX_{t-1}^\prime\widehat{\valpha}_n(r,s)\}}\big[\rho_{\beta}(\delta_{t} - \vw^\prime\mD_n^{-1}\mX_{t-1}) - \rho_{\beta}(\delta_{t})\big],
\end{equation}
then the minimizer $\vw_n(r,s)$ of $f_n(\cdot,r,s)$ satisfies that 
\[
	\vw_n(r,s)=\mD_n\big[\widehat{\vbeta}_n(r,s) - \vbeta_0\big].
\]
As in the proof of Theorem~\ref{thm:std est}, we invoke Theorem~2 of \citet{Kat09} to derive the weak limit of $\vw_n(\cdot,\cdot)$.

Rewrite $f_n(\cdot, r,s)$ by noting that
\begin{align}
	\1_{\{Y_t>\mX_{t-1}^\prime\widehat{\valpha}_n(r,s)\}} &= \1_{\{Y_t-\mX_{t-1}^\prime\valpha_0>\mX_{t-1}^\prime[\widehat{\valpha}_n(r,s)-\valpha_0]\}}\notag\\
	&\overset{\eqref{eq:quantile model}}{=} \1_{\{\epsilon_t>\mX_{t-1}^\prime[\widehat{\valpha}_n(r,s)-\valpha_0]\}} - \1_{\{\epsilon_t>0\}} + \1_{\{\epsilon_t>0\}}\notag\\
	&= \1_{\{\mX_{t-1}^\prime[\widehat{\valpha}_n(r,s) - \valpha_0]<\epsilon_t\leq0\}} - \1_{\{0<\epsilon_t\leq \mX_{t-1}^\prime[\widehat{\valpha}_n(r,s) - \valpha_0]\}} + \1_{\{\epsilon_t>0\}}.\label{eq:ind trafo}
\end{align}
Observe that from Assumption~\ref{ass:innov CoVaR},
\[
	\P\big\{\exists\ t\in\mathbb{N}\colon \delta_t=0\big\}=\P\bigg\{\bigcup_{t\in\mathbb{N}}\{ \delta_t=0\}\bigg\}\leq\sum_{t\in\mathbb{N}}\P\big\{ \delta_t=0\big\}=0.
\]
Using \eqref{eq:(1)} and in a second step \eqref{eq:ind trafo}, we deduce from \eqref{eq:(p.16W)} that a.s.
\begin{align*}
	f_n(\vw,r,s) &= -\vw^\prime\sum_{t=\lfloor nr\rfloor+1}^{\lfloor ns\rfloor}\1_{\{Y_t>\mX_{t-1}^\prime\widehat{\valpha}_n(r,s)\}}\psi_{\beta}(\delta_{t})\mD_n^{-1}\mX_{t-1}\\
	& \hspace{0.5cm} + \sum_{t=\lfloor nr\rfloor+1}^{\lfloor ns\rfloor}\1_{\{Y_t>\mX_{t-1}^\prime\widehat{\valpha}_n(r,s)\}}\big(\delta_{t} - \vw^\prime\mD_n^{-1}\mX_{t-1}\big)\times\\
	&\hspace{6.8cm}\times\big[\1_{\{\vw^\prime\mD_n^{-1}\mX_{t-1}<\delta_{t}<0\}} - \1_{\{0<\delta_{t}<\vw^\prime\mD_n^{-1}\mX_{t-1}\}}\big]\\
	&= -\vw^\prime\sum_{t=\lfloor nr\rfloor+1}^{\lfloor ns\rfloor}\1_{\{\epsilon_t>0\}}\psi_{\beta}(\delta_{t})\mD_n^{-1}\mX_{t-1}\\
	& \hspace{0.5cm} + \sum_{t=\lfloor nr\rfloor+1}^{\lfloor ns\rfloor}\1_{\{\epsilon_t>0\}}\big(\delta_{t} - \vw^\prime\mD_n^{-1}\mX_{t-1}\big)\big[\1_{\{\vw^\prime\mD_n^{-1}\mX_{t-1}<\delta_{t}<0\}} - \1_{\{0<\delta_{t}<\vw^\prime\mD_n^{-1}\mX_{t-1}\}}\big]\\
	&\hspace{0.5cm} - \vw^\prime \sum_{t=\lfloor nr\rfloor+1}^{\lfloor ns\rfloor}\big[\1_{\{\mX_{t-1}^\prime[\widehat{\valpha}_n(r,s) - \valpha_0]<\epsilon_t\leq0\}} - \1_{\{0<\epsilon_t\leq \mX_{t-1}^\prime[\widehat{\valpha}_n(r,s) - \valpha_0]\}}\big]\psi_{\beta}(\delta_{t})\mD_n^{-1}\mX_{t-1}\\
	&\hspace{0.5cm} + \sum_{t=\lfloor nr\rfloor+1}^{\lfloor ns\rfloor}\big[\1_{\{\mX_{t-1}^\prime[\widehat{\valpha}_n(r,s) - \valpha_0]<\epsilon_t\leq0\}} - \1_{\{0<\epsilon_t\leq \mX_{t-1}^\prime[\widehat{\valpha}_n(r,s) - \valpha_0]\}}\big]\big(\delta_{t} - \vw^\prime\mD_n^{-1}\mX_{t-1}\big)\times\\
	&\hspace{6.8cm}\times\big[\1_{\{\vw^\prime\mD_n^{-1}\mX_{t-1}<\delta_{t}<0\}} - \1_{\{0<\delta_{t}<\vw^\prime\mD_n^{-1}\mX_{t-1}\}}\big]\\
	&= -\vw^\prime\bigg\{\sum_{t=\lfloor nr\rfloor+1}^{\lfloor ns\rfloor}\1_{\{\epsilon_t>0\}}\psi_{\beta}(\delta_{t})\mD_n^{-1}\mX_{t-1} \\
	&\hspace{0.5cm} + \sum_{t=\lfloor nr\rfloor+1}^{\lfloor ns\rfloor}\big[\1_{\{\mX_{t-1}^\prime[\widehat{\valpha}_n(r,s) - \valpha_0]<\epsilon_t\leq0\}} - \1_{\{0<\epsilon_t\leq \mX_{t-1}^\prime[\widehat{\valpha}_n(r,s) - \valpha_0]\}}\big]\psi_{\beta}(\delta_{t})\mD_n^{-1}\mX_{t-1}\bigg\}\\
	&\hspace{0.5cm}+ \sum_{t=\lfloor nr\rfloor+1}^{\lfloor ns\rfloor}\1_{\{\epsilon_t>0\}}\big(\delta_{t} - \vw^\prime\mD_n^{-1}\mX_{t-1}\big)\big[\1_{\{\vw^\prime\mD_n^{-1}\mX_{t-1}<\delta_{t}<0\}} - \1_{\{0<\delta_{t}<\vw^\prime\mD_n^{-1}\mX_{t-1}\}}\big]\\
	&\hspace{0.5cm} + \sum_{t=\lfloor nr\rfloor+1}^{\lfloor ns\rfloor}\big[\1_{\{\mX_{t-1}^\prime[\widehat{\valpha}_n(r,s) - \valpha_0]<\epsilon_t\leq0\}} - \1_{\{0<\epsilon_t\leq \mX_{t-1}^\prime[\widehat{\valpha}_n(r,s) - \valpha_0]\}}\big]\big(\delta_{t} - \vw^\prime\mD_n^{-1}\mX_{t-1}\big)\times\\
	&\hspace{6.8cm}\times\big[\1_{\{\vw^\prime\mD_n^{-1}\mX_{t-1}<\delta_{t}<0\}} - \1_{\{0<\delta_{t}<\vw^\prime\mD_n^{-1}\mX_{t-1}\}}\big].
\end{align*}
Define
\begin{equation*}
	g_n(\vw,r,s) = -\vw^\prime\mW_n(r,s) +\frac{1}{2}\vw^\prime(s-r)\mK_{\ast}\vw,
\end{equation*}
where
\begin{multline*}
	\mW_n(r,s):=\sum_{t=\lfloor nr\rfloor+1}^{\lfloor ns\rfloor}\1_{\{\epsilon_t>0\}}\psi_{\beta}(\delta_{t})\mD_n^{-1}\mX_{t-1}  \\
		+\sum_{t=\lfloor nr\rfloor+1}^{\lfloor ns\rfloor}\big[\1_{\{\mX_{t-1}^\prime[\widehat{\valpha}_n(r,s) - \valpha_0]<\epsilon_t\leq0\}} - \1_{\{0<\epsilon_t\leq \mX_{t-1}^\prime[\widehat{\valpha}_n(r,s) - \valpha_0]\}}\big]\psi_{\beta}(\delta_{t})\mD_n^{-1}\mX_{t-1}.
\end{multline*}
By Propositions~\ref{lem:LLN2} and \ref{lem:LLN3}, as $n\to\infty$,
\[
	\sup_{(r,s)\in\mathcal{D}_{\iota}}\big|f_n(\vw,r,s) - g_n(\vw,r,s)\big|=o_{\P}(1).
\]
Also, by Propositions~\ref{lem:CLT2} and \ref{lem:LLN CLT}, 
\begin{multline}\label{eq:WN2}
	\mW_n(r,s)\overset{d}{\longrightarrow}\mW_{\infty}(r,s): = \sqrt{(1-\alpha)(1-\beta)\beta}\mOmega^{1/2}\big[\mW_{\dagger}(s)-\mW_{\dagger}(r)\big]\\
	+ \big[(1-\beta)\mK- \mK_{\dagger}\big]\mSigma^{1/2}\big[\mW(s)-\mW(r)\big]\qquad\text{in }(\ell^{\infty}(\mathcal{D}_{\iota}))^{k+1}.
\end{multline}
Therefore, the CMT implies that for every $\eta>0$ there exists $M>0$, such that
\[
	\limsup_{n\to\infty}\P\bigg\{\sup_{(r,s)\in\mathcal{D}_{\iota}}\big\Vert\mW_n(r,s)\big\Vert>M\bigg\}\leq\eta.
\]

Hence, we may apply Theorem~2 in \citet{Kat09} to deduce that
\[
	\vw_n(r,s)=\frac{1}{s-r}\mK_{\ast}^{-1}\mW_n(r,s)+\vr_n(r,s),
\]
where $\sup_{(r,s)\in\mathcal{D}_{\iota}}\big\Vert\vr_n(r,s)\big\Vert=o_{\P}(1)$.
Therefore, it follows by \eqref{eq:WN2} that, as $n\to\infty$,
\begin{equation}\label{eq:w conv CoVaR}
	\vw_n(r,s)\overset{d}{\longrightarrow}\vw_{\infty}(r,s)\qquad\text{in }(\ell^{\infty}(\mathcal{D}_{\iota}))^{k+1},
\end{equation}
where
\begin{align}
	\vw_{\infty}(r,s)&= \frac{1}{s-r}\mK_{\ast}^{-1}\Big\{\sqrt{(1-\alpha)(1-\beta)\beta}\mOmega^{1/2}\big[\mW_{\dagger}(s)-\mW_{\dagger}(r)\big]\notag\\
	&\hspace{5cm} + \big[(1-\beta)\mK- \mK_{\dagger}\big]\mSigma^{1/2}\big[\mW(s)-\mW(r)\big]\Big\}.\label{eq:w def}
\end{align}

It remains to show that $(s-r)\vw_{\infty}(r,s)$ is distributionally equivalent to the limit in \eqref{eq:thm2 simple}.
The key step is to show that $\mW_{\infty}(r,s)$ is Brownian motion. 
To do so, we show that the Brownian motions $\mW_{\dagger}(\cdot)$ and $\mW(\cdot)$ in $\mW_{\infty}(r,s)$ are independent of each other.
Here, $\mW(\cdot)$ is the Brownian motion from Proposition~\ref{lem:CLT} (and Proposition~\ref{lem:LLN CLT}) and $\mW_{\dagger}(\cdot)$ that of Proposition~\ref{lem:CLT2}.
To show independence, observe that mimicking the proofs of Propositions~\ref{lem:CLT} and \ref{lem:CLT2} yields that
\begin{equation}\label{eq:(p.15.2.1)}
	\sum_{t=1}^{\lfloor ns\rfloor}\vzeta_{nt}:=\sum_{t=1}^{\lfloor ns\rfloor}\begin{pmatrix}\psi_{\alpha}(\epsilon_t)\mD_n^{-1}\mX_{t-1}\\ \1_{\{\epsilon_t>0\}}\psi_{\beta}(\delta_t)\mD_n^{-1}\mX_{t-1} \end{pmatrix}\overset{d}{\longrightarrow}\underline{\mOmega}^{1/2}\underline{\mW}(s)\qquad\text{in }(D[0,1])^{2k+2},
\end{equation}
where $\underline{\mW}(s)=\big(\underline{\mW}_1^\prime(s), \underline{\mW}_2^\prime(s)\big)^\prime$ is a $(2k+2)$-variate standard Brownian motion, and $\underline{\mOmega}$ is the probability limit of
\begin{multline}
	\sum_{t=1}^{n}\E_{t-1}\big[\vzeta_{nt} \vzeta_{nt}^\prime\big]\\
	= \sum_{t=1}^{n}\begin{pmatrix} \E_{t-1}\big[\psi_{\alpha}^2(\epsilon_t)\big] & \E_{t-1}\big[\1_{\{\epsilon_t>0\}}\psi_{\alpha}(\epsilon_t)\psi_{\beta}(\delta_t)\big]\\
	\E_{t-1}\big[\1_{\{\epsilon_t>0\}}\psi_{\alpha}(\epsilon_t)\psi_{\beta}(\delta_t)\big] & \E_{t-1}\big[\1_{\{\epsilon_t>0\}}\psi_{\beta}^2(\delta_t)\big]\end{pmatrix}\otimes \mD_n^{-1}\mX_{t-1}\mX_{t-1}^\prime\mD_n^{-1},\label{eq:(C.3p)}
\end{multline}
where $\otimes$ denotes the Kronecker product.

To simplify \eqref{eq:(C.3p)}, exploit the assumptions on the errors in \eqref{eq:(QRalt)}--\eqref{eq:(CoVaRalt)} to deduce that
\begin{align}
	\E_{t-1}\big[\1_{\{\epsilon_t>0\}}\1_{\{\delta_t\leq0\}}\big] &= \P_{t-1}\big\{\epsilon_t>0\big\}\E_{t-1}\big[\1_{\{\delta_t\leq0\}}\mid\epsilon_t>0\big]\notag\\
	&= \P_{t-1}\big\{\epsilon_t>0\big\}\P_{t-1}\big\{\delta_t\leq0\mid\epsilon_t>0\big\}\notag\\
	&= (1-\alpha)\beta.\label{eq:(p.16.2)}
\end{align}
Using this and $Q_{\alpha}(\epsilon_t\mid\mathcal{F}_{t-1})=0$, we obtain that
\begin{align}
	\E_{t-1}\big[\1_{\{\epsilon_t>0\}}\psi_{\beta}^2(\delta_t)\big] &= \E_{t-1}\big[\1_{\{\epsilon_t>0\}}(\beta-\1_{\{\delta_t\leq0\}})^2\big]\notag\\
	&= \E_{t-1}\big[\1_{\{\epsilon_t>0\}}(\beta^2-2\beta\1_{\{\delta_t\leq0\}} + \1_{\{\delta_t\leq0\}})\big]\notag\\
	&= \beta^2\E_{t-1}\big[\1_{\{\epsilon_t>0\}}\big] - 2\beta\E_{t-1}\big[\1_{\{\epsilon_t>0\}}\1_{\{\delta_t\leq0\}}\big] + \E_{t-1}\big[\1_{\{\epsilon_t>0\}}\1_{\{\delta_t\leq0\}}\big] \notag\\
	&=\beta^2(1-\alpha) -2\beta(1-\alpha)\beta + (1-\alpha)\beta\notag\\
	&=-\beta^2(1-\alpha) + (1-\alpha)\beta\notag\\
	&=(1-\alpha)(1-\beta)\beta.\label{eq:(p.16)}
\end{align}
Moreover,
\begin{align}
	\E_{t-1}\big[\1_{\{\epsilon_t>0\}}\psi_{\alpha}(\epsilon_t)\psi_{\beta}(\delta_t)\big] &= \E_{t-1}\big[\1_{\{\epsilon_t>0\}}(\alpha-\1_{\{\epsilon_t\leq0\}})(\beta-\1_{\{\delta_t\leq0\}})\big]\notag\\
	&=\alpha\E_{t-1}\big[\1_{\{\epsilon_t>0\}}(\beta-\1_{\{\delta_t\leq0\}})\big]\notag\\
	&=\alpha\beta\P_{t-1}\big\{\epsilon_t>0\big\} - \alpha \E_{t-1}\big[\1_{\{\epsilon_t>0\}}\1_{\{\delta_t\leq0\}}\big]\notag\\
	&=\alpha\beta(1-\alpha) - \alpha(1-\alpha)\beta\notag\\
	&=0\label{eq:(p.18)}
\end{align}
from $Q_{\alpha}(\epsilon_t\mid\mathcal{F}_{t-1})=0$ and, in the penultimate step, \eqref{eq:(p.16.2)}.
Also, from $Q_{\alpha}(\epsilon_t\mid\mathcal{F}_{t-1})=0$,
\begin{align}
	\E_{t-1}\big[\psi_{\alpha}^2(\epsilon_t)\big] &= \E_{t-1}\big[(\alpha-\1_{\{\epsilon_t\leq 0\}})^2\big]\notag\\
	&= \E_{t-1}\big[\alpha^2-2\alpha \1_{\{\epsilon_t\leq 0\}} + \1_{\{\epsilon_t\leq 0\}}^2\big]\notag\\
	&= \alpha^2-2\alpha\P_{t-1}\{\epsilon_t\leq 0\} + \P_{t-1}\{\epsilon_t\leq 0\}\notag\\
	&= \alpha^2-2\alpha^2 + \alpha\notag\\
	&= \alpha(1-\alpha). \label{eq:(pp.4)}
\end{align}
Hence, using \eqref{eq:(p.16)}--\eqref{eq:(pp.4)} and Lemma~\ref{lem:SUM} (from Appendix~\ref{sec:QRest Lemmas}) in connection with \eqref{eq:(C.3p)}, we deduce that 
\[
	\sum_{t=1}^{n}\E_{t-1}\big[\vzeta_{nt} \vzeta_{nt}^\prime\big]\overset{\P}{\longrightarrow}\begin{pmatrix}\alpha(1-\alpha)\mOmega & \vzero \\
	\vzero & (1-\alpha)(1-\beta)\beta\mOmega\end{pmatrix}=\underline{\mOmega}.
\]
From \eqref{eq:(p.15.2.1)} and Propositions~\ref{lem:CLT} and \ref{lem:CLT2}, it must be the case that
\[
	\underline{\mOmega}^{1/2}\underline{\mW}(s)=\begin{pmatrix}\sqrt{\alpha(1-\alpha)}\mOmega^{1/2}\underline{\mW}_{1}(s)\\
	\sqrt{(1-\alpha)(1-\beta)\beta}\mOmega^{1/2}\overline{\mW}_{2}(s)\end{pmatrix}\overset{d}{=}\begin{pmatrix}\sqrt{\alpha(1-\alpha)}\mOmega^{1/2}\mW(s)\\
	\sqrt{(1-\alpha)(1-\beta)\beta}\mOmega^{1/2}\mW_{\dagger}(s)\end{pmatrix}.
\]
Since $\overline{\mW}_{1}(s)$ and $\overline{\mW}_{2}(s)$ are independent Brownian motions, $\mW(s)$ and $\mW_{\dagger}(s)$ must also be independent Brownian motions, such that $\mW_{\infty}(r,s)$ is a Brownian motion with covariance matrix
\begin{align}
\Cov\big(\mW_{\infty}(0,1), \mW_{\infty}(0,1)\big) &= \E\bigg[\Big\{\sqrt{(1-\alpha)(1-\beta)\beta}\mOmega^{1/2}\mW_{\dagger}(1) + \big[(1-\beta)\mK- \mK_{\dagger}\big]\mSigma^{1/2}\mW(1)\Big\}\times\notag\\
&\hspace{0.4cm}\times \Big\{\sqrt{(1-\alpha)(1-\beta)\beta}\mOmega^{1/2}\mW_{\dagger}(1) + \big[(1-\beta)\mK- \mK_{\dagger}\big]\mSigma^{1/2}\mW(1)\Big\}^\prime\bigg]\notag\\
&= (1-\alpha)(1-\beta)\beta\mOmega + \big[(1-\beta)\mK - \mK_{\dagger}\big]\mSigma\big[(1-\beta)\mK - \mK_{\dagger}\big]\notag\\
&=\alpha(1-\alpha)\mOmega_{\ast},\label{eq:(p.19a)}
\end{align}
where the final step uses \eqref{eq:(E.2)}.
Hence, $\mW_{\infty}(r,s)\overset{d}{=}\sqrt{\alpha(1-\alpha)}\mOmega_{\ast}^{1/2}\big[\mW_{\ast}(s)-\mW_{\ast}(r)\big]$ for a $(k+1)$-variate standard Brownian motion $\mW_{\ast}(\cdot)$.
This completes the proof of \eqref{eq:thm2 simple} by \eqref{eq:w conv CoVaR}--\eqref{eq:w def}.

Now, in the second and final part of this proof, we establish the joint (functional) convergence of the QR and CoVaR parameter estimator.
The above shows that the convergences in \eqref{eq:thm1} and \eqref{eq:w conv CoVaR} hold jointly, implying that, as $n\to\infty$,
\begin{align*}
	&(s-r)\begin{pmatrix}
		\mD_n \big[\widehat{\valpha}_n(r,s) - \valpha_0\big]\\
		\mD_n \big[\widehat{\vbeta}_n(r,s) - \vbeta_0\big]
	\end{pmatrix}\\
	&\overset{d}{\longrightarrow}\begin{pmatrix}
		\sqrt{\alpha(1-\alpha)}\mK^{-1}\mOmega^{1/2}\big[\mW(s)-\mW(r)\big]\\
		\sqrt{(1-\alpha)(1-\beta)\beta}\mK_{\ast}^{-1}\mOmega^{1/2}\big[\mW_{\dagger}(s)-\mW_{\dagger}(r)\big] + \mK_{\ast}^{-1}\big[(1-\beta)\mK - \mK_{\dagger}\big]\mSigma^{1/2}\big[\mW(s)-\mW(r)\big]
	\end{pmatrix}\\
	&=\begin{pmatrix}
		\mK^{-1} & \vzero\\
		\vzero & \mK_{\ast}^{-1}
	\end{pmatrix}
	\begin{pmatrix}
		\sqrt{\alpha(1-\alpha)}\mOmega^{1/2} & \vzero\\
		\big[(1-\beta)\mK-\mK_{\dagger}\big]\mSigma^{1/2} & \sqrt{(1-\alpha)(1-\beta)\beta}\mOmega^{1/2} 
	\end{pmatrix}
	\begin{pmatrix}
		\mW(s)-\mW(r)\\
		\mW_{\dagger}(s)-\mW_{\dagger}(r)
	\end{pmatrix}\\
	&=:\overline{\va}(r,s)\qquad\text{in }(\ell^{\infty}(\mathcal{D}_{\iota}))^{2k+2},
\end{align*}
where $\mW(\cdot)$ and $\mW_{\dagger}(\cdot)$ are $(k+1)$-variate standard Brownian motions that are independent of each other.
The covariance matrix of the above limiting Brownian motion is given by
\begin{align*}
	&\Cov\big(\overline{\va}(0,1),\overline{\va}(0,1)\big)\\
	&	=\begin{pmatrix}
		\mK^{-1} & \vzero\\
		\vzero & \mK_{\ast}^{-1}
	\end{pmatrix}
	\begin{pmatrix}
		\sqrt{\alpha(1-\alpha)}\mOmega^{1/2} & \vzero\\
		\big[(1-\beta)\mK-\mK_{\dagger}\big]\mSigma^{1/2} & \sqrt{(1-\alpha)(1-\beta)\beta}\mOmega^{1/2}
	\end{pmatrix}\times \\
	&\hspace{3cm}\times
	\begin{pmatrix}
		\sqrt{\alpha(1-\alpha)}(\mOmega^{1/2})^\prime & (\mSigma^{1/2})^\prime\big[(1-\beta)\mK-\mK_{\dagger}\big]\\
		\vzero & \sqrt{(1-\alpha)(1-\beta)\beta}(\mOmega^{1/2})^\prime
	\end{pmatrix}
	\begin{pmatrix}
		\mK^{-1} & \vzero\\
		\vzero & \mK_{\ast}^{-1}
	\end{pmatrix}\\
	&= \overline{\mK}^{-1}
	\begin{pmatrix}
		\alpha(1-\alpha)\mOmega & \alpha(1-\alpha)\mOmega\mK^{-1}\big[(1-\beta)\mK-\mK_{\dagger}\big]\\
		\alpha(1-\alpha)\big[(1-\beta)\mK-\mK_{\dagger}\big]\mK^{-1}\mOmega & \alpha(1-\alpha)\mOmega_{\ast}
	\end{pmatrix}
	\overline{\mK}^{-1}\\
	&= \alpha(1-\alpha)\overline{\mK}^{-1}
	\begin{pmatrix}
		\mOmega & \mOmega\mK^{-1}\big[(1-\beta)\mK-\mK_{\dagger}\big]\\
		\big[(1-\beta)\mK-\mK_{\dagger}\big]\mK^{-1}\mOmega & \mOmega_{\ast}
	\end{pmatrix}\overline{\mK}^{-1}\\
	&=\alpha(1-\alpha)\overline{\mK}^{-1}\cdot \overline{\mOmega}\cdot \overline{\mK}^{-1},
\end{align*}
where we have used in the second equality that $\mSigma^{1/2}=\sqrt{\alpha(1-\alpha)}\mK^{-1}\mOmega^{1/2}$ (by definition of $\mSigma$ in Theorem~\ref{thm:std est}), such that
\[
	\sqrt{\alpha(1-\alpha)}\big[(1-\beta)\mK-\mK_{\dagger}\big]\mSigma^{1/2}(\mOmega^{1/2})^\prime = \alpha(1-\alpha)\big[(1-\beta)\mK-\mK_{\dagger}\big]\mK^{-1}\mOmega.
\]
This concludes the proof.
\end{proof}

\begin{proof}[{\textbf{Proof of Corollary~\ref{cor:SBT CoVaR}:}}]
Define
\begin{equation}\label{eq:overline D}
	\overline{\mD}_n=\begin{pmatrix}\mD_n & \vzero\\ \vzero & \mD_n\end{pmatrix}.
\end{equation}
Observing that $\overline{\mD}_n$ is invertible, we can rewrite the test statistic under $\mathcal{H}_{0}^{\CoVaR}$ as
\begin{multline*}
	\mathcal{U}_{n,\vgamma} = \sup_{s\in[\iota,1-\iota]} s^2(1-s)^2\Big\{\overline{\mD}_n\big[\widehat{\vgamma}_n(0,s)-\vgamma_0\big]- \overline{\mD}_n\big[\widehat{\vgamma}_n(s,1)-\vgamma_0\big]\Big\}^\prime \overline{\bm{\mathcal{N}}}_{n,\vgamma}^{-1}(s)\times\\
	\times\Big\{\overline{\mD}_n\big[\widehat{\vgamma}_n(0,s)-\vgamma_0\big]- \overline{\mD}_n\big[\widehat{\vgamma}_n(s,1)-\vgamma_0\big]\Big\}
\end{multline*}
with $\vgamma_0=(\valpha_0^\prime, \vbeta_0^\prime)^\prime$ and normalizer
\begin{align*}
	\overline{\bm{\mathcal{N}}}_{n,\vgamma}(s) &=\int_{\iota}^{s}r^2\Big\{\overline{\mD}_n\big[\widehat{\vgamma}_n(0,r)-\vgamma_0\big] - \overline{\mD}_n\big[\widehat{\vgamma}_n(0,s)-\vgamma_0\big]\Big\}\times\\
		&\hspace{1.2cm}\times\Big\{\overline{\mD}_n\big[\widehat{\vgamma}_n(0,r)-\vgamma_0\big] - \overline{\mD}_n\big[\widehat{\vgamma}_n(0,s)-\vgamma_0\big]\Big\}^\prime\D r \\
		& \hspace{0.5cm}+ \int_{s}^{1-\iota}(1-r)^2\Big\{\overline{\mD}_n\big[\widehat{\vgamma}_n(r,1)-\vgamma_0\big] - \overline{\mD}_n\big[\widehat{\vgamma}_n(s,1)-\vgamma_0\big]\Big\}\times\\
		&\hspace{2.9cm}\times\Big\{\overline{\mD}_n\big[\widehat{\vgamma}_n(r,1)-\vgamma_0\big] - \overline{\mD}_n\big[\widehat{\vgamma}_n(s,1)-\vgamma_0\big]\Big\}^\prime\D r.
\end{align*}
The conclusion then follows immediately from Theorem~\ref{thm:CoVaR est} and the CMT.
\end{proof}

\section{Proofs of Propositions~\ref{lem:CLT2}--\ref{lem:LLN3}}\label{sec:CoVaRest Lemmas}

\begin{proof}[{\textbf{Proof of Proposition~\ref{lem:CLT2}:}}] 
\textit{\textbf{(NS) case:}} We first show that, as $n\to\infty$,
\begin{equation}\label{eq:(p.30.0)}
	\sum_{t=1}^{\lfloor ns\rfloor}\1_{\{\epsilon_t>0\}}\psi_{\beta}(\delta_t)\mD_n^{-1}\mX_{t-1}
	\overset{d}{\longrightarrow}\sqrt{(1-\alpha)(1-\beta)\beta}\mOmega^{1/2}\mW_{\dagger}(s)\qquad\text{in }(D[0,1])^{k+1}.
\end{equation}
From similar arguments as in the proof of Proposition~\ref{lem:CLT}, it follows that we only have to show that, as $n\to\infty$,
\begin{equation}\label{eq:To show CLT2}
	\sum_{t=1}^{\lfloor ns\rfloor} \vzeta_{nt}:=\sum_{t=1}^{\lfloor ns\rfloor} \1_{\{\epsilon_t>0\}}\psi_{\beta}(\delta_t)\begin{pmatrix} n^{-1/2}\\ n^{-1/2}\vxi_{n,t-1}\end{pmatrix}\overset{d}{\longrightarrow}\sqrt{(1-\alpha)(1-\beta)\beta}\mOmega^{1/2}\mW_{\dagger}(s)\qquad\text{in }(D[0,1])^{k+1},
\end{equation}
where $\vxi_{n,t-1}$ is defined in \eqref{eq:(p.4.decomp)}.
To apply Theorem~3.33 of \citet[Chapter~VIII]{JS87}, we first consider the conditional variances: As $n\to\infty$,
\begin{align*}
	\sum_{t=1}^{\lfloor ns\rfloor} \E_{t-1}\big[\vzeta_{nt}\vzeta_{nt}^\prime\big] &= \sum_{t=1}^{\lfloor ns\rfloor}\E_{t-1}\big[\1_{\{\epsilon_t>0\}}\psi_{\beta}^2(\delta_t)\big]\mD_n^{-1}\mX_{t-1}(\vzero,\vzero)\mX_{t-1}^\prime(\vzero,\vzero)\mD_n^{-1}\\
	&= (1-\alpha)(1-\beta)\beta\sum_{t=1}^{\lfloor ns\rfloor}\mD_n^{-1}\mX_{t-1}(\vzero,\vzero)\mX_{t-1}^\prime(\vzero,\vzero)\mD_n^{-1}\\
	&\overset{\P}{\longrightarrow}s(1-\alpha)(1-\beta)\beta\mOmega,
\end{align*}
where we used Lemma~\ref{lem:SUM} and \eqref{eq:(p.16)}.

As the CLC for $\vzeta_{nt}$ also follows similarly as in the proof of Proposition~\ref{lem:CLT}, \eqref{eq:To show CLT2} follows from Theorem~3.33 in \citet[Chapter~VIII]{JS87}, establishing \eqref{eq:(p.30.0)}.

In view of the continuity of the limit process, the claim of the proposition then follows once more from standard arguments \citep{VS14}.

\textit{\textbf{(I0) case:}} The proof follows as in the treatment of the (I0) case in Proposition~\ref{lem:CLT} with similar modifications as just outlined for the (NS) case.
\end{proof}

\begin{proof}[{\textbf{Proof of Proposition~\ref{lem:LLN2}:}}] 
The proof is similar to that of Proposition~\ref{lem:LLN}.
By similar arguments used there, it suffices to show that
\begin{multline*}
	\sup_{0\leq s\leq 1}\bigg|\sum_{t=1}^{\lfloor ns\rfloor}\1_{\{\epsilon_{t}>0\}}\big(\delta_{t} - \vw^\prime\mD_n^{-1}\mX_{t-1}\big)\big[\1_{\{\vw^\prime\mD_n^{-1}\mX_{t-1}<\delta_{t}<0\}} - \1_{\{0<\delta_{t}<\vw^\prime\mD_n^{-1}\mX_{t-1}\}}\big]
\\
-\frac{1}{2}s\vw^\prime\mK_\ast \vw\bigg|=o_{\P}(1).
\end{multline*}
We slightly overload notation introduced in the proof of Proposition~\ref{lem:LLN} by redefining
\begin{align*}
\nu_t(\vw) &:=\1_{\{\epsilon_{t}>0\}}\big(\delta_{t} - \vw^\prime\mD_n^{-1}\mX_{t-1}\big)\big[\1_{\{\vw^\prime\mD_n^{-1}\mX_{t-1}<\delta_{t}<0\}} - \1_{\{0<\delta_{t}<\vw^\prime\mD_n^{-1}\mX_{t-1}\}}\big],\\
\overline{\nu}_{t}(\vw) &:= \E_{t-1}\big[\nu_t(\vw)\big],\\
V_n(\vw,s) &:= \sum_{t=1}^{\lfloor ns\rfloor} \nu_t(\vw),\\
\overline{V}_n(\vw,s) &:= \sum_{t=1}^{\lfloor ns\rfloor} \overline{\nu}_t(\vw).
\end{align*}
We establish the proposition by showing that, uniformly in $s\in[0,1]$,
\begin{align}
	\overline{V}_n(\vw,s) &\overset{\P}{\longrightarrow}\frac{1}{2}s\vw^\prime\mK_{\ast} \vw,\label{eq:(pp.6.1)}\\
	V_n(\vw,s) - \overline{V}_n(\vw,s) &=o_{\P}(1).\label{eq:(pp.6.2)}
\end{align}

We first show \eqref{eq:(pp.6.1)}. 
As in the proof of Proposition~\ref{lem:LLN}, it suffices to show this convergence on the set $\big\{\max_{t=1,\ldots,n}|\vw^\prime\mD_n^{-1}\mX_{t-1}|\leq d\big\}$, which ensures that the density of $(\epsilon_t,\delta_t)^\prime\mid\mathcal{F}_{t-1}$ satisfies the properties set out in Assumption~\ref{ass:innov CoVaR}.
Decompose
\begin{align*}
	\overline{V}_n(\vw,s) &= \sum_{t=1}^{\lfloor ns\rfloor}\E_{t-1}\Big[\1_{\{\epsilon_t>0\}}(\delta_t - \vw^\prime\mD_n^{-1}\mX_{t-1})\1_{\{\vw^\prime\mD_n^{-1}\mX_{t-1}<\delta_t<0\}}\Big]\\
	&\hspace{0.4cm} + \sum_{t=1}^{\lfloor ns\rfloor}\E_{t-1}\Big[\1_{\{\epsilon_t>0\}}(\vw^\prime\mD_n^{-1}\mX_{t-1}-\delta_t)\1_{\{0<\delta_t< \vw^\prime\mD_n^{-1}\mX_{t-1}\}}\Big]\\
	&=:\overline{V}_{1n}(\vw,s) + \overline{V}_{2n}(\vw,s).
\end{align*}
Both terms can be treated similarly, so we only deal with $\overline{V}_{1n}(\vw,s)$. 
For this, note that
\begin{align*}
	\overline{V}_{1n}(\vw,s) &= 
	\sum_{t=1}^{\lfloor ns\rfloor}\1_{\{\vw^\prime\mD_n^{-1}\mX_{t-1}<0\}}\int_{0}^{\infty} \bigg(\int_{\vw^\prime\mD_n^{-1}\mX_{t-1}}^{0}(y-\vw^\prime\mD_n^{-1}\mX_{t-1})f_{(\epsilon_t,\delta_t)^\prime\mid\mathcal{F}_{t-1}}(x,y)\D y\bigg)\D x\\
	&=\sum_{t=1}^{\lfloor ns\rfloor}\1_{\{\vw^\prime\mD_n^{-1}\mX_{t-1}<0\}} \int_{\vw^\prime\mD_n^{-1}\mX_{t-1}}^{0}(y-\vw^\prime\mD_n^{-1}\mX_{t-1})\bigg(\int_{0}^{\infty}f_{(\epsilon_t,\delta_t)^\prime\mid\mathcal{F}_{t-1}}(x,y)\D x\bigg)\D y\\
	&= \sum_{t=1}^{\lfloor ns\rfloor}\1_{\{\vw^\prime\mD_n^{-1}\mX_{t-1}<0\}} \int_{\vw^\prime\mD_n^{-1}\mX_{t-1}}^{0}(y-\vw^\prime\mD_n^{-1}\mX_{t-1})\bigg(\int_{0}^{\infty}f_{(\epsilon_t,\delta_t)^\prime\mid\mathcal{F}_{t-1}}(x,0)\D x\bigg)\D y\\
	&\hspace{2cm} + \sum_{t=1}^{\lfloor ns\rfloor}\1_{\{\vw^\prime\mD_n^{-1}\mX_{t-1}<0\}} \int_{\vw^\prime\mD_n^{-1}\mX_{t-1}}^{0}(y-\vw^\prime\mD_n^{-1}\mX_{t-1})\times\\
	&\hspace{4cm} \times\bigg(\int_{0}^{\infty}f_{(\epsilon_t,\delta_t)^\prime\mid\mathcal{F}_{t-1}}(x,y)\D x - \int_{0}^{\infty}f_{(\epsilon_t,\delta_t)^\prime\mid\mathcal{F}_{t-1}}(x,0)\D x\bigg)\D y\\
	&=:\overline{V}_{11n}(\vw,s) + \overline{V}_{12n}(\vw,s).
\end{align*}
Simple integration yields that
\[
	\overline{V}_{11n}(\vw,s)=\frac{1}{2}\sum_{t=1}^{\lfloor ns\rfloor}\1_{\{\vw^\prime\mD_n^{-1}\mX_{t-1}<0\}} \bigg(\int_{0}^{\infty}f_{(\epsilon_t,\delta_t)^\prime\mid\mathcal{F}_{t-1}}(x,0)\D x\bigg)\big(\vw^\prime\mD_n^{-1}\mX_{t-1}\big)^2.
\]
Also, uniformly in $s\in[0,1]$,
\begin{align*}
	\big|\overline{V}_{12n}(\vw,s)\big| &\leq \sum_{t=1}^{\lfloor ns\rfloor}\1_{\{\vw^\prime\mD_n^{-1}\mX_{t-1}<0\}} \int_{\vw^\prime\mD_n^{-1}\mX_{t-1}}^{0}(y-\vw^\prime\mD_n^{-1}\mX_{t-1})|y-0|\times\\
	&\hspace{4cm} \times\frac{\Big|\int_{0}^{\infty}f_{(\epsilon_t,\delta_t)^\prime\mid\mathcal{F}_{t-1}}(x,y)\D x - \int_{0}^{\infty}f_{(\epsilon_t,\delta_t)^\prime\mid\mathcal{F}_{t-1}}(x,0)\D x\Big|}{|y-0|}\D y\\
	&\leq L\sum_{t=1}^{\lfloor ns\rfloor}\1_{\{\vw^\prime\mD_n^{-1}\mX_{t-1}<0\}} \int_{\vw^\prime\mD_n^{-1}\mX_{t-1}}^{0}(y-\vw^\prime\mD_n^{-1}\mX_{t-1})|y-0|\D y\\
	&= L\sum_{t=1}^{\lfloor ns\rfloor}\1_{\{\vw^\prime\mD_n^{-1}\mX_{t-1}<0\}} (-1/6)\big(\vw^\prime\mD_n^{-1}\mX_{t-1}\big)^3\\
	&\leq K\sum_{t=1}^{\lfloor ns\rfloor}\big|\vw^\prime\mD_n^{-1}\mX_{t-1}\big|^3\\
	&\leq K\max_{t=1,\ldots,n}\Vert\mD_n^{-1}\mX_{t-1}\Vert \vw^\prime\bigg(\sum_{t=1}^{n}\mD_n^{-1}\mX_{t-1}\mX_{t-1}^\prime\mD_n^{-1}\bigg)\vw\\
	&=Ko_{\P}(1)O_{\P}(1)\\
	&=o_{\P}(1),
\end{align*}
where we used Assumption~\ref{ass:innov CoVaR}~\eqref{it:Lipschitz CoVaR} in the second step, and Lemmas~\ref{lem:MAX}--\ref{lem:SUM} in the penultimate step.
The previous two displays imply that
\[
	\overline{V}_{1n}(\vw,s)=\frac{1}{2}\sum_{t=1}^{\lfloor ns\rfloor}\1_{\{\vw^\prime\mD_n^{-1}\mX_{t-1}<0\}} \bigg(\int_{0}^{\infty}f_{(\epsilon_t,\delta_t)^\prime\mid\mathcal{F}_{t-1}}(x,0)\D x\bigg)\big(\vw^\prime\mD_n^{-1}\mX_{t-1}\big)^2 + o_{\P}(1)
\]
uniformly in $s\in[0,1]$.
By similar arguments,
\[
	\overline{V}_{2n}(\vw,s)=\frac{1}{2}\sum_{t=1}^{\lfloor ns\rfloor}\1_{\{\vw^\prime\mD_n^{-1}\mX_{t-1}>0\}} \bigg(\int_{0}^{\infty}f_{(\epsilon_t,\delta_t)^\prime\mid\mathcal{F}_{t-1}}(x,0)\D x\bigg)\big(\vw^\prime\mD_n^{-1}\mX_{t-1}\big)^2 + o_{\P}(1)
\]
uniformly in $s\in[0,1]$, such that
\begin{align*}
	\overline{V}_{n}(\vw,s)&=\overline{V}_{1n}(\vw,s) + \overline{V}_{2n}(\vw,s)\\
	&=\frac{1}{2}\sum_{t=1}^{\lfloor ns\rfloor} \bigg(\int_{0}^{\infty}f_{(\epsilon_t,\delta_t)^\prime\mid\mathcal{F}_{t-1}}(x,0)\D x\bigg)\big(\vw^\prime\mD_n^{-1}\mX_{t-1}\big)^2 + o_{\P}(1)\\
	&= \frac{1}{2}\vw^\prime\Bigg[\sum_{t=1}^{\lfloor ns\rfloor} \bigg(\int_{0}^{\infty}f_{(\epsilon_t,\delta_t)^\prime\mid\mathcal{F}_{t-1}}(x,0)\D x\bigg)\mD_n^{-1}\mX_{t-1}\mX_{t-1}^\prime\mD_n^{-1}\Bigg]\vw + o_{\P}(1)\\
	&= \frac{1}{2}s\vw^\prime \mK_{\ast}\vw+o_{\P}(1)
\end{align*}
by Assumption~\ref{ass:K ast}.

It remains to show \eqref{eq:(pp.6.2)}.
Clearly,
\[
	V_n(\vw,s) - \overline{V}_n(\vw,s) = \sum_{t=1}^{\lfloor ns\rfloor}\big[\nu_t(\vw) - \overline{\nu}_t(\vw)\big]
\]
is a sum of MDAs.
We again invoke Theorem~3.33 of \citet[Chapter~VIII]{JS87}.
That 
\begin{align*}
	\sum_{t=1}^{\lfloor ns\rfloor}& \E_{t-1}\big[\{\nu_t(\vw) - \overline{\nu}_t(\vw)\}^2\big]\\
	&\leq \sum_{t=1}^{\lfloor ns\rfloor} \E_{t-1}\big[\nu_t^2(\vw)\big]\\
	&= \sum_{t=1}^{\lfloor ns\rfloor} \E_{t-1}\Big[\1_{\{\epsilon_t>0\}}(\delta_t - \vw^\prime\mD_n^{-1}\mX_{t-1})^2 \big(\1_{\{\vw^\prime\mD_n^{-1}\mX_{t-1}<\delta_t<0\}}+\1_{\{0<\delta_t<\vw^\prime\mD_n^{-1}\mX_{t-1}\}}\big)\Big]\\
	&=o_{\P}(1)
\end{align*}
follows as in the proof of Proposition~\ref{lem:LLN}, as does the CLC.
Hence, said Theorem~3.33 implies that 
\[	
	\sup_{s\in[0,1]}\big|V_n(\vw,s)-\overline{V}_n(\vw,s)\big|=o_{\P}(1),
\]
as desired.
\end{proof}

The proof of Proposition~\ref{lem:LLN CLT} requires the three preliminary Lemmas~\ref{lem:7}--\ref{lem:9}.
To introduce these, define
\begin{align*}
	\vnu_t(\vv) &:= \big[\1_{\{\vv^\prime\mD_n^{-1}\mX_{t-1}<\epsilon_t\leq0\}} - \1_{\{0<\epsilon_t\leq \vv^\prime\mD_n^{-1}\mX_{t-1}\}}\big]\psi_{\beta}(\delta_{t})\mD_n^{-1}\mX_{t-1},\\
	\overline{\vnu}_t(\vv) &:= \E_{t-1}\big[\vnu_t(\vv)\big],\\
	\mV_n(\vv,r,s) &:= \sum_{t=\lfloor nr\rfloor+1}^{\lfloor ns\rfloor} \vnu_t(\vv),\\
	\overline{\mV}_n(\vv,r,s) &:= \sum_{t=\lfloor nr\rfloor+1}^{\lfloor ns\rfloor} \overline{\vnu}_t(\vv).
\end{align*}

\begin{lem}\label{lem:7}
Under the assumptions of Theorem~\ref{thm:CoVaR est} it holds for any $K>0$ that, as $n\to\infty$,
\[
	\sup_{\Vert\vv\Vert\leq K}\sup_{0\leq r < s\leq1}\big\Vert\overline{\mV}_n(\vv,r,s) - (s-r)\big[(1-\beta)\mK-\mK_{\dagger}\big]\vv\big\Vert=o_{\P}(1).
\]
\end{lem}

\begin{proof}
See Appendix~\ref{sec:CoVaRest Lemmas help}.
\end{proof}

\begin{lem}\label{lem:8}
Under the assumptions of Theorem~\ref{thm:CoVaR est} it holds for any fixed $\vv\in\mathbb{R}^{k+1}$ that, as $n\to\infty$,
\[
	\sup_{0\leq r < s\leq1}\big\Vert\mV_n(\vv,r,s) - \overline{\mV}_n(\vv,r,s)\big\Vert=o_{\P}(1).
\]
\end{lem}

\begin{proof}
See Appendix~\ref{sec:CoVaRest Lemmas help}.
\end{proof}

\begin{lem}\label{lem:9}
Under the assumptions of Theorem~\ref{thm:CoVaR est} it holds for any $K>0$ that, as $n\to\infty$,
\[
	\sup_{\Vert\vv\Vert\leq K}\sup_{0\leq r < s\leq1}\big\Vert\mV_n(\vv,r,s) - \overline{\mV}_n(\vv,r,s)\big\Vert=o_{\P}(1).
\]
\end{lem}

\begin{proof}
See Appendix~\ref{sec:CoVaRest Lemmas help}.
\end{proof}

\begin{proof}[{\textbf{Proof of Proposition~\ref{lem:LLN CLT}:}}] 
Plugging in $\widehat{\vv}_n(r,s)=\mD_n\big[\widehat{\valpha}_n(r,s)-\valpha_0\big]$ for $\vv$ in Lemma~\ref{lem:7} yields by Theorem~\ref{thm:std est} that, as $n\to\infty$,
\[
	\overline{\mV}_n\big(\widehat{\vv}_n(r,s), r,s\big)\overset{d}{\longrightarrow}\big[(1-\beta)\mK - \mK_{\dagger}\big] \mSigma^{1/2}\big[\mW(s)-\mW(r)\big]\qquad\text{in}\ (\ell^{\infty}(\mathcal{D}_{\iota}))^{k+1}.
\]
It also holds that
\begin{align*}
	\P&\bigg\{\sup_{0\leq r<s\leq1}\Big\Vert \mV_n\big(\widehat{\vv}_n(r,s), r,s\big) - \overline{\mV}_n(\widehat{\vv}_n(r,s), r,s) \Big\Vert>\varepsilon\bigg\} \\
	&= \P\bigg\{\sup_{0\leq r<s\leq1}\Big\Vert \mV_n\big(\widehat{\vv}_n(r,s), r,s\big) - \overline{\mV}_n\big(\widehat{\vv}_n(r,s), r,s\big) \Big\Vert>\varepsilon, \sup_{0\leq r<s\leq1}\big\Vert\widehat{\vv}_n(r,s)\big\Vert\leq K\bigg\}\\
	&\hspace{1.5cm} + \P\bigg\{\sup_{0\leq r<s\leq1}\Big\Vert \mV_n\big(\widehat{\vv}_n(r,s), r,s\big) - \overline{\mV}_n\big(\widehat{\vv}_n(r,s), r,s\big) \Big\Vert>\varepsilon, \sup_{0\leq r<s\leq1}\big\Vert\widehat{\vv}_n(r,s)\big\Vert> K\bigg\}\\
	& \leq \P\bigg\{\sup_{\Vert\vv\Vert\leq K}\sup_{0\leq r<s\leq1}\big\Vert \mV_n(\vv, r,s) - \overline{\mV}_n(\vv, r,s) \big\Vert>\varepsilon\bigg\}+ \P\bigg\{ \sup_{0\leq r<s\leq1}\big\Vert\widehat{\vv}_n(r,s)\big\Vert> K\bigg\}\\
	&=o(1)+o(1),
\end{align*}
as $n\to\infty$, followed by $K\to\infty$, where the last line follows from Lemma~\ref{lem:9} and Theorem~\ref{thm:std est}.
Therefore, $\sup_{0\leq r<s\leq1}\big\Vert \mV_n(\widehat{\vv}_n(r,s), r,s) - \overline{\mV}_n(\widehat{\vv}_n(r,s), r,s) \big\Vert=o_{\P}(1)$, whence
\begin{align*}
	\mV_n\big(\widehat{\vv}_n(r,s), r,s\big) &= \overline{\mV}_n\big(\widehat{\vv}_n(r,s), r,s\big) + \Big[\mV_n\big(\widehat{\vv}_n(r,s), r,s\big) - \overline{\mV}_n\big(\widehat{\vv}_n(r,s), r,s\big)\Big]\\
	&=\overline{\mV}_n\big(\widehat{\vv}_n(r,s), r,s\big) + o_{\P}(1)\\
	&\overset{d}{\longrightarrow}\big[(1-\beta)\mK - \mK_{\dagger}\big] \mSigma^{1/2}\big[\mW(s)-\mW(r)\big]\qquad\text{in}\ (\ell^{\infty}(\mathcal{D}_{\iota}))^{k+1}.
\end{align*}
This is the claimed result.
\end{proof}

\begin{proof}[{\textbf{Proof of Proposition~\ref{lem:LLN3}:}}] 
We only show that
\[
	\sup_{(r,s)\in\mathcal{D}_{\iota}}\bigg|\sum_{t=\lfloor nr\rfloor+1}^{\lfloor ns\rfloor}\1_{\{0<\epsilon_t\leq \mX_{t-1}^\prime[\widehat{\valpha}_n(r,s) - \valpha_0]\}}\big(\vw^\prime\mD_n^{-1}\mX_{t-1} - \delta_{t}\big) \1_{\{0<\delta_{t}<\vw^\prime\mD_n^{-1}\mX_{t-1}\}}\bigg|=o_{\P}(1),
\]
as the convergences involving the other indicator functions follow similarly.
It holds that
\begin{align*}
	\P&\bigg\{\sup_{(r,s)\in\mathcal{D}_{\iota}}\bigg|\sum_{t=\lfloor nr\rfloor+1}^{\lfloor ns\rfloor}\1_{\{0<\epsilon_t\leq \mX_{t-1}^\prime[\widehat{\valpha}_n(r,s) - \valpha_0]\}}\big(\vw^\prime\mD_n^{-1}\mX_{t-1} - \delta_{t}\big) \1_{\{0<\delta_{t}<\vw^\prime\mD_n^{-1}\mX_{t-1}\}}\bigg|>\varepsilon\bigg\}\\
	&\leq \P\bigg\{\sup_{(r,s)\in\mathcal{D}_{\iota}}\bigg|\sum_{t=\lfloor nr\rfloor+1}^{\lfloor ns\rfloor}\1_{\{0<\epsilon_t\leq \mX_{t-1}^\prime[\widehat{\valpha}_n(r,s) - \valpha_0]\}}\big(\vw^\prime\mD_n^{-1}\mX_{t-1} - \delta_{t}\big) \1_{\{0<\delta_{t}<\vw^\prime\mD_n^{-1}\mX_{t-1}\}}\bigg|>\varepsilon,\\
	&\hspace{9cm}\sup_{(r,s)\in\mathcal{D}_{\iota}}\big\Vert\mD_n[\widehat{\valpha}_n(r,s) - \valpha_0]\big\Vert\leq K\bigg\}\\
	&\hspace{1cm} + \P\bigg\{\sup_{(r,s)\in\mathcal{D}_{\iota}}\big\Vert\mD_n[\widehat{\valpha}_n(r,s) - \valpha_0]\big\Vert> K\bigg\}\\
	&\leq \P\bigg\{\sup_{(r,s)\in\mathcal{D}_{\iota}}\bigg|\sum_{t=\lfloor nr\rfloor+1}^{\lfloor ns\rfloor}\1_{\{0<\epsilon_t\leq K\Vert\mD_n^{-1}\mX_{t-1}\Vert\}}\big(\vw^\prime\mD_n^{-1}\mX_{t-1} - \delta_{t}\big) \1_{\{0<\delta_{t}<\vw^\prime\mD_n^{-1}\mX_{t-1}\}}\bigg|>\varepsilon\bigg\} + o(1)\\
	&\leq \P\bigg\{\sum_{t=1}^{n}\1_{\{0<\epsilon_t\leq K\Vert\mD_n^{-1}\mX_{t-1}\Vert\}}\big(\vw^\prime\mD_n^{-1}\mX_{t-1} - \delta_{t}\big) \1_{\{0<\delta_{t}<\vw^\prime\mD_n^{-1}\mX_{t-1}\}}>\varepsilon\bigg\} + o(1),
\end{align*}
as $n\to\infty$, followed by $K\to\infty$, where the penultimate step follows from Theorem~\ref{thm:std est}, which implies that $\sup_{(r,s)\in\mathcal{D}_{\iota}}\big\Vert\mD_n[\widehat{\valpha}_n(r,s) - \valpha_0]\big\Vert=O_{\P}(1)$.

In light of this, it suffices to show that
\begin{equation}\label{eq:(p.21)}
	\sum_{t=1}^{n}\omega_t(\vw):=\sum_{t=1}^{n}\1_{\{0<\epsilon_t\leq K\Vert\mD_n^{-1}\mX_{t-1}\Vert\}}\big(\vw^\prime\mD_n^{-1}\mX_{t-1} - \delta_{t}\big) \1_{\{0<\delta_{t}<\vw^\prime\mD_n^{-1}\mX_{t-1}\}}=o_{\P}(1).
\end{equation}
To do so, use Assumption~\ref{ass:innov CoVaR}~\eqref{it:dens bound CoVaR} to deduce that
\begin{align*}
	\E_{t-1}\big[\omega_t(\vw)\big]&= \int_{0}^{K\Vert \mD_n^{-1}\mX_{t-1}\Vert}\bigg[\int_{0}^{\vw^\prime\mD_n^{-1}\mX_{t-1}}(\vw^\prime\mD_n^{-1}\mX_{t-1}-y)f_{(\epsilon_t,\delta_t)^\prime\mid\mathcal{F}_{t-1}}(x,y)\D y\bigg]\D x\\
	&\leq \int_{0}^{K\Vert \mD_n^{-1}\mX_{t-1}\Vert}\bigg[\frac{\overline{f}}{2}(\vw^\prime\mD_n^{-1}\mX_{t-1})^2\bigg]\D x\\
	&\leq K\Vert\mD_n^{-1}\mX_{t-1}\Vert \vw^\prime\big(\mD_n^{-1}\mX_{t-1}\mX_{t-1}^\prime\mD_n^{-1}\big)\vw.
\end{align*}
Therefore,
\begin{align*}
	\sum_{t=1}^{n}\E_{t-1}\big[\omega_t(\vw)\big]&\leq K \max_{t=1,\ldots,n}\Vert\mD_n^{-1}\mX_{t-1}\Vert\vw^\prime\bigg(\sum_{t=1}^{n}\mD_n^{-1}\mX_{t-1}\mX_{t-1}^\prime\mD_n^{-1}\bigg)\vw\\
	&=o_{\P}(1)O_{\P}(1)\\
	&=o_{\P}(1)
\end{align*}
by Lemmas~\ref{lem:MAX}--\ref{lem:SUM}.
In view of this, \eqref{eq:(p.21)} follows if we can show that, as $n\to\infty$,
\begin{equation}\label{eq:(p.4.L.7)}
	\sum_{t=1}^{n}\Big\{\omega_t(\vw) - \E_{t-1}\big[\omega_t(\vw)\big]\Big\}=o_{\P}(1).
\end{equation}
To prove this, we use Corollary~3.1 of \citet{HH80}.
Note for this that $\big\{\omega_t(\vw) - \E_{t-1}\big[\omega_t(\vw)\big]\big\}$ is a sequence of MDAs by construction.
First, the sum of the conditional variances is asymptotically negligible, because from by now familiar arguments,
\begin{align*}
	\sum_{t=1}^{n}&\E_{t-1}\Big[\big\{\omega_t(\vw) - \E_{t-1}[\omega_t(\vw)]\big\}^2\Big]\\
	&\leq \sum_{t=1}^{n}\E_{t-1}\big[\omega_t^2(\vw)\big]\\
	&=\sum_{t=1}^{n}\E_{t-1}\Big[\1_{\{0<\epsilon_t\leq K\Vert\mD_n^{-1}\mX_{t-1}\Vert\}}\big(\vw^\prime\mD_n^{-1}\mX_{t-1} - \delta_{t}\big)^2 \1_{\{0<\delta_{t}<\vw^\prime\mD_n^{-1}\mX_{t-1}\}}\Big]\\
	&=\sum_{t=1}^{n} \int_{0}^{K\Vert\mD_n^{-1}\mX_{t-1}\Vert}\bigg[\int_{0}^{\vw^\prime\mD_n^{-1}\mX_{t-1}}(\vw^\prime\mD_n^{-1}\mX_{t-1} - y)^2f_{(\epsilon_t,\delta_t)^\prime\mid\mathcal{F}_{t-1}}(x,y)\D y\bigg]\D x\\
	&\leq \sum_{t=1}^{n} \int_{0}^{K\Vert\mD_n^{-1}\mX_{t-1}\Vert}\bigg[ \int_{0}^{\vw^\prime\mD_n^{-1}\mX_{t-1}}\overline{f}(\vw^\prime\mD_n^{-1}\mX_{t-1})^2\D y\bigg]\D x\\
	&\leq K\max_{t=1,\ldots,n}\Vert\mD_n^{-1}\mX_{t-1}\Vert^2\vw^\prime\bigg(\sum_{t=1}^{n}\mD_n^{-1}\mX_{t-1}\mX_{t-1}^\prime\mD_n^{-1}\bigg)\vw\\
	&= o_{\P}(1) O_{\P}(1)\\
	&= o_{\P}(1).
\end{align*}
Second, the CLC follows from the LC. 
The LC follows similarly as below \eqref{eq:(p.12.LC)}, since
\begin{align*}
	\big|\omega_t(\vw) - \E_{t-1}[\omega_t(\vw)]\big|^2 &\leq \big|2\vw^\prime\mD_n^{-1}\mX_{t-1}\big|^2\\
	&\leq K\Vert\mD_n^{-1}\mX_{t-1}\Vert^2\\
	&\leq 2K\max\big\{n^{-1}, n^{-(1+\kappa)}\Vert\vx_{t-1}\Vert^2\big\}.
\end{align*}
Therefore, Corollary~3.1 of \citet{HH80} implies \eqref{eq:(p.4.L.7)}, concluding the proof.
\end{proof}

\subsection{Proofs of Lemmas~\ref{lem:7}--\ref{lem:9}}\label{sec:CoVaRest Lemmas help}

It suffices to prove the convergences in Lemmas~\ref{lem:7}--\ref{lem:9} on the set $\big\{\max_{t=1,\ldots,n}\Vert\mD_n^{-1}\mX_{t-1}\Vert\leq d^\ast\big\}$, where $d^\ast=d^{\ast}(d,\vv)$ is chosen to ensure that $\max_{t=1,\ldots,n}|\vv^\prime\mD_n^{-1}\mX_{t-1}|\leq d$.
This suffices because Lemma~\ref{lem:MAX} implies that $\P\big\{\max_{t=1,\ldots,n}\Vert\mD_n^{-1}\mX_{t-1}\Vert\leq d^\ast\big\}\longrightarrow1$, as $n\to\infty$.
Once again, working on the set $\big\{\max_{t=1,\ldots,n}\Vert\mD_n^{-1}\mX_{t-1}\Vert\leq d^\ast\big\}$ ensures that the densities of Assumptions~\ref{ass:innov} and \ref{ass:innov CoVaR} exist and satisfy the properties set out in Appendices~\ref{Assumptions on the Quantile Regression}--\ref{Assumptions on the CoVaR Regression}.

\begin{proof}[{\textbf{Proof of Lemma~\ref{lem:7}:}}] 
Since
\begin{align*}
\E_{t-1}&\Big[\big(\1_{\{\vv^\prime\mD_n^{-1}\mX_{t-1}<\epsilon_t\leq0\}} - \1_{\{0<\epsilon_t\leq\vv^\prime\mD_n^{-1}\mX_{t-1}\}}\big)\psi_{\beta}(\delta_t)\Big]\\
&=\E_{t-1}\Big[\big(\1_{\{\vv^\prime\mD_n^{-1}\mX_{t-1}<\epsilon_t\leq0\}} - \1_{\{0<\epsilon_t\leq\vv^\prime\mD_n^{-1}\mX_{t-1}\}}\big)\big(\beta-1+\1_{\{\delta_t>0\}})\Big]\\
&=(\beta-1)\1_{\{\vv^\prime\mD_n^{-1}\mX_{t-1}<0\}}\int_{\vv^\prime\mD_n^{-1}\mX_{t-1}}^{0}f_{\epsilon_t\mid\mathcal{F}_{t-1}}(x)\D x \\
&\hspace{1cm}- (\beta-1)\1_{\{\vv^\prime\mD_n^{-1}\mX_{t-1}>0\}}\int_{0}^{\vv^\prime\mD_n^{-1}\mX_{t-1}}f_{\epsilon_t\mid\mathcal{F}_{t-1}}(x)\D x\\
&\hspace{0.2cm} + \1_{\{\vv^\prime\mD_n^{-1}\mX_{t-1}<0\}}\int_{\vv^\prime\mD_n^{-1}\mX_{t-1}}^{0}\bigg(\int_{0}^{\infty}f_{(\epsilon_t,\delta_t)^\prime\mid\mathcal{F}_{t-1}}(x,y)\D y\bigg)\D x\\
&\hspace{1cm} - \1_{\{\vv^\prime\mD_n^{-1}\mX_{t-1}>0\}}\int_{0}^{\vv^\prime\mD_n^{-1}\mX_{t-1}}\bigg(\int_{0}^{\infty}f_{(\epsilon_t,\delta_t)^\prime\mid\mathcal{F}_{t-1}}(x,y)\D y\bigg)\D x,
\end{align*}
we may write
\begin{align}
	\overline{\mV}_n(\vv,r,s) &= \sum_{t=\lfloor nr\rfloor+1}^{\lfloor ns\rfloor}\E_{t-1}\Big[\big(\1_{\{\vv^\prime\mD_n^{-1}\mX_{t-1}<\epsilon_t\leq0\}} - \1_{\{0<\epsilon_t\leq\vv^\prime\mD_n^{-1}\mX_{t-1}\}}\big)\psi_{\beta}(\delta_t)\Big]\mD_n^{-1}\mX_{t-1}\notag\\	
	&=\bigg[(\beta-1)\sum_{t=\lfloor nr\rfloor+1}^{\lfloor ns\rfloor} \1_{\{\vv^\prime\mD_n^{-1}\mX_{t-1}<0\}}\int_{\vv^\prime\mD_n^{-1}\mX_{t-1}}^{0}f_{\epsilon_t\mid\mathcal{F}_{t-1}}(x)\D x\,\mD_n^{-1}\mX_{t-1}\notag\\
	&\hspace{1cm} - (\beta-1)\sum_{t=\lfloor nr\rfloor+1}^{\lfloor ns\rfloor} \1_{\{\vv^\prime\mD_n^{-1}\mX_{t-1}>0\}}\int_{0}^{\vv^\prime\mD_n^{-1}\mX_{t-1}}f_{\epsilon_t\mid\mathcal{F}_{t-1}}(x)\D x\,\mD_n^{-1}\mX_{t-1}\bigg]\notag\\
	&\hspace{0.2cm} + \bigg[\sum_{t=\lfloor nr\rfloor+1}^{\lfloor ns\rfloor}\1_{\{\vv^\prime\mD_n^{-1}\mX_{t-1}<0\}}\int_{\vv^\prime\mD_n^{-1}\mX_{t-1}}^{0}\bigg(\int_{0}^{\infty}f_{(\epsilon_t,\delta_t)^\prime\mid\mathcal{F}_{t-1}}(x,y)\D y\bigg)\D x\,\mD_n^{-1}\mX_{t-1}\notag\\
	&\hspace{1cm} - \sum_{t=\lfloor nr\rfloor+1}^{\lfloor ns\rfloor}\1_{\{\vv^\prime\mD_n^{-1}\mX_{t-1}>0\}}\int_{0}^{\vv^\prime\mD_n^{-1}\mX_{t-1}}\bigg(\int_{0}^{\infty}f_{(\epsilon_t,\delta_t)^\prime\mid\mathcal{F}_{t-1}}(x,y)\D y\bigg)\D x\,\mD_n^{-1}\mX_{t-1}\bigg]\notag\\
	&=: \big[\overline{\mV}_{1n}(\vv,r,s) - \overline{\mV}_{2n}(\vv,r,s)\big] + \big[\overline{\mV}_{3n}(\vv,r,s) - \overline{\mV}_{4n}(\vv,r,s)\big].\label{eq:bV decomp}
\end{align}
Consider each $\overline{\mV}_{in}(\vv,r,s)$ $(i=1,\ldots,4)$ separately. First,
\begin{align*}
	\overline{\mV}_{1n}(\vv,r,s) &= (\beta-1)\sum_{t=\lfloor nr\rfloor+1}^{\lfloor ns\rfloor} \1_{\{\vv^\prime\mD_n^{-1}\mX_{t-1}<0\}}\int_{\vv^\prime\mD_n^{-1}\mX_{t-1}}^{0}f_{\epsilon_t\mid\mathcal{F}_{t-1}}(0)\D x\,\mD_n^{-1}\mX_{t-1}\\
	&\hspace{1cm} + (\beta-1)\sum_{t=\lfloor nr\rfloor+1}^{\lfloor ns\rfloor} \1_{\{\vv^\prime\mD_n^{-1}\mX_{t-1}<0\}}\int_{\vv^\prime\mD_n^{-1}\mX_{t-1}}^{0}\big[f_{\epsilon_t\mid\mathcal{F}_{t-1}}(x) - f_{\epsilon_t\mid\mathcal{F}_{t-1}}(0)\big]\D x\,\mD_n^{-1}\mX_{t-1}\\
	&=: \overline{\mV}_{11n}(\vv,r,s) + \overline{\mV}_{12n}(\vv,r,s).
\end{align*}
Simple integration yields that
\begin{align*}
	\overline{\mV}_{11n}(\vv,r,s) &= (1-\beta)\sum_{t=\lfloor nr\rfloor+1}^{\lfloor ns\rfloor} \1_{\{\vv^\prime\mD_n^{-1}\mX_{t-1}<0\}}f_{\epsilon_t\mid\mathcal{F}_{t-1}}(0)\big(\vv^\prime\mD_n^{-1}\mX_{t-1}\big)\big(\mD_n^{-1}\mX_{t-1}\big)\\
	&= (1-\beta)\sum_{t=\lfloor nr\rfloor+1}^{\lfloor ns\rfloor} \1_{\{\vv^\prime\mD_n^{-1}\mX_{t-1}<0\}}f_{\epsilon_t\mid\mathcal{F}_{t-1}}(0)\big(\mD_n^{-1}\mX_{t-1}\mX_{t-1}^\prime\mD_n^{-1}\big)\vv.
\end{align*}
Moreover, using Assumption~\ref{ass:innov}~\eqref{it:Lipschitz} and Lemmas~\ref{lem:MAX} and \ref{lem:SUM},
\begin{align*}
	\big\Vert\overline{\mV}_{12n}(\vv,r,s)\big\Vert &\leq |\beta-1|\sum_{t=\lfloor nr\rfloor+1}^{\lfloor ns\rfloor} \1_{\{\vv^\prime\mD_n^{-1}\mX_{t-1}<0\}}\int_{\vv^\prime\mD_n^{-1}\mX_{t-1}}^{0}|x-0|\frac{\big|f_{\epsilon_t\mid\mathcal{F}_{t-1}}(x) - f_{\epsilon_t\mid\mathcal{F}_{t-1}}(0)\big|}{|x-0|}\D x\times\\
	&\hspace{11cm}\times \big\Vert\mD_n^{-1}\mX_{t-1}\big\Vert\\
	&\leq L|\beta-1|\sum_{t=1}^{n}\frac{1}{2}\big(\vv^\prime\mD_n^{-1}\mX_{t-1}\big)^2\big\Vert\mD_n^{-1}\mX_{t-1}\big\Vert\\
	&\leq \frac{L|\beta-1|}{2}\max_{t=1,\ldots,n}\Vert\mD_n^{-1}\mX_{t-1}\Vert \vv^\prime\bigg(\sum_{t=1}^{n}\mD_n^{-1}\mX_{t-1}\mX_{t-1}^\prime\mD_n^{-1}\bigg)\vv\\
	&=o_{\P}(1)
\end{align*}
uniformly in $0\leq r<s\leq 1$ and $\Vert\vv\Vert\leq K$.
The previous three displays imply that
\begin{equation}\label{eq:bV1}
	\overline{\mV}_{1n}(\vv,r,s) = (1-\beta)\sum_{t=\lfloor nr\rfloor+1}^{\lfloor ns\rfloor} \1_{\{\vv^\prime\mD_n^{-1}\mX_{t-1}<0\}}f_{\epsilon_t\mid\mathcal{F}_{t-1}}(0)\big(\mD_n^{-1}\mX_{t-1}\mX_{t-1}^\prime\mD_n^{-1}\big)\vv+o_{\P}(1)
\end{equation}
uniformly in $0\leq r<s\leq 1$ and $\Vert\vv\Vert\leq K$.

By similar arguments we obtain that
\begin{equation}\label{eq:bV2}
	\overline{\mV}_{2n}(\vv,r,s) = (\beta-1)\sum_{t=\lfloor nr\rfloor+1}^{\lfloor ns\rfloor} \1_{\{\vv^\prime\mD_n^{-1}\mX_{t-1}>0\}}f_{\epsilon_t\mid\mathcal{F}_{t-1}}(0)\big(\mD_n^{-1}\mX_{t-1}\mX_{t-1}^\prime\mD_n^{-1}\big)\vv+o_{\P}(1)
\end{equation}
uniformly in $0\leq r < s \leq 1$ and $\Vert\vv\Vert\leq K$.

For $\overline{\mV}_{3n}(\vv,r,s)$ we get the decomposition
\begin{align*}
	\overline{\mV}_{3n}(\vv,r,s) &= \sum_{t=\lfloor nr\rfloor+1}^{\lfloor ns\rfloor}\1_{\{\vv^\prime\mD_n^{-1}\mX_{t-1}<0\}}\int_{\vv^\prime\mD_n^{-1}\mX_{t-1}}^{0}\bigg(\int_{0}^{\infty}f_{(\epsilon_t,\delta_t)^\prime\mid\mathcal{F}_{t-1}}(0,y)\D y\bigg)\D x\,\mD_n^{-1}\mX_{t-1}\\
	&\hspace{0.5cm} + \sum_{t=\lfloor nr\rfloor+1}^{\lfloor ns\rfloor}\1_{\{\vv^\prime\mD_n^{-1}\mX_{t-1}<0\}}\int_{\vv^\prime\mD_n^{-1}\mX_{t-1}}^{0}\bigg(\int_{0}^{\infty}f_{(\epsilon_t,\delta_t)^\prime\mid\mathcal{F}_{t-1}}(x,y)\D y\\
	&\hspace{7cm}- \int_{0}^{\infty}f_{(\epsilon_t,\delta_t)^\prime\mid\mathcal{F}_{t-1}}(0,y)\D y\bigg)\D x\,\mD_n^{-1}\mX_{t-1}\\
	&=: \overline{\mV}_{31n}(\vv,r,s) + \overline{\mV}_{32n}(\vv,r,s).
\end{align*}
For the first right-hand side term, it follows that
\begin{align*}
\overline{\mV}_{31n}(\vv,r,s) &= -\sum_{t=\lfloor nr\rfloor+1}^{\lfloor ns\rfloor}\1_{\{\vv^\prime\mD_n^{-1}\mX_{t-1}<0\}}\bigg(\int_{0}^{\infty}f_{(\epsilon_t,\delta_t)^\prime\mid\mathcal{F}_{t-1}}(0,y)\D y\bigg)\big(\vv^\prime\mD_n^{-1}\mX_{t-1}\big)\mD_n^{-1}\mX_{t-1}\\
&= -\sum_{t=\lfloor nr\rfloor+1}^{\lfloor ns\rfloor}\1_{\{\vv^\prime\mD_n^{-1}\mX_{t-1}<0\}}\bigg(\int_{0}^{\infty}f_{(\epsilon_t,\delta_t)^\prime\mid\mathcal{F}_{t-1}}(0,y)\D y\bigg)\big(\mD_n^{-1}\mX_{t-1}\mX_{t-1}^\prime\mD_n^{-1}\big)\vv
\end{align*}
and, by Assumption~\ref{ass:innov CoVaR}~\eqref{it:Lipschitz CoVaR} and Lemmas~\ref{lem:MAX} and \ref{lem:SUM},
\begin{align*}
	\big\Vert\overline{\mV}_{32n}(\vv,r,s)\big\Vert &\leq \sum_{t=\lfloor nr\rfloor+1}^{\lfloor ns\rfloor}\1_{\{\vv^\prime\mD_n^{-1}\mX_{t-1}<0\}}\int_{\vv^\prime\mD_n^{-1}\mX_{t-1}}^{0}|x-0|\times\\
	&\hspace{2cm}\times\frac{\Big|\int_{0}^{\infty}f_{(\epsilon_t,\delta_t)^\prime\mid\mathcal{F}_{t-1}}(x,y)\D y- \int_{0}^{\infty}f_{(\epsilon_t,\delta_t)^\prime\mid\mathcal{F}_{t-1}}(0,y)\D y\Big|}{|x-0|}\D x\,\big\Vert\mD_n^{-1}\mX_{t-1}\big\Vert\\
	&\leq L \sum_{t=1}^{n} \frac{1}{2}\big(\vv^\prime\mD_n^{-1}\mX_{t-1}\big)^2 \big\Vert\mD_n^{-1}\mX_{t-1}\big\Vert\\
	&\leq \frac{L}{2}\max_{t=1,\ldots,n}\Vert\mD_n^{-1}\mX_{t-1}\Vert \vv^\prime\bigg(\sum_{t=1}^{n}\mD_n^{-1}\mX_{t-1}\mX_{t-1}^\prime\mD_n^{-1}\bigg)\vv\\
	&=o_{\P}(1)
\end{align*}
uniformly in $0\leq r < s\leq 1$ and $\Vert\vv\Vert\leq K$.
Therefore,
\begin{equation}\label{eq:bV3}
	\overline{\mV}_{3n}(\vv,s) = -\sum_{t=\lfloor nr\rfloor+1}^{\lfloor ns\rfloor}\1_{\{\vv^\prime\mD_n^{-1}\mX_{t-1}<0\}}\bigg(\int_{0}^{\infty}f_{(\epsilon_t,\delta_t)^\prime\mid\mathcal{F}_{t-1}}(0,y)\D y\bigg)\big(\mD_n^{-1}\mX_{t-1}\mX_{t-1}^\prime\mD_n^{-1}\big)\vv + o_{\P}(1)
\end{equation}
uniformly in $0\leq r < s\leq 1$ and $\Vert\vv\Vert\leq K$.

Similarly, 
\begin{equation}\label{eq:bV4}
	\overline{\mV}_{4n}(\vv,s) = \sum_{t=\lfloor nr\rfloor+1}^{\lfloor ns\rfloor}\1_{\{\vv^\prime\mD_n^{-1}\mX_{t-1}>0\}}\bigg(\int_{0}^{\infty}f_{(\epsilon_t,\delta_t)^\prime\mid\mathcal{F}_{t-1}}(0,y)\D y\bigg)\big(\mD_n^{-1}\mX_{t-1}\mX_{t-1}^\prime\mD_n^{-1}\big)\vv + o_{\P}(1)
\end{equation}
uniformly in $0\leq r < s\leq 1$ and $\Vert\vv\Vert\leq K$.

Plugging \eqref{eq:bV1}--\eqref{eq:bV4} into \eqref{eq:bV decomp} yields that
\begin{multline}\label{eq:(p.30)}
	\overline{\mV}_n(\vv,r,s) =(1-\beta)\sum_{t=\lfloor nr\rfloor+1}^{\lfloor ns\rfloor} f_{\epsilon_t\mid\mathcal{F}_{t-1}}(0)\big(\mD_n^{-1}\mX_{t-1}\mX_{t-1}^\prime\mD_n^{-1}\big)\vv\\
	-\sum_{t=\lfloor nr\rfloor+1}^{\lfloor ns\rfloor}\bigg(\int_{0}^{\infty}f_{(\epsilon_t,\delta_t)^\prime\mid\mathcal{F}_{t-1}}(0,y)\D y\bigg)\big(\mD_n^{-1}\mX_{t-1}\mX_{t-1}^\prime\mD_n^{-1}\big)\vv + o_{\P}(1)
\end{multline}
uniformly in $0\leq r < s\leq 1$ and $\Vert\vv\Vert\leq K$, such that the conclusion follows from Assumptions~\ref{ass:K} and \ref{ass:K ast}.
\end{proof}

\begin{proof}[{\textbf{Proof of Lemma~\ref{lem:8}:}}] 
It is easy to check that
\begin{align*}
	&\sup_{0\leq r<s\leq 1}\big\Vert \mV_n(\vv,r,s)-\overline{\mV}_n(\vv,r,s) \big\Vert\\
	&=\sup_{0\leq r<s\leq 1}\Big\Vert \mV_n(\vv,0,s)-\overline{\mV}_n(\vv,0,s) - \big[\mV_n(\vv,0,r)-\overline{\mV}_n(\vv,0,r)\big] \Big\Vert\\
	&\leq 2 \sup_{s\in[0,1]}\big\Vert \mV_n(\vv,0,s)-\overline{\mV}_n(\vv,0,s) \big\Vert.
\end{align*}
Therefore, it suffices to show that
\[
	\sup_{s\in[0,1]}\big\Vert \mV_n(\vv,0,s)-\overline{\mV}_n(\vv,0,s) \big\Vert=o_{\P}(1).
\]
The proof of this relation is similar to that of \eqref{eq:(P.7.1)}.
Recall the definitions of $\vnu_t(\vv)$ and $\overline{\vnu}_t(\vv)$ above Lemma~\ref{lem:7}.
To apply Theorem~3.33 of \citet[Chapter~VIII]{JS87} to the (vector) MDA $\big\{\vnu_t(\vv)-\overline{\vnu}_t(\vv)\big\}$, we first show asymptotic negligibility of the sum of the conditional variances:
\begin{align*}
	\sum_{t=1}^{\lfloor ns\rfloor}&\E_{t-1}\Big[\big\{\vnu_t(\vv)-\overline{\vnu}_t(\vv)\big\}\big\{\vnu_t(\vv)-\overline{\vnu}_t(\vv)\big\}^\prime\Big]\\
	&= \sum_{t=1}^{\lfloor ns\rfloor}\Big\{\E_{t-1}\big[\vnu_t(\vv)\vnu_t^\prime(\vv)\big] - \overline{\vnu}_t(\vv)\overline{\vnu}_t^\prime(\vv)\Big\}\\
	&\leq \sum_{t=1}^{\lfloor ns\rfloor}\E_{t-1}\big[\vnu_t(\vv)\vnu_t^\prime(\vv)\big]\\
	&= \sum_{t=1}^{\lfloor ns\rfloor}\E_{t-1}\Big[\big(\1_{\{\vv^\prime\mD_n^{-1}\mX_{t-1}<\epsilon_t\leq0\}} - \1_{\{0<\epsilon_t\leq\vv^\prime\mD_n^{-1}\mX_{t-1}\}}\big)^2\psi_{\beta}^2(\delta_t)\Big]\mD_n^{-1}\mX_{t-1}\mX_{t-1}^\prime\mD_n^{-1}\\
	&\leq \sum_{t=1}^{\lfloor ns\rfloor}\E_{t-1}\Big[\big(\1_{\{\vv^\prime\mD_n^{-1}\mX_{t-1}<\epsilon_t\leq0\}} + \1_{\{0<\epsilon_t\leq\vv^\prime\mD_n^{-1}\mX_{t-1}\}}\big)\Big]\mD_n^{-1}\mX_{t-1}\mX_{t-1}^\prime\mD_n^{-1}\\
	&= \sum_{t=1}^{\lfloor ns\rfloor}\bigg[\1_{\{\vv^\prime\mD_n^{-1}\mX_{t-1}<0\}}\int_{\vv^\prime\mD_n^{-1}\mX_{t-1}}^{0}f_{\epsilon_t\mid\mathcal{F}_{t-1}}(x)\D x\\
	&\hspace{3cm}+ \1_{\{\vv^\prime\mD_n^{-1}\mX_{t-1}>0\}}\int_{0}^{\vv^\prime\mD_n^{-1}\mX_{t-1}}f_{\epsilon_t\mid\mathcal{F}_{t-1}}(x)\D x\bigg]\mD_n^{-1}\mX_{t-1}\mX_{t-1}^\prime\mD_n^{-1}\\
	&\leq \sum_{t=1}^{\lfloor ns\rfloor}2\overline{f}|\vv^\prime\mD_n^{-1}\mX_{t-1}|\mD_n^{-1}\mX_{t-1}\mX_{t-1}^\prime\mD_n^{-1}\\
	&\leq K \max_{t=1,\ldots,n}\Vert\mD_n^{-1}\mX_{t-1}\Vert \sum_{t=1}^{\lfloor ns\rfloor}\mD_n^{-1}\mX_{t-1}\mX_{t-1}^\prime\mD_n^{-1}\\
	&=Ko_{\P}(1)O_{\P}(1)\\
	&=o_{\P}(1),
\end{align*}
where the above inequalities are to be understood with respect to the Loewner order, and the penultimate step follows from Lemmas~\ref{lem:MAX} and \ref{lem:SUM}.

Second, we verify the CLC
\[
	\sum_{t=1}^{n}\E_{t-1}\Big[\Vert\vnu_t(\vv) - \overline{\vnu}_t(\vv)\Vert^2\1_{\{\Vert\vnu_t(\vv) - \overline{\vnu}_t(\vv)\Vert^2>\delta^2\}}\Big]\overset{\P}{\underset{(n\to\infty)}{\longrightarrow}}0
\]
by verifying the LC
\[
	\sum_{t=1}^{n}\E\Big[\Vert\vnu_t(\vv) - \overline{\vnu}_t(\vv)\Vert^2\1_{\{\Vert\vnu_t(\vv) - \overline{\vnu}_t(\vv)\Vert^2>\delta^2\}}\Big]\underset{(n\to\infty)}{\longrightarrow}0.
\]
Since
\begin{align*}
	\Vert\vnu_t(\vv) - \overline{\vnu}_t(\vv)\Vert^2 &\leq K\Vert\mD_n^{-1}\mX_{t-1}\Vert^2\\
	&\leq K\max\big\{n^{-1}, n^{-(1+\kappa)}\Vert\vx_{t-1}\Vert^2\big\},
\end{align*}
where $\kappa=0$ corresponds to the (I0) case, we only have to verify the LC for $n^{-(1+\kappa)}\Vert\vx_{t-1}\Vert^2$.
Once again, this follows as in the proof of Proposition~\ref{lem:CLT} via Proposition~A1~(ii) of \citet{MP20}.

Thus, we may apply Theorem~3.33 of \citet[Chapter~VIII]{JS87} to conclude that
\[
	\mV_n(\vv,0,s) - \overline{\mV}_n(\vv,0,s)\overset{d}{\longrightarrow}\vzero\qquad\text{in }(D[0,1])^{k+1},
\]
from which $\sup_{s\in[0,1]}\Vert\mV_n(\vv,0,s) - \overline{\mV}_n(\vv,0,s)\Vert=o_{\P}(1)$ follows via the CMT.
\end{proof}

\begin{proof}[{\textbf{Proof of Lemma~\ref{lem:9}:}}] 
By a similar argument as in the proof of Lemma~\ref{lem:8}, it suffices to show that
\[
	\sup_{\Vert\vv\Vert\leq K}\sup_{0\leq s\leq1}\big\Vert\mV_n(\vv,0,s) - \overline{\mV}_n(\vv,0,s)\big\Vert=o_{\P}(1).
\]
We do so for real-valued $\mX_t$ here, as the vector-valued case is only notationally more complicated.
To reflect this in the notation, we write $X_t$ instead of $\mX_t$, $V_n(v,0,s) - \overline{V}_n(v,0,s)$ instead of $\mV_n(\vv,0,s) - \overline{\mV}_n(\vv,0,s)$, etc.
Decompose
\[
	V_n(v,0,s) - \overline{V}_n(v,0,s) = \sum_{t=1}^{\lfloor ns\rfloor}\big[\nu_{1t}(v)- \overline{\nu}_{1t}(v)\big] - \sum_{t=1}^{\lfloor ns\rfloor}\big[\nu_{2t}(v)- \overline{\nu}_{2t}(v)\big],
\]
where
\begin{align*}
	\nu_{1t}(v) &:= \1_{\{v^\prime D_n^{-1} X_{t-1}<\epsilon_t\leq 0\}}\psi_{\beta}(\delta_t)D_n^{-1}X_{t-1},\\
	\overline{\nu}_{1t}(v) &:= \E_{t-1}\big[\nu_{1t}(v)\big],\\
	\nu_{2t}(v) &:= \1_{\{0<\epsilon_t\leq v^\prime D_n^{-1} X_{t-1}\}}\psi_{\beta}(\delta_t)D_n^{-1}X_{t-1},\\
	\overline{\nu}_{2t}(v) &:= \E_{t-1}\big[\nu_{2t}(v)\big].
\end{align*}
In light of the above decomposition, it suffices to show that
\[
	\sup_{|v|\leq K}\sup_{s\in[0,1]}\bigg|\sum_{t=1}^{\lfloor ns\rfloor}\big[\nu_{it}(v)-\overline{\nu}_{it}(v)\big]\bigg|=o_{\P}(1),\qquad i=1,2.
\]
We only do so for $i=2$, as the case $i=1$ can be dealt with similarly.

Fix some $\rho>0$ and assume without loss of generality that $2K/\rho$ is an integer.
Since
\begin{align*}
&\sup_{|v|\leq K}\sup_{s\in[0,1]}\bigg|\sum_{t=1}^{\lfloor ns\rfloor} \big[\nu_{2t}(v)-\overline{\nu}_{2t}(v)\big]\bigg|\\
&\hspace{1cm} \leq \max_{\substack{j\in\mathbb{Z}\\ j\rho\in[-K,K]}}\sup_{s\in[0,1]}\bigg|\sum_{t=1}^{\lfloor ns\rfloor} \big[\nu_{2t}(j\rho)-\overline{\nu}_{2t}(j\rho)\big]\bigg|\\
&\hspace{2cm}+ \sup_{\substack{-K\leq v_1,v_2\leq K\\ |v_1-v_2|\leq\rho}}\sup_{s\in[0,1]}\bigg|\sum_{t=1}^{\lfloor ns\rfloor} \big[\nu_{2t}(v_1)-\overline{\nu}_{2t}(v_1)\big] - \sum_{t=1}^{\lfloor ns\rfloor} \big[\nu_{2t}(v_2)-\overline{\nu}_{2t}(v_2)\big]\bigg|,
\end{align*}
it follows that
\begin{align}
	\P&\bigg\{\sup_{|v|\leq K}\sup_{s\in[0,1]}\bigg|\sum_{t=1}^{\lfloor ns\rfloor} \big[\nu_{2t}(v)-\overline{\nu}_{2t}(v)\big]\bigg|>\varepsilon\bigg\}\notag\\
	& \leq \P\Bigg\{\max_{\substack{j\in\mathbb{Z}\\ j\rho\in[-K,K]}}\sup_{s\in[0,1]}\bigg|\sum_{t=1}^{\lfloor ns\rfloor} \big[\nu_{2t}(j\rho)-\overline{\nu}_{2t}(j\rho)\big]\bigg|>\frac{\varepsilon}{2}\Bigg\}\notag\\
	&\hspace{1cm} + \P\Bigg\{\sup_{\substack{-K\leq v_1,v_2\leq K\\ |v_1-v_2|\leq\rho}}\sup_{s\in[0,1]}\bigg|\sum_{t=1}^{\lfloor ns\rfloor} \big[\nu_{2t}(v_1)-\overline{\nu}_{2t}(v_1)\big] - \sum_{t=1}^{\lfloor ns\rfloor} \big[\nu_{2t}(v_2)-\overline{\nu}_{2t}(v_2)\big]\bigg|>\frac{\varepsilon}{2}\Bigg\}.\label{eq:decomp nu2}
\end{align}
By subadditivity and Lemma~\ref{lem:8},
\begin{multline*}
	\P\Bigg\{\max_{\substack{j\in\mathbb{Z}\\ j\rho\in[-K,K]}}\sup_{s\in[0,1]}\bigg|\sum_{t=1}^{\lfloor ns\rfloor} \big[\nu_{2t}(j\rho)-\overline{\nu}_{2t}(j\rho)\big]\bigg|>\frac{\varepsilon}{2}\Bigg\}\\
		\leq \sum_{\substack{j\in\mathbb{Z}\\ j\rho\in[-K,K]}}\P\bigg\{\sup_{s\in[0,1]}\bigg|\sum_{t=1}^{\lfloor ns\rfloor} \big[\nu_{2t}(j\rho)-\overline{\nu}_{2t}(j\rho)\big]\bigg|>\frac{\varepsilon}{2}\bigg\}=o(1).
\end{multline*}

Therefore, we only have to show that the final right-hand side term in \eqref{eq:decomp nu2} converges to zero.
Note that
\begin{align}
	\sup_{\substack{-K\leq v_1,v_2\leq K\\ |v_1-v_2|\leq\rho}}&\sup_{s\in[0,1]}\bigg|\sum_{t=1}^{\lfloor ns\rfloor} \big[\nu_{2t}(v_1)-\overline{\nu}_{2t}(v_1)\big] - \sum_{t=1}^{\lfloor ns\rfloor} \big[\nu_{2t}(v_2)-\overline{\nu}_{2t}(v_2)\big]\bigg|\notag\\
	&=\sup_{\substack{-K\leq v_1,v_2\leq K\\ |v_1-v_2|\leq\rho}}\sup_{s\in[0,1]}\bigg|\sum_{t=1}^{\lfloor ns\rfloor} \big[\nu_{2t}(v_1)-\nu_{2t}(v_2)\big] - \sum_{t=1}^{\lfloor ns\rfloor} \big[\overline{\nu}_{2t}(v_1)-\overline{\nu}_{2t}(v_2)\big]\bigg|\notag\\
	&\leq \sup_{\substack{-K\leq v_1,v_2\leq K\\ |v_1-v_2|\leq\rho}}\sum_{t=1}^{n} \big|\nu_{2t}(v_1)-\nu_{2t}(v_2)\big| + \sup_{\substack{-K\leq v_1,v_2\leq K\\ |v_1-v_2|\leq\rho}}\sum_{t=1}^{n} \big|\overline{\nu}_{2t}(v_1)-\overline{\nu}_{2t}(v_2)\big|\notag\\
	&=:V_{1n} + V_{2n}.\notag
\end{align}
For $V_{2n}$, we get by monotonicity of the indicator function in $\overline{\nu}_{2t}(\cdot)$ that
\begin{align}
	V_{2n} &\leq\max_{\ell}\sum_{t=1}^{n}\big|\overline{\nu}_{2t}((\ell+2)\rho)-\overline{\nu}_{2t}(\ell\rho)\big|\notag\\
	&\leq\max_{\ell}\sum_{t=1}^{n} \E_{t-1}\Big[\big(\1_{\{0<\epsilon_t\leq(\ell+2)\rho D_n^{-1}X_{t-1}\}}-\1_{\{0<\epsilon_t\leq\ell\rho D_n^{-1}X_{t-1}\}}\big)\big|\psi_{\beta}(\delta_t)\big|\Big]|D_n^{-1}X_{t-1}|,\label{eq:(p.19)}
\end{align}
where the maximum is taken over the integers $\ell$ satisfying $[\ell\rho,(\ell+2)\rho]\subset[-K,K]$.
The conditional expectation in the above expression may be bounded as follows:
\begin{align*}
\E_{t-1}\Big[\big(\1_{\{0<\epsilon_t\leq(\ell+2)\rho D_n^{-1}X_{t-1}\}}-\1_{\{0<\epsilon_t\leq\ell\rho D_n^{-1}X_{t-1}\}}\big)\big|\psi_{\beta}(\delta_t)\big|\Big]
&\leq \int_{|\ell\rho D_n^{-1}X_{t-1}|}^{|(\ell+2)\rho D_n^{-1}X_{t-1}|}f_{\epsilon_t\mid\mathcal{F}_{t-1}}(x)\D x\\
&\leq \overline{f}2\rho|D_n^{-1}X_{t-1}|.
\end{align*}
Insert this into \eqref{eq:(p.19)} to get that
\begin{equation}\label{eq:A4n B4n}
	V_{2n} \leq K\rho\sum_{t=1}^{n}|D_n^{-1}X_{t-1}X_{t-1}^\prime D_n^{-1}|=\rho O_{\P}(1)
\end{equation}
by Lemma~\ref{lem:SUM}.
This implies that
\begin{equation}\label{eq:(p.20)}
	\lim_{\rho\downarrow0}\limsup_{n\to\infty} \P\Bigg\{\sup_{\substack{-K\leq v_1,v_2\leq K\\ |v_1-v_2|\leq\rho}}\sum_{t=1}^{n}\big|\overline{\nu}_{2t}(v_1) - \overline{\nu}_{2t}(v_2)\big|>\varepsilon\Bigg\}=0.
\end{equation}

For $V_{1n}$ we obtain a similar bound as in \eqref{eq:(p.19)}:
\begin{equation*}
	V_{1n} \leq\max_{\ell} \sum_{t=1}^{n}\big[\1_{\{0<\epsilon_t\leq(\ell+2)\rho D_n^{-1}X_{t-1}\}}-\1_{\{0<\epsilon_t\leq\ell\rho D_n^{-1}X_{t-1}\}}\big]\big|\psi_{\beta}(\delta_t)D_n^{-1}X_{t-1}\big|.
\end{equation*}
By arguments similar to those in the proof of Lemma~\ref{lem:8}, we obtain for each fixed $\ell$ that
\begin{multline*}
	\sum_{t=1}^{n}\bigg\{\Big[\1_{\{0<\epsilon_t\leq(\ell+2)\rho D_n^{-1}X_{t-1}\}}\big|\psi_{\beta}(\delta_t)D_n^{-1}X_{t-1}\big| - \1_{\{0<\epsilon_t\leq\ell\rho D_n^{-1}X_{t-1}\}}\big|\psi_{\beta}(\delta_t)D_n^{-1}X_{t-1}\big|\Big]\\
	-\E_{t-1}\Big[\1_{\{0<\epsilon_t\leq(\ell+2)\rho D_n^{-1}X_{t-1}\}}\big|\psi_{\beta}(\delta_t)D_n^{-1}X_{t-1}\big| - \1_{\{0<\epsilon_t\leq\ell\rho D_n^{-1}X_{t-1}\}}\big|\psi_{\beta}(\delta_t)D_n^{-1}X_{t-1}\big|\Big]\bigg\}\\
	=o_{\P}(1).
\end{multline*}
Using this and the fact that
\begin{multline*}
	\sum_{t=1}^{n}\E_{t-1}\Big[\1_{\{0<\epsilon_t\leq(\ell+2)\rho D_n^{-1}X_{t-1}\}}\big|\psi_{\beta}(\delta_t)D_n^{-1}X_{t-1}\big| - \1_{\{0<\epsilon_t\leq\ell\rho D_n^{-1}X_{t-1}\}}\big|\psi_{\beta}(\delta_t)D_n^{-1}X_{t-1}\big|\Big]\\
	\leq \rho O_{\P}(1)
\end{multline*}
from the arguments leading up to \eqref{eq:A4n B4n}, we deduce that
\[
	\lim_{\rho\downarrow0}\limsup_{n\to\infty} \P\Bigg\{\sup_{\substack{-K\leq v_1,v_2\leq K\\ |v_1-v_2|\leq\rho}}\sum_{t=1}^{n}\big|\nu_{2t}(v_1) - \nu_{2t}(v_2)\big|>\varepsilon\Bigg\}=0.
\]
From this and \eqref{eq:(p.20)} it then follows that
\[
	\sup_{\substack{-K\leq v_1,v_2\leq K\\ |v_1-v_2|\leq\rho}}\sup_{s\in[0,1]}\bigg|\sum_{t=1}^{\lfloor ns\rfloor}\big[\nu_{2t}(v_1) - \overline{\nu}_{2t}(v_1)\big] - \sum_{t=1}^{\lfloor ns\rfloor}\big[\nu_{2t}(v_2) - \overline{\nu}_{2t}(v_2)\big]\bigg|=o_{\P}(1),
\]
as $n\to\infty$, followed by $\rho\downarrow0$. 
Therefore, the final right-hand side term in \eqref{eq:decomp nu2} can be made arbitrarily small, concluding the proof.
\end{proof}

\section{Functional Convergence of the QR Estimator Under Local Alternatives}
\label{sec:thm1 alt}

Here, we establish the functional convergence of the QR estimator under the local alternative
\[
	\mathcal{H}_1^{Q}\colon \valpha_{0,t}=\valpha_0 + \mD_n^{-1}\va(t/n),\qquad t=1,\ldots,n,
\]
where $\va(\cdot)$ is a $\mathbb{R}^{k+1}$-valued, componentwise step function on the interval $[0,1]$.
Again, we do so as a first stepping stone towards deriving local power for our structural break test in CoVaR regressions in Theorem~\ref{thm:CoVaR est alt}.
The analog result to Theorem~\ref{thm:std est} under the local alternative $\mathcal{H}_1^{Q}$ reads as follows.

\begin{thm}\label{thm:std est alt}
Suppose $\mathcal{H}_1^{Q}$ holds true for the model \eqref{eq:(QRalt)}.
If Assumptions~\ref{ass:N}--\ref{ass:K} are satisfied, then, as $n\to\infty$,
\begin{equation*}
(s-r)\mD_n[\widehat{\valpha}_n(r,s) - \valpha_0]\overset{d}{\longrightarrow}\mSigma^{1/2}\big[\mW(s)-\mW(r)\big]+\int_{r}^{s}\va(x)\D x\qquad\text{in }(\ell^{\infty}(\mathcal{D}_{\iota}))^{k+1},
\end{equation*}
where $\mW(\cdot)$ and $\mSigma$ are as in Theorem~\ref{thm:std est}.
\end{thm}

Before proving Theorem~\ref{thm:std est alt}, we require the following analog of Proposition~\ref{lem:LLN}.

\begin{prop}\label{lem:LLN alt}
Under the assumptions of Theorem~\ref{thm:std est alt} it holds for any $\vw\in\mathbb{R}^{k+1}$ that, as $n\to\infty$,
\begin{multline*}
	\sup_{0\leq r< s\leq 1}\bigg|\sum_{t=\lfloor nr\rfloor+1}^{\lfloor ns\rfloor}\Big(\epsilon_{t} - \big[\vw-\va(t/n)\big]^\prime\mD_n^{-1}\mX_{t-1}\Big)\times\\
	\times\big[\1_{\{\vw^\prime\mD_n^{-1}\mX_{t-1}<\epsilon_{t}+\va^\prime(t/n)\mD_n^{-1}\mX_{t-1}<0\}} - \1_{\{0<\epsilon_{t}+\va^\prime(t/n)\mD_n^{-1}\mX_{t-1}<\vw^\prime\mD_n^{-1}\mX_{t-1}\}}\big]
\\
-\frac{1}{2}(s-r)\vw^\prime\mK\vw\bigg|=o_{\P}(1).
\end{multline*}
\end{prop}

\begin{proof}
See Appendix~\ref{sec:Prop alt QR}.
\end{proof}

\begin{prop}\label{prop:1alt}
Under the assumptions of Theorem~\ref{thm:std est alt} it holds that, as $n\to\infty$,
\begin{multline*}
	\sup_{0\leq r<s\leq1} \bigg|\sum_{t=\lfloor nr\rfloor+1}^{\lfloor ns\rfloor}\big[\1_{\{-\va^\prime(t/n)\mD_n^{-1}\mX_{t-1}<\epsilon_t\leq0\}} - \1_{\{0<\epsilon_t\leq -\va^\prime(t/n)\mD_n^{-1}\mX_{t-1}\}}\big]\mD_n^{-1}\mX_{t-1} \\
	- \mK \Big(\int_{r}^{s}\va(x)\D x\Big)\bigg|=o_{\P}(1).
\end{multline*}
\end{prop}

\begin{proof}
See Appendix~\ref{sec:Prop alt QR}.
\end{proof}

\begin{proof}[{\textbf{Proof of Theorem~\ref{thm:std est alt}:}}]
The proof is roughly similar to that of Theorem~\ref{thm:std est}.
To highlight the analogy, we often overload notation by redefining quantities from the proof of Theorem~\ref{thm:std est}.
Under $\mathcal{H}_{1}^{Q}$ the estimator $\widehat{\valpha}_n(r,s)$ can equivalently be written as
\begin{align*}
\widehat{\valpha}_n(r,s) &= \argmin_{\valpha\in\mathbb{R}^{k+1}} \sum_{t=\lfloor nr\rfloor+1}^{\lfloor ns\rfloor}\Big[\rho_{\alpha}(Y_t-\mX_{t-1}^\prime\valpha) - \rho_{\alpha}\big(\epsilon_{t}+\va^\prime(t/n)\mD_n^{-1}\mX_{t-1}\big)\Big]\\
&= \argmin_{\valpha\in\mathbb{R}^{k+1}} \sum_{t=\lfloor nr\rfloor+1}^{\lfloor ns\rfloor}\Big[\rho_{\alpha}\big(Y_t-\mX_{t-1}^\prime\valpha_{0,t}-\mX_{t-1}^\prime(\valpha-\valpha_{0,t})\big) - \rho_{\alpha}\big(\epsilon_{t}+\va^\prime(t/n)\mD_n^{-1}\mX_{t-1}\big)\Big]\\
&\overset{\eqref{eq:(QRalt)}}{=} \argmin_{\valpha\in\mathbb{R}^{k+1}} \sum_{t=\lfloor nr\rfloor+1}^{\lfloor ns\rfloor}\Big[\rho_{\alpha}\big(\epsilon_{t} - (\valpha-\valpha_0)^\prime\mX_{t-1} + \va^\prime(t/n)\mD_n^{-1}\mX_{t-1}\big) - \rho_{\alpha}\big(\epsilon_{t}+\va^\prime(t/n)\mD_n^{-1}\mX_{t-1}\big)\Big],
\end{align*}
where we used in the last step that $\valpha_{0,t}=\valpha_0 + \mD_n^{-1}\va(t/n)$.
Therefore, if we define 
\[
 f_n(\vw,r,s)=\sum_{t=\lfloor nr\rfloor+1}^{\lfloor ns\rfloor}\Big[\rho_{\alpha}\big(\epsilon_{t} - \vw^\prime\mD_n^{-1}\mX_{t-1} + \va^\prime(t/n)\mD_n^{-1}\mX_{t-1}\big) - \rho_{\alpha}\big(\epsilon_{t}+\va^\prime(t/n)\mD_n^{-1}\mX_{t-1}\big)\Big],
\] 
then the minimizer $\vw_n(r,s)$ of $f_n(\cdot,r,s)$ satisfies that 
\[
	\vw_n(r,s)=\mD_n\big[\widehat{\valpha}_n(r,s) - \valpha_0\big].
\]

To derive the weak limit of $\vw_n(\cdot,\cdot)$, we again invoke Theorem~2 of \citet{Kat09}.
Note that from Assumption~\ref{ass:innov},
\begin{align*}
	\P&\big\{\exists\ t\in\mathbb{N}\colon \epsilon_t=-\va^\prime(t/n)\mD_n^{-1}\mX_{t-1},\ \max_{t=1,\ldots,n}\big|\va^\prime(t/n)\mD_n^{-1}\mX_{t-1}\big|\leq d\big\}\\
	&=\P\bigg\{\bigcup_{t\in\mathbb{N}}\{ \epsilon_t=-\va^\prime(t/n)\mD_n^{-1}\mX_{t-1}\},\ \max_{t=1,\ldots,n}\big|\va^\prime(t/n)\mD_n^{-1}\mX_{t-1}\big|\leq d\bigg\}\\
	& \leq\sum_{t\in\mathbb{N}}\P\big\{ \epsilon_t=-\va^\prime(t/n)\mD_n^{-1}\mX_{t-1},\ \max_{t=1,\ldots,n}\big|\va^\prime(t/n)\mD_n^{-1}\mX_{t-1}\big|\leq d\big\}\\
	&=0,
\end{align*}
such that, by Lemma~\ref{lem:MAX}, $\big\{\epsilon_t\neq-\va^\prime(t/n)\mD_n^{-1}\mX_{t-1}$ for all $t\in\mathbb{N}\big\}$ occurs with probability approaching 1 (w.p.a.~1), as $n\to\infty$.
Therefore, we may use \eqref{eq:(1)} to get that w.p.a.~1, as $n\to\infty$,
\begin{align}
\rho_{\alpha}&\big(\epsilon_{t} - \vw^\prime\mD_n^{-1}\mX_{t-1} + \va^\prime(t/n)\mD_n^{-1}\mX_{t-1}\big) - \rho_{\alpha}\big(\epsilon_{t}+\va^\prime(t/n)\mD_n^{-1}\mX_{t-1}\big)\notag\\
&=-\vw^\prime\psi_{\alpha}\big(\epsilon_{t}+\va^\prime(t/n)\mD_n^{-1}\mX_{t-1}\big)\mD_n^{-1}\mX_{t-1} + \big[\epsilon_t + \va^\prime(t/n)\mD_n^{-1}\mX_{t-1} - \vw^\prime\mD_n^{-1}\mX_{t-1}\big]\times\notag\\
&\hspace{2cm}\times\big[\1_{\{\vw^\prime\mD_n^{-1}\mX_{t-1}<\epsilon_{t}+\va^\prime(t/n)\mD_n^{-1}\mX_{t-1}<0\}} - \1_{\{0<\epsilon_{t}+\va^\prime(t/n)\mD_n^{-1}\mX_{t-1}<\vw^\prime\mD_n^{-1}\mX_{t-1}\}}\big].\label{eq:decomp loss}
\end{align}
Also note that
\begin{align}
	\psi_{\alpha}\big(\epsilon_{t}+\va^\prime(t/n)\mD_n^{-1}\mX_{t-1}\big) - \psi_{\alpha}\big(\epsilon_{t}\big)
	&=\big[\alpha - \1_{\{\epsilon_t+\va^\prime(t/n)\mD_n^{-1}\mX_{t-1}\leq0\}}\big] - (\alpha - \1_{\{\epsilon_t\leq0\}})\notag\\
	&= \1_{\{\epsilon_t\leq0\}} - \1_{\{\epsilon_t+\va^\prime(t/n)\mD_n^{-1}\mX_{t-1}\leq0\}}\notag\\
	&= \1_{\{-\va^\prime(t/n)\mD_n^{-1}\mX_{t-1}<\epsilon_t\leq0\}} - \1_{\{0<\epsilon_t\leq -\va^\prime(t/n)\mD_n^{-1}\mX_{t-1}\}}.\label{eq:diff psi}
\end{align}
Thus, using \eqref{eq:decomp loss} first and then \eqref{eq:diff psi},
\begin{align*}
	f_n(\vw,r,s)&= -\vw^\prime\sum_{t=\lfloor nr\rfloor+1}^{\lfloor ns\rfloor}\psi_{\alpha}\big(\epsilon_{t}\big)\mD_n^{-1}\mX_{t-1}\\
	&\hspace{0.5cm} - \vw^\prime\sum_{t=\lfloor nr\rfloor+1}^{\lfloor ns\rfloor}\Big[\psi_{\alpha}\big(\epsilon_{t}+\va^\prime(t/n)\mD_n^{-1}\mX_{t-1}\big)-\psi_{\alpha}\big(\epsilon_{t}\big)\Big]\mD_n^{-1}\mX_{t-1}\\
	&\hspace{0.5cm} + \sum_{t=\lfloor nr\rfloor+1}^{\lfloor ns\rfloor}\big[\epsilon_t + \va^\prime(t/n)\mD_n^{-1}\mX_{t-1} - \vw^\prime\mD_n^{-1}\mX_{t-1}\big]\times\\
	&\hspace{1.5cm}\times\big[\1_{\{\vw^\prime\mD_n^{-1}\mX_{t-1}<\epsilon_{t}+\va^\prime(t/n)\mD_n^{-1}\mX_{t-1}<0\}} - \1_{\{0<\epsilon_{t}+\va^\prime(t/n)\mD_n^{-1}\mX_{t-1}<\vw^\prime\mD_n^{-1}\mX_{t-1}\}}\big]\\
	&=-\vw^\prime\bigg\{\sum_{t=\lfloor nr\rfloor+1}^{\lfloor ns\rfloor}\psi_{\alpha}\big(\epsilon_{t}\big)\mD_n^{-1}\mX_{t-1}\\
	&\hspace{0.5cm} +\sum_{t=\lfloor nr\rfloor+1}^{\lfloor ns\rfloor}\Big[\1_{\{-\va^\prime(t/n)\mD_n^{-1}\mX_{t-1}<\epsilon_t\leq0\}} - \1_{\{0<\epsilon_t\leq -\va^\prime(t/n)\mD_n^{-1}\mX_{t-1}\}}\Big]\mD_n^{-1}\mX_{t-1}\bigg\}\\
	&\hspace{0.5cm} + \sum_{t=\lfloor nr\rfloor+1}^{\lfloor ns\rfloor}\big[\epsilon_t + \va^\prime(t/n)\mD_n^{-1}\mX_{t-1} - \vw^\prime\mD_n^{-1}\mX_{t-1}\big]\times\\
	&\hspace{1.5cm}\times\big[\1_{\{\vw^\prime\mD_n^{-1}\mX_{t-1}<\epsilon_{t}+\va^\prime(t/n)\mD_n^{-1}\mX_{t-1}<0\}} - \1_{\{0<\epsilon_{t}+\va^\prime(t/n)\mD_n^{-1}\mX_{t-1}<\vw^\prime\mD_n^{-1}\mX_{t-1}\}}\big].
\end{align*}

Define
\begin{equation*}
	g_n(\vw,r,s) = -\vw^\prime\mW_n(r,s) + \frac{1}{2}\vw^\prime(s-r)\mK\vw,
\end{equation*}
where
\begin{multline*}
	\mW_n(r,s) = \sum_{t=\lfloor nr\rfloor+1}^{\lfloor ns\rfloor}\psi_{\alpha}\big(\epsilon_{t}\big)\mD_n^{-1}\mX_{t-1}\\
	 +\sum_{t=\lfloor nr\rfloor+1}^{\lfloor ns\rfloor}\Big[\1_{\{-\va^\prime(t/n)\mD_n^{-1}\mX_{t-1}<\epsilon_t\leq0\}} - \1_{\{0<\epsilon_t\leq -\va^\prime(t/n)\mD_n^{-1}\mX_{t-1}\}}\Big]\mD_n^{-1}\mX_{t-1}.
\end{multline*}
Then, by Proposition~\ref{lem:LLN alt}, we obtain that for each $\vw\in\mathbb{R}^{k+1}$,
\[
	\sup_{(r,s)\in\mathcal{D}_{\iota}}\big|f_n(\vw,r,s)-g_n(\vw,r,s)\big|=o_{\P}(1).
\]
(This is the equivalent of equation (10) in \citet{Kat09}.)
Additionally, by Propositions~\ref{lem:CLT} and \ref{prop:1alt},
\begin{equation}\label{eq:convWN2}
	\mW_n(r,s) \overset{d}{\longrightarrow}\sqrt{\alpha(1-\alpha)}\mOmega^{1/2}\big[\mW(s)-\mW(r)\big] + \mK\Big(\int_{r}^{s}\va(x)\D x\Big)\qquad\text{in }(\ell^{\infty}(\mathcal{D}_{\iota}))^{k+1}.
\end{equation}
In particular, by the CMT,
\begin{multline*}
	\limsup_{n\to\infty}\P\bigg\{\sup_{(r,s)\in\mathcal{D}_{\iota}}\big\Vert\mW_n(r,s)\big\Vert>M\bigg\}\\
	=\P\bigg\{ \sup_{(r,s)\in\mathcal{D}_{\iota}}\Big\Vert\sqrt{\alpha(1-\alpha)}\mOmega^{1/2}\big[\mW(s)-\mW(r)\big] + \mK\Big(\int_{r}^{s}\va(x)\D x\Big)\Big\Vert>M\bigg\},
\end{multline*}
which can be made arbitrarily small by choosing $M>0$ sufficiently large. (This is the analog of equation (11) in \citet{Kat09}.)

We therefore may invoke \citet[Theorem~2]{Kat09} to get that
\[
	\vw_n(r,s)=\frac{1}{s-r}\mK^{-1}\mW_n(r,s) + \vr_{n}(r,s)
\]
with $\sup_{(r,s)\in\mathcal{D}_{\iota}}\Vert\vr_n(r,s)\Vert=o_{\P}(1)$.
Hence, it holds by \eqref{eq:convWN2} that
\begin{equation*}
	\vw_n(r,s)\overset{d}{\longrightarrow}\frac{\sqrt{\alpha(1-\alpha)}}{s-r}\mK^{-1}\mOmega^{1/2}\big[\mW(s)-\mW(r)\big]+\frac{1}{s-r}\Big(\int_{r}^{s}\va(x)\D x\Big)\qquad\text{in }(\ell^{\infty}(\mathcal{D}_{\iota}))^{k+1}.
\end{equation*}
From this, the conclusion easily follows.
\end{proof}

\section{Proofs of Propositions~\ref{lem:LLN alt}--\ref{prop:1alt}}\label{sec:Prop alt QR}

\begin{proof}[{\textbf{Proof of Proposition~\ref{lem:LLN alt}:}}]
The proof is very similar to that of Proposition~\ref{lem:LLN}.
To highlight the similarities, we frequently overload notation by redefining quantities already introduced in the proof of Proposition~\ref{lem:LLN}.

Once again, it suffices to show that
\begin{multline}\label{eq:(SC alt)}
	\sup_{0\leq s\leq 1}\bigg|\sum_{t=1}^{\lfloor ns\rfloor}\Big(\epsilon_{t} - \big[\vw-\va(t/n)\big]^\prime\mD_n^{-1}\mX_{t-1}\Big)\times\\
	\times\big[\1_{\{\vw^\prime\mD_n^{-1}\mX_{t-1}<\epsilon_{t}+\va^\prime(t/n)\mD_n^{-1}\mX_{t-1}<0\}} - \1_{\{0<\epsilon_{t}+\va^\prime(t/n)\mD_n^{-1}\mX_{t-1}<\vw^\prime\mD_n^{-1}\mX_{t-1}\}}\big]
\\
-\frac{1}{2}s\vw^\prime\mK\vw \bigg|=o_{\P}(1).
\end{multline}
To do so, define
\begin{align*}
	\nu_{t}(\vw) &:= \Big(\epsilon_{t} - \big[\vw-\va(t/n)\big]^\prime\mD_n^{-1}\mX_{t-1}\Big)\times\\
		&\hspace{1.5cm}\times\big(\1_{\{\vw^\prime\mD_n^{-1}\mX_{t-1}< \epsilon_{t}+\va^\prime(t/n)\mD_n^{-1}\mX_{t-1}<0\}} - \1_{\{0<\epsilon_{t}+\va^\prime(t/n)\mD_n^{-1}\mX_{t-1}< \vw^\prime\mD_n^{-1}\mX_{t-1}\}}\big),\\
	\overline{\nu}_{t}(\vw) &:= \E_{t-1}\Big[\big(\epsilon_{t} - [\vw-\va(t/n)]^\prime\mD_n^{-1}\mX_{t-1}\big)\times\\
	&\hspace{1.5cm}\times\big(\1_{\{\vw^\prime\mD_n^{-1}\mX_{t-1}< \epsilon_{t}+\va^\prime(t/n)\mD_n^{-1}\mX_{t-1}<0\}} - \1_{\{0<\epsilon_{t}+\va^\prime(t/n)\mD_n^{-1}\mX_{t-1}< \vw^\prime\mD_n^{-1}\mX_{t-1}\}}\big)\Big],\\
	V_n(\vw,s) &:= \sum_{t=1}^{\lfloor ns\rfloor}\nu_t(\vw),\\
	\overline{V}_n(\vw,s) &:= \sum_{t=1}^{\lfloor ns\rfloor}\overline{\nu}_t(\vw).
\end{align*}
We establish \eqref{eq:(SC alt)} by showing that, uniformly in $s\in[0,1]$,
\begin{align}
	\overline{V}_n(\vw,s) &\overset{\P}{\longrightarrow}\frac{1}{2}s\vw^\prime\mK\vw, \label{eq:(P.7.0 alt)}\\
	V_n(\vw,s) - \overline{V}_n(\vw,s) &= o_{\P}(1).  \label{eq:(P.7.1 alt)}
\end{align}

The relation \eqref{eq:(P.7.1 alt)} can be shown in a similar fashion as \eqref{eq:(P.7.1)} in the proof of Proposition~\ref{lem:LLN}.
Therefore, we only derive \eqref{eq:(P.7.0 alt)}, whose proof deviates slightly from that of \eqref{eq:(P.7.0)}.
Similarly as in the proof of Proposition~\ref{lem:LLN}, we may show the convergence on a set, where the conditional density of the QR errors $\epsilon_t$ exists.
Hence,
\begin{align*}
	&\overline{V}_n(\vw,s)\\
	&= \sum_{t=1}^{\lfloor ns\rfloor}\E_{t-1}\Big[\big(\epsilon_{t} - [\vw-\va(t/n)]^\prime\mD_n^{-1}\mX_{t-1}\big)\1_{\{\vw^\prime\mD_n^{-1}\mX_{t-1}< \epsilon_{t}+\va^\prime(t/n)\mD_n^{-1}\mX_{t-1}<0\}}\Big]\\
	&\hspace{0.5cm} + \sum_{t=1}^{\lfloor ns\rfloor}\E_{t-1}\Big[\big([\vw-\va(t/n)]^\prime\mD_n^{-1}\mX_{t-1} - \epsilon_{t}\big) \1_{\{0<\epsilon_{t}+\va^\prime(t/n)\mD_n^{-1}\mX_{t-1}< \vw^\prime\mD_n^{-1}\mX_{t-1}\}}\Big]\\
	&= \sum_{t=1}^{\lfloor ns\rfloor}\1_{\{\vw^\prime\mD_n^{-1}\mX_{t-1}<0\}}\int_{[\vw-\va(t/n)]^\prime\mD_n^{-1}\mX_{t-1}}^{-\va^\prime(t/n)\mD_n^{-1}\mX_{t-1}}\Big(x-\big[\vw-\va(t/n)\big]^\prime\mD_n^{-1}\mX_{t-1}\Big)f_{\epsilon_t\mid\mathcal{F}_{t-1}}(x)\D x\\
	&\hspace{0.5cm} +\sum_{t=1}^{\lfloor ns\rfloor}\1_{\{\vw^\prime\mD_n^{-1}\mX_{t-1}>0\}}\int_{-\va^\prime(t/n)\mD_n^{-1}\mX_{t-1}}^{[\vw-\va(t/n)]^\prime\mD_n^{-1}\mX_{t-1}}\Big(\big[\vw-\va(t/n)\big]^\prime\mD_n^{-1}\mX_{t-1} - x\Big)f_{\epsilon_t\mid\mathcal{F}_{t-1}}(x)\D x\\
	&=:\overline{V}_{1n}(\vw,s) + \overline{V}_{2n}(\vw,s).
\end{align*}	
Write
\begin{align*}
	\overline{V}_{1n}(\vw,s) &= \sum_{t=1}^{\lfloor ns\rfloor}\1_{\{\vw^\prime\mD_n^{-1}\mX_{t-1}<0\}}\int_{[\vw-\va(t/n)]^\prime\mD_n^{-1}\mX_{t-1}}^{-\va^\prime(t/n)\mD_n^{-1}\mX_{t-1}}\Big(x-\big[\vw-\va(t/n)\big]^\prime\mD_n^{-1}\mX_{t-1}\Big)f_{\epsilon_t\mid\mathcal{F}_{t-1}}(0)\D x\\
	&\hspace{0.5cm}+ \sum_{t=1}^{\lfloor ns\rfloor}\1_{\{\vw^\prime\mD_n^{-1}\mX_{t-1}<0\}}\int_{[\vw-\va(t/n)]^\prime\mD_n^{-1}\mX_{t-1}}^{-\va^\prime(t/n)\mD_n^{-1}\mX_{t-1}}\Big(x-\big[\vw-\va(t/n)\big]^\prime\mD_n^{-1}\mX_{t-1}\Big)\times\\
	&\hspace{9cm}\times\big[f_{\epsilon_t\mid\mathcal{F}_{t-1}}(x)-f_{\epsilon_t\mid\mathcal{F}_{t-1}}(0)\big]\D x\\
	&=:\overline{V}_{11n}(\vw,s) + \overline{V}_{12n}(\vw,s).
\end{align*}	
Integrating out yields that
\begin{align*}
	\overline{V}_{11n}(\vw,s) &= \sum_{t=1}^{\lfloor ns\rfloor}\1_{\{\vw^\prime\mD_n^{-1}\mX_{t-1}<0\}}f_{\epsilon_t\mid\mathcal{F}_{t-1}}(0)\frac{1}{2}(\vw^\prime\mD_n^{-1}\mX_{t-1})^2\\
	&=\frac{1}{2}\vw^\prime\bigg(\sum_{t=1}^{\lfloor ns\rfloor}\1_{\{\vw^\prime\mD_n^{-1}\mX_{t-1}<0\}}f_{\epsilon_t\mid\mathcal{F}_{t-1}}(0)\mD_n^{-1}\mX_{t-1}\mX_{t-1}^\prime\mD_n^{-1}\bigg)\vw.
\end{align*}
For the other term, Assumption~\ref{ass:innov}~\eqref{it:Lipschitz} implies that
\begin{align*}
	\big|\overline{V}_{12n}(\vw,s)\big| &\leq\sum_{t=1}^{\lfloor ns\rfloor}\1_{\{\vw^\prime\mD_n^{-1}\mX_{t-1}<0\}}\int_{[\vw-\va(t/n)]^\prime\mD_n^{-1}\mX_{t-1}}^{-\va^\prime(t/n)\mD_n^{-1}\mX_{t-1}}\Big(x-\big[\vw-\va(t/n)\big]^\prime\mD_n^{-1}\mX_{t-1}\Big)L|x|\D x\\
	&\leq K\max_{t=1,\ldots,n}\big\Vert\mD_n^{-1}\mX_{t-1}\big\Vert\sum_{t=1}^{\lfloor ns\rfloor}\1_{\{\vw^\prime\mD_n^{-1}\mX_{t-1}<0\}}\times\\
	&\hspace{4cm}\times\int_{[\vw-\va(t/n)]^\prime\mD_n^{-1}\mX_{t-1}}^{-\va^\prime(t/n)\mD_n^{-1}\mX_{t-1}}\Big(x-\big[\vw-\va(t/n)\big]^\prime\mD_n^{-1}\mX_{t-1}\Big)\D x\\
	&= K\max_{t=1,\ldots,n}\big\Vert\mD_n^{-1}\mX_{t-1}\big\Vert\sum_{t=1}^{\lfloor ns\rfloor}\1_{\{\vw^\prime\mD_n^{-1}\mX_{t-1}<0\}}\frac{1}{2}(\vw^\prime\mD_n^{-1}\mX_{t-1})^2\\
	&\leq K\max_{t=1,\ldots,n}\big\Vert\mD_n^{-1}\mX_{t-1}\big\Vert\vw^\prime\bigg(\sum_{t=1}^{n}\mD_n^{-1}\mX_{t-1}\mX_{t-1}^\prime\mD_n^{-1}\bigg)\vw\\
	&=o_{\P}(1)O_{\P}(1)\\
	&=o_{\P}(1)
\end{align*}
uniformly in $s\in[0,1]$ by Lemmas~\ref{lem:MAX}--\ref{lem:SUM}.
Therefore,
\begin{equation*} 
\overline{V}_{1n}(\vw,s) =\frac{1}{2}\vw^\prime\bigg(\sum_{t=1}^{\lfloor ns\rfloor}\1_{\{\vw^\prime\mD_n^{-1}\mX_{t-1}<0\}}f_{\epsilon_t\mid\mathcal{F}_{t-1}}(0)\mD_n^{-1}\mX_{t-1}\mX_{t-1}^\prime\mD_n^{-1}\bigg)\vw + o_{\P}(1)
\end{equation*}
uniformly in $s\in[0,1]$.

We can deduce using almost identical arguments that
\begin{equation*} 
\overline{V}_{2n}(\vw,s) =\frac{1}{2}\vw^\prime\bigg(\sum_{t=1}^{\lfloor ns\rfloor}\1_{\{\vw^\prime\mD_n^{-1}\mX_{t-1}>0\}}f_{\epsilon_t\mid\mathcal{F}_{t-1}}(0)\mD_n^{-1}\mX_{t-1}\mX_{t-1}^\prime\mD_n^{-1}\bigg)\vw + o_{\P}(1)
\end{equation*}
uniformly in $s\in[0,1]$.

Combining the previous two displays then yields, using Assumption~\ref{ass:K}, that
\begin{align*}
	\overline{V}_{n}(\vw,s) &= \overline{V}_{1n}(\vw,s) + \overline{V}_{2n}(\vw,s)\\
	&=\frac{1}{2}\vw^\prime\bigg(\sum_{t=1}^{\lfloor ns\rfloor}f_{\epsilon_t\mid\mathcal{F}_{t-1}}(0)\mD_n^{-1}\mX_{t-1}\mX_{t-1}^\prime\mD_n^{-1}\bigg)\vw + o_{\P}(1)\\
	&\overset{\P}{\longrightarrow}\frac{1}{2}s\vw^\prime\mK\vw
\end{align*}
uniformly in $s\in[0,1]$, which finishes the proof.
\end{proof}

\begin{proof}[{\textbf{Proof of Proposition~\ref{prop:1alt}:}}]
This proof also follows along similar lines as that of Proposition~\ref{lem:LLN}.
Define
\begin{align*}
	\vnu_t &:= \big[\1_{\{-\va^\prime(t/n)\mD_n^{-1}\mX_{t-1}<\epsilon_t\leq0\}} - \1_{\{0<\epsilon_t\leq -\va^\prime(t/n)\mD_n^{-1}\mX_{t-1}\}}\big]\mD_n^{-1}\mX_{t-1},\\
	\overline{\vnu}_t &:= \E_{t-1}[\vnu_t],\\
	\mV_n(s) &:= \sum_{t=1}^{\lfloor ns\rfloor}\vnu_t,\\
	\overline{\mV}_n(s) &:= \sum_{t=1}^{\lfloor ns\rfloor}\overline{\vnu}_t.
\end{align*}
An argument similar to that in \eqref{eq:(p.9)} shows that for the proposition to be established it suffices to prove 
\begin{align}
\overline{\mV}_n(s) &\overset{\P}{\longrightarrow} \mK \Big(\int_{0}^{s}\va(x)\D x\Big),\label{eq:(1.1)}\\
\mV_n(s) - \overline{\mV}_n(s) &\overset{\P}{\longrightarrow}\vzero\label{eq:(1.2)}
\end{align}
uniformly in $s\in[0,1]$.
Since \eqref{eq:(1.2)} follows along the lines of \eqref{eq:(P.7.1)}, we only show \eqref{eq:(1.1)}.
Write
\begin{align*}
	&\overline{\mV}_n(s)\\
	&= \sum_{t=1}^{\lfloor ns\rfloor}\mD_n^{-1}\mX_{t-1}\Big[F_{\epsilon_t\mid\mathcal{F}_{t-1}}(0) - F_{\epsilon_t\mid\mathcal{F}_{t-1}}\big(-\va^\prime(t/n)\mD_n^{-1}\mX_{t-1}\big)\Big]\1_{\{\va^\prime(t/n)\mD_n^{-1}\mX_{t-1}>0\}}\\
	&\hspace{0.5cm} -\sum_{t=1}^{\lfloor ns\rfloor}\mD_n^{-1}\mX_{t-1}\Big[F_{\epsilon_t\mid\mathcal{F}_{t-1}}\big(-\va^\prime(t/n)\mD_n^{-1}\mX_{t-1}\big) - F_{\epsilon_t\mid\mathcal{F}_{t-1}}(0)\Big]\1_{\{\va^\prime(t/n)\mD_n^{-1}\mX_{t-1}<0\}}\\
	&=\sum_{t=1}^{\lfloor ns\rfloor}\mD_n^{-1}\mX_{t-1}\mX_{t-1}^\prime\mD_n^{-1}\va(t/n)\frac{F_{\epsilon_t\mid\mathcal{F}_{t-1}}(0) - F_{\epsilon_t\mid\mathcal{F}_{t-1}}\big(-\va^\prime(t/n)\mD_n^{-1}\mX_{t-1}\big)}{0-\big(-\va^\prime(t/n)\mD_n^{-1}\mX_{t-1}\big)}\1_{\{\va^\prime(t/n)\mD_n^{-1}\mX_{t-1}>0\}}\\
	&\hspace{0.5cm} +\sum_{t=1}^{\lfloor ns\rfloor}\mD_n^{-1}\mX_{t-1}\mX_{t-1}^\prime\mD_n^{-1}\va(t/n)\frac{F_{\epsilon_t\mid\mathcal{F}_{t-1}}\big(-\va^\prime(t/n)\mD_n^{-1}\mX_{t-1}\big) - F_{\epsilon_t\mid\mathcal{F}_{t-1}}(0)}{-\va^\prime(t/n)\mD_n^{-1}\mX_{t-1}-0}\1_{\{\va^\prime(t/n)\mD_n^{-1}\mX_{t-1}<0\}}\\
	&=:\overline{\mV}_{1n}(s) + \overline{\mV}_{2n}(s).
\end{align*}
Consider $\overline{\mV}_{1n}(s)$ and use \eqref{eq:MVT} to deduce that
\begin{align*}
	\overline{\mV}_{1n}(s) &= \sum_{t=1}^{\lfloor ns\rfloor}\mD_n^{-1}\mX_{t-1}\mX_{t-1}^\prime\mD_n^{-1}\va(t/n)f_{\epsilon_t\mid\mathcal{F}_{t-1}}(0)\1_{\{\va^\prime(t/n)\mD_n^{-1}\mX_{t-1}>0\}}\\
	&\hspace{0.5cm} +\sum_{t=1}^{\lfloor ns\rfloor}\mD_n^{-1}\mX_{t-1}\mX_{t-1}^\prime\mD_n^{-1}\va(t/n)\big[f_{\epsilon_t\mid\mathcal{F}_{t-1}}(x^{\ast})-f_{\epsilon_t\mid\mathcal{F}_{t-1}}(0)\big]\1_{\{\va^\prime(t/n)\mD_n^{-1}\mX_{t-1}>0\}},
\end{align*}
where $x^\ast$ is some mean value between $-\va^\prime(t/n)\mD_n^{-1}\mX_{t-1}$ and $0$.
Using Assumption~\ref{ass:innov}~\eqref{it:Lipschitz}, the norm of the second term on the right-hand side can be bounded by
\begin{align*}
	\bigg\Vert\sum_{t=1}^{\lfloor ns\rfloor}&\mD_n^{-1}\mX_{t-1}\mX_{t-1}^\prime\mD_n^{-1}\va(t/n)\big[f_{\epsilon_t\mid\mathcal{F}_{t-1}}(x^{\ast})-f_{\epsilon_t\mid\mathcal{F}_{t-1}}(0)\big]\1_{\{\va^\prime(t/n)\mD_n^{-1}\mX_{t-1}>0\}}\bigg\Vert\\
	& \leq \sum_{t=1}^{n}\big\Vert\mD_n^{-1}\mX_{t-1}\mX_{t-1}^\prime\mD_n^{-1}\va(t/n)\big\Vert \cdot\big|f_{\epsilon_t\mid\mathcal{F}_{t-1}}(x^{\ast})-f_{\epsilon_t\mid\mathcal{F}_{t-1}}(0)\big|\\
	& \leq K\sum_{t=1}^{n}\big\Vert\mD_n^{-1}\mX_{t-1}\mX_{t-1}^\prime\mD_n^{-1}\big\Vert\cdot \big\Vert\mD_n^{-1}\mX_{t-1}\big\Vert\\
	&\leq K\max_{t=1,\ldots,n}\big\Vert\mD_n^{-1}\mX_{t-1}\big\Vert \sum_{t=1}^{n}\big\Vert\mD_n^{-1}\mX_{t-1}\mX_{t-1}^\prime\mD_n^{-1}\big\Vert\\
	&=o_{\P}(1)O_{\P}(1)\\
	&=o_{\P}(1)
\end{align*}
uniformly in $s\in[0,1]$.
Here, we have used that, with $\mD_n=\diag(D_{0,n},D_{1,n},\ldots,D_{k,n})$ and $\mX_{t-1}=(X_{0,t-1},X_{1,t-1},\ldots,X_{k,t-1})^\prime$,
\begin{align}
	\sum_{t=1}^{n}\big\Vert\mD_n^{-1}\mX_{t-1}\mX_{t-1}^\prime\mD_n^{-1}\big\Vert &\leq K \sum_{i,j=0}^{k} \sum_{t=1}^{n}\big|D_{i,n}^{-1}X_{i,t-1}X_{j,t-1}D_{j,n}^{-1}\big|\notag\\
	&\leq K  \sum_{i,j=0}^{k} \sum_{t=1}^{n}\big[D_{i,n}^{-2}X_{i,t-1}^2 + D_{j,n}^{-2}X_{j,t-1}^2\big]\notag\\
	&= O_{\P}(1),\label{eq:Op1 cross}
\end{align}
where the first line follows from norm equivalence in finite-dimensional spaces, the second from the inequality $ab\leq\frac{1}{2}(a^2 + b^2)$, and the third from the fact that all the diagonal elements of
\[
	\sum_{t=1}^{n}\mD_n^{-1}\mX_{t-1}\mX_{t-1}^\prime\mD_n^{-1}=\sum_{t=1}^{n}\begin{pmatrix}
		D_{n,1}^{-2}X_{1,t-1}^2 &  & \ast \\
		&                        \ddots & \\
		\ast && D_{n,k+1}^{-2}X_{k+1,t-1}^2
	\end{pmatrix}
\]
converge in probability by virtue of Lemma~\ref{lem:SUM}.
We get that, as $n\to\infty$,
\[
	\overline{\mV}_{1n}(s)=\sum_{t=1}^{\lfloor ns\rfloor}\mD_n^{-1}\mX_{t-1}\mX_{t-1}^\prime\mD_n^{-1}\va(t/n)f_{\epsilon_t\mid\mathcal{F}_{t-1}}(0)\1_{\{\va^\prime(t/n)\mD_n^{-1}\mX_{t-1}>0\}} + o_{\P}(1)
\]
uniformly in $s\in[0,1]$ and, by similar arguments,
\[
	\overline{\mV}_{2n}(s)=\sum_{t=1}^{\lfloor ns\rfloor}\mD_n^{-1}\mX_{t-1}\mX_{t-1}^\prime\mD_n^{-1}\va(t/n)f_{\epsilon_t\mid\mathcal{F}_{t-1}}(0)\1_{\{\va^\prime(t/n)\mD_n^{-1}\mX_{t-1}<0\}} + o_{\P}(1).
\]
Combining these two convergences, we obtain from Assumption~\ref{ass:K} and the fact that $\va(\cdot)$ is a step function that
\begin{align*}
	\overline{\mV}_{n}(s)&=\sum_{t=1}^{\lfloor ns\rfloor}\mD_n^{-1}\mX_{t-1}\mX_{t-1}^\prime\mD_n^{-1}\va(t/n)f_{\epsilon_t\mid\mathcal{F}_{t-1}}(0) + o_{\P}(1)\\
	&\overset{\P}{\longrightarrow} \mK \Big(\int_{0}^{s}\va(x)\D x\Big)
\end{align*}
uniformly in $s\in[0,1]$.
\end{proof}

\section{Proofs of Theorem~\ref{thm:CoVaR est alt} and Corollary~\ref{cor:one-break}}\label{CoVaR alt}

The proof of Theorem~\ref{thm:CoVaR est alt} requires the following preliminary propositions.
The first is the analog of Proposition~\ref{prop:1alt}.

\begin{prop}\label{prop:1alt CoVaR}
Under the assumptions of Theorem~\ref{thm:CoVaR est alt} it holds that, as $n\to\infty$,
\begin{multline*}
	\sup_{0\leq r< s\leq 1}\bigg|\sum_{t=\lfloor nr\rfloor+1}^{\lfloor ns\rfloor}\1_{\{\epsilon_t>0\}}\Big[\1_{\{-\vb^\prime(t/n)\mD_n^{-1}\mX_{t-1}<\delta_t\leq0\}} - \1_{\{0<\delta_t\leq -\vb^\prime(t/n)\mD_n^{-1}\mX_{t-1}\}}\Big]\mD_n^{-1}\mX_{t-1}\\	
-\mK_{\ast}\Big(\int_{r}^{s}\vb(x)\D x\Big)\bigg|=o_{\P}(1).
\end{multline*}
\end{prop}

\begin{proof}
See Appendix~\ref{sec:Prop alt CoVaR}.
\end{proof}

The next proposition plays a role analogous to that of Propositions~\ref{lem:LLN2} and \ref{lem:LLN alt}.

\begin{prop}\label{lem:LLN alt CoVaR}
Under the assumptions of Theorem~\ref{thm:CoVaR est alt} it holds for fixed $\vw\in\mathbb{R}^{k+1}$ that, as $n\to\infty$,
\begin{multline*}
	\sup_{0\leq r< s\leq 1}\bigg|\sum_{t=\lfloor nr\rfloor+1}^{\lfloor ns\rfloor}\1_{\{\epsilon_t>0\}}\big[\delta_t + \vb^\prime(t/n)\mD_n^{-1}\mX_{t-1} - \vw^\prime\mD_n^{-1}\mX_{t-1}\big]\times\\
\times\big[\1_{\{\vw^\prime\mD_n^{-1}\mX_{t-1}<\delta_{t}+\vb^\prime(t/n)\mD_n^{-1}\mX_{t-1}<0\}} - \1_{\{0<\delta_{t}+\vb^\prime(t/n)\mD_n^{-1}\mX_{t-1}<\vw^\prime\mD_n^{-1}\mX_{t-1}\}}\big]\\	
-\frac{1}{2}(s-r)\vw^\prime\mK_{\ast}\vw \bigg|=o_{\P}(1).
\end{multline*}
\end{prop}

\begin{proof}
See Appendix~\ref{sec:Prop alt CoVaR}.
\end{proof}

The next proposition is similar to Proposition~\ref{lem:LLN CLT}.

\begin{prop}\label{lem:LLN CLT Alt}
Under the assumptions of Theorem~\ref{thm:CoVaR est alt} it holds that, as $n\to\infty$,
\begin{multline*}
	\sum_{t=\lfloor nr\rfloor+1}^{\lfloor ns\rfloor}\big[\1_{\{\mX_{t-1}^\prime[\widehat{\valpha}_n(r,s) - \valpha_{0,t}]<\epsilon_t\leq0\}} - \1_{\{0<\epsilon_t\leq \mX_{t-1}^\prime[\widehat{\valpha}_n(r,s) - \valpha_{0,t}]\}}\big]\psi_{\beta}(\delta_{t})\mD_n^{-1}\mX_{t-1}\\	
\overset{d}{\longrightarrow}\big[(1-\beta)\mK - \mK_{\dagger}\big]\mSigma^{1/2}\big[\mW(s)-\mW(r)\big]\qquad\text{in }(\ell^{\infty}(\mathcal{D}_{\iota}))^{k+1},
\end{multline*}
\end{prop}

\begin{proof}
See Appendix~\ref{sec:Prop alt CoVaR}.
\end{proof}

The next two propositions are similar to Proposition~\ref{lem:LLN3}.

\begin{prop}\label{lem:LLN3 Alt1}
Under the assumptions of Theorem~\ref{thm:CoVaR est alt} it holds that, as $n\to\infty$,
\begin{multline*}
	\sup_{(r,s)\in\mathcal{D}_{\iota}}\bigg|\sum_{t=\lfloor nr\rfloor+1}^{\lfloor ns\rfloor}\big[\1_{\{\mX_{t-1}^\prime[\widehat{\valpha}_n(r,s) - \valpha_{0,t}]<\epsilon_t\leq0\}} - \1_{\{0<\epsilon_t\leq \mX_{t-1}^\prime[\widehat{\valpha}_n(r,s) - \valpha_{0,t}]\}}\big]\times\\
	\times\Big[\1_{\{0<\delta_t\leq -\vb^\prime(t/n)\mD_n^{-1}\mX_{t-1}\}} - \1_{\{-\vb^\prime(t/n)\mD_n^{-1}\mX_{t-1}<\delta_t\leq0\}}\Big]\mD_n^{-1}\mX_{t-1}\bigg|=o_{\P}(1).
\end{multline*}
\end{prop}

\begin{proof}
See Appendix~\ref{sec:Prop alt CoVaR}.
\end{proof}

\begin{prop}\label{lem:LLN3 Alt2}
Under the assumptions of Theorem~\ref{thm:CoVaR est alt} it holds that, as $n\to\infty$,
\begin{multline*}
	\sup_{(r,s)\in\mathcal{D}_{\iota}}\bigg|\sum_{t=\lfloor nr\rfloor+1}^{\lfloor ns\rfloor}\big[\1_{\{\mX_{t-1}^\prime[\widehat{\valpha}_n(r,s) - \valpha_{0,t}]<\epsilon_t\leq0\}} - \1_{\{0<\epsilon_t\leq \mX_{t-1}^\prime[\widehat{\valpha}_n(r,s) - \valpha_{0,t}]\}}\big]\times\\
	\times\big[\delta_t + \vb^\prime(t/n)\mD_n^{-1}\mX_{t-1} - \vw^\prime\mD_n^{-1}\mX_{t-1}\big]\times\notag\\
\times\big[\1_{\{\vw^\prime\mD_n^{-1}\mX_{t-1}<\delta_{t}+\vb^\prime(t/n)\mD_n^{-1}\mX_{t-1}<0\}} - \1_{\{0<\delta_{t}+\vb^\prime(t/n)\mD_n^{-1}\mX_{t-1}<\vw^\prime\mD_n^{-1}\mX_{t-1}\}}\big]	
\bigg|=o_{\P}(1).
\end{multline*}
\end{prop}

\begin{proof}
See Appendix~\ref{sec:Prop alt CoVaR}.
\end{proof}

\begin{proof}[{\textbf{Proof of Theorem~\ref{thm:CoVaR est alt}:}}]
The first part of the proof is roughly similar to that of Theorem~\ref{thm:CoVaR est}.
To highlight the analogy, we often overload notation by redefining quantities from that proof.
In the first part of this proof, we show that, as $n\to\infty$,
\begin{equation}\label{eq:(p.42)}
(s-r)\mD_n\big[\widehat{\vbeta}_n(r,s) - \vbeta_0\big]\overset{d}{\longrightarrow}\mSigma_{\ast}^{1/2}\big[\mW_{\ast}(s)-\mW_{\ast}(r)\big]+ \int_{r}^{s}\vb(x)\D x\qquad\text{in }(\ell^{\infty}(\mathcal{D}_{\iota}))^{k+1},
\end{equation}
where $\mW_{\ast}(\cdot)$ is a $(k+1)$-variate standard Brownian motion, and $\mSigma_{\ast}=\alpha(1-\alpha)\mK_{\ast}^{-1}\mOmega_{\ast}\mK_{\ast}^{-1}$ with 
\[
	\mOmega_{\ast}=\alpha^{-1}\beta(1-\beta)\mOmega + \alpha^{-1}(1-\alpha)^{-1}\big[(1-\beta)\mK - \mK_{\dagger}\big]\mK^{-1}\mOmega\mK^{-1}\big[(1-\beta)\mK - \mK_{\dagger}\big]
\]
as defined in Appendix~\ref{Asymptotic Variance-Covariance Matrices}.
The estimator $\widehat{\vbeta}_n(r,s)$ can equivalently be written as
\begin{align*}
\widehat{\vbeta}_n(r,s) &= \argmin_{\vbeta\in\mathbb{R}^{k+1}} \sum_{t=\lfloor nr\rfloor+1}^{\lfloor ns\rfloor}\1_{\{Y_t>\mX_{t-1}^\prime\widehat{\valpha}_n(r,s)\}}\Big[\rho_{\beta}(Z_t-\mX_{t-1}^\prime\vbeta) - \rho_{\beta}\big(\delta_{t}+\vb^\prime(t/n)\mD_n^{-1}\mX_{t-1}\big)\Big]\\
&= \argmin_{\vbeta\in\mathbb{R}^{k+1}} \sum_{t=\lfloor nr\rfloor+1}^{\lfloor ns\rfloor}\1_{\{Y_t>\mX_{t-1}^\prime\widehat{\valpha}_n(r,s)\}}\Big[\rho_{\beta}\big(Z_t-\mX_{t-1}^\prime\vbeta_{0,t}-\mX_{t-1}^\prime(\vbeta-\vbeta_{0,t})\big)\\
&\hspace{9cm} - \rho_{\beta}\big(\delta_{t}+\vb^\prime(t/n)\mD_n^{-1}\mX_{t-1}\big)\Big]\\
&\overset{\eqref{eq:(CoVaRalt)}}{=} \argmin_{\vbeta\in\mathbb{R}^{k+1}} \sum_{t=\lfloor nr\rfloor+1}^{\lfloor ns\rfloor}\1_{\{Y_t>\mX_{t-1}^\prime\widehat{\valpha}_n(r,s)\}}\Big[\rho_{\beta}\big(\delta_{t} - (\vbeta-\vbeta_{0})^\prime\mX_{t-1} + \vb^\prime(t/n)\mD_n^{-1}\mX_{t-1}\big) \\
&\hspace{9cm}- \rho_{\beta}\big(\delta_{t}+\vb^\prime(t/n)\mD_n^{-1}\mX_{t-1}\big)\Big],
\end{align*}
where we used in the last step that $\vbeta_{0,t}=\vbeta_{0}+\mD_n^{-1}\vb(t/n)$ under $\mathcal{H}_1^{\CoVaR}$.
Therefore, if we define 
\begin{multline}\label{eq:(p.16W) alt}
	f_n(\vw,r,s)=\sum_{t=\lfloor nr\rfloor+1}^{\lfloor ns\rfloor}\1_{\{Y_t>\mX_{t-1}^\prime\widehat{\valpha}_n(r,s)\}}\times\\
\times\Big[\rho_{\beta}\big(\delta_{t} - \vw^\prime\mD_n^{-1}\mX_{t-1} + \vb^\prime(t/n)\mD_n^{-1}\mX_{t-1}\big) - \rho_{\beta}\big(\delta_{t} + \vb^\prime(t/n)\mD_n^{-1}\mX_{t-1}\big)\Big],
\end{multline}
then the minimizer $\vw_n(r,s)$ of $f_n(\cdot,r,s)$ satisfies that 
\[
	\vw_n(r,s)=\mD_n\big[\widehat{\vbeta}_n(r,s) - \vbeta_0\big].
\]
As in the proof of Theorem~\ref{thm:CoVaR est}, we invoke \citet[Theorem~2]{Kat09} to derive the weak limit of $\vw_n(\cdot,\cdot)$.

We first rewrite $f_n(\cdot, r,s)$. 
Note that by Assumption~\ref{ass:innov CoVaR},
\begin{align*}
	\P&\big\{\exists\ t\in\mathbb{N}\colon \delta_t=-\vb^\prime(t/n)\mD_n^{-1}\mX_{t-1}\big\}\\
	&=\P\bigg\{\bigcup_{t\in\mathbb{N}}\{ \delta_t=-\vb^\prime(t/n)\mD_n^{-1}\mX_{t-1}\}\bigg\}\\
	& \leq\sum_{t\in\mathbb{N}}\P\big\{ \delta_t=-\vb^\prime(t/n)\mD_n^{-1}\mX_{t-1}\big\}\\
	&=0,
\end{align*}
such that, by \eqref{eq:(1)}, almost surely (a.s.),
\begin{align}
\rho_{\beta}&\big(\delta_{t} - \vw^\prime\mD_n^{-1}\mX_{t-1} + \vb^\prime(t/n)\mD_n^{-1}\mX_{t-1}\big) - \rho_{\beta}\big(\delta_{t}+\vb^\prime(t/n)\mD_n^{-1}\mX_{t-1}\big)\notag\\
&=-\vw^\prime\psi_{\beta}\big(\delta_{t}+\vb^\prime(t/n)\mD_n^{-1}\mX_{t-1}\big)\mD_n^{-1}\mX_{t-1} + \big[\delta_t + \vb^\prime(t/n)\mD_n^{-1}\mX_{t-1} - \vw^\prime\mD_n^{-1}\mX_{t-1}\big]\times\notag\\
&\hspace{2cm}\times\big[\1_{\{\vw^\prime\mD_n^{-1}\mX_{t-1}<\delta_{t}+\vb^\prime(t/n)\mD_n^{-1}\mX_{t-1}<0\}} - \1_{\{0<\delta_{t}+\vb^\prime(t/n)\mD_n^{-1}\mX_{t-1}<\vw^\prime\mD_n^{-1}\mX_{t-1}\}}\big].\label{eq:pbeta}
\end{align}
Moreover, similarly as in \eqref{eq:ind trafo}, we get that
\begin{equation}\label{eq:(H.3)}
	\1_{\{Y_t>\mX_{t-1}^\prime\widehat{\valpha}_n(r,s)\}} =\1_{\{\mX_{t-1}^\prime[\widehat{\valpha}_n(r,s) - \valpha_{0,t}]<\epsilon_t\leq0\}} - \1_{\{0<\epsilon_t\leq \mX_{t-1}^\prime[\widehat{\valpha}_n(r,s) - \valpha_{0,t}]\}} + \1_{\{\epsilon_t>0\}}.
\end{equation}
Also, from \eqref{eq:diff psi},
\begin{equation}\label{eq:diff psi beta}
	\psi_{\beta}\big(\delta_{t}+\vb^\prime(t/n)\mD_n^{-1}\mX_{t-1}\big) - \psi_{\beta}\big(\delta_{t}\big)	= \1_{\{-\vb^\prime(t/n)\mD_n^{-1}\mX_{t-1}<\delta_t\leq0\}} - \1_{\{0<\delta_t\leq -\vb^\prime(t/n)\mD_n^{-1}\mX_{t-1}\}}.
\end{equation}
Using \eqref{eq:pbeta} first, \eqref{eq:(H.3)} in the next step, and finally \eqref{eq:diff psi beta}, we deduce from \eqref{eq:(p.16W) alt} that a.s.,
\begin{align*}
	f_n(\vw,r,s) &= -\vw^\prime\sum_{t=\lfloor nr\rfloor+1}^{\lfloor ns\rfloor}\1_{\{Y_t>\mX_{t-1}^\prime\widehat{\valpha}_n(r,s)\}}\psi_{\beta}(\delta_{t})\mD_n^{-1}\mX_{t-1}\\
	& \hspace{0.5cm}-\vw^\prime\sum_{t=\lfloor nr\rfloor+1}^{\lfloor ns\rfloor}\1_{\{Y_t>\mX_{t-1}^\prime\widehat{\valpha}_n(r,s)\}}\Big[\psi_{\beta}\big(\delta_{t}+\vb^\prime(t/n)\mD_n^{-1}\mX_{t-1}\big) - \psi_{\beta}(\delta_{t})\Big]\mD_n^{-1}\mX_{t-1}\\
	& \hspace{0.5cm} + \sum_{t=\lfloor nr\rfloor+1}^{\lfloor ns\rfloor}\1_{\{Y_t>\mX_{t-1}^\prime\widehat{\valpha}_n(r,s)\}}\big[\delta_t + \vb^\prime(t/n)\mD_n^{-1}\mX_{t-1} - \vw^\prime\mD_n^{-1}\mX_{t-1}\big]\times\notag\\
&\hspace{2cm}\times\big[\1_{\{\vw^\prime\mD_n^{-1}\mX_{t-1}<\delta_{t}+\vb^\prime(t/n)\mD_n^{-1}\mX_{t-1}<0\}} - \1_{\{0<\delta_{t}+\vb^\prime(t/n)\mD_n^{-1}\mX_{t-1}<\vw^\prime\mD_n^{-1}\mX_{t-1}\}}\big]\\
	&= -\vw^\prime\sum_{t=\lfloor nr\rfloor+1}^{\lfloor ns\rfloor}\1_{\{\epsilon_t>0\}}\psi_{\beta}(\delta_{t})\mD_n^{-1}\mX_{t-1}\\
	& \hspace{0.5cm}-\vw^\prime\sum_{t=\lfloor nr\rfloor+1}^{\lfloor ns\rfloor}\1_{\{\epsilon_t>0\}}\Big[\psi_{\beta}\big(\delta_{t}+\vb^\prime(t/n)\mD_n^{-1}\mX_{t-1}\big) - \psi_{\beta}(\delta_{t})\Big]\mD_n^{-1}\mX_{t-1}\\
	& \hspace{0.5cm} + \sum_{t=\lfloor nr\rfloor+1}^{\lfloor ns\rfloor}\1_{\{\epsilon_t>0\}}\big[\delta_t + \vb^\prime(t/n)\mD_n^{-1}\mX_{t-1} - \vw^\prime\mD_n^{-1}\mX_{t-1}\big]\times\\
&\hspace{2cm}\times\big[\1_{\{\vw^\prime\mD_n^{-1}\mX_{t-1}<\delta_{t}+\vb^\prime(t/n)\mD_n^{-1}\mX_{t-1}<0\}} - \1_{\{0<\delta_{t}+\vb^\prime(t/n)\mD_n^{-1}\mX_{t-1}<\vw^\prime\mD_n^{-1}\mX_{t-1}\}}\big]\\
	&\hspace{0.5cm} - \vw^\prime \sum_{t=\lfloor nr\rfloor+1}^{\lfloor ns\rfloor}\big[\1_{\{\mX_{t-1}^\prime[\widehat{\valpha}_n(r,s) - \valpha_{0,t}]<\epsilon_t\leq0\}} - \1_{\{0<\epsilon_t\leq \mX_{t-1}^\prime[\widehat{\valpha}_n(r,s) - \valpha_{0,t}]\}}\big]\psi_{\beta}(\delta_{t})\mD_n^{-1}\mX_{t-1}\\
	&\hspace{0.5cm} + \vw^\prime \sum_{t=\lfloor nr\rfloor+1}^{\lfloor ns\rfloor}\big[\1_{\{\mX_{t-1}^\prime[\widehat{\valpha}_n(r,s) - \valpha_{0,t}]<\epsilon_t\leq0\}} - \1_{\{0<\epsilon_t\leq \mX_{t-1}^\prime[\widehat{\valpha}_n(r,s) - \valpha_{0,t}]\}}\big]\times\\
	&\hspace{5cm}\times\Big[\psi_{\beta}(\delta_{t})-\psi_{\beta}\big(\delta_{t}-\vb^\prime(t/n)\mD_n^{-1}\mX_{t-1}\big)\Big]\mD_n^{-1}\mX_{t-1}\\
	&\hspace{0.5cm} + \sum_{t=\lfloor nr\rfloor+1}^{\lfloor ns\rfloor}\big[\1_{\{\mX_{t-1}^\prime[\widehat{\valpha}_n(r,s) - \valpha_{0,t}]<\epsilon_t\leq0\}} - \1_{\{0<\epsilon_t\leq \mX_{t-1}^\prime[\widehat{\valpha}_n(r,s) - \valpha_{0,t}]\}}\big]\times\\
	&\hspace{5cm}\times\big[\delta_t + \vb^\prime(t/n)\mD_n^{-1}\mX_{t-1} - \vw^\prime\mD_n^{-1}\mX_{t-1}\big]\times\notag\\
&\hspace{2cm}\times\big[\1_{\{\vw^\prime\mD_n^{-1}\mX_{t-1}<\delta_{t}+\vb^\prime(t/n)\mD_n^{-1}\mX_{t-1}<0\}} - \1_{\{0<\delta_{t}+\vb^\prime(t/n)\mD_n^{-1}\mX_{t-1}<\vw^\prime\mD_n^{-1}\mX_{t-1}\}}\big]\\
	&= -\vw^\prime\bigg\{\sum_{t=\lfloor nr\rfloor+1}^{\lfloor ns\rfloor}\1_{\{\epsilon_t>0\}}\psi_{\beta}(\delta_{t})\mD_n^{-1}\mX_{t-1}\\
	& \hspace{1cm} + \sum_{t=\lfloor nr\rfloor+1}^{\lfloor ns\rfloor}\1_{\{\epsilon_t>0\}}\Big[\1_{\{-\vb^\prime(t/n)\mD_n^{-1}\mX_{t-1}<\delta_t\leq0\}} - \1_{\{0<\delta_t\leq -\vb^\prime(t/n)\mD_n^{-1}\mX_{t-1}\}}\Big]\mD_n^{-1}\mX_{t-1}\\
	&\hspace{1cm} + \sum_{t=\lfloor nr\rfloor+1}^{\lfloor ns\rfloor}\big[\1_{\{\mX_{t-1}^\prime[\widehat{\valpha}_n(r,s) - \valpha_{0,t}]<\epsilon_t\leq0\}} - \1_{\{0<\epsilon_t\leq \mX_{t-1}^\prime[\widehat{\valpha}_n(r,s) - \valpha_{0,t}]\}}\big]\psi_{\beta}(\delta_{t})\mD_n^{-1}\mX_{t-1}\\
	&\hspace{1cm} - \sum_{t=\lfloor nr\rfloor+1}^{\lfloor ns\rfloor}\big[\1_{\{\mX_{t-1}^\prime[\widehat{\valpha}_n(r,s) - \valpha_{0,t}]<\epsilon_t\leq0\}} - \1_{\{0<\epsilon_t\leq \mX_{t-1}^\prime[\widehat{\valpha}_n(r,s) - \valpha_{0,t}]\}}\big]\times\\
	&\hspace{5cm}\times\Big[\psi_{\beta}(\delta_{t})-\psi_{\beta}\big(\delta_{t}-\vb^\prime(t/n)\mD_n^{-1}\mX_{t-1}\big)\Big]\mD_n^{-1}\mX_{t-1}\bigg\}\\
	& \hspace{0.5cm} + \sum_{t=\lfloor nr\rfloor+1}^{\lfloor ns\rfloor}\1_{\{\epsilon_t>0\}}\big[\delta_t + \vb^\prime(t/n)\mD_n^{-1}\mX_{t-1} - \vw^\prime\mD_n^{-1}\mX_{t-1}\big]\times\\
&\hspace{2cm}\times\big[\1_{\{\vw^\prime\mD_n^{-1}\mX_{t-1}<\delta_{t}+\vb^\prime(t/n)\mD_n^{-1}\mX_{t-1}<0\}} - \1_{\{0<\delta_{t}+\vb^\prime(t/n)\mD_n^{-1}\mX_{t-1}<\vw^\prime\mD_n^{-1}\mX_{t-1}\}}\big]\\
	&\hspace{0.5cm} + \sum_{t=\lfloor nr\rfloor+1}^{\lfloor ns\rfloor}\big[\1_{\{\mX_{t-1}^\prime[\widehat{\valpha}_n(r,s) - \valpha_{0,t}]<\epsilon_t\leq0\}} - \1_{\{0<\epsilon_t\leq \mX_{t-1}^\prime[\widehat{\valpha}_n(r,s) - \valpha_{0,t}]\}}\big]\times\\
	&\hspace{5cm}\times\big[\delta_t + \vb^\prime(t/n)\mD_n^{-1}\mX_{t-1} - \vw^\prime\mD_n^{-1}\mX_{t-1}\big]\times\\
&\hspace{2cm}\times\big[\1_{\{\vw^\prime\mD_n^{-1}\mX_{t-1}<\delta_{t}+\vb^\prime(t/n)\mD_n^{-1}\mX_{t-1}<0\}} - \1_{\{0<\delta_{t}+\vb^\prime(t/n)\mD_n^{-1}\mX_{t-1}<\vw^\prime\mD_n^{-1}\mX_{t-1}\}}\big].
\end{align*}

Define
\[
	g_n(\vw,r,s) := -\vw^\prime\mW_n(r,s) + \frac{1}{2}\vw^\prime(s-r)\mK_{\ast}\vw,
\]
where
\begin{align*}
	\mW_n(r,s) &:= \sum_{t=\lfloor nr\rfloor+1}^{\lfloor ns\rfloor}\1_{\{\epsilon_t>0\}}\psi_{\beta}(\delta_{t})\mD_n^{-1}\mX_{t-1}\\
	& \hspace{1cm} + \sum_{t=\lfloor nr\rfloor+1}^{\lfloor ns\rfloor}\1_{\{\epsilon_t>0\}}\Big[\1_{\{-\vb^\prime(t/n)\mD_n^{-1}\mX_{t-1}<\delta_t\leq0\}} - \1_{\{0<\delta_t\leq -\vb^\prime(t/n)\mD_n^{-1}\mX_{t-1}\}}\Big]\mD_n^{-1}\mX_{t-1}\\
	&\hspace{1cm} + \sum_{t=\lfloor nr\rfloor+1}^{\lfloor ns\rfloor}\big[\1_{\{\mX_{t-1}^\prime[\widehat{\valpha}_n(r,s) - \valpha_{0,t}]<\epsilon_t\leq0\}} - \1_{\{0<\epsilon_t\leq \mX_{t-1}^\prime[\widehat{\valpha}_n(r,s) - \valpha_{0,t}]\}}\big]\psi_{\beta}(\delta_{t})\mD_n^{-1}\mX_{t-1}\\
	&\hspace{1cm} - \sum_{t=\lfloor nr\rfloor+1}^{\lfloor ns\rfloor}\big[\1_{\{\mX_{t-1}^\prime[\widehat{\valpha}_n(r,s) - \valpha_{0,t}]<\epsilon_t\leq0\}} - \1_{\{0<\epsilon_t\leq \mX_{t-1}^\prime[\widehat{\valpha}_n(r,s) - \valpha_{0,t}]\}}\big]\times\\
	&\hspace{5cm}\times\Big[\psi_{\beta}(\delta_{t})-\psi_{\beta}\big(\delta_{t}-\vb^\prime(t/n)\mD_n^{-1}\mX_{t-1}\big)\Big]\mD_n^{-1}\mX_{t-1}.
\end{align*}
By Propositions~\ref{lem:LLN alt CoVaR} and \ref{lem:LLN3 Alt2}, 
\[
	\sup_{(r,s)\in\mathcal{D}_{\iota}}\big|f_n(\vw,r,s) - g_n(\vw,r,s)\big|=o_{\P}(1)
\]
for each $\vw\in\mathbb{R}^{k+1}$ (which is the equivalent of equation (10) in \citet{Kat09}).
Moreover, by Propositions~\ref{lem:CLT2}, \ref{prop:1alt CoVaR}, \ref{lem:LLN CLT Alt} and \ref{lem:LLN3 Alt1},
\begin{multline}\label{eq:convWN3}
	\mW_n(r,s)\overset{d}{\longrightarrow}\mW_{\infty}(r,s):=\sqrt{(1-\alpha)(1-\beta)\beta}\mOmega^{1/2}\big[\mW_{\dagger}(s)-\mW_{\dagger}(r)\big] + \mK_{\ast}\int_{r}^{s}\vb(x)\D x\\
	 +\big[(1-\beta)\mK - \mK_{\dagger}\big]\mSigma^{1/2}\big[\mW(s)-\mW(r)\big]\qquad\text{in }(\ell^{\infty}(\mathcal{D}_{\iota}))^{k+1}.
\end{multline}
As in the proof of Theorem~\ref{thm:CoVaR est}, the CMT hence implies that for every $\eta>0$ there exists $M>0$, such that
\[
	\limsup_{n\to\infty}\P\bigg\{\sup_{(r,s)\in\mathcal{D}_{\iota}}\big\Vert\mW_n(r,s)\big\Vert>M\bigg\}\leq\eta.
\]
(This is the analog of equation (11) in \citet{Kat09}.)
We obtain from Theorem~2 in \citet{Kat09} that
\[
	\vw_n(r,s)=\frac{1}{s-r}\mK_{\ast}^{-1}\mW_n(r,s) + \vr_n(r,s),
\]
where $\sup_{(r,s)\in\mathcal{D}_{\iota}}\Vert\vr_n(r,s)\Vert=o_{\P}(1)$.
Therefore, by \eqref{eq:convWN3},
\begin{equation*}
	\vw_n(r,s)\overset{d}{\longrightarrow}\vw_{\infty}(r,s)\qquad\text{in }(\ell^{\infty}(\mathcal{D}_{\iota}))^{k+1},
\end{equation*}
where
\begin{multline*}
	\vw_{\infty}(r,s) = \frac{1}{s-r}\mK_{\ast}^{-1}\Big(\sqrt{(1-\alpha)(1-\beta)\beta}\mOmega^{1/2}\big[\mW_{\dagger}(s)-\mW_{\dagger}(r)\big]+ \mK_{\ast}\int_{r}^{s}\vb(x)\D x \\
	 +\big[(1-\beta)\mK - \mK_{\dagger}\big]\mSigma^{1/2}\big[\mW(s)-\mW(r)\big]\Big).
\end{multline*}

By similar arguments leading up to \eqref{eq:(p.19a)} in the proof of Theorem~\ref{thm:CoVaR est}, it may be shown that $\mW(\cdot)$ and $\mW_{\dagger}(\cdot)$ are independent Brownian motions and, hence,
\[
	\mW_{\infty}(r,s)\overset{d}{=}\sqrt{\alpha(1-\alpha)}\mOmega_{\ast}^{1/2}\big[\mW_{\ast}(s)-\mW_{\ast}(r)\big] + \mK_{\ast}\int_{r}^{s}\vb(x)\D x
\]
for a $(k+1)$-variate standard Brownian motion $\mW_{\ast}(\cdot)$.
Thus, \eqref{eq:(p.42)} follows.

It remains to establish the joint convergence of $\widehat{\valpha}_n(r,s)$ and $\widehat{\vbeta}_n(r,s)$.
For this, we follow similar steps as in the proof of Theorem~\ref{thm:CoVaR est} to obtain that, as $n\to\infty$,
\begin{equation*}
	(s-r)\begin{pmatrix}
		\mD_n \big[\widehat{\valpha}_n(r,s) - \valpha_0\big]\\
		\mD_n \big[\widehat{\vbeta}_n(r,s) - \vbeta_0\big]
	\end{pmatrix}\overset{d}{\longrightarrow}\overline{\mSigma}^{1/2}\big[\overline{\mW}(s)-\overline{\mW}(r)\big]\\
	 + \begin{pmatrix}\int_{r}^{s}\va(x)\D x\\ \int_{r}^{s}\vb(x)\D x\end{pmatrix}\qquad\text{in }(\ell^{\infty}(\mathcal{D}_{\iota}))^{2k+2}.
\end{equation*}
This finishes proof.
\end{proof}

\begin{proof}[{\textbf{Proof of Corollary~\ref{cor:one-break}:}}]
Consider the expression over which the supremum is taken in the definition of $\mathcal{U}_{n,\vgamma}$ evaluated at the break point $s=s^{\ast}$, i.e.,
\begin{equation}\label{eq:s ast eval}
	(s^{\ast})^2(1-s^{\ast})^2\big[\widehat{\vgamma}_n(0,s^{\ast}) - \widehat{\vgamma}_n(s^{\ast},1)\big]^\prime \bm{\mathcal{N}}_{n,\vgamma}^{-1}(s^{\ast})\big[\widehat{\vgamma}_n(0,s^{\ast}) - \widehat{\vgamma}_n(s^{\ast},1)\big].
\end{equation}
Recall the definition of $\overline{\mD}_n$ from \eqref{eq:overline D}.
Then, note that from the CMT and Theorem~\ref{thm:CoVaR est alt},
\begin{align*}
	s^{\ast}&(1-s^{\ast})\overline{\mD}_n\big[\widehat{\vgamma}_n(0,s^{\ast}) - \widehat{\vgamma}_n(s^{\ast},1)\big]\\
	&= (1-s^{\ast})s^{\ast}\begin{pmatrix}\mD_n\big[\widehat{\valpha}_n(0,s^\ast)-\valpha_0\big]\\ \mD_n\big[\widehat{\vbeta}_n(0,s^\ast)-\vbeta_0\big]\end{pmatrix}	
	- s^{\ast}(1-s^{\ast})\begin{pmatrix}\mD_n\big[\widehat{\valpha}_n(s^\ast,1)-\valpha_0\big]\\ \mD_n\big[\widehat{\vbeta}_n(s^\ast,1)-\vbeta_0\big]\end{pmatrix}\\
	&\overset{d}{\longrightarrow}(1-s^{\ast})\big\{\overline{\mSigma}^{1/2}\overline{\mW}(s^{\ast})+ s^{\ast}\vc_1 \big\}\\
	&\hspace{1.9cm}- s^{\ast}\big\{\overline{\mSigma}^{1/2}\big[\overline{\mW}(1)-\overline{\mW}(s^{\ast})\big] + (1-s^{\ast})\vc_2\big\}\\
	&= \overline{\mSigma}^{1/2}\Big\{\big[\overline{\mW}(s^{\ast})-s^{\ast}\overline{\mW}(1)\big] + s^\ast(1-s^{\ast})\overline{\mSigma}^{-1/2}\big(\vc_1-\vc_2\big)\Big\}.
\end{align*}
Therefore, the CMT and Theorem~\ref{thm:CoVaR est alt} imply that the random variable in \eqref{eq:s ast eval} converges weakly to
\begin{multline*}
	\Big\{\big[\overline{\mW}(s^{\ast})-s^{\ast}\overline{\mW}(1)\big] + s^{\ast}(1-s^{\ast})\overline{\mSigma}^{-1/2}\big(\vc_1-\vc_2\big)\Big\}^{\prime}\overline{\bm{\mathcal{\mW}}}^{-1}(s^{\ast})\times\\
	\times\Big\{\big[\overline{\mW}(s^{\ast})-s^{\ast}\overline{\mW}(1)\big] + s^{\ast}(1-s^{\ast})\overline{\mSigma}^{-1/2}\big(\vc_1-\vc_2\big)\Big\},
\end{multline*}
where $\overline{\bm{\mathcal{\mW}}}(\cdot)$ is defined in Corollary~\ref{cor:SBT CoVaR}.
Because $\overline{\bm{\mathcal{\mW}}}^{-1}(s^{\ast})$ does not depend on $\vc_i$ ($i=1,2$), this term diverges in probability to $\infty$, as $\Vert\vc_2-\vc_1\Vert\to\infty$.
Thus, consistency follows.
\end{proof}

\section{Proofs of Propositions~\ref{prop:1alt CoVaR}--\ref{lem:LLN3 Alt2}}\label{sec:Prop alt CoVaR}

\begin{proof}[{\textbf{Proof of Proposition~\ref{prop:1alt CoVaR}:}}]
The proof follows along similar lines as that of Proposition~\ref{prop:1alt}.
Define
\begin{align*}
	\vnu_t &:= \1_{\{\epsilon_t>0\}}\big[\1_{\{-\vb^\prime(t/n)\mD_n^{-1}\mX_{t-1}<\delta_t\leq0\}} - \1_{\{0<\delta_t\leq -\vb^\prime(t/n)\mD_n^{-1}\mX_{t-1}\}}\big]\mD_n^{-1}\mX_{t-1},\\
	\overline{\vnu}_t &:= \E_{t-1}[\vnu_t],\\
	\mV_n(s) &:= \sum_{t=1}^{\lfloor ns\rfloor}\vnu_t,\\
	\overline{\mV}_n(s) &:= \sum_{t=1}^{\lfloor ns\rfloor}\overline{\vnu}_t.
\end{align*}
As in the proof of Proposition~\ref{prop:1alt}, it suffices to show that, uniformly in $s\in[0,1]$, 
\begin{align}
\overline{\mV}_n(s) &\overset{\P}{\longrightarrow} \mK_{\ast}\Big(\int_{0}^{s}\vb(x)\D x\Big),\label{eq:(1.1) CoVaR}\\
\mV_n(s) - \overline{\mV}_n(s) &\overset{\P}{\longrightarrow}\vzero.\label{eq:(1.2) CoVaR}
\end{align}
The proof of \eqref{eq:(1.2) CoVaR} is analogous to that of \eqref{eq:(pp.6.2)}, so we only show \eqref{eq:(1.1) CoVaR}.
Write
\begin{align*}
\overline{\mV}_n(s)&= \sum_{t=1}^{\lfloor ns\rfloor}\mD_n^{-1}\mX_{t-1}\E_{t-1}\Big[\1_{\{\epsilon_t>0\}}\1_{\{-\vb^\prime(t/n)\mD_n^{-1}\mX_{t-1}<\delta_t\leq0\}}\Big]\\
	&\hspace{0.5cm} - \sum_{t=1}^{\lfloor ns\rfloor}\mD_n^{-1}\mX_{t-1}\E_{t-1}\Big[\1_{\{\epsilon_t>0\}}\1_{\{0<\delta_t\leq -\vb^\prime(t/n)\mD_n^{-1}\mX_{t-1}\}}\Big]\\	
	&=:\overline{\mV}_{1n}(s) - \overline{\mV}_{2n}(s).
\end{align*}
Both terms can be treated similarly, so we only deal with $\overline{\mV}_{2n}(s)$.
For this, write
\begin{align*}
	\overline{\mV}_{2n}(s)
	&= \sum_{t=1}^{\lfloor ns\rfloor}\1_{\{-\vb^\prime(t/n)\mD_n^{-1}\mX_{t-1}>0\}}\mD_n^{-1}\mX_{t-1}\int_{0}^{\infty}\Big(\int_{0}^{-\vb^\prime(t/n)\mD_n^{-1}\mX_{t-1}}f_{(\epsilon_t,\delta_t)^\prime\mid\mathcal{F}_{t-1}}(x,y) \D y\Big)\D x\\
	&= \sum_{t=1}^{\lfloor ns\rfloor}\1_{\{-\vb^\prime(t/n)\mD_n^{-1}\mX_{t-1}>0\}}\mD_n^{-1}\mX_{t-1}\int_{0}^{-\vb^\prime(t/n)\mD_n^{-1}\mX_{t-1}}\Big(\int_{0}^{\infty}f_{(\epsilon_t,\delta_t)^\prime\mid\mathcal{F}_{t-1}}(x,y) \D x\Big)\D y\\	
	&= \sum_{t=1}^{\lfloor ns\rfloor}\1_{\{-\vb^\prime(t/n)\mD_n^{-1}\mX_{t-1}>0\}}\mD_n^{-1}\mX_{t-1}\int_{0}^{-\vb^\prime(t/n)\mD_n^{-1}\mX_{t-1}}\Big(\int_{0}^{\infty}f_{(\epsilon_t,\delta_t)^\prime\mid\mathcal{F}_{t-1}}(x,0) \D x\Big)\D y\\	
	&\hspace{0.5cm} + \sum_{t=1}^{\lfloor ns\rfloor}\1_{\{-\vb^\prime(t/n)\mD_n^{-1}\mX_{t-1}>0\}}\mD_n^{-1}\mX_{t-1}\times\\
	&\hspace{1cm}\times\int_{0}^{-\vb^\prime(t/n)\mD_n^{-1}\mX_{t-1}}\Big(\int_{0}^{\infty}f_{(\epsilon_t,\delta_t)^\prime\mid\mathcal{F}_{t-1}}(x,y) \D x - \int_{0}^{\infty}f_{(\epsilon_t,\delta_t)^\prime\mid\mathcal{F}_{t-1}}(x,0) \D x\Big)\D y\\
	&=:\overline{\mV}_{21n}(s) + \overline{\mV}_{22n}(s).
\end{align*}
Integrating out yields that
\[
	\overline{\mV}_{21n}(s)=-\sum_{t=1}^{\lfloor ns\rfloor}\1_{\{-\vb^\prime(t/n)\mD_n^{-1}\mX_{t-1}>0\}}\mD_n^{-1}\mX_{t-1}\mX_{t-1}^\prime\mD_n^{-1}\vb(t/n)\Big(\int_{0}^{\infty}f_{(\epsilon_t,\delta_t)^\prime\mid\mathcal{F}_{t-1}}(x,0) \D x\Big).
\]
Also, Assumption~\ref{ass:innov CoVaR}~\eqref{it:Lipschitz CoVaR} implies that uniformly in $s\in[0,1]$, 
\begin{align*}
	\big|\overline{\mV}_{22n}(s)\big| &\leq \sum_{t=1}^{\lfloor ns\rfloor}\1_{\{-\vb^\prime(t/n)\mD_n^{-1}\mX_{t-1}>0\}}\big\Vert\mD_n^{-1}\mX_{t-1}\big\Vert\times\\
	&\hspace{1cm}\times\int_{0}^{-\vb^\prime(t/n)\mD_n^{-1}\mX_{t-1}}|y-0|\frac{\Big|\int_{0}^{\infty}f_{(\epsilon_t,\delta_t)^\prime\mid\mathcal{F}_{t-1}}(x,y) \D x - \int_{0}^{\infty}f_{(\epsilon_t,\delta_t)^\prime\mid\mathcal{F}_{t-1}}(x,0) \D x\Big|}{|y-0|}\D y\\
	&\leq L\max_{t=1,\ldots,n}\big\Vert\mD_n^{-1}\mX_{t-1}\big\Vert \sum_{t=1}^{n}\1_{\{-\vb^\prime(t/n)\mD_n^{-1}\mX_{t-1}>0\}}\int_{0}^{-\vb^\prime(t/n)\mD_n^{-1}\mX_{t-1}}|y|\D y\\
	&\leq L\max_{t=1,\ldots,n}\big\Vert\mD_n^{-1}\mX_{t-1}\big\Vert\sum_{t=1}^{n}\frac{1}{2}\big(-\vb^\prime(t/n)\mD_n^{-1}\mX_{t-1}\big)^2\\
	&= \frac{L}{2}\max_{t=1,\ldots,n}\big\Vert\mD_n^{-1}\mX_{t-1}\big\Vert\sum_{t=1}^{n}\vb^\prime(t/n)\mD_n^{-1}\mX_{t-1}\mX_{t-1}^\prime\mD_n^{-1}\vb(t/n)\\
	&\leq K\max_{t=1,\ldots,n}\big\Vert\mD_n^{-1}\mX_{t-1}\big\Vert\sum_{t=1}^{n}\big\Vert\mD_n^{-1}\mX_{t-1}\mX_{t-1}^\prime\mD_n^{-1}\big\Vert\\
	&=o_{\P}(1)O_{\P}(1)\\
	&= o_{\P}(1),
\end{align*}
where the penultimate step follows from Lemma~\ref{lem:MAX} and \eqref{eq:Op1 cross}.
Taken together, we obtain that
\[
	\overline{\mV}_{2n}(s) = -\sum_{t=1}^{\lfloor ns\rfloor}\1_{\{-\vb^\prime(t/n)\mD_n^{-1}\mX_{t-1}>0\}}\mD_n^{-1}\mX_{t-1}\mX_{t-1}^\prime\mD_n^{-1}\vb(t/n)\Big(\int_{0}^{\infty}f_{(\epsilon_t,\delta_t)^\prime\mid\mathcal{F}_{t-1}}(x,0) \D x\Big) + o_{\P}(1)
\]
uniformly in $s\in[0,1]$.

Similar arguments give that
\[
	\overline{\mV}_{1n}(s) = \sum_{t=1}^{\lfloor ns\rfloor}\1_{\{-\vb^\prime(t/n)\mD_n^{-1}\mX_{t-1}<0\}}\mD_n^{-1}\mX_{t-1}\mX_{t-1}^\prime\mD_n^{-1}\vb(t/n)\Big(\int_{0}^{\infty}f_{(\epsilon_t,\delta_t)^\prime\mid\mathcal{F}_{t-1}}(x,0) \D x\Big) + o_{\P}(1).
\]
Combining the previous two displays and then using Assumption~\ref{ass:K ast} yields that
\begin{align*}
	\overline{\mV}_{n}(s) &= \sum_{t=1}^{\lfloor ns\rfloor}\mD_n^{-1}\mX_{t-1}\mX_{t-1}^\prime\mD_n^{-1}\vb(t/n)\Big(\int_{0}^{\infty}f_{(\epsilon_t,\delta_t)^\prime\mid\mathcal{F}_{t-1}}(x,0) \D x\Big) + o_{\P}(1)\\
	&\overset{\P}{\longrightarrow} \mK_{\ast}\Big(\int_{0}^{s}\vb(x)\D x\Big)
\end{align*}
uniformly in $s\in[0,1]$, finishing the proof.
\end{proof}

\begin{proof}[{\textbf{Proof of Proposition~\ref{lem:LLN alt CoVaR}:}}]
The proof follows along similar lines as that of Proposition~\ref{lem:LLN2}.
Once more, it suffices to prove that
\begin{multline*}
	\sup_{0\leq s\leq 1}\bigg|\sum_{t=1}^{\lfloor ns\rfloor}\1_{\{\epsilon_t>0\}}\big[\delta_t + \vb^\prime(t/n)\mD_n^{-1}\mX_{t-1} - \vw^\prime\mD_n^{-1}\mX_{t-1}\big]\times\\
\times\big[\1_{\{\vw^\prime\mD_n^{-1}\mX_{t-1}<\delta_{t}+\vb^\prime(t/n)\mD_n^{-1}\mX_{t-1}<0\}} - \1_{\{0<\delta_{t}+\vb^\prime(t/n)\mD_n^{-1}\mX_{t-1}<\vw^\prime\mD_n^{-1}\mX_{t-1}\}}\big]\\	
-\frac{1}{2}s\vw^\prime\mK_{\ast}\vw \bigg|=o_{\P}(1).
\end{multline*}
To do so, define
\begin{align*}
	\nu_t(\vw) &:= \1_{\{\epsilon_t>0\}}\big[\delta_t + \vb^\prime(t/n)\mD_n^{-1}\mX_{t-1} - \vw^\prime\mD_n^{-1}\mX_{t-1}\big]\times\\
&\hspace{1cm}\times\big[\1_{\{\vw^\prime\mD_n^{-1}\mX_{t-1}<\delta_{t}+\vb^\prime(t/n)\mD_n^{-1}\mX_{t-1}<0\}} - \1_{\{0<\delta_{t}+\vb^\prime(t/n)\mD_n^{-1}\mX_{t-1}<\vw^\prime\mD_n^{-1}\mX_{t-1}\}}\big],\\
	\overline{\nu}_t(\vw) &:= \E_{t-1}\big[\nu_t(\vw)\big],\\
	V_n(\vw,s) &:= \sum_{t=1}^{\lfloor ns\rfloor}\nu_t(\vw),\\
	\overline{V}_n(\vw,s) &:= \sum_{t=1}^{\lfloor ns\rfloor}\overline{\nu}_t(\vw).
\end{align*}
Then, we have to show that, uniformly in $s\in[0,1]$, 
\begin{align}
\overline{V}_n(\vw,s) &\overset{\P}{\longrightarrow} \frac{1}{2}s\vw^\prime\mK_{\ast}\vw,\label{eq:(1.1) CoVaR1}\\
V_n(\vw,s) - \overline{V}_n(\vw,s) &\overset{\P}{\longrightarrow}0.\label{eq:(1.2) CoVaR1}
\end{align}
Since the claim \eqref{eq:(1.2) CoVaR1} follows along similar lines as that of \eqref{eq:(pp.6.2)}, it only remains to show \eqref{eq:(1.1) CoVaR1}.
Write
\begin{align*}
	\overline{V}_n(\vw,s)&= \sum_{t=1}^{\lfloor ns\rfloor}\E_{t-1}\Big[\1_{\{\epsilon_t>0\}}\big(\delta_t + \vb^\prime(t/n)\mD_n^{-1}\mX_{t-1} - \vw^\prime\mD_n^{-1}\mX_{t-1}\big)\times\\
	&\hspace{7cm}\times\1_{\{\vw^\prime\mD_n^{-1}\mX_{t-1}<\delta_{t}+\vb^\prime(t/n)\mD_n^{-1}\mX_{t-1}<0\}}\Big]\\
	&\hspace{0.5cm} - \sum_{t=1}^{\lfloor ns\rfloor}\E_{t-1}\Big[\1_{\{\epsilon_t>0\}}\big(\delta_t + \vb^\prime(t/n)\mD_n^{-1}\mX_{t-1} - \vw^\prime\mD_n^{-1}\mX_{t-1}\big)\times\\
	&\hspace{7cm}\times\1_{\{0<\delta_{t}+\vb^\prime(t/n)\mD_n^{-1}\mX_{t-1}<\vw^\prime\mD_n^{-1}\mX_{t-1}\}}\Big]\\	
	&=:\overline{V}_{1n}(\vw,s) - \overline{V}_{2n}(\vw,s).
\end{align*}

Both terms can be treated similarly, so we only deal with $\overline{V}_{1n}(\vw,s)$.
For this, decompose
\begin{align*}
	&\overline{V}_{1n}(\vw,s)\\
	&= \sum_{t=1}^{\lfloor ns\rfloor}\1_{\{\vw^\prime\mD_n^{-1}\mX_{t-1}<0\}}\int_{0}^{\infty}\bigg(\int_{[\vw-\vb(t/n)]^\prime\mD_n^{-1}\mX_{t-1}}^{-\vb^\prime(t/n)\mD_n^{-1}\mX_{t-1}}\Big(y-\big[\vw-\vb(t/n)\big]^\prime\mD_n^{-1}\mX_{t-1}\Big)\\
	&\hspace{8cm}\times f_{(\epsilon_t,\delta_t)^\prime\mid\mathcal{F}_{t-1}}(x,y) \D y\bigg)\D x\\
	&= \sum_{t=1}^{\lfloor ns\rfloor}\1_{\{\vw^\prime\mD_n^{-1}\mX_{t-1}<0\}}\int_{[\vw-\vb(t/n)]^\prime\mD_n^{-1}\mX_{t-1}}^{-\vb^\prime(t/n)\mD_n^{-1}\mX_{t-1}}\Big(y-\big[\vw-\vb(t/n)\big]^\prime\mD_n^{-1}\mX_{t-1}\Big)\times\\
	&\hspace{8cm}\times\bigg(\int_{0}^{\infty}f_{(\epsilon_t,\delta_t)^\prime\mid\mathcal{F}_{t-1}}(x,y) \D x\bigg)\D y\\
	&= \sum_{t=1}^{\lfloor ns\rfloor}\1_{\{\vw^\prime\mD_n^{-1}\mX_{t-1}<0\}}\int_{[\vw-\vb(t/n)]^\prime\mD_n^{-1}\mX_{t-1}}^{-\vb^\prime(t/n)\mD_n^{-1}\mX_{t-1}}\Big(y-\big[\vw-\vb(t/n)\big]^\prime\mD_n^{-1}\mX_{t-1}\Big)\times\\
	&\hspace{8cm}\times\bigg(\int_{0}^{\infty}f_{(\epsilon_t,\delta_t)^\prime\mid\mathcal{F}_{t-1}}(x,0) \D x\bigg)\D y\\
	&\hspace{0.5cm}+\sum_{t=1}^{\lfloor ns\rfloor}\1_{\{\vw^\prime\mD_n^{-1}\mX_{t-1}<0\}}\int_{[\vw-\vb(t/n)]^\prime\mD_n^{-1}\mX_{t-1}}^{-\vb^\prime(t/n)\mD_n^{-1}\mX_{t-1}}\Big(y-\big[\vw-\vb(t/n)\big]^\prime\mD_n^{-1}\mX_{t-1}\Big)\times\\
	&\hspace{4cm}\times\bigg(\int_{0}^{\infty}f_{(\epsilon_t,\delta_t)^\prime\mid\mathcal{F}_{t-1}}(x,y) \D x - \int_{0}^{\infty}f_{(\epsilon_t,\delta_t)^\prime\mid\mathcal{F}_{t-1}}(x,0) \D x\bigg)\D y\\	
	&=:\overline{V}_{11n}(\vw,s) + \overline{V}_{12n}(\vw,s).
\end{align*}
Integrating out yields that
\begin{align*}
	\overline{V}_{11n}(\vw,s)&=\frac{1}{2}\sum_{t=1}^{\lfloor ns\rfloor}\1_{\{\vw^\prime\mD_n^{-1}\mX_{t-1}<0\}}\big(\vw^\prime\mD_n^{-1}\mX_{t-1}\big)^2\bigg(\int_{0}^{\infty}f_{(\epsilon_t,\delta_t)^\prime\mid\mathcal{F}_{t-1}}(x,0) \D x\bigg)\\
	&= \frac{1}{2}\sum_{t=1}^{\lfloor ns\rfloor}\1_{\{\vw^\prime\mD_n^{-1}\mX_{t-1}<0\}}\bigg(\int_{0}^{\infty}f_{(\epsilon_t,\delta_t)^\prime\mid\mathcal{F}_{t-1}}(x,0) \D x\bigg)\vw^\prime\mD_n^{-1}\mX_{t-1}\mX_{t-1}^\prime\mD_n^{-1}\vw.
\end{align*}
Also, Assumption~\ref{ass:innov CoVaR}~\eqref{it:Lipschitz CoVaR} implies that uniformly in $s\in[0,1]$, 
\begin{align*}
	\big|\overline{V}_{12n}(\vw,s)\big| &\leq \sum_{t=1}^{\lfloor ns\rfloor}\1_{\{\vw^\prime\mD_n^{-1}\mX_{t-1}<0\}}\int_{[\vw-\vb(t/n)]^\prime\mD_n^{-1}\mX_{t-1}}^{-\vb^\prime(t/n)\mD_n^{-1}\mX_{t-1}}\Big(y-\big[\vw-\vb(t/n)\big]^\prime\mD_n^{-1}\mX_{t-1}\Big)\times\\
	&\hspace{3.5cm}\times|y-0|\frac{\Big|\int_{0}^{\infty}f_{(\epsilon_t,\delta_t)^\prime\mid\mathcal{F}_{t-1}}(x,y) \D x - \int_{0}^{\infty}f_{(\epsilon_t,\delta_t)^\prime\mid\mathcal{F}_{t-1}}(x,0) \D x\Big|}{|y-0|}\D y\\
		&\leq L\sum_{t=1}^{\lfloor ns\rfloor}\1_{\{\vw^\prime\mD_n^{-1}\mX_{t-1}<0\}}\int_{[\vw-\vb(t/n)]^\prime\mD_n^{-1}\mX_{t-1}}^{-\vb^\prime(t/n)\mD_n^{-1}\mX_{t-1}}\Big(y-\big[\vw-\vb(t/n)\big]^\prime\mD_n^{-1}\mX_{t-1}\Big)|y-0|\D y\\
		&\leq K\max_{t=1,\ldots,n}\big\Vert\mD_n^{-1}\mX_{t-1}\big\Vert\sum_{t=1}^{\lfloor ns\rfloor}\1_{\{\vw^\prime\mD_n^{-1}\mX_{t-1}<0\}}\times\\
		&\hspace{3.5cm}\times\int_{[\vw-\vb(t/n)]^\prime\mD_n^{-1}\mX_{t-1}}^{-\vb^\prime(t/n)\mD_n^{-1}\mX_{t-1}}\Big(y-\big[\vw-\vb(t/n)\big]^\prime\mD_n^{-1}\mX_{t-1}\Big)\D y\\
		&\leq K\max_{t=1,\ldots,n}\big\Vert\mD_n^{-1}\mX_{t-1}\big\Vert\sum_{t=1}^{\lfloor ns\rfloor}\1_{\{\vw^\prime\mD_n^{-1}\mX_{t-1}<0\}}\big(\vw^\prime\mD_n^{-1}\mX_{t-1}\big)^2\\
	&\leq K\max_{t=1,\ldots,n}\big\Vert\mD_n^{-1}\mX_{t-1}\big\Vert\vw^\prime\bigg(\sum_{t=1}^{n}\mD_n^{-1}\mX_{t-1}\mX_{t-1}^\prime\mD_n^{-1}\bigg)\vw\\
	&=Ko_{\P}(1)O_{\P}(1)\\
	&=o_{\P}(1),
\end{align*}
where we used Assumption~\ref{ass:innov CoVaR}~\eqref{it:Lipschitz CoVaR} in the second step, and Lemmas~\ref{lem:MAX} and \ref{lem:SUM} in the penultimate step.
Taken together, we obtain that
\begin{multline*}
	\overline{V}_{1n}(\vw,s) = \frac{1}{2}\sum_{t=1}^{\lfloor ns\rfloor}\1_{\{\vw^\prime\mD_n^{-1}\mX_{t-1}<0\}}\bigg(\int_{0}^{\infty}f_{(\epsilon_t,\delta_t)^\prime\mid\mathcal{F}_{t-1}}(x,0) \D x\bigg)\vw^\prime\mD_n^{-1}\mX_{t-1}\mX_{t-1}^\prime\mD_n^{-1}\vw  + o_{\P}(1)
\end{multline*}
uniformly in $s\in[0,1]$.

Similar arguments yield that
\begin{multline*}
	\overline{V}_{2n}(\vw,s) = -\frac{1}{2}\sum_{t=1}^{\lfloor ns\rfloor}\1_{\{\vw^\prime\mD_n^{-1}\mX_{t-1}>0\}}\bigg(\int_{0}^{\infty}f_{(\epsilon_t,\delta_t)^\prime\mid\mathcal{F}_{t-1}}(x,0) \D x\bigg)\times\\
	\times\vw^\prime\mD_n^{-1}\mX_{t-1}\mX_{t-1}^\prime\mD_n^{-1}\vw  + o_{\P}(1).
\end{multline*}
Combining the previous two displays and then using Assumption~\ref{ass:K ast} implies that
\begin{align*}
	\overline{V}_{n}(\vw,s) &= \frac{1}{2}\sum_{t=1}^{\lfloor ns\rfloor}\bigg(\int_{0}^{\infty}f_{(\epsilon_t,\delta_t)^\prime\mid\mathcal{F}_{t-1}}(x,0) \D x\bigg)\vw^\prime\mD_n^{-1}\mX_{t-1}\mX_{t-1}^\prime\mD_n^{-1}\vw + o_{\P}(1)\\
	&\overset{\P}{\longrightarrow} \frac{1}{2}s\vw^\prime\mK_{\ast}\vw
\end{align*}
uniformly in $s\in[0,1]$. 
This establishes \eqref{eq:(1.1) CoVaR1}, finishing the proof.
\end{proof}

\begin{proof}[{\textbf{Proof of Proposition~\ref{lem:LLN CLT Alt}:}}]
Exploiting $\mathcal{H}_1^{\CoVaR}$, we may write
\begin{align*}
	&\1_{\{\mX_{t-1}^\prime[\widehat{\valpha}_n(r,s) - \valpha_{0,t}]<\epsilon_t\leq0\}} - \1_{\{0<\epsilon_t\leq \mX_{t-1}^\prime[\widehat{\valpha}_n(r,s) - \valpha_{0,t}]\}}\\
	&	=	\1_{\{\mX_{t-1}^\prime[\widehat{\valpha}_n(r,s) - \valpha_{0}]-\va^\prime(t/n)\mD_n^{-1}\mX_{t-1}<\epsilon_t\leq0\}} - \1_{\{0<\epsilon_t\leq \mX_{t-1}^\prime[\widehat{\valpha}_n(r,s) - \valpha_{0}] - \va^\prime(t/n)\mD_n^{-1}\mX_{t-1}\}}\\
	&= \1_{\{[\mD_n(\widehat{\valpha}_n(r,s) - \valpha_{0}) -\va(t/n)]^\prime\mD_n^{-1}\mX_{t-1}<\epsilon_t\leq0\}} - \1_{\{0<\epsilon_t\leq [\mD_n(\widehat{\valpha}_n(r,s) - \valpha_{0}) -\va(t/n)]^\prime\mD_n^{-1}\mX_{t-1}\}}.
\end{align*}
Define
\begin{align*}
	\vnu_t^{1}(\vv) &:= \big[\1_{\{[\vv-\va(t/n)]^\prime\mD_n^{-1}\mX_{t-1}<\epsilon_t\leq0\}} - \1_{\{0<\epsilon_t\leq[\vv-\va(t/n)]^\prime\mD_n^{-1}\mX_{t-1}\}}\big]\psi_{\beta}(\delta_{t})\mD_n^{-1}\mX_{t-1},\\
	\overline{\vnu}_t^{1}(\vv) &:= \E_{t-1}\big[\vnu_t^{1}(\vv)\big],\\
	\mV_n^{1}(\vv,r,s) &:= \sum_{t=\lfloor nr\rfloor +1}^{\lfloor ns\rfloor}\vnu_t^{1}(\vv),\\
	\overline{\mV}_n^{1}(\vv,r,s) &:= \sum_{t=\lfloor nr\rfloor +1}^{\lfloor ns\rfloor}\overline{\vnu}_t^{1}(\vv).
\end{align*}
Then, similarly as in Lemma~\ref{lem:7}, we may show that, as $n\to\infty$,
\begin{equation}\label{eq:lem7 ext}
	\sup_{\Vert\vv\Vert\leq K}\sup_{0\leq r < s\leq1}\bigg\Vert\overline{\mV}_n^{1}(\vv,r,s) - (s-r)\big[(1-\beta)\mK-\mK_{\dagger}\big]\vv + \big[(1-\beta)\mK-\mK_{\dagger}\big]\int_{r}^{s}\va(x)\D x\bigg\Vert=o_{\P}(1).
\end{equation}
To see this, replace $\vv$ with $\vv-\va(t/n)$ throughout the proof of Lemma~\ref{lem:7} to obtain that (cf.~\eqref{eq:(p.30)})
\begin{align*}
	\overline{\mV}_n^{1}(\vv,r,s) &=(1-\beta)\sum_{t=\lfloor nr\rfloor+1}^{\lfloor ns\rfloor} f_{\epsilon_t\mid\mathcal{F}_{t-1}}(0)\big(\mD_n^{-1}\mX_{t-1}\mX_{t-1}^\prime\mD_n^{-1}\big)\big[\vv-\va(t/n)\big]\\
	&\hspace{1cm}-\sum_{t=\lfloor nr\rfloor+1}^{\lfloor ns\rfloor}\bigg(\int_{0}^{\infty}f_{(\epsilon_t,\delta_t)^\prime\mid\mathcal{F}_{t-1}}(0,y)\D y\bigg)\big(\mD_n^{-1}\mX_{t-1}\mX_{t-1}^\prime\mD_n^{-1}\big)\big[\vv-\va(t/n)\big] + o_{\P}(1)
\end{align*}
uniformly in $0\leq r<s\leq1$ and $\Vert\vv\Vert\leq K$.
Therefore, Assumptions~\ref{ass:K} and \ref{ass:K ast} imply \eqref{eq:lem7 ext}.
Moreover, it is easy to show that the conclusions of Lemmas~\ref{lem:8}--\ref{lem:9} continue to hold with $\mV_n(\vv,r,s)$ and $\overline{\mV}_n(\vv,r,s)$ replaced by $\mV_n^{1}(\vv,r,s)$ and $\overline{\mV}_n^{1}(\vv,r,s)$, respectively.

With these results at hand, the remainder of the proof now follows along similar lines as that of Proposition~\ref{lem:LLN CLT}.
Plugging in $\widehat{\vv}_n^1(r,s)=\mD_n\big[\widehat{\valpha}_n(r,s)-\valpha_0\big]$ for $\vv$ in \eqref{eq:lem7 ext}, yields by Theorem~\ref{thm:std est alt} that
\begin{align*}
	\overline{\mV}_n^1\big(\widehat{\vv}_n^1(r,s),r,s\big) &\overset{d}{\longrightarrow}\big[(1-\beta)\mK-\mK_{\dagger}\big]\Big\{\mSigma^{1/2}\big[\mW(s)-\mW(r)\big]+ \int_{r}^{s}\va(x)\D x\Big\}\\
	&\hspace{2cm} - \big[(1-\beta)\mK-\mK_{\dagger}\big] \int_{r}^{s}\va(x)\D x\\
	&=\big[(1-\beta)\mK-\mK_{\dagger}\big]\mSigma^{1/2}\big[\mW(s)-\mW(r)\big]\qquad\text{in }(\ell^{\infty}(\mathcal{D}_{\iota}))^{k+1}.
\end{align*}
The remainder of the proof now follows virtually line-by-line that of Proposition~\ref{lem:LLN CLT}.
We omit details to save space.
\end{proof}

\begin{proof}[{\textbf{Proof of Proposition~\ref{lem:LLN3 Alt1}:}}]
The proof closely follows that of Proposition~\ref{lem:LLN3}.
We only show that
\[
	\sup_{(r,s)\in\mathcal{D}_{\iota}}\bigg|\sum_{t=\lfloor nr\rfloor+1}^{\lfloor ns\rfloor}\1_{\{0<\epsilon_t\leq \mX_{t-1}^\prime[\widehat{\valpha}_n(r,s) - \valpha_{0,t}]\}}\1_{\{0<\delta_{t}\leq-\vb^\prime(t/n)\mD_n^{-1}\mX_{t-1}\}}\mD_n^{-1}\mX_{t-1}\bigg|=o_{\P}(1),
\]
as the convergences involving the other indicator functions follow similarly.
It holds under $\mathcal{H}_1^{\CoVaR}$ that
\begin{align*}
	\P&\bigg\{\sup_{(r,s)\in\mathcal{D}_{\iota}}\bigg\Vert\sum_{t=\lfloor nr\rfloor+1}^{\lfloor ns\rfloor}\1_{\{0<\epsilon_t\leq \mX_{t-1}^\prime[\widehat{\valpha}_n(r,s) - \valpha_{0,t}]\}} \1_{\{0<\delta_{t}\leq-\vb^\prime(t/n)\mD_n^{-1}\mX_{t-1}\}}\mD_n^{-1}\mX_{t-1}\bigg\Vert>\varepsilon\bigg\}\\
	&\leq \P\bigg\{\sup_{(r,s)\in\mathcal{D}_{\iota}}\bigg\Vert\sum_{t=\lfloor nr\rfloor+1}^{\lfloor ns\rfloor}\1_{\{0<\epsilon_t\leq [\mD_n(\widehat{\valpha}_n(r,s) - \valpha_0)-\va(t/n)]^\prime\mD_n^{-1}\mX_{t-1}\}}\times\\
	&\hspace{3cm}\times\1_{\{0<\delta_{t}\leq-\vb^\prime(t/n)\mD_n^{-1}\mX_{t-1}\}}\mD_n^{-1}\mX_{t-1}\bigg\Vert>\varepsilon,\ \sup_{(r,s)\in\mathcal{D}_{\iota}}\big\Vert\mD_n[\widehat{\valpha}_n(r,s) - \valpha_0]\big\Vert\leq K\bigg\}\\
	&\hspace{1cm} + \P\bigg\{\sup_{(r,s)\in\mathcal{D}_{\iota}}\big\Vert\mD_n[\widehat{\valpha}_n(r,s) - \valpha_0]\big\Vert> K\bigg\}\\
	&\leq \P\bigg\{\sup_{(r,s)\in\mathcal{D}_{\iota}}\sum_{t=\lfloor nr\rfloor+1}^{\lfloor ns\rfloor}\1_{\{0<\epsilon_t\leq K\Vert\mD_n^{-1}\mX_{t-1}\Vert\}} \1_{\{0<\delta_{t}\leq-\vb^\prime(t/n)\mD_n^{-1}\mX_{t-1}\}}\big\Vert\mD_n^{-1}\mX_{t-1}\big\Vert>\varepsilon\bigg\} + o(1)\\
	&\leq \P\bigg\{\sum_{t=1}^{n}\1_{\{0<\epsilon_t\leq K\Vert\mD_n^{-1}\mX_{t-1}\Vert\}} \1_{\{0<\delta_{t}\leq-\vb^\prime(t/n)\mD_n^{-1}\mX_{t-1}\}}\big\Vert\mD_n^{-1}\mX_{t-1}\big\Vert>\varepsilon\bigg\} + o(1),
\end{align*}
as $n\to\infty$, followed by $K\to\infty$, where the penultimate step follows from Lemma~\ref{lem:MAX} and Theorem~\ref{thm:std est alt}, which implies that $\sup_{(r,s)\in\mathcal{D}_{\iota}}\big\Vert\mD_n[\widehat{\valpha}_n(r,s) - \valpha_0]\big\Vert=O_{\P}(1)$.

In light of this, it suffices to show that
\begin{equation}\label{eq:(p.21) alt}
	\sum_{t=1}^{n}\omega_t:=\sum_{t=1}^{n}\1_{\{0<\epsilon_t\leq K\Vert\mD_n^{-1}\mX_{t-1}\Vert\}} \1_{\{0<\delta_{t}\leq-\vb^\prime(t/n)\mD_n^{-1}\mX_{t-1}\}}\big\Vert\mD_n^{-1}\mX_{t-1}\big\Vert=o_{\P}(1).
\end{equation}
To do so, use Assumption~\ref{ass:innov CoVaR}~\eqref{it:dens bound CoVaR} to deduce that
\begin{align*}
	\E_{t-1}\big[\omega_t\big]&= \1_{\{0<-\vb^\prime(t/n)\mD_n^{-1}\mX_{t-1}\}}\big\Vert\mD_n^{-1}\mX_{t-1}\big\Vert \times \\
	&\hspace{4cm}\times\int_{0}^{K\Vert \mD_n^{-1}\mX_{t-1}\Vert}\bigg[\int_{0}^{-\vb^\prime(t/n)\mD_n^{-1}\mX_{t-1}}f_{(\epsilon_t,\delta_t)^\prime\mid\mathcal{F}_{t-1}}(x,y)\D y\bigg]\D x\\
	&\leq \1_{\{0<-\vb^\prime(t/n)\mD_n^{-1}\mX_{t-1}\}}\big\Vert\mD_n^{-1}\mX_{t-1}\big\Vert\int_{0}^{K\Vert \mD_n^{-1}\mX_{t-1}\Vert}\Big[\overline{f}\big\{-\vb^\prime(t/n)\mD_n^{-1}\mX_{t-1}\big\}\Big]\D x\\
	&\leq K\Vert\mD_n^{-1}\mX_{t-1}\Vert \big\Vert\mD_n^{-1}\mX_{t-1}\mX_{t-1}^\prime\mD_n^{-1}\vb(t/n)\big\Vert,
\end{align*}
where the final step follows because $-\vb^\prime(t/n)\mD_n^{-1}\mX_{t-1}>0$ is simply a positive scalar.
Therefore,
\begin{align*}
	\sum_{t=1}^{n}\E_{t-1}\big[\omega_t\big]&\leq K \max_{t=1,\ldots,n}\Vert\mD_n^{-1}\mX_{t-1}\Vert\sum_{t=1}^{n} \big\Vert\mD_n^{-1}\mX_{t-1}\mX_{t-1}^\prime\mD_n^{-1}\big\Vert\\
	&=o_{\P}(1)O_{\P}(1)\\
	&=o_{\P}(1)
\end{align*}
by Lemma~\ref{lem:MAX} and \eqref{eq:Op1 cross}.
In view of this, \eqref{eq:(p.21) alt} follows if we can show that, as $n\to\infty$,
\begin{equation}\label{eq:(p.4.L.7) alt}
	\sum_{t=1}^{n}\big\{\omega_t - \E_{t-1}[\omega_t]\big\}=o_{\P}(1).
\end{equation}
We prove this using Corollary~3.1 of \citet{HH80}.
The details are as in the proof of \eqref{eq:(p.4.L.7)} and, hence, are omitted.
\end{proof}

\begin{proof}[{\textbf{Proof of Proposition~\ref{lem:LLN3 Alt2}:}}]
The proof is similar to that of Proposition~\ref{lem:LLN3}, with similar required changes as in the proof of Proposition~\ref{lem:LLN3 Alt1}.
We omit details for brevity.
\end{proof}

\section{Proofs of Corollaries~\ref{thm:CoVaR est unsupervised}--\ref{thm:CoVaR est unsupervised loc alt}}\label{sec:proofs of cors}

\begin{proof}[{\textbf{Proof of Corollary~\ref{thm:CoVaR est unsupervised}:}}]
Under the stated conditions, Theorem~\ref{thm:CoVaR est} implies that, as $n\to\infty$,
\[
	\overline{\mD}_n\widehat{\mS}_n(r,s)= (s-r)\begin{pmatrix}
		\mD_n\big[\widehat{\valpha}_n(r,s)-\valpha_0\big]\\
		\mD_n\big[\widehat{\vbeta}_n(r,s)-\vbeta_0\big]
	\end{pmatrix}	
	\overset{d}{\longrightarrow}\overline{\mSigma}^{1/2}\Delta\overline{\mW}(r,s)\qquad\text{in }(\ell^{\infty}(\mathcal{D}_{\iota}))^{2k+2},
\]
where $\overline{\mW}(\cdot)$ is a $(2k+2)$-variate standard Brownian motion and $\overline{\mD}_n$ is from \eqref{eq:overline D}.
From this and the CMT, we obtain the desired result that
\begin{align*}
	\mathcal{V}_{n,\vgamma} &= \sup_{(r_1,r_2)\in\mathcal{D}_{\iota}}\vd^\prime(\overline{\mD}_n\widehat{\mS}_n, 0,r_1,r_2)\mXi^{-1}(\overline{\mD}_n\widehat{\mS}_n, 0,r_1,r_2)\vd(\overline{\mD}_n\widehat{\mS}_n, 0,r_1,r_2)\\
	& \hspace{0.7cm}+\sup_{(s_1,s_2)\in\mathcal{D}_{\iota}}\vd^\prime(\overline{\mD}_n\widehat{\mS}_n, s_1,s_2,1)\mXi^{-1}(\overline{\mD}_n\widehat{\mS}_n,s_1,s_2,1)\vd(\overline{\mD}_n\widehat{\mS}_n, s_1,s_2,1)\\
	&\overset{d}{\longrightarrow}\sup_{(r_1,r_2)\in\mathcal{D}_{\iota}}\vd^\prime(\overline{\mSigma}^{1/2}\Delta\overline{\mW}, 0,r_1,r_2)\mXi^{-1}(\overline{\mSigma}^{1/2}\Delta\overline{\mW}, 0,r_1,r_2)\vd(\overline{\mSigma}^{1/2}\Delta\overline{\mW}, 0,r_1,r_2)\\
	& \hspace{0.7cm} +\sup_{(s_1,s_2)\in\mathcal{D}_{\iota}}\vd^\prime(\overline{\mSigma}^{1/2}\Delta\overline{\mW}, s_1,s_2,1)\mXi^{-1}(\overline{\mSigma}^{1/2}\Delta\overline{\mW},s_1,s_2,1)\vd(\overline{\mSigma}^{1/2}\Delta\overline{\mW}, s_1,s_2,1)=\mathcal{V}_{2k+2}.
\end{align*}
The two equalities in the above display follow by construction of the test statistic and the limiting distribution, where the pre-multiplied invertible matrices in the first argument of $\vd(\cdot,\cdot,\cdot,\cdot)$ (that is, $\overline{\mD}_n$ and $\overline{\mSigma}^{1/2}$) cancel each other out.
\end{proof}

\begin{proof}[{\textbf{Proof of Corollary~\ref{thm:CoVaR est unsupervised loc alt}:}}]
The proof follows along similar lines as that of Theorem~3.2 in \citet{ZL18}.
Consider the first part of the test statistic in \eqref{eq:Unsupervised T} evaluated at $(r_1,r_2) = (s_{1}^{\ast}, s_{2}^{\ast})$, i.e., the first two break times.
Then, the normalizing matrix $\mXi^{-1}(\widehat{\mS}_n, 0,s_{1}^{\ast}, s_{2}^{\ast})$ remains bounded in probability by similar arguments used in the proof of Corollary~\ref{thm:CoVaR est unsupervised}.
Moreover, from \eqref{eq:(L.2)},
\begin{align*}
	\vd(\widehat{\mS}_n, 0,s_{1}^{\ast}, s_{2}^{\ast}) &= \frac{s_{1}^{\ast}(s_{2}^{\ast}-s_{1}^{\ast})}{(s_{2}^{\ast})^{3/2}}\big[\widehat{\vgamma}_n(0,s_{1}^{\ast}) - \widehat{\vgamma}_n(s_{1}^{\ast},s_{2}^{\ast})\big]\\
	&= \frac{s_{1}^{\ast}(s_{2}^{\ast}-s_{1}^{\ast})}{(s_{2}^{\ast})^{3/2}}\Big\{\big[\widehat{\vgamma}_n(0,s_{1}^{\ast})-\vgamma_0\big] - \big[\widehat{\vgamma}_n(s_{1}^{\ast},s_{2}^{\ast})-\vgamma_0\big]\Big\},
\end{align*}
where $\vgamma_0=(\valpha_0^\prime, \vbeta_0^\prime)^\prime$.
The remainder of the proof now follows similarly as that of Corollary~\ref{cor:one-break}.
\end{proof}

\section{Monte Carlo Simulations}\label{sec:Main Simulation}

This section investigates the finite-sample size and power of the tests based on $\mathcal{U}_{n,\valpha}$ and $\mathcal{U}_{n,\vgamma}$.
We do not present simulation results for $\mathcal{V}_{n,\gamma}$ because of the excessive computational time this would require.
Indeed, the computation of $\mathcal{V}_{n,\gamma}$ takes several orders of magnitude longer than the computation of $\mathcal{U}_{n,\gamma}$.
All simulations are run in \texttt{R} \citep{R.4.3.3} and the estimators $\widehat{\valpha}_n(\cdot,\cdot)$ and $\widehat{\vbeta}_n(\cdot,\cdot)$ are computed via the \texttt{quantreg} package \citep{quantreg}.
We use 5,000 replications throughout.
In the following, we denote an $(m\times1)$-vector of ones by $\vone_{m}$.

We proceed as follows.
To start, Section~\ref{DGP} presents the main data-generating process (DGP).
Then, in Sections~\ref{Size} and \ref{Power} we report results on size and power, respectively.
Section~\ref{Comparison with Qu08} compares our persistence-robust quantile regression (QR) stability test with the non-robust test of \citet{Qu08}.
There, we also investigate size of our QR stability test for smaller sample sizes (as often encountered in QRs in macroeconomics) than those considered in Sections~\ref{Size} and \ref{Power}.
The final two sections extend the simulations of Sections~\ref{Size} and \ref{Power} by considering early and late breaks under the alternative (Section~\ref{Early and Late Breaks}) and by considering a larger number of predictors (Section~\ref{Additional Predictors}).

\subsection{Data-Generating Process}\label{DGP}

As the DGP we use the predictive CoVaR regression with $k=2$, $\alpha=\beta=0.9$ and 
\begin{align*}
	\begin{pmatrix}Y_t \\
	Z_t
	\end{pmatrix}
	&=\begin{pmatrix}\valpha_{0,t}^\prime \\ \vbeta_{0,t}^\prime\end{pmatrix}
	\begin{pmatrix}
	1\\ \vx_{t-1}\end{pmatrix}+
		\begin{pmatrix}
			\epsilon_t\\ \delta_t
		\end{pmatrix},\qquad t=1,\ldots,n.
\end{align*}
Here, the predictors are generated by
\begin{equation*}
	\vx_t=\vmu_x + \vxi_t,\qquad \vxi_t = \mR_n\vxi_{t-1} + \vu_t,\qquad t\in\mathbb{N},
\end{equation*}
with $\vmu_x=\vzero$ and initialization at $\vxi_0=\vzero$.
In the (I0) setting, we use $\mR_n=\mR=\mI_{k} + \mC$ for $\mC=-c\mI_{k}$ with $c$ varying in the interval $(0,\, 0.5]$.
These choices of $\mC$ ensure that the spectral radius of $\mR$ satisfies that $\rho(\mR)=1-c<1$, as required for stationarity.
Therefore, $\mR_n=r_n\mI_{k}$ with $r_n:=r:=1-c$ under (I0).
Similarly, to generate (NS) predictors, we put $r_n=1-n^{-1/2}$ in the autoregressive matrix $\mR_n=r_n\mI_{k}$, such that $\kappa=1/2$ and $\mC_n=n^{1/2}(\mR_n - \mI_k)=-\mI_k=:\mC$ is negatively stable.

The predictor innovations are generated from the VAR(1) model
\begin{equation}\label{eq:PIsim}
	\vu_t=\mPhi\vu_{t-1} + \vvarepsilon_t,
\end{equation}
and the CoVaR errors are 
\[
\begin{pmatrix}
			\epsilon_t\\ \delta_t
		\end{pmatrix}=\begin{pmatrix}
			v_{1t} - Q_{\alpha}(v_{1t})\\ v_{2t} - \CoVaR_{\beta\mid\alpha}\big((v_{2t}, v_{1t})^\prime\big)
		\end{pmatrix},
\]
where $\vv_t=(v_{1t},v_{2t})^\prime$ and $\vvarepsilon_t$ are independent, identically distributed (i.i.d.) draws from a joint normal distribution.
More specifically, 
\[
	\begin{pmatrix} \vv_t\\ \vvarepsilon_t\end{pmatrix}\overset{\text{i.i.d.}}{\sim} N(\vzero,\mPsi),\qquad\mPsi=\begin{pmatrix}\mPsi_{\vv} & \mPsi_{\vv\vvarepsilon}^\prime\\
	\mPsi_{\vv\vvarepsilon} & \mPsi_{\vvarepsilon}
		\end{pmatrix},\quad
		\begin{cases}
\mPsi_{\vv} &= (1-\rho_{\vv})\mI_{2} + \rho_{\vv} \vone_{2}\vone_{2}^\prime,\\
	\mPsi_{\vvarepsilon} &= (1-\rho_{\vvarepsilon})\mI_{k} + \rho_{\vvarepsilon} \vone_{k}\vone_{k}^\prime,\\
	\mPsi_{\vv\vvarepsilon} &= \rho_{\vv\vvarepsilon}\vone_{k}\vone_{2}^\prime,
	\end{cases}
\]
such that $\mPsi_{\vv}$ and $\mPsi_{\vvarepsilon}$ are equicorrelation matrices and $\mPsi$ is a multiple-block equicorrelation matrix.
For the parameters of the normal distribution of $(\vv_t^\prime, \vvarepsilon_t^\prime)^\prime$, we choose $\rho_{\vv}=\rho_{\vvarepsilon}=0$ and $\rho_{\vv\vvarepsilon}=-0.25$, ensuring that $\mPsi$ is positive definite \citep[Lemma~2.3~(iii)]{EK12}.
For the autoregressive matrix of the predictor innovations in \eqref{eq:PIsim} we put $\mPhi=-0.95\mI_k$.

In the spirit of \citet{Sta99}, the above setup leads to regression disturbances $(\epsilon_t,\delta_t)^\prime$ that correlate with the regressors' innovation $\vu_t$, because
\[
	\Cov\big((\epsilon_t,\delta_t)^\prime,\,\vu_t\big) = \Cov\big(\vv_t,\mPhi\vu_{t-1}+\vvarepsilon_t\big)=\Cov(\vv_t,\,\vvarepsilon_t)=\mPsi_{\vv\vvarepsilon}=-0.25\cdot\vone_{2}\vone_{2}^\prime.
\]

\subsection{Size}\label{Size}

\begin{figure}[t!]
	\centering
	\includegraphics[width=\linewidth]{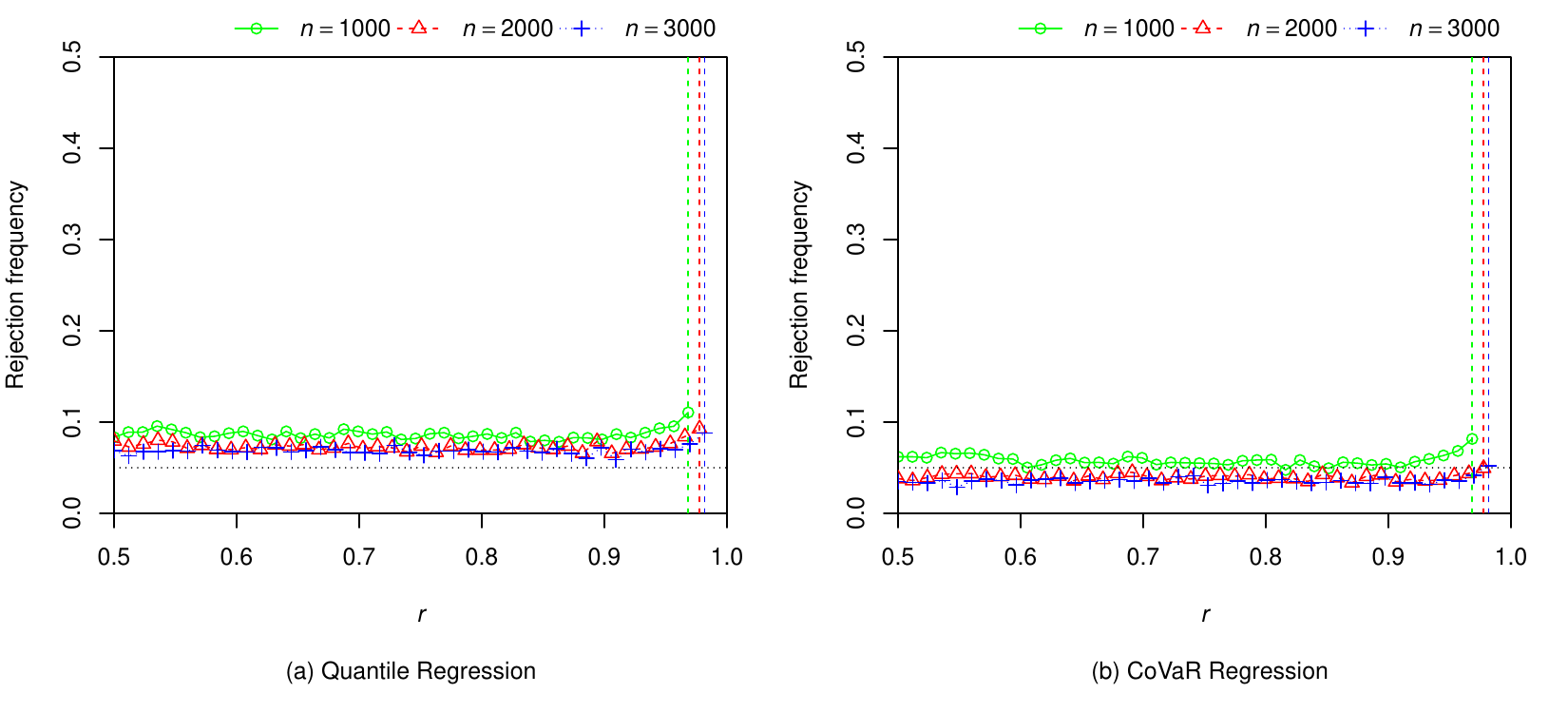}
	\caption{Panel (a): Empirical size of test based on $\mathcal{U}_{n,\valpha}$ for the quantile regression. Panel (b): Empirical size of test based on $\mathcal{U}_{n,\vgamma}$ for the CoVaR regression. 
	In both panels, size is plotted as a function of the autoregressive parameter $r$ of the (I0) predictors.	
	The dashed vertical lines correspond to the values of $r_n=1-n^{-0.5}$ in the (NS) setting.
	The dotted horizontal lines indicate the tests' nominal level of 5\%. }
	\label{fig:Size}
\end{figure}

To simulate under the no-break null hypothesis, we put $\valpha_{0,t}=\vbeta_{0,t}=(0,\vone_k^\prime)^\prime$ for all $t\geq1$.
We set $\iota=0.1$ for our test statistics $\mathcal{U}_{n,\valpha}$ and $\mathcal{U}_{n,\vgamma}$.
Figure~\ref{fig:Size} shows the empirical rejection frequencies of our tests as a function of the persistence parameter $r$ for sample sizes $n\in\{1000,\ 2000,\ 3000\}$.
The break tests for the quantile regression are somewhat oversized for the smaller sample size. 
Yet, the size distortions diminish as $n$ gets larger.
Interestingly, size is even better for the CoVaR regression although the CoVaR is a more ``extreme'' quantity and, hence, harder to estimate than the quantiles.
Even for $r$ close to unity (i.e., for highly persistent predictors) do we see little remaining distortions in large samples. 
In particular, in the (NS) setting with $r_n=1-n^{-0.5}$ (such that $\kappa=1/2$ lies in the middle of the allowed interval $(0,1)$ in Assumption~\ref{ass:N}) size distortions of our self-normalized tests remain negligible, as predicted by Corollaries~\ref{cor:SBT CoVaR} and \ref{cor:SBT} (see the dashed vertical lines in Figure~\ref{fig:Size}).

\subsection{Power}\label{Power}

\begin{figure}[t!]
	\centering
	\includegraphics[width=\linewidth]{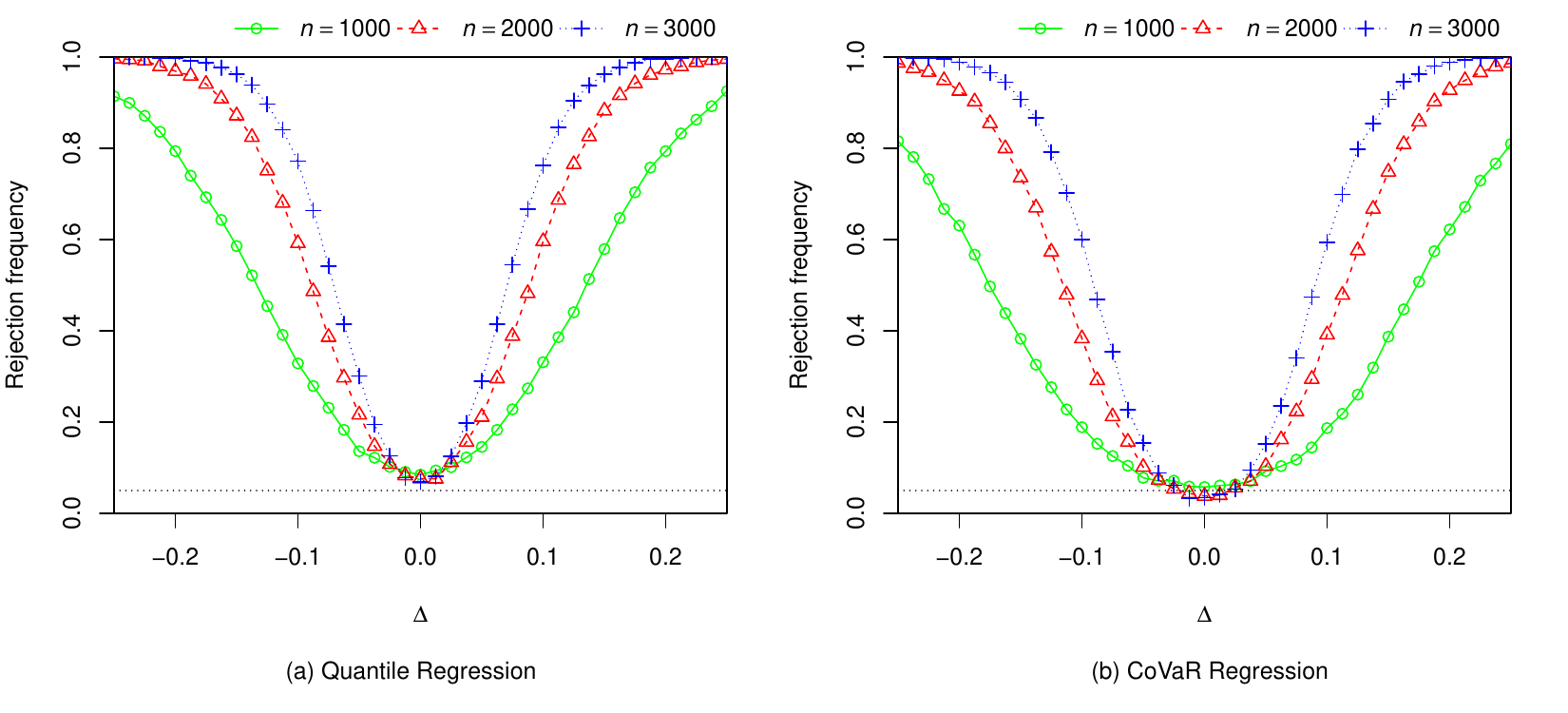}
	\caption{Panel (a): Empirical power of test based on $\mathcal{U}_{n,\valpha}$ for the quantile regression. Panel (b): Empirical power of test based on $\mathcal{U}_{n,\vgamma}$ for the CoVaR regression. 
	In both panels, power is plotted as a function of $\Delta$, i.e., the deviation from the null.
 The dotted horizontal lines indicate the tests' nominal level of 5\%. Predictors are stationary with $r=0.5$ (such that $\kappa=0$).}
	\label{fig:Power1}
\end{figure}

To study the finite-sample power of our tests, we vary the predictive content of the covariates in the DGP.
Specifically, we consider the single-break setting with $\valpha_{0,t}=\vbeta_{0,t}=(0,\vone_k^\prime)^\prime+\mDelta_t$, where
\[
	\mDelta_t=\begin{cases}
							\vzero,& t=1,\ldots, \lfloor ns^\ast\rfloor,\\
							\Delta\vone_{k+1},& t=\lfloor ns^\ast\rfloor+1,\ldots,n,\\
						\end{cases}
\]
and vary $\Delta$ in the interval $[-1/4, 1/4]$. 
We fix the timing of the break in the middle of the sample ($s^\ast=0.5$).
Power is somewhat lower for earlier and later breaks ($s^\ast\in\{0.25,\, 0.75\}$), but the results are qualitatively similar; see Section~\ref{Early and Late Breaks}.

\begin{figure}[t!]
	\centering
	\includegraphics[width=\linewidth]{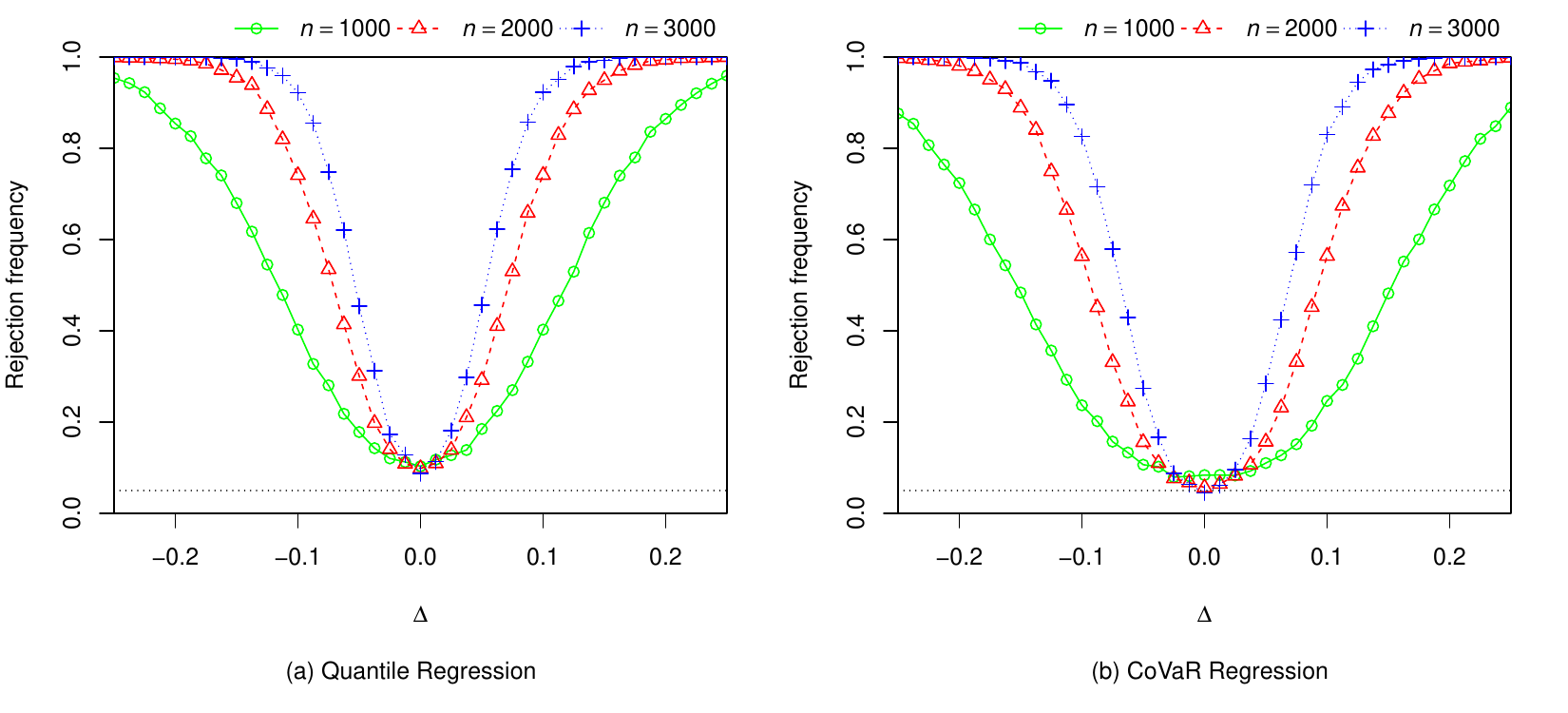}
	\caption{Panel (a): Empirical power of test based on $\mathcal{U}_{n,\valpha}$ for the quantile regression. Panel (b): Empirical power of test based on $\mathcal{U}_{n,\vgamma}$ for the CoVaR regression. 
	In both panels, power is plotted as a function of $\Delta$, i.e., the deviation from the null.
	The dotted horizontal lines indicate the tests' nominal level of 5\%. Predictors are near-stationary with $r_n=1-n^{-0.5}$ (such that $\kappa=1/2$).}
	\label{fig:Power2}
\end{figure}

Figure~\ref{fig:Power1} plots the empirical rejection frequencies as a function of $\Delta$ for $\kappa=0$ and $r=0.5$ (i.e., for stationary predictors) and Figure~\ref{fig:Power2} those for $r_n=1-n^{-\kappa}$ for $\kappa=1/2$ (i.e., for near-stationary predictors).
As predicted by the discussion at the end of Section~\ref{Power Analysis}, a comparison of Figures~\ref{fig:Power1} and \ref{fig:Power2} reveals that power is somewhat higher for near-stationary predictors.
As also expected, the rejection frequency increases in the sample size $n$, and in the distance of the alternative from the null (i.e., in $|\Delta|$).
Power is also roughly symmetric in $\Delta$.
Of course, for $\Delta=0$ the results of Figure~\ref{fig:Power1} correspond to size (already plotted in Figure~\ref{fig:Size} for $r=0.5$) and, likewise, $\Delta=0$ in Figure~\ref{fig:Power2} indicates size (already plotted in Figure~\ref{fig:Size} for $r_n=1-n^{-1/2}$).

\subsection{Comparison with \citet{Qu08} for Small and Large $n$}\label{Comparison with Qu08}

As pointed out in the main paper, for predictive quantile regressions one may use the stability test of \citet{Qu08}.
However, his test is only valid for stationary predictors.
Therefore, here we compare our $\mathcal{U}_{n,\valpha}$-based test with that of \citet{Qu08} to illustrate the advantages of the persistence-robustness.
In fact, \citet{Qu08} proposes several tests.
The one closest in spirit to ours is his $SW_{\tau}$ test, because---unlike his $SQ_{\tau}$ test---it is also based on comparing subsample estimates for the QR. 
Hence, we use \citeauthor{Qu08}'s \citeyearpar{Qu08} $SW_{\tau}$ test for the comparison with ours.\footnote{In computing the $SW_{\tau}$ test statistic, we follow the outline given in \citet[Section~6]{Qu08}.
In particular, we use the bandwidth choice from \citet{HS88}.}
In doing so, we use exactly the same simulation setups as described in Sections~\ref{Size} and \ref{Power}.

There is one exception however.
Our simulations in Sections~\ref{Size} and \ref{Power} only consider sample sizes of $n\in\{1000,\, 2000,\, 3000\}$. 
Such sample sizes are typically required to reliably estimate CoVaR regressions, where the effective sample size (due to the conditioning in the CoVaR) is reduced to $n(1-\alpha)$, which may be very small for $\alpha$ close to 1.
Yet, quantile regressions can be estimated already with smaller sample sizes.
For instance, in the QR application in Section~\ref{Quantile Predictability of the Equity Premium}, we have $n=586$.
Therefore, here we investigate size of our QR stability test for $n\in\{500,\, 750,\, 1000,\, 2000\}$.

\begin{figure}[t!]
	\centering
	\includegraphics[width=\linewidth]{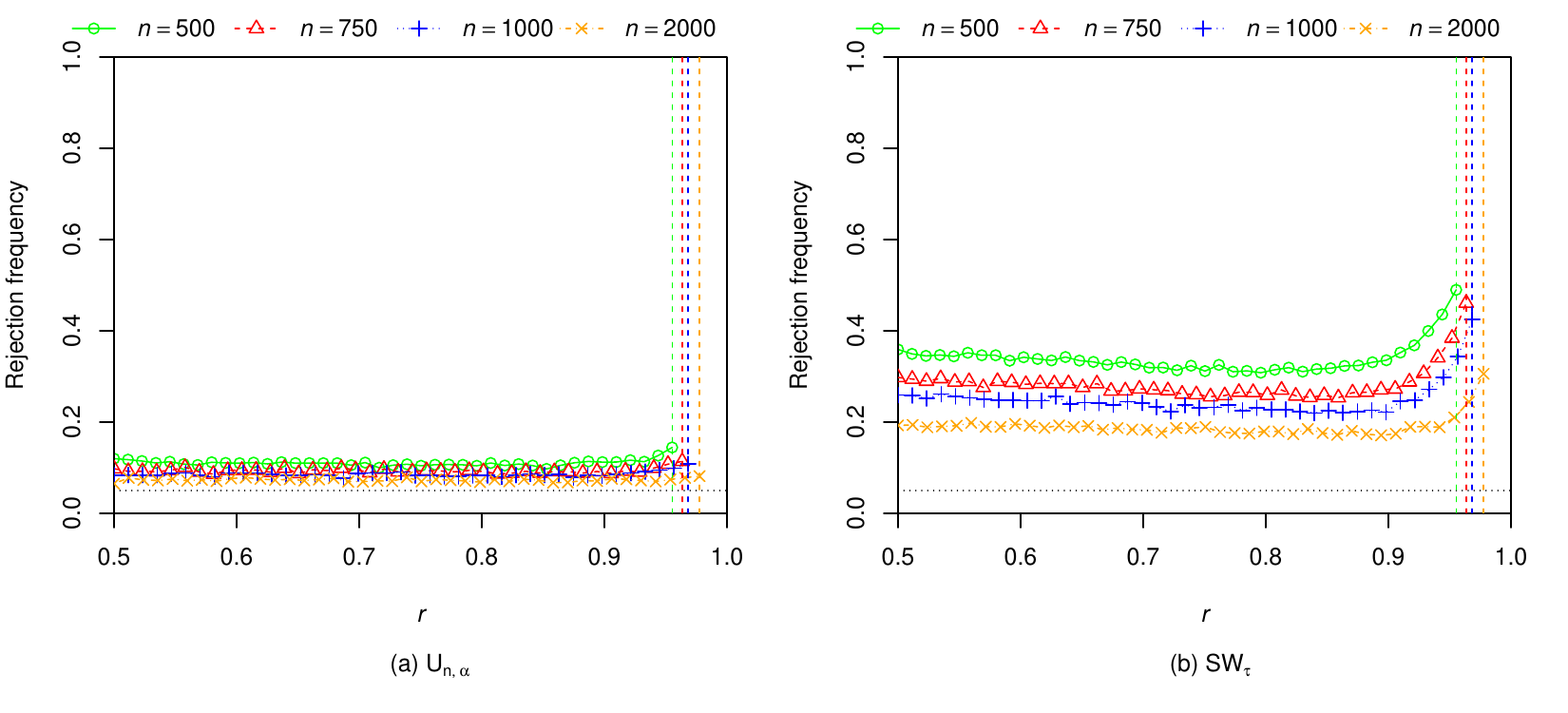}
	\caption{Panel (a): Empirical size of test based on $\mathcal{U}_{n,\valpha}$. 
	Panel (b): Empirical size of test based on $SW_{\tau}$.
	In both panels, size is plotted as a function of the autoregressive parameter $r$ of the (I0) predictors.
	The dashed vertical lines correspond to the values of $r_n=1-n^{-0.5}$ in the (NS) setting.
	The dotted horizontal lines indicate the tests' nominal level of 5\%.}
	\label{fig:SizeQu2008}
\end{figure}

Figure~\ref{fig:SizeQu2008} plots the size of the $\mathcal{U}_{n,\valpha}$ and the $SW_{\tau}$ test in panels (a) and (b), respectively. 
Our test has some liberal tendencies for these smaller sample sizes, yet still performs satisfactorily.
In particular, as the autoregressive coefficient $r_n=1-n^{-0.5}$ tends to unity in larger samples (indicated by the dashed vertical lines), the behavior of our test is almost unaffected, thus illustrating its persistence-robustness.
In contrast, the $SW_{\tau}$ test exhibits a clear upward spike in rejections for $r_n=1-n^{-0.5}$, indicating that it is not persistence-robust.
For the ``less extreme'' values of $r$, the $SW_{\tau}$ test is also markedly oversized, particularly for the small values of $n$.

Overall, the above comparison demonstrates the persistence-robustness of the self-normalized test based on $\mathcal{U}_{n,\valpha}$ and its excellent size relative to the $SW_{\tau}$ test.
Regarding the latter advantage, it is well known in the literature that SN-based tests, such as ours, typically possess very good finite-sample size \citep{SZ10,Sha15}.

\begin{figure}[t!]
	\centering
	\includegraphics[width=\linewidth]{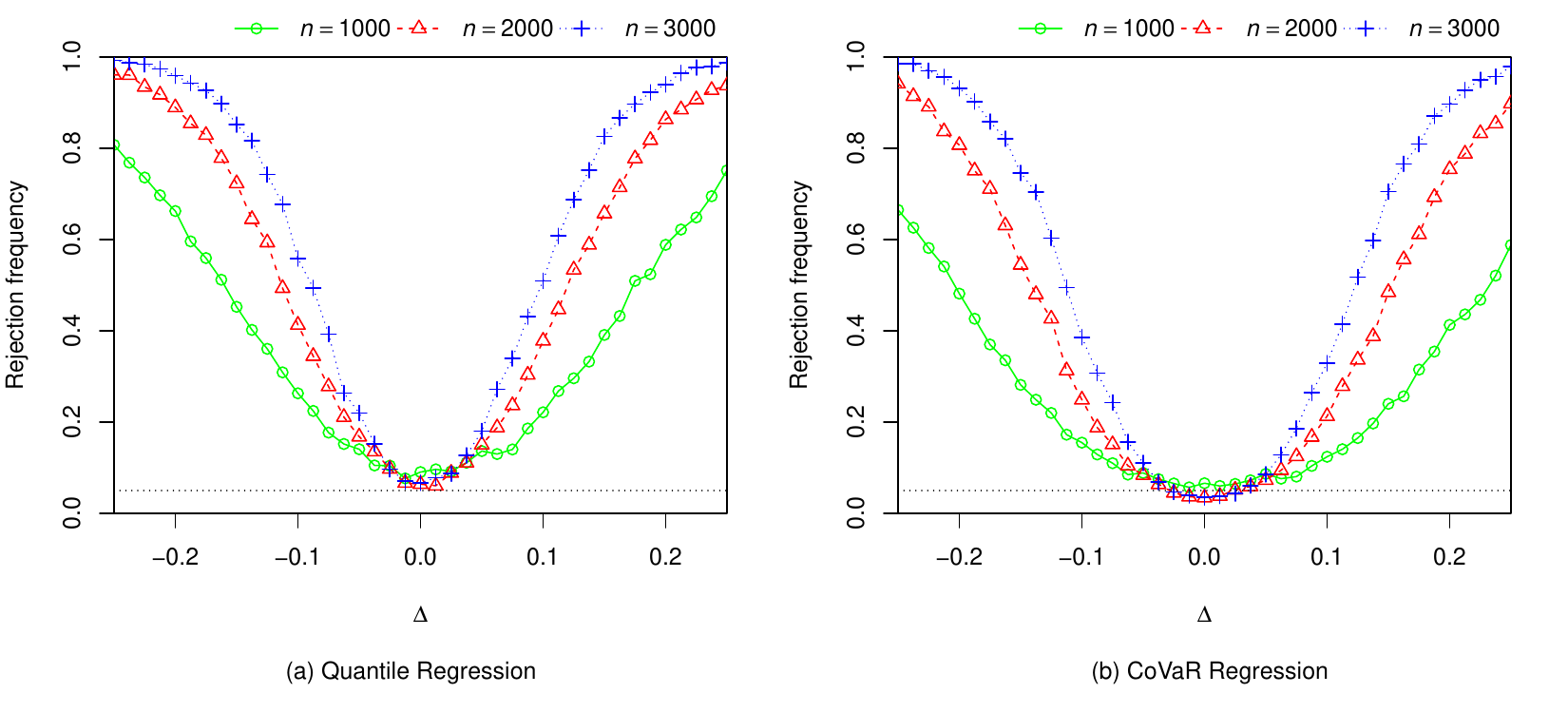}
	\caption{Panel (a): Empirical power of test based on $\mathcal{U}_{n,\valpha}$ for the quantile regression. Panel (b): Empirical power of test based on $\mathcal{U}_{n,\vgamma}$ for the CoVaR regression. 
	In both panels, power is plotted as a function of $\Delta$, i.e., the deviation from the null, and the break is located at $s^{\ast}=0.25$.
 The dotted horizontal lines indicate the tests' nominal level of 5\%. Predictors are stationary with $r=0.5$ (such that $\kappa=0$).}
	\label{fig:Power1Early}
\end{figure}

\begin{figure}[t!]
	\centering
	\includegraphics[width=\linewidth]{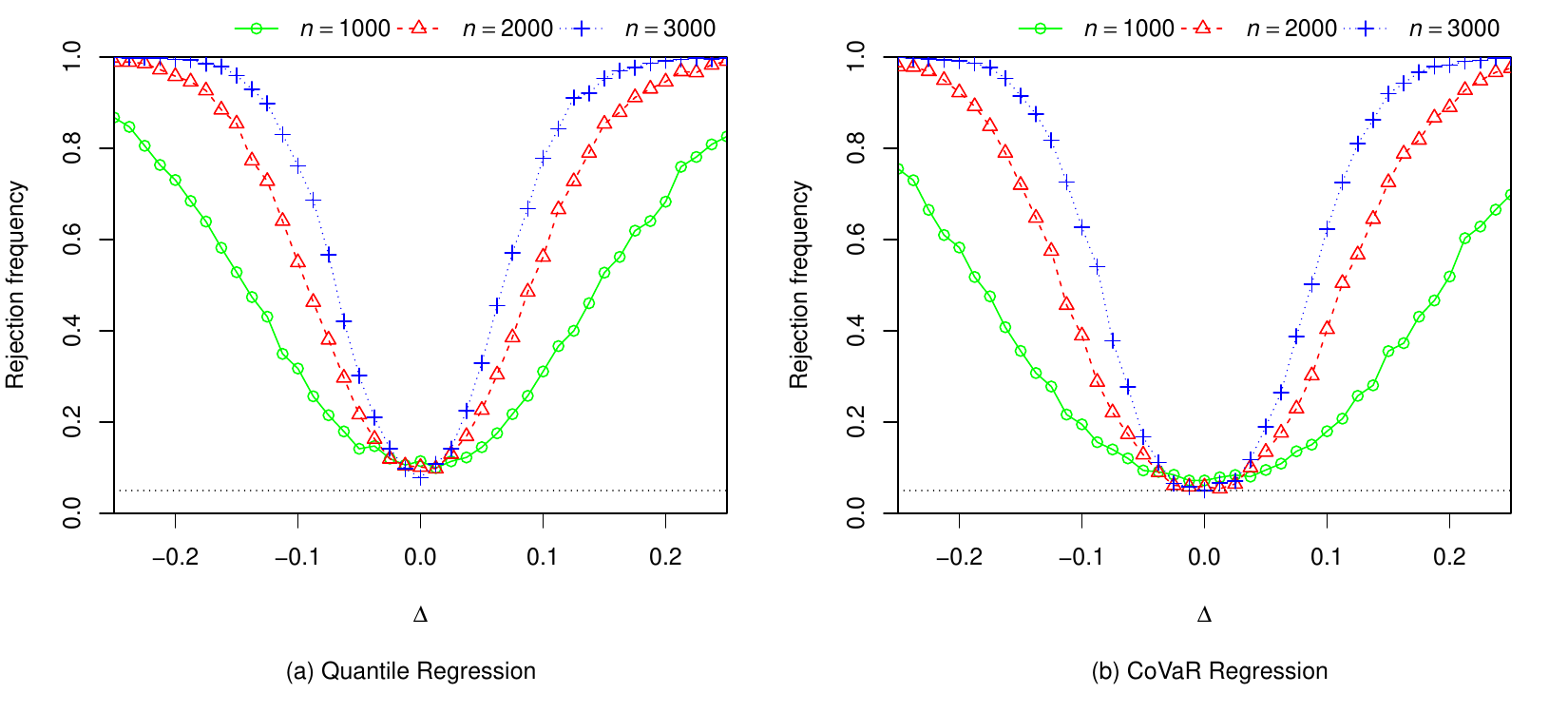}
	\caption{Panel (a): Empirical power of test based on $\mathcal{U}_{n,\valpha}$ for the quantile regression. Panel (b): Empirical power of test based on $\mathcal{U}_{n,\vgamma}$ for the CoVaR regression. 
	In both panels, power is plotted as a function of $\Delta$, i.e., the deviation from the null, and the break is located at $s^{\ast}=0.25$.
	The dotted horizontal lines indicate the tests' nominal level of 5\%. Predictors are near-stationary with $r_n=1-n^{-0.5}$ (such that $\kappa=1/2$).}
	\label{fig:Power2Early}
\end{figure}

\subsection{Early and Late Breaks}\label{Early and Late Breaks}

The simulations in Section~\ref{Power} investigate power when the break occurs in the middle of the sample, i.e., for $s^\ast=0.5$.
Here, we supplement these results by considering early breaks ($s^\ast=0.25$) and late breaks ($s^\ast=0.75$), but otherwise keep the setup unchanged.

Figures~\ref{fig:Power1Early} and \ref{fig:Power2Early} are the analogs of Figures~\ref{fig:Power1} and \ref{fig:Power2} for early breaks, i.e., for $s^\ast=0.25$.
They plot power as a function of the magnitude of the break $\Delta$.
Of course, for $\Delta=0$ the results correspond to size and, hence, are unchanged relative to Figures~\ref{fig:Power1}--\ref{fig:Power2}.
However, for $\Delta\neq0$, we see that power is reduced when the break occurs earlier.

\begin{figure}[t!]
	\centering
	\includegraphics[width=\linewidth]{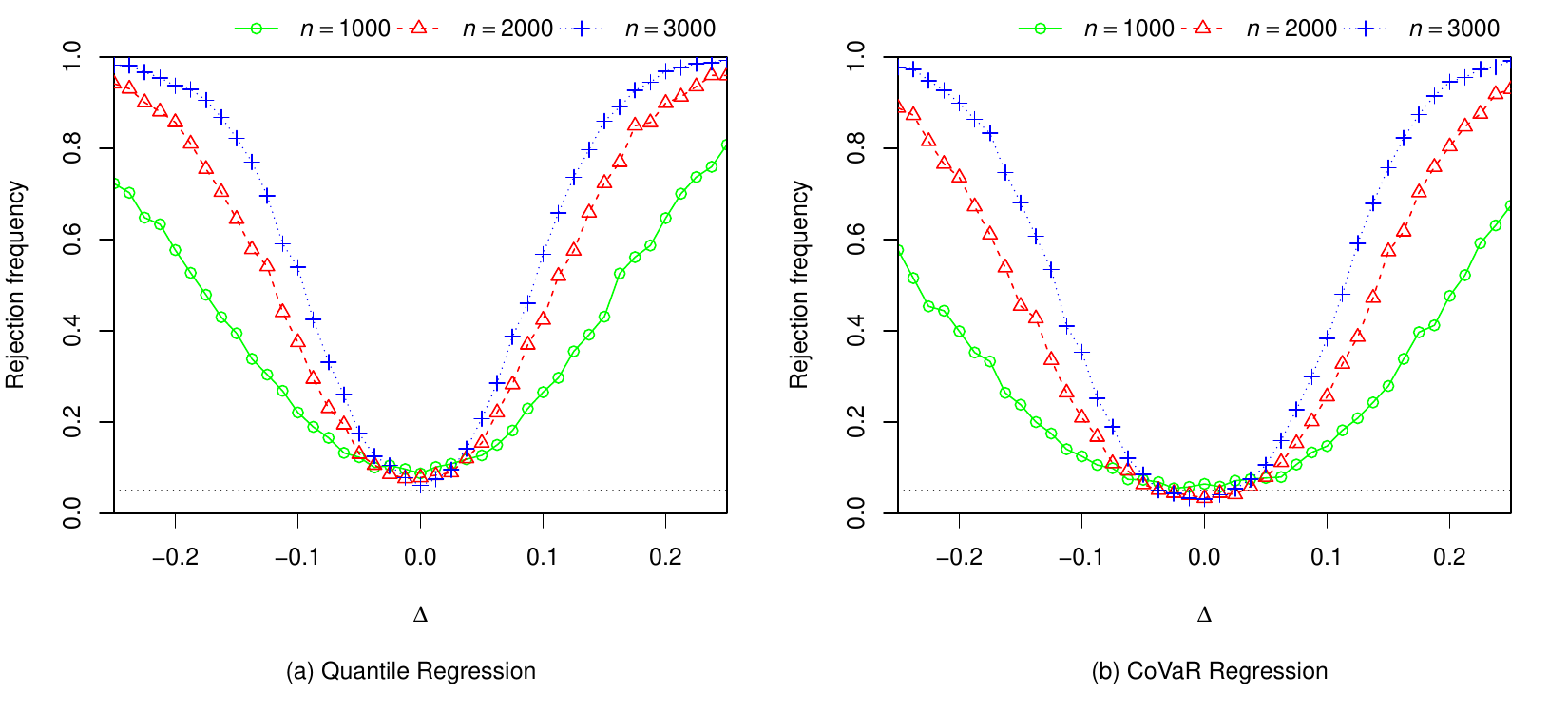}
	\caption{Panel (a): Empirical power of test based on $\mathcal{U}_{n,\valpha}$ for the quantile regression. Panel (b): Empirical power of test based on $\mathcal{U}_{n,\vgamma}$ for the CoVaR regression. 
	In both panels, power is plotted as a function of $\Delta$, i.e., the deviation from the null, and the break is located at $s^{\ast}=0.75$.
 The dotted horizontal lines indicate the tests' nominal level of 5\%. Predictors are stationary with $r=0.5$ (such that $\kappa=0$).}
	\label{fig:Power1Late}
\end{figure}

\begin{figure}[t!]
	\centering
	\includegraphics[width=\linewidth]{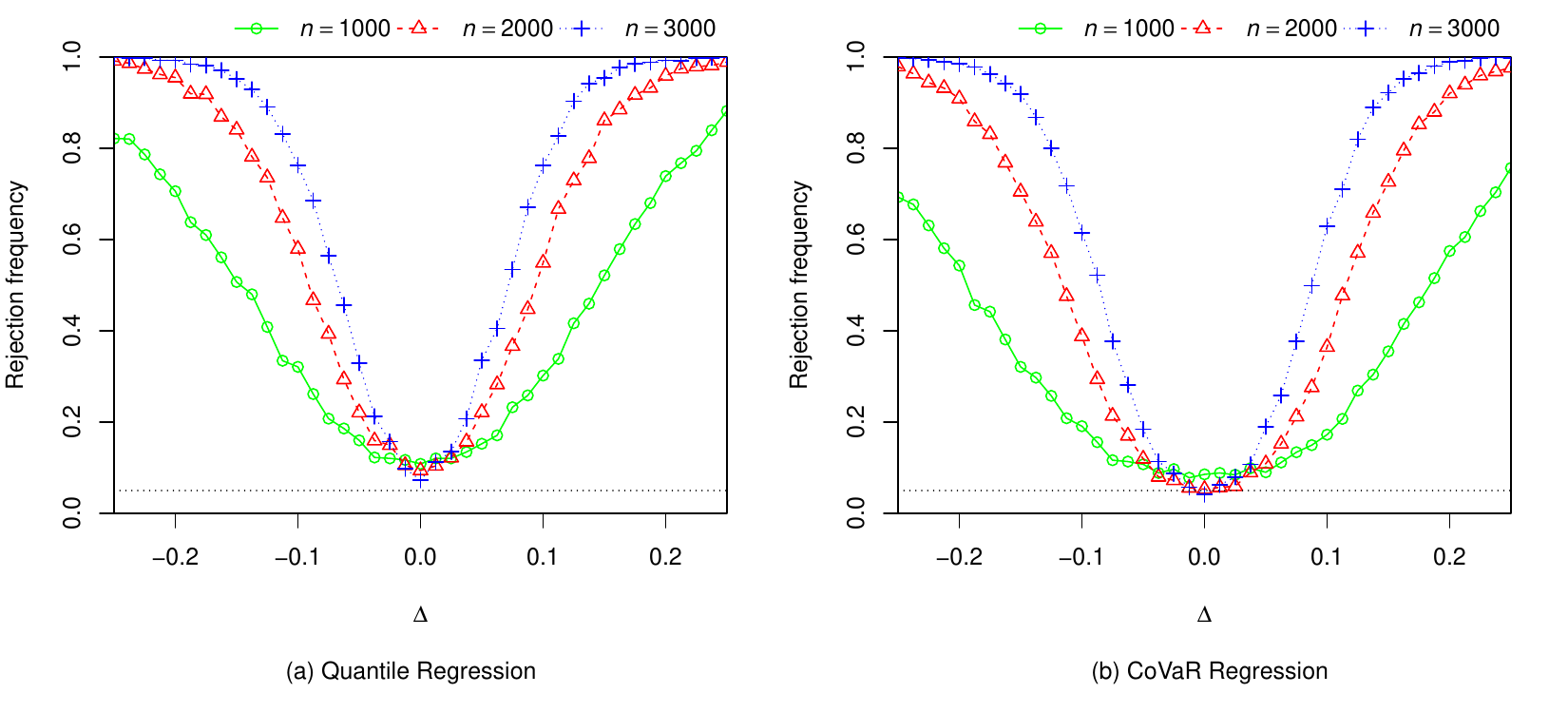}
	\caption{Panel (a): Empirical power of test based on $\mathcal{U}_{n,\valpha}$ for the quantile regression. Panel (b): Empirical power of test based on $\mathcal{U}_{n,\vgamma}$ for the CoVaR regression. 
	In both panels, power is plotted as a function of $\Delta$, i.e., the deviation from the null, and the break is located at $s^{\ast}=0.75$.
	The dotted horizontal lines indicate the tests' nominal level of 5\%. Predictors are near-stationary with $r_n=1-n^{-0.5}$ (such that $\kappa=1/2$).}
	\label{fig:Power2Late}
\end{figure}

Similarly, Figures~\ref{fig:Power1Late} and \ref{fig:Power2Late} are the equivalents of Figures~\ref{fig:Power1} and \ref{fig:Power2} for late breaks, i.e., for $s^{\ast}=0.75$.
Once again, power is lower compared with the case when the break occurs in the middle of the sample.
By symmetry of our test statistics it is not surprising to find that the detection probability for late and early breaks is roughly equal.

\subsection{Additional Predictors}\label{Additional Predictors}

The simulations in Sections~\ref{Size} and \ref{Power} only consider $k=2$ stochastic predictors.
Here, we examine the effect on our results of increasing this to $k=4$, such that the number of (stochastic) predictors is doubled.
Apart from the change in $k$, we keep all the parameters as in Sections~\ref{Size} and \ref{Power}.
In particular, under the alternative, we assume the break to occur in the middle of the sample ($s^\ast=0.5$), unlike in the previous section.

\begin{figure}[t!]
	\centering
	\includegraphics[width=\linewidth]{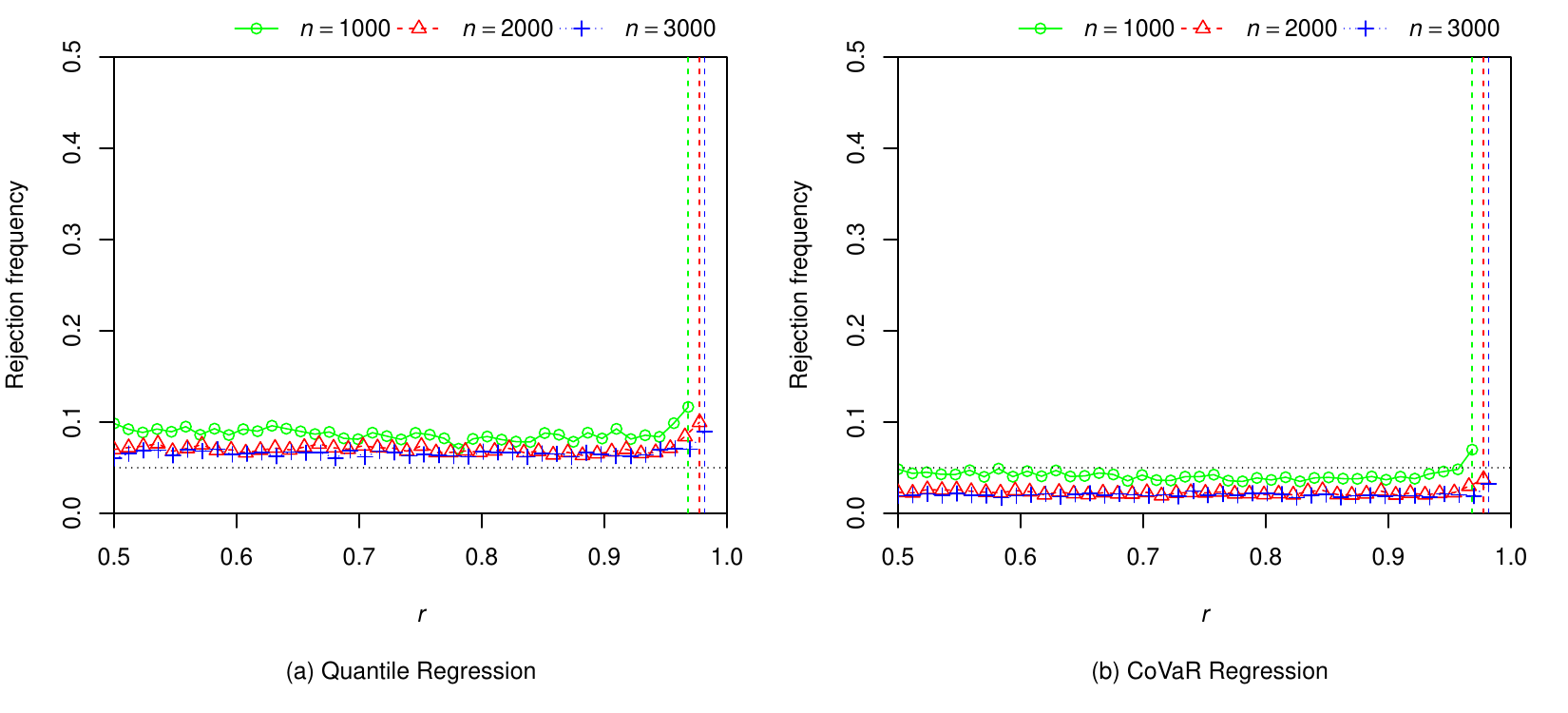}
	\caption{Panel (a): Empirical size of test based on $\mathcal{U}_{n,\valpha}$ for the quantile regression. Panel (b): Empirical size of test based on $\mathcal{U}_{n,\vgamma}$ for the CoVaR regression. 
	In both panels, size is plotted as a function of the autoregressive parameter $r$ of the (I0) predictors.	
	The dashed vertical lines correspond to the values of $r_n=1-n^{-0.5}$ in the (NS) setting.
	The dotted horizontal lines indicate the tests' nominal level of 5\%. }
	\label{fig:Sizek4}
\end{figure}

For $k=4$, Figures~\ref{fig:Sizek4}--\ref{fig:Power2k4} are the analogs of Figures~\ref{fig:Size}--\ref{fig:Power2} for $k=2$.
Comparing Figure~\ref{fig:Sizek4} with Figure~\ref{fig:Size}, we find that for $k=4$ the QR stability test is slightly more liberal, while the CoVaR test is somewhat more conservative.
There are also only minor differences between Figures~\ref{fig:Power1k4}--\ref{fig:Power2k4} and Figures~\ref{fig:Power1}--\ref{fig:Power2}.
Overall, the differences between the plots appear modest, such that the effect on our test of including more predictors are very minor.

\begin{figure}[t!]
	\centering
	\includegraphics[width=\linewidth]{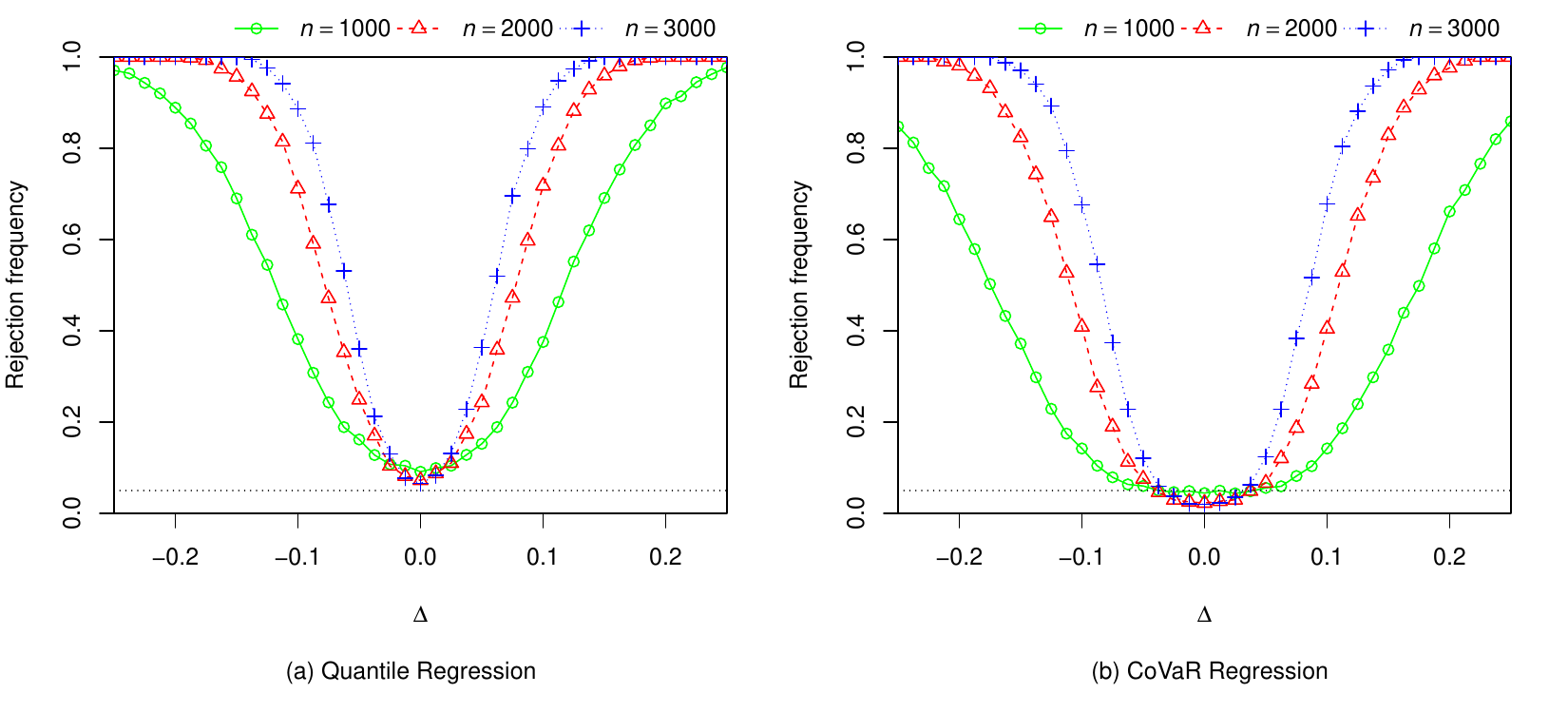}
	\caption{Panel (a): Empirical power of test based on $\mathcal{U}_{n,\valpha}$ for the quantile regression. Panel (b): Empirical power of test based on $\mathcal{U}_{n,\vgamma}$ for the CoVaR regression. 
	In both panels, power is plotted as a function of $\Delta$, i.e., the deviation from the null.
 The dotted horizontal lines indicate the tests' nominal level of 5\%. Predictors are stationary with $r=0.5$ (such that $\kappa=0$).}
	\label{fig:Power1k4}
\end{figure}

\begin{figure}[t!]
	\centering
	\includegraphics[width=\linewidth]{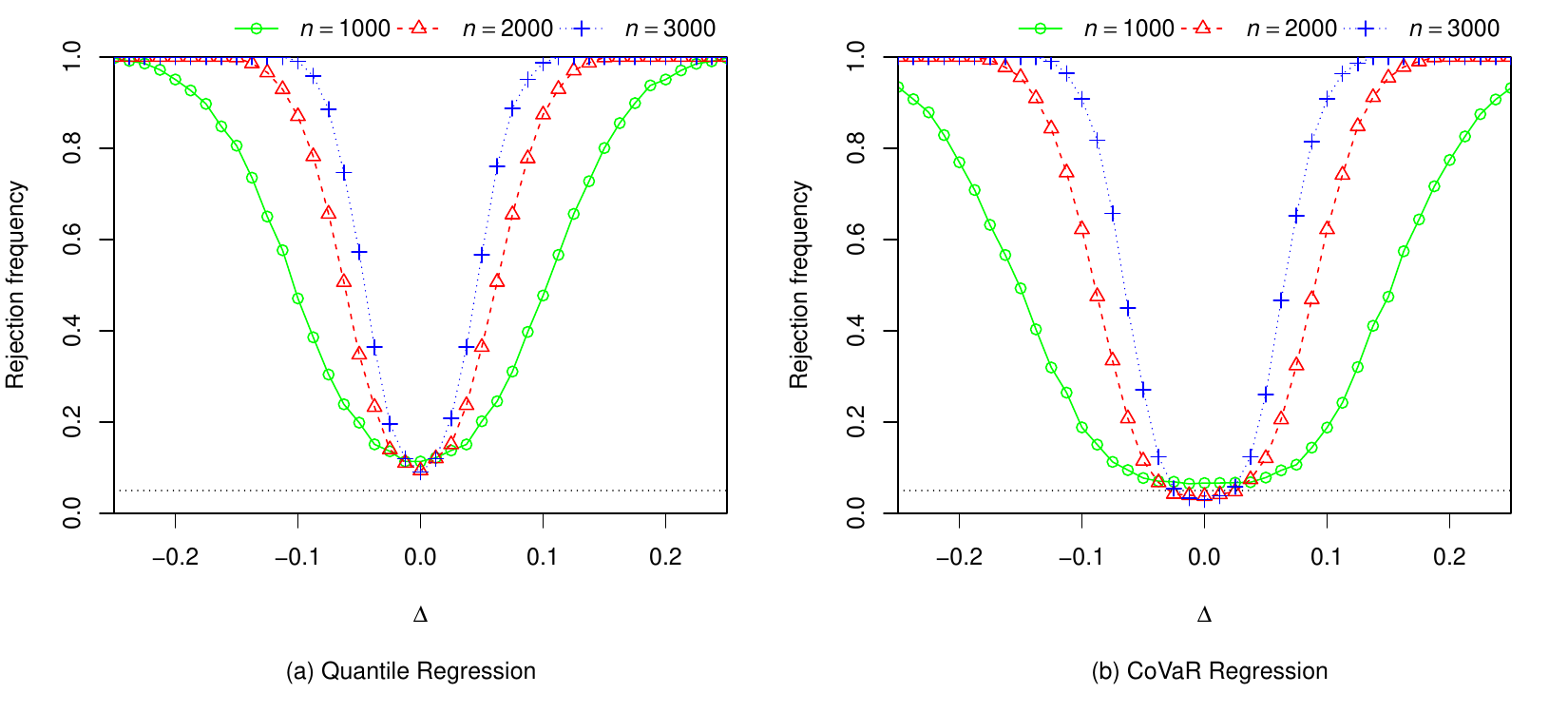}
	\caption{Panel (a): Empirical power of test based on $\mathcal{U}_{n,\valpha}$ for the quantile regression. Panel (b): Empirical power of test based on $\mathcal{U}_{n,\vgamma}$ for the CoVaR regression. 
	In both panels, power is plotted as a function of $\Delta$, i.e., the deviation from the null.
	The dotted horizontal lines indicate the tests' nominal level of 5\%. Predictors are near-stationary with $r_n=1-n^{-0.5}$ (such that $\kappa=1/2$).}
	\label{fig:Power2k4}
\end{figure}

\section{Quantile Predictability of the Equity Premium}\label{Quantile Predictability of the Equity Premium}

The equity premium, defined as the excess return of equities over riskless bonds, is on average around 6--9\% per year in developed countries \citep[Table~3]{MP08}.
Much work in recent decades has investigated whether the equity premium can be forecasted by a wide range of predictors \citep{WG08}.
The growing body of evidence, summarized by \citet{GWZ24}, suggests that the historical mean as a forecast of the average equity premium remains a hard-to-beat benchmark, implying that few---if any---variables truly contain predictive content for the mean of the equity premium going forward.

However, non-predictability of the \textit{mean} does not preclude predictability of equity premium \textit{quantiles} to exist. 
In fact, \citet{Lee16} finds evidence for quantile predictability of the US equity premium for a large set of predictors.
Generally speaking, he finds that predictability of the median is weakest, with stronger degrees of predictability emerging in the tails.
Subsequently, many authors have extended inference methods for predictive quantile regressions in several directions \citep{FL19,cai2022new,FLS23,Lea24a,MSK24}.
However, all the aforementioned authors assume structural stability of the forecasting models.

Therefore, here we investigate the stability of predictive quantile models for the equity premium.
Indeed, there may be a number of reasons why equity premium predictability may be liable to change over time. 
\citet{PT02} mention speculative bubbles, business cycles, rare disasters, time-varying risk aversion and changes in monetary policy as possible causes. 
For instance, monetary policy changes could shift the fundamental value of stocks due to its influence on economic growth.
Accordingly, \citet{PT06} and \citet{Tim08} find that return prediction models may work sometimes, but---more often than not---fail.
More recently, \citet{FST23} have explained this finding of ``pockets of predictability'' by a sticky expectations model with slow updating of investor beliefs about the long-lived component of future cash flows.

To formally study the stability of predictive QRs for the equity premium, we use our $\mathcal{U}_{n,\valpha}$ test from Corollary~\ref{cor:SBT}.
We consider the monthly data analyzed by \citet{WG08}, which (updated to 2023) is available from \href{https://sites.google.com/view/agoyal145}{sites.google.com/view/agoyal145}.
Specifically, we focus on the post-oil crisis period from 1975--2023.
The equity premium is defined as 
\[
	Y_t=\log\bigg(\frac{P_t + D_t}{P_{t-1}}\bigg)-\log\big(1+R^{f}_{t}\big),
\]
where $P_{t}$ denotes the index value of the S\&P~500, $D_t$ the dividends paid out by all S\&P~500 constituents in month $t$, and $R_{t}^{f}$ the Treasury-bill rate.

\begin{table}[t!]
	\centering
		\begin{tabular}{lccccccc}
			\toprule
	& \textit{dp} & \textit{dy} & \textit{ep} & \textit{bm} & \textit{tbl} & \textit{lty}  \\
	\cline{2-7} \noalign{\vspace{0.5ex}}
	AR(1)	  & 0.993 & 0.993 & 0.988 & 0.993 & 0.991 &  0.994 \\        
	KPSS	  & $2.22^{\ast\ast\ast}$ & $2.21^{\ast\ast\ast}$ & $1.71^{\ast\ast\ast}$ & $2.08^{\ast\ast\ast}$ & $2.22^{\ast\ast\ast}$ & $2.74^{\ast\ast\ast}$  \\
	ADF--GLS& $-0.065$ & $-0.045$ & $-1.685^{\ast}$ & $0.262$ & $-1.542$ & $-0.828$    \\
	\midrule\midrule
		& \textit{dfy} & \textit{dfr} & \textit{de} & \textit{ntis} & \textit{infl} & \textit{svar} &\\
	\cline{2-7} \noalign{\vspace{0.5ex}}
	AR(1)	  & 0.959 & 0.959 & 0.985 & 0.982 & 0.608 & 0.384 & \\
	KPSS	  & $0.52^{\ast\ast}$ & 0.07 & $0.13$ & $0.80^{\ast\ast\ast}$ & $1.17^{\ast\ast\ast}$ & $0.24$ & \\
	ADF--GLS& $-1.685^{\ast}$ & $-1.861^{\ast}$ & $-5.499^{\ast\ast\ast}$ & $-3.205^{\ast\ast\ast}$ & $-7.174^{\ast\ast\ast}$ & $-7.683^{\ast\ast\ast}$ & \\	
			\bottomrule
		\end{tabular}
	\caption{AR(1) coefficient estimates, KPSS test statistics, and ADF--GLS test statistics. For the KPSS test and the ADF--GLS test, significances at the 10\%, 5\% and 1\% level are indicated by $^\ast$, $^{\ast\ast}$ and $^{\ast\ast\ast}$, respectively.}
\label{tab:pred pers}
\end{table}

As predictors $\vx_{t-1}$, we use the dividend--price ratio (\textit{dp}), dividend yields (\textit{dy}), earnings--price ratio (\textit{ep}), book-to-market ratio (\textit{bm}), treasury-bill rate (\textit{tbl}), long-term government bond yield (\textit{lty}), default yield spread (\textit{dfy}), default return spread (\textit{dfr}), dividend-payout ratio (\textit{de}), net equity expansion (\textit{ntis}), inflation (\textit{infl}) and stock variance (\textit{svar}).
These variables are described in more detail by \citet[pp.~1457--1459]{WG08}.

As in \citet{CWW15} and in line with Assumption~\ref{ass:N}, we report AR(1) coefficient estimates for all predictors in Table~\ref{tab:pred pers}.
Except for inflation and stock market variance, the predictors all display a high level of persistence, with point estimates ranging from 0.959 to 0.994.
The null of stationarity can even be rejected for these series (with the exception of \textit{dfr} and \textit{de}), as indicated by the KPSS test of \citet{Kea92} also shown in Table~\ref{tab:pred pers}.
To shed more light on the serial dependence properties of the predictors, we also run the ADF--GLS test of \citet{ERS96} to test the null of a unit root.
While the evidence of the KPSS and the ADF--GLS test often points in the same direction, the tests produce conflicting results for four variables (\textit{ep}, \textit{dfy}, \textit{ntis}, \textit{infl}), where in each case the KPSS test rejects the null of stationarity and the ADF--GLS test rejects the null of non-stationarity.

Overall, we find the predictors employed by \citet{WG08} to exhibit quite different degrees of persistence where sometimes the exact degree is difficult to determine.
Therefore, applying existent structural break tests for QRs which are only valid for stationary predictors---such as those of \citet{Qu08} and \citet{OQ11}---may paint a misleading picture.
In contrast, our persistence-robust tests do provide more solid evidence.

\begin{table}[t!]
		\centering
		\begin{tabular}{lcd{3.4}d{4.4}d{4.4}}
			\toprule
		Classification &Predictor $\vx_{t-1}$ & \multicolumn{3}{c}{$\alpha$} 	\\
	\cline{3-5} \noalign{\vspace{0.5ex}}
		&& 0.1  & 0.5 &  0.9 \\
\midrule
Valuation ratios		&	\textit{dp}		&  312.5^{\ast\ast\ast}  &  675.8^{\ast\ast\ast} &   141.5^{\ast\ast} 	\\
										&	\textit{dy}		&  327.4^{\ast\ast\ast}  &   599.5^{\ast\ast\ast} &   182.8^{\ast\ast\ast} 	\\
										&	\textit{ep}		&  206.5^{\ast\ast\ast}  &   146.0^{\ast\ast} &    43.0 	\\
										&	\textit{bm}		&  1107.8^{\ast\ast\ast}  &  1956.1^{\ast\ast\ast} &  1853.6^{\ast\ast\ast} 	\\
										\midrule
Bond yield measures	&	\textit{tbl}	&  37.8   &    40.8 &    56.6 	\\
										&	\textit{lty}	&  476.7^{\ast\ast\ast}  &   444.8^{\ast\ast\ast}  &    23.0 	\\
										&	\textit{dfy}	&   13.5  &    86.8^{\ast} &    38.0 	\\
										&	\textit{dfr}	&   6.9  &   42.6  &  95.6^{\ast\ast} 	\\
										\midrule
Corporate finance 	&	\textit{de}		&   39.5  &    38.1 &    24.7 	\\
										&	\textit{ntis}	&   13.2  &   152.1^{\ast\ast\ast} &   16.9 	\\
										\midrule
Macro variables			&	\textit{infl}	&   26.9  &    33.9 &   25.0 	\\
										\midrule                        
Equity risk					&	\textit{svar}	&  138.9^{\ast\ast}  &   223.7^{\ast\ast\ast} &   14.0 	\\		
			\bottomrule
		\end{tabular}
\caption{Values of test statistic $\mathcal{U}_{n,\valpha}$ for predictive QR with probability level $\alpha$. Significances at the 10\%, 5\% and 1\% level are indicated by $^\ast$, $^{\ast\ast}$ and $^{\ast\ast\ast}$, respectively.}
\label{tab:EP pred}
\end{table}

We apply our structural break test for the linear predictive QR
\[
	Y_t=\alpha_0+\alpha_1 \vx_{t-1}+\epsilon_t,\qquad Q_{\alpha}(\epsilon_t\mid\mathcal{F}_{t-1}),
\]
where $\vx_{t-1}$ corresponds to a single one of the predictors mentioned above.
The use of such simple linear models with one predictor is standard in the literature on equity premium predictability \citep{WG08,Lee16,GWZ24}.
Table~\ref{tab:EP pred} shows the test results.
Roughly speaking, the forecasting models using valuation ratios and equity risk variables seem to be most prone to structural change. 
In contrast, bond yield measures, corporate finance variables and macro variables display much less variation in their predictive content for the equity premium. 

Our evidence that valuation ratios have time-varying forecasting power for the equity premium is consistent with the vast evidence in \textit{mean} (as opposed to \textit{quantile}) regressions \citep[e.g.,][]{HMN11,CWW15}.
Similarly, our findings on bond yield measures and corporate finance variables support the results of \citet{Gea18} obtained for mean regressions.
They detect no instabilities in the predictive power of dividend earnings (\textit{de}) and some bond yield measures.

We now compare the results of our persistence-robust test with the $SW_{\tau}$ test of \citet{Qu08}, which---among the tests proposed by \citet{Qu08}---is closest in spirit to our test.
Of course, the $SW_{\tau}$ test is only valid for stationary covariates; see \citet[Proposition~1]{Qu08} for details.
Therefore, it is particularly instructive to compare the results of the $SW_{\tau}$ test, displayed in Table~\ref{tab:EP pred Qu08}, with those of our robust test in Table~\ref{tab:EP pred}.

Overall, we find that the $SW_{\tau}$ test rejects the null of a stable predictive relationship much more frequently.
Broadly speaking, this is particularly true for the corporate finance variables and the bond yield measures, which are highly persistent as seen in Table~\ref{tab:pred pers}.
This is as expected, because the simulations in Section~\ref{Comparison with Qu08} show that the $SW_{\tau}$ test is particularly liberal for highly persistent predictors.
In contrast, for the less persistent variables \textit{infl} and \textit{svar}, the evidence of both tests points in a similar direction.
These results were to be expected, because both the $SW_{\tau}$ test and our test are valid for stationary predictors, yet when it comes to highly persistent predictors only our proposal guarantees valid inference.

\begin{table}[t!]
		\centering
		\begin{tabular}{lcd{3.4}d{4.4}d{4.4}}
			\toprule
		Classification &Predictor $\vx_{t-1}$ & \multicolumn{3}{c}{$\alpha$} 	\\
	\cline{3-5} \noalign{\vspace{0.5ex}}
		&& 0.1  & 0.5 &  0.9 \\
\midrule
Valuation ratios		&	\textit{dp}		&    85.5^{\ast\ast\ast}  &   171.9^{\ast\ast\ast}  &    71.4^{\ast\ast\ast} \\
										&	\textit{dy}		&    84.2^{\ast\ast\ast}  &   159.2^{\ast\ast\ast}  &    72.0^{\ast\ast\ast} \\
										&	\textit{ep}		&    43.2^{\ast\ast\ast}  &   115.9^{\ast\ast\ast}  &    21.9^{\ast\ast\ast} \\
										&	\textit{bm}		&   645.8^{\ast\ast\ast}  &   682.6^{\ast\ast\ast}  &   240.1^{\ast\ast\ast} \\
										\midrule           
Bond yield measures	&	\textit{tbl}	&    43.8^{\ast\ast\ast}  &    34.9^{\ast\ast\ast}  &    37.1^{\ast\ast\ast} \\
										&	\textit{lty}	&    89.9^{\ast\ast\ast}  &   146.3^{\ast\ast\ast}  &    30.7^{\ast\ast\ast} \\
										&	\textit{dfy}	&    12.5^{\ast\ast}  &    12.7^{\ast\ast}  &    19.5^{\ast\ast\ast} \\
										&	\textit{dfr}	&     5.2  &    13.4^{\ast\ast}  &    51.0^{\ast\ast\ast} \\
										\midrule           
Corporate finance 	&	\textit{de}		&     9.1  &    61.5^{\ast\ast\ast}  &   131.6^{\ast\ast\ast} \\
										&	\textit{ntis}	&     9.3  &    68.4^{\ast\ast\ast}  &    17.5^{\ast\ast\ast} \\
										\midrule           
Macro variables			&	\textit{infl}	&     8.6  &    29.5^{\ast\ast\ast}  &     5.9 \\
										\midrule           
Equity risk					&	\textit{svar}	&    15.8^{\ast\ast\ast}  &    13.7^{\ast\ast}  &    15.3^{\ast\ast} \\
			\bottomrule
		\end{tabular}
\caption{Values of test statistic $SW_{\tau}$ for predictive QR with probability level $\alpha$. Significances at the 10\%, 5\% and 1\% level are indicated by $^\ast$, $^{\ast\ast}$ and $^{\ast\ast\ast}$, respectively.}
\label{tab:EP pred Qu08}
\end{table}

Figure~\ref{fig:CPT QR} illustrates our findings for the stock return predictor with perhaps the longest history, namely the dividend--price ratio $\vx_t=\divp_t=\log\big(D_t/P_t\big)$ \citep{WG08}.
We focus here on the median of the equity premium, i.e., $\alpha=0.5$.
For this case, the top panel of Figure~\ref{fig:CPT QR} plots the function $s\mapsto s^2(1-s)^2\big[\widehat{\valpha}_n(0,s) - \widehat{\valpha}_n(s,1)\big]^\prime \bm{\mathcal{N}}_{n,\valpha}^{-1}(s)\big[\widehat{\valpha}_n(0,s) - \widehat{\valpha}_n(s,1)\big]$, which forms the basis of the test statistic $\mathcal{U}_{n,\valpha}$.
The dashed line indicates the 5\%-critical value.
We see that particularly large differences occur at the beginning and the end of the sample, where both peaks extend well above the critical value.
As a robustness check, we also apply our $\mathcal{V}_{n,\vgamma}$ test from Section~\ref{Multiple Breaks}, which is designed to handle multiple possible breaks.
This test also rejects the null of constancy at the 1\%-level, leaving little doubt about instabilities in the predictive content of dividend--prices for the equity premium median.

\begin{figure}[t!]
	\centering
	\includegraphics[width=\linewidth]{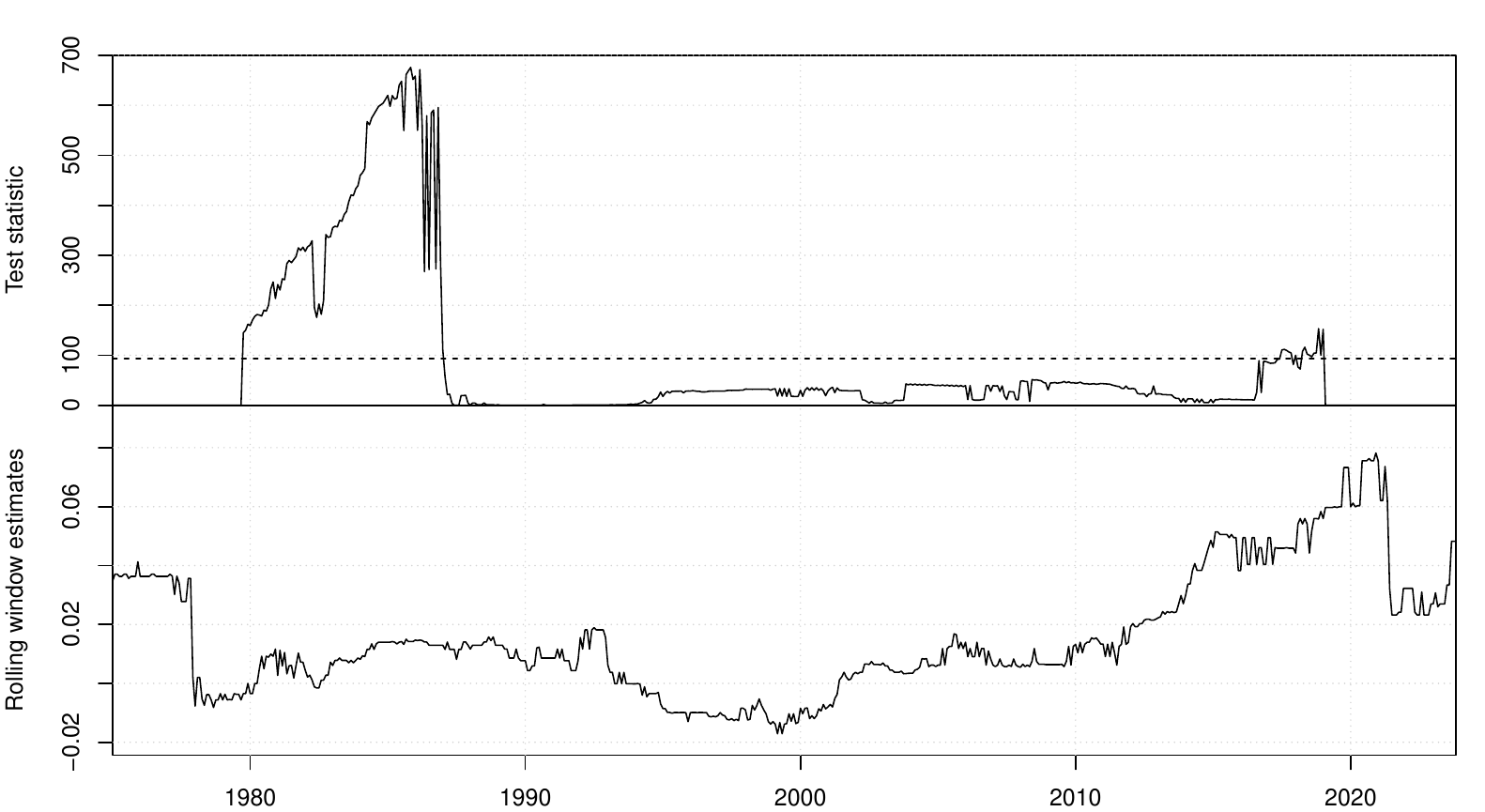}
	\caption{Top panel: Plot of the function $s\mapsto s^2(1-s)^2\big[\widehat{\valpha}_n(0,s) - \widehat{\valpha}_n(s,1)\big]^\prime \bm{\mathcal{N}}_{n,\valpha}^{-1}(s)\big[\widehat{\valpha}_n(0,s) - \widehat{\valpha}_n(s,1)\big]$ for $\alpha=0.5$. The 5\%-critical value is indicated by the dashed horizontal line. Bottom panel: Rolling window estimates of the slope coefficient of the dividend--price ratio in the linear predictive QR. The rolling window estimates are based on 240 months of data.}
	\label{fig:CPT QR}
\end{figure}

More informally, the bottom panel of Figure~\ref{fig:CPT QR} shows rolling window estimates of the coefficient of $\divp_{t-1}$ in the predictive quantile regression $Y_t=\alpha_0 + \alpha_1\divp_{t-1} + \epsilon_t$ (based on a window of 240 months).
From the late 1970s to the early 2010s, estimates of $\alpha_1$ hover around $\pm0$.
In particular, dividend--prices were neither able to predict the bust of the dot-com bubble in 2000 nor the ``lost decade'' of the 2000s for equities.
Only after 2010, when prices were still recovering from the global financial crisis and, hence, dividend--prices were high, did the predictive signal of $\divp_{t-1}$ pick up in strength.
This suggests that, to some extent, the high dividend--prices were able to forecast the long bull market that started afterward.
We see this increase in predictive content reflected in the test statistic in the top panel of Figure~\ref{fig:CPT QR}, which exceeds the 5\%-critical value during early 2019.
There are also strong indications for another break in the forecasting ability of dividend--prices in the earlier part of our sample.
Specifically, dividend--prices seem to have done well in predicting the recovery in stock prices after the oil shock of 1973--74.
Summing up the evidence, we find the predictive content of dividend--prices for the equity premium median to have changed over time.

The evidence for instabilities in mean prediction models of stock returns is overwhelming \citep{PT06,FST23}.
Overall, our results suggest that quantile prediction models of the equity premium are likewise liable to change over time. 
However, our evidence can be seen as more robust than that for mean predictions of the equity premium, as our test is valid both for stationary as well as near-stationary predictors.

\singlespacing

\bibliographystyle{APA}
\bibliography{thebib}

\end{document}